%% file: may2021.tex
\documentclass[11pt]{article}
\usepackage[T1]{fontenc}  
\usepackage[ttscale=.85]{libertine}

\usepackage{amsmath,amsfonts,amsthm,amssymb,color}
\usepackage{fancybox}
\usepackage{framed}
\usepackage{algpseudocode}
\usepackage{comment}
\usepackage[margin=1in]{geometry}
\usepackage{mathrsfs}
\usepackage{hyperref}
\usepackage{cleveref}
\usepackage{subcaption}

\usepackage{enumitem}
\usepackage[ruled,vlined]{algorithm2e}
\usepackage{color,soul}
\usepackage{multicol}

\usepackage{footnote}
\makesavenoteenv{algorithm}
\LinesNumbered
\DontPrintSemicolon
\SetKwFor{RepTimes}{repeat}{times}{end}
\SetKwFor{Loop}{loop}{}{end}
\SetKwInput{Input}{Input}
\SetKwInput{Output}{Output}

\newtheorem{theorem}{Theorem}[section]
\newtheorem{corollary}[theorem]{Corollary}
\newtheorem{lemma}[theorem]{Lemma}

\newtheorem{claim}[theorem]{Claim}
\newtheorem{fact}[theorem]{Fact}

\theoremstyle{definition}
\newtheorem{definition}[theorem]{Definition}
\newtheorem{remark}[theorem]{Remark}

\newenvironment{fminipage}%
  {\begin{Sbox}\begin{minipage}}%
  {\end{minipage}\end{Sbox}\fbox{\TheSbox}}

\def\defeq{\stackrel{\mathrm{def}}{=}}

\def\abs#1{\left|#1  \right|}

\newcommand{\wrap}[1]{\left(#1\right)}

\newcommand{\set}[1]{\left\{#1\right\}}

\newcommand{\cut}{\setminus}
\renewcommand{\subset}{\subseteq}

\newcommand{\cc}{\mathrm{CC}}

\renewcommand{\bar}[1]{\overline{#1}}

\newcommand{\calP}{\mathcal{P}}
\newcommand{\calQ}{\mathcal{Q}}
\newcommand{\calR}{\mathcal{R}}
\def\simple{\mathsf{Simple}}
\def\sparsifier{\mathsf{Sparsifier}}
\def\path{\mathsf{ConnectingPath}}
\def\superedge{\mathsf{SuperEdge}}
\def\contract{\mathsf{Contract}}

\def\insertt{\mathsf{insert}}

\def\delete{\mathsf{delete}}
\def\query{\mathsf{query}}

\def\datastructure{\mathfrak{DS}}
\def\oneleveldatastructure{\mathfrak{ODS}}
\def\multileveldatastructure{\mathfrak{MDS}}

\def\pruning{\textsc{Pruning}}
\def\expanderdecomposition{\textsc{Expander-Decomposition}}
\def\decsingleexpander{\textsc{Decremental-Single-Expander}}
\def\oneshotupdate{\textsc{One-Level-Update}}

\def\atomiccutverification{\textsc{Atomic-Cut-Verification}}
\def\enumeratecuts{\textsc{Enumerate-Simple-Cuts}}
\def\enumeratenonsimplecuts{\textsc{Enumerate-Cuts}}
\def\bipartitionsystem{\textsc{Bipartition-System}}
\def\enumerateparallelcuts{\textsc{Elimination}}

\def\typeonerepairset{\textsc{Type-One-Repair-Set}}
\def\typetworepairset{\textsc{Type-Two-Repair-Set}}
\def\typethreerepairset{\textsc{Type-Three-Repair-Set}}
\def\sequivalent{\textsc{Equivalent}}
\def\sequivalentsimplecut{\textsc{EquivalentSimpleCut}}

\def\repairsetalgo{\textsc{Repair-Set}}
\def\updatepartition{\textsc{Update-Partition}}

\def\vol{\mathbf{vol}}
\def\factor{2^{O(\log^{1/3} n \log^{2/3} \log n)}}
\def\truefactor{2^{\delta \log^{1/3} n \log^{2/3} \log n}}

\def\parametertimes{\xi}
\def\parametertimesub{\zeta}
\def\parameterlength{w}

\def\polylog{\operatorname{polylog}}
\def\poly{\operatorname{poly}}

\def\updateseq{\mathtt{UpdateSeq}}
\def\newupdateseq{\mathtt{NewUpdateSeq}}
\def\seqone{\mathsf{Seq}}
\def\endpoints{\mathsf{End}}

\newcommand{\IA}{\mathrm{IA}}

\newcommand{\eqdef}{\stackrel{\rm def}{=}}

\usepackage{tikzit}
\input{simple.tikzstyles}

\setlist[enumerate,1]{label={(\arabic*)}}

\title{Fully Dynamic $s$-$t$ Edge Connectivity in Subpolynomial Time}

\author{Wenyu Jin\\University of Illinois at Chicago\\{\tt wjin9@uic.edu} \and Xiaorui Sun\\University of Illinois at Chicago\\{\tt xiaorui@uic.edu}}

\date{} 
\begin{document}

\begin{titlepage}
\clearpage\maketitle
 \thispagestyle{empty}

\begin{abstract}
We present a deterministic fully dynamic algorithm to answer $c$-edge connectivity queries on pairs of vertices in $n^{o(1)}$ worst case update and query time for any positive integer $c = (\log n)^{o(1)}$ for a graph with $n$ vertices. Previously, only polylogarithmic and $O(\sqrt{n})$ worst case update time fully dynamic algorithms were known for answering $1$, $2$ and $3$-edge connectivity queries respectively [Henzinger and King 1995, Frederikson 1997, Galil and Italiano 1991].

Our result extends the  $c$-edge connectivity vertex sparsifier [Chalermsook et al. 2021] to a  multi-level sparsification framework. As our main technical contribution, we present a novel update algorithm for the multi-level $c$-edge connectivity vertex sparsifier with subpolynomial update time.

\end{abstract}

\newpage
\pagestyle{empty}
\tableofcontents
\end{titlepage}

\newpage

\section{Introduction}\label{sec:intro}
In the study of dynamic graphs, a 
fully dynamic algorithm for some property $\mathscr{P}$ supports the following operations on a given graph $G$:
\begin{enumerate}\itemsep -3pt
\item $\mathtt{preprocess}(G)$: initialize the algorithm with input graph $G$
\item $\insertt(u, v)$: insert edge $(u, v)$ into the graph
\item $\delete(u,v )$: delete edge $(u, v)$ from the graph
\item $\query(\mathscr{P})$: answer if property $\mathscr{P}$ holds for the graph 
\end{enumerate}
The goal is to process graph update operations (edge insertions and deletions) and query operations as efficiently as possible.
Update and query time can be  categorized into two types:
worst case, i.e.,  the upper bound on the running time of any  update or query operation, 
and amortized, i.e., the running time amortized over a sequence of operations.

In this paper, we study the fully dynamic $s$-$t$ edge connectivity problem.
Here, 
the graph receives edge insertion/deletion updates, as well as queries each containing two vertices. 
The goal is to determine if the two queried vertices are $c$-edge connected, i.e., the two queried vertices cannot be disconnected by removing fewer than $c$ edges. 
In the literature such as~\cite{galil1991fully, galil1991fully_3, eppstein1997sparsification, frederickson1997ambivalent, henzinger1997fully}, the dynamic $s$-$t$ edge connectivity problem was also referred to as the dynamic $c$-edge connectivity problem. In order to distinguish this problem from the dynamic problem of determining if every two vertices are $c$-edge connected, we refer to the problem studied in this paper the dynamic $s$-$t$ edge connectivity problem (see Section~\ref{sec:related_work}).


Fully dynamic $s$-$t$ edge connectivity has been studied for more than three decades~\cite{frederickson1985data, galil1991fully, galil1991fully_3, westbrook1992maintaining,
henzinger1995randomized,
eppstein1997sparsification, frederickson1997ambivalent, henzinger1997fully,  henzinger1997sampling, henzinger1999randomized, thorup2000near, holm2001poly, kapron2013dynamic, wulff2013faster, kejlberg2016faster, nanongkai2017dynamica, nanongkai2017dynamic,  wulff2017fully, holm2018dynamic, chuzhoy2019deterministic}. 
For $c = 1$, it is  the classic fully dynamic  connectivity problem. For a  graph with $n$ vertices and $m$ edges,
the best known algorithms have Monte Carlo randomized polylogarithmic worst case update and query time by Kapron et al.~\cite{kapron2013dynamic},
and $n^{o(1)}$ deterministic worst case update and query time by Chuzhoy et al.~\cite{chuzhoy2019deterministic}.

The study of fully dynamic $s$-$t$ 2-edge connectivity dates back to the work by Westbrook and Tarjan~\cite{westbrook1992maintaining} in a context of maintaining 2-edge connected components. 
Galil and Italiano~\cite{galil1991fully} obtained 
an  algorithm
with  $O(m^{2/3})$ update time and $O(1)$ query time. 
The update time  was improved to  $ O(\sqrt{m})$ by 
Frederickson~\cite{frederickson1997ambivalent}, and  $O(\sqrt{n})$ by Eppstein et al.~\cite{eppstein1997sparsification}. All these running times are worst case. 
To date, the best known worst case update time is still $O(\sqrt{n})$.
For the amortized case, 
polylogarithmic amortized update time algorithm were proposed in~\cite{henzinger1997fully,thorup2000near,holm2001poly,  holm2018dynamic}. 

For fully dynamic $s$-$t$ 3-edge connectivity, the best known result is an $O(n^{2/3})$ worst case update time and $O(\log n)$ query time algorithm by Galil and Italiano \cite{galil1991fully_3}, combining the sparsification technique by Eppstein et al.~\cite{eppstein1997sparsification}.

To our knowledge, 
for fully dynamic $s$-$t$ $c$-edge connectivity with  $c> 3$,
no algorithm with $o(n)$ update and query time was known, even for the amortized case.
Dinitz,  Westbrook, and Vainshtein~\cite{dinitz1994connectivity, vainshtein1995locally, dinitz1998maintaining} studied the incremental case (only edge insertions are allowed).
Molina and Sandlund~\cite{molina18} gave an offline fully dynamic algorithm for $c = 4, 5$  with $O(\sqrt{n})$ time per query.
Recently, Chalermsook et al.~\cite{chalermsook2020vertex} presented an offline fully dynamic algorithm for $s$-$t$ $c$-edge connectivity with $c^{O(c)}$ time per query.



We obtain a deterministic fully dynamic $s$-$t$ $c$-edge connectivity algorithm with subpolynomial worst case update and query time for any $c = (\log n)^{o(1)}$.

\begin{theorem}\label{thm:2-edge-connectivity}
There is a deterministic fully dynamic $s$-$t$ $c$-edge connectivity algorithm on a graph of $n$ vertices and $m$ edges
with $m^{1+o(1)}$ preprocessing time
and 
 $n^{o(1)}$ worst case update and query time for any positive integer $c = (\log n)^{o(1)}$.
\end{theorem}

Instead of maintaining data structures that preserve $c$-edge connectivity between all vertices of the graph as previous works did on fully dynamic $s$-$t$ $c$-edge connectivity for $c = 2$ and $3$~\cite{galil1991fully, eppstein1997sparsification, henzinger1997fully, thorup2000near,holm2001poly,  holm2018dynamic, galil1991fully_3}, we carefully select a special subset of vertices, called terminals, and only maintain a $c$-edge connectivity vertex sparsifier on these terminals, i.e. a smaller graph that maintains $c$-edge connectivity between any two terminals.
In order to answer a $c$-edge connectivity query, 
we update the data structure by adding the queried vertices to the terminal set.

Our result builds upon $c$-edge connectivity vertex sparsifier recently proposed by Chalermsook et al. in~\cite{chalermsook2020vertex}. 
As our main technical contribution, we give a fully dynamic algorithm for the $c$-edge connectivity vertex sparsifier with subpolynomial update and query time. 
The backbone of our result is an efficient algorithm 
to update a set of edges  that capture all the minimum cuts of size at most $c$ between terminals (called cut containment set) in a decremental update setting. 
The fully dynamic algorithm is obtained by combining expander decomposition and pruning~\cite{chuzhoy2019deterministic, saranurak2019expander} to reduce general graph edge insertions and deletions to the decremental updates of cut containment set, and 
a multi-level sparsification framework. 

\vspace{-.2cm}
\subsection{Other related work}\label{sec:related_work}

\noindent \textbf{Dynamic $s$-$t$ $c$-edge connectivity v.s. dynamic $c$-edge connectivity\ \ }
A closely related problem is the fully dynamic $c$-edge connectivity problem, that is, after each update of the graph, dynamic algorithm determines if every two vertices in the graph are $c$-edge connected.
Fully dynamic $c$-edge connectivity for constant $c$ has been studied in~\cite{eppstein1997sparsification, frederickson1985data, frederickson1997ambivalent, galil1991fully_3}.  
In~\cite{10.1007/s00493-007-0045-2}, Thorup proposed the best known fully dynamic $c$-edge connectivity  algorithm for polylogarithmic $c$ with $O\wrap{\sqrt n}$ update time. 
Although these two problems are closely related, to our knowledge, the idea and algorithm presented in~\cite{10.1007/s00493-007-0045-2} cannot be used to answer  dynamic $s$-$t$ $c$-edge connectivity queries. 

\vspace{.15cm} 
\noindent \textbf{Vertex sparsifier\ \ } Sparsifiers of graphs that preserve certain properties have been extensively studied~\cite{aho1972transitive, althofer1993sparse, spielman2011spectral, andoni2019solving, bernstein2020fully, goranci2020improved}. 
Vertex sparsifiers have been applied to dynamic graph algorithms in~\cite{ goranci2017power, goranci2018dynamic, durfee2019fully, chen2020fast}. 
The thesis by Goranci~\cite{goranci2019} gives a 
comprehensive exposition of the connections between vertex sparsifiers and dynamic
graph data structures.
Recently, $c$-edge connectivity vertex sparsifiers have been proposed in~\cite{chalermsook2020vertex}, but prior to our work,  it was  unclear how to  use such $c$-edge connectivity sparsifiers to design fully dynamic $c$-edge connectivity algorithms. 

\vspace{.15cm} \noindent \textbf{Techniques used in this paper\ \ } This paper makes use of  developments on dynamic algorithms, including fully dynamic spanning tree~\cite{frederickson1985data, eppstein1997sparsification, henzinger1997fully, henzinger1997sampling, henzinger1999randomized, thorup2000near, holm2001poly, wulff2013faster,  holm2015faster,huang2017fully, nanongkai2017dynamica, nanongkai2017dynamic, wulff2017fully, chuzhoy2019deterministic}, tree contraction~\cite{henzinger1997fully,holm2001poly,nanongkai2017dynamic}, expander decomposition and pruning~\cite{kannan2004clusterings,nanongkai2017dynamica, nanongkai2017dynamic, wulff2017fully,goranci2020expander,chuzhoy2019deterministic, saranurak2019expander}, and techniques to turn amortized dynamic algorithms into worst case dynamic algorithms
\cite{henzinger2001maintaining,thorup2005worst, patrascu2007planning,demetrescu2008oracles, bernstein2009nearly,  duan2009dual, chechik2010f, khanna2010approximate, acar2011parallelism, duan2016faster, henzinger2016incremental, duan2017connectivity, henzinger2017conditional,
nanongkai2017dynamic,brand2019sensitive, simsiri2016work, acar2017brief, acar2019parallel, tseng2019batch}.

\vspace{.15cm} 
\noindent \textbf{Organization\ \ }
Section~\ref{sec:preliminary} gives notations used in the paper. 
Section~\ref{sec:review_cdl} reviews the $c$-edge connectivity sparsifier defined in~\cite{chalermsook2020vertex}.
Section~\ref{sec:overview} gives an overview of the main results of this paper.
Section~\ref{sec:update_ia_set} presents our main technical contribution, an update algorithm for the cut containment  sets in the decremental update setting.
Section~\ref{sec:sparsifier} defines $c$-edge connectivity sparsifiers and gives some useful properties. 
Section~\ref{sec:one_level_update} and Section~\ref{sec:multilevel_update} present our one-level and multi-level sparsifier update algorithms respectively. 
Section~\ref{sec:fully} gives the fully dynamic algorithm and proves Theorem~\ref{thm:2-edge-connectivity}.

\vspace{.15cm} 
\noindent \textbf{Acknowledgement\ \ }
We thank Richard Peng for helpful discussions.
We thank Monika Henzinger and  Mikkel Thorup for insightful comments for an early version of this paper. 
We also thank Thatchaphol Saranurak for clarifying several lemmas about tree contraction technique.

\section{Notations}\label{sec:preliminary}

Throughout this paper, unless specified, we assume the graphs are multigraphs. 
A multigraph $G = (V, E)$ is defined by a vertex set $V$ and an edge multiset $E$. 
The multiplicity of an edge is the number of appearances of the edge in $E$. 
$\simple(G)$ denotes the  simple graph derived from $G$, i.e., graph on the same set of vertices and edges with all edges having multiplicity 1. 
For a subset $S \subseteq V$, $G[S]$ denotes the induced sub-multigraph of $G$ on $S$ and $\partial_G(S)$ denotes the multiset of edges with one endpoint in $S$ and another endpoint in $V\setminus S$. 
For a multiset of edges $E' \subset E$,
$\endpoints(E')$ denotes the set of all endpoints of edges in $E'$ and $G\setminus E'$ denotes the graph $(V, E\cut E')$.
We also use $E' |_{G[S]}$ to denote the multiset of edges in $E'$ that are in $G[S]$, i.e., $E' |_{G[S]} = \{(x, y) \in E' : x, y \in S\}$.

A cut is a bipartition of graph vertices. For a cut $C$, the cut-set of $C$ is the multiset of
edges that have one endpoint on each side of the cut, and the size of $C$ is the cardinality of the cut-set.
For a set of vertices $T$ and a bipartition $(T', T\setminus T')$ of $T$ for some $T'\subset T$, a cut $C$ is a $(T', T\setminus T')$-cut if $C$ partitions $T$ into $T'$ and $T\setminus T'$. 
We use $mincut_G(T', T\setminus T')$ to denote the size of minimum $(T', T\setminus T')$-cut for graph $G$. If the size of the minimum cut is upper bounded by $c$, we call the bipartition $(T', T\cut T')$ of $T$  a $(T, c)$-bipartition.
A few useful properties of cuts are presented in Appendix~\ref{sec:useful_lemmas_cuts}.

For a vertex partition $\calP$ of graph $G$, 
a cluster is a vertex set in $\calP$. 
An edge of $G$ is an intercluster edge with respect to $\calP$ if the endpoints of the edge are in different clusters, otherwise, the edge is an inner edge.
We use $G[\calP]$ to denote the union of induced subgraph on each cluster, and $\partial_G(\calP)$ to denote the set of intercluster edges with respect to $\calP$. 
Boundary vertices of $G$ with respect to $\calP$ are the endpoints of intercluster edges, i.e., $\endpoints(\partial_G(\calP))$.

For a graph update sequence $\updateseq$ (consisting of vertex/edge insertions/deletions) on a dynamic graph,  $|\updateseq|$  denotes the number of updates in the update sequence.

\vspace{.2cm} \noindent \textbf{Expander and expander decomposition\ \ } For a graph $G = (V, E)$, the \emph{conductance} of $G$ 
is defined as 
\[
\min_{
\emptyset \subsetneq S \subsetneq V} \frac{|\partial_G(S)|}{\min\{\vol_G(S), \vol_G(V\cut  S)\}},\]
where 
$\vol_G(S)$ is the volume of $S$ in $G$, i.e.,  the sum of  degrees of vertices in $S$.
A graph $G$ is a $\phi$-expander if its conductance is at least $\phi$. 
A $\phi$\emph{-expander decomposition} of a graph $G = (V, E)$ is a vertex partition $\calP$ of $V$ such that for each $P \in \calP$,
the induced  subgraph $G[P]$ is a $\phi$-expander.
A $(\phi, \epsilon)$\emph{-expander decomposition} of  $G$ is a $\phi$-expander decomposition of $G$, 
and 
the number of intercluster edges 
is at most an $\epsilon$ factor of the total number of edges in $G$. 

\section{Review of \texorpdfstring{$c$}{c}-Edge Connectivity Sparsifier}\label{sec:review_cdl}
In \cite{chalermsook2020vertex}, 
Chalermsook et al. gave a construction of 
\textbf{$c$-edge connectivity vertex sparsifier}, which is a $c$-edge connectivity equivalent graph with respect to a terminal set $T$ of $|T|O(c)^{O(c)}$ vertices and edges.

\begin{definition}\label{def:c-edge-equivalent-graph-intro}
A multigraph $H = (V_H, E_H)$ is a \emph{$c$-edge connectivity equivalent graph} of multigraph $G = (V, E)$ with respect to a set of terminal vertices $T$
if $T$ is a subset of both $V$ and $V_H$ and 
the $c$-edge connectivity between any pair of vertices in $T$ are the same in $G$ and in $H$. 
\end{definition}

The construction of vertex sparsifier in \cite{chalermsook2020vertex} is based on finding a  \textbf{cut containment set} 
with respect to the given terminal set.

\begin{definition}[Definition 5.1 of \cite{chalermsook2020vertex}]\label{def:cut-containment-set}
    Given a graph $G=(V, E)$, a terminal set $T \subset V$ and a positive integer $c$,
    a set of edges $\cc \subset E$ is a $(T, c)$-\textit{cut containment set} if for every bipartition $(T', T \setminus T')$ of $T$ such that $mincut_G(T', T\setminus T') \leq c$, there is a minimum $(T', T\setminus T')$-cut whose cut-set is a subset of $\cc$. 
\end{definition}
Given a $(T, c)$-cut containment set $\cc$, 
the $c$-edge connectivity sparsifier with respect to $T$ was obtained by shrinking each connected component of $G \setminus \cc$ into a single vertex.
Hence, to construct a $c$-edge connectivity vertex sparsifier with respect to $T$ of $|T|O(c)^{O(c)}$ vertices and edges, it is sufficient to give a $(T, c)$-cut containment set of size $|T|O(c)^{O(c)}$. 
The cut-containment set is constructed 
from intermediate ``intersecting all'' sets. We denote these ``intersecting all'' sets by $\IA$.


\begin{definition}\label{def:intersect-all}\footnote{Definition~\ref{def:intersect-all} is generalized from Definition 5.2 of \cite{chalermsook2020vertex} with slight modification for convenience.}
Given a connected graph $G$ with terminals $T$ and integers $d \geq c > 0$,
a set of edges $R$ is an $\IA_G(T, d, c)$ set if
for every bipartition $(T', T\setminus T')$ satisfying $mincut_G(T', T\setminus T') \leq d$, there is a minimum $(T', T\setminus T')$-cut $C$ such that
every connected component of $G\setminus R$ contains at most $\max\{mincut_G(T', T\setminus T')-c, 0\}$ edges of the cut-set of $C$. 


For a graph $G$ with terminals $T$ such that $G$ is not necessarily connected, and two positive integer $d \geq c > 0$, 
a set of edges $R$ is an $\IA_G(T, d, c)$ set if 
$R$ is a union of $\IA_{G[Q]}(T\cap Q, d, c)$ sets for each $Q$ that corresponds to a connected component of $G$.

\end{definition}

Chalermsook et al. showed that 
any $\IA_{G}(T, d, c)$ set is a $(T, c)$-cut containment set for any $d \geq c$.
In addition, an $\IA_{G}(T, d, c)$ can be obtained by  
\begin{equation}\label{equ:recursive_IA}
\IA_G(T, d, i) \cup \IA_{G\setminus \IA_G(T,d, i)}(T\cup \endpoints(\IA(T, d, i)), d-i, c-i)
\end{equation}
for any $1 \leq i < c$. 
Hence, to construct an $\IA_G(T, d, c)$ set, 
one can first construct an $\IA_G(T, d, 1)$ set, and then construct an $\IA_{G\setminus \IA_G(T, d, 1)}(T \cup \endpoints(\IA_G(T, d, 1)), d-1, c-1)$ set recursively on each connected component of $G\setminus \IA_G(T, d, 1)$. 
Another way to view the edge set obtained is that
all edges taken in the recursion of depth at most $i$ form an $\IA_G(T, d, i)$ set for any $1 \leq i \leq c$. Together with the observation that there is always an $\IA_G(T, d, 1)$ set of size $O(|T|d)$, the $(T, c)$-cut containment set of size at most $|T|O(c)^{O(c)}$ is obtained by an $\IA_G(T, c, c)$ set.



Throughout this paper, we assume that cut containment sets are recursively constructed from $\IA$ sets. 
We also ensure that $\IA$ sets are intercluster edges of some vertex partition of the graph so that each connected component of the graph after removing a cut containment set or an $\IA$ set
is the induced subgraph on that vertex set.

\section{Overview of Our Technique}\label{sec:overview}
The main obstacle of using  $c$-edge connectivity sparsifier defined in \cite{chalermsook2020vertex} for dynamic algorithms
is that \cite{chalermsook2020vertex} did not give an update algorithm for the cut containment set with respect to the updates of the input graph.
As our main technical contribution, we give an efficient algorithm to update a cut containment set in a decremental update setting. 
The \textbf{decremental cut containment set update problem} in this paper is defined as follows: 
Let $G_0 = (V_0, E_0)$ be a graph,  $T_0 \subset V_0$ be a subset of vertices.
Suppose some vertices of $G_0$ are removed, and only the induced subgraph on $V \subset V_0$, denoted by $G = (V, E)$, is left. 
Let $S = \endpoints(\partial_{G_0}(V)) \cap V$ and $T = (T_0 \cap V) \setminus S$.
%
Given a $(T_0, c_0)$-cut containment set of $G_0$, 
the goal is to
obtain a $(S\cup T, c)$-cut containment set of $G$ by adding a set of additional edges to $\cc_0 \cap E$ for some $c \leq c_0$.
We call the newly added edge set a \textbf{cut containment repair set}. 

\begin{lemma}
\label{lem:representation_update}
If $c_0 \geq c^2 + 2c$, 
there is an algorithm to compute a cut containment repair set of size 
$(10c)^{O(c)}\cdot |S|$
in time $(1 / \phi)^{O(c^2)} \cdot |S| \cdot  n^{o(1)}$ time if both $G_0$ and $G$ are $\phi$-expanders.


\end{lemma}


Lemma~\ref{lem:representation_update} suggests that the cut containment set can be decrementally updated such that 
the running time is irrelevant to the number of terminals in the original cut containment set (i.e., $T_0$),
and is related to the number of new terminals (i.e., $S$).
This is crucial to our algorithm, because it relates the algorithm running time to the number of updates on the graph, which is required by worst case fully dynamic algorithm.

In the rest of this section, we assume Lemma~\ref{lem:representation_update} and use it to construct our fully dynamic $s$-$t$ $c$-edge connectivity algorithm. 
The existance of small cut containment repair set is given in Section~\ref{sec:overview:cut_containment}, and the proof of Lemma~\ref{lem:representation_update} is given in Section~\ref{sec:update_ia_set}. This ennables the possibliity of efficient update \\


Throughout the paper, we use the standard degree reduction technique~\cite{harary6graph} to turn the dynamic input graph of $n$ vertices and $m$ edges with arbitrary maximum degree
into a multigraph of $O(n+m)$ vertices and distinct edges~\footnote{Two edges are distinct if they connect different pairs of vertices.} such that in the new graph, every vertex has at most a constant number of different neighbors.
In addition, each update or query to the input graph can be transformed into a constant number of updates or queries to the multigraph (see Section~\ref{sec:sub_fully_dynamic}). 
Thus, we work on the multigraph for the updates and queries with the following update operations: 


\begin{itemize}\itemsep -3pt
\item $\insertt(u, v, \alpha)$: insert edge $(u, v)$ with edge multiplicity $\alpha$ into the graph
\item $\delete(u, v)$: delete edge $(u, v)$ from the graph (no matter what the edge multiplicity is)
\item $\insertt(v)$: insert a new vertex $v$  to the graph
\item $\delete(v)$: delete isolated vertex $v$ from the graph
\end{itemize}

As an intermediate step to the fully dynamic algorithm, 
we first give update algorithms for $c$-edge connectivity sparsifier in the online-batch update setting: The updates are partitioned into a few batches, and each batch may contain $O(m)$ updates. 
The batches come one-by-one, and 
the algorithm needs to maintain the data structure accordingly after receiving each batch (without seeing future update batches) as efficiently as possible. Ideally, the update time for each batch is proportional to the batch size.


In this subsection, we present an algorithm for $c$-edge connectivity vertex sparsifier 
in the online-batch setting for $O(\log \log \log n)$ update batches,
each with update time $n^{o(1)}$ times the batch size.
The fully dynamic algorithm is obtained by applying a general framework to turn an online-batch algorithm to a fully dynamic algorithm developed by Nanongkai et al.~\cite{nanongkai2017dynamic} in a blackbox way.

\paragraph{One-level $c$-edge connectivity sparsifier} 
Instead of defining a $c$-edge connectivity sparsifier on selected terminals directly, we define our sparsifier with respect to an expander decomposition of the input graph.
Formally, 
let $G = (V, E)$ be a multigraph,
$\mathcal{P}$ be an expander decomposition of $G$,
and $\cc$ be a union of $c$-cut containment sets for induced subgraphs $G[P]$ with respect to boundary vertices ($\endpoints(\partial_G(\mathcal P)) \cap P$) for all $P \in \mathcal{P}$. 
The sparsifier of $G, \calP$, and $\cc$, denoted by $\sparsifier(G, \calP, \cc)$, 
is obtained as follows: for each $P \in\calP$, 
replace the induced subgraph on $G[P]$ by 
a $c$-edge connectivity vertex sparsifier of $G[P]$ with boundary vertices as terminals\footnote{The $c$-edge connectivity sparsifier for a graph with respect to given terminals used in this paper is slightly different to  the sparsifier defined in Section~\ref{sec:review_cdl}, and is 
formally defined in Section~\ref{sec:sparsifier}. However, for simplicity, one can think of such a sparsifier as a sparsifier defined in Section~\ref{sec:review_cdl}}. All intercluster edges of $G$ with respect to $\calP$ are kept in $\sparsifier(G, \calP, \cc)$.

There are two reasons of using expander decomposition. 
First, the conductance of each cluster in the expander decomposition is good and thus the cut containment set for each cluster can be updated efficiently. 
Second, $\sparsifier(G, \calP, \cc)$ is significantly smaller than the original graph. 

Since the sparsifier for each cluster is obtained with respect to the boundary vertices of the cluster, and the intercluster edges are kept in $\sparsifier(G, \calP, \cc)$, $\sparsifier(G, \calP, \cc)$ is a $c$-edge connectivity equivalent graph for all the boundary vertices $\endpoints(\partial_G(\calP))$. 
Therefore, a $c$-edge connectivity query on any two boundary vertices of $G$ can be answered by $\sparsifier(G, \calP, \cc)$.

\paragraph{One-level sparsifier update algorithm}
We start with an online-batch update algorithm for one-level sparsifiers:  Given a multigraph $G$, a $\phi$-expander decomposition $\calP$ of $G$, and an edge set $\cc$ which is a union of $(c^2+2c)$-cut containment sets for each cluster with boundary vertices as terminals, and an update sequence $\updateseq$ of $G$, we want to update $G, \calP$ and  $\cc$ according to the update sequence such that 
updating 
$\sparsifier(G, \calP, \cc)$ is almost as efficient as updating $G$. 

We give an algorithm with running time $O(|\updateseq|n^{o(1)} / \phi^3 )$
to update $G, \calP$, and $ \cc$.
The algorithm also outputs an update sequence  $\newupdateseq$ of length  $|\updateseq|(10c)^{O(c)}$ such that  $\sparsifier(G, \calP, \cc)$ is updated according to the changes made to $G, \calP$, and $\cc$.
In addition, our algorithm has a nice property that every vertex involved in the update sequence becomes a boundary vertex in the updated expander decomposition. 


We give the high-level idea of our algorithm. 
We first update $\calP$ to a refined partition
$\calP'$  of $G$. (Every cluster of $\calP'$ is a subset of a cluster in $\calP$.)
We collect all edges incident to the vertices involved in $\updateseq$, and run the expander pruning algorithm~\cite{saranurak2019expander} with these edges. On the pruned vertices, we further run expander decomposition~\cite{chuzhoy2019deterministic} to make sure the number of distinct new intercluster edges is linear  in $\abs{\updateseq}$ (since each vertex has a constant number of neighbors). 
The updated expander decomposition, denoted by $\calP'$ satisfies the following properties:
(1) Any vertex involved in the multigraph update sequence $\updateseq$ is in a singleton (a cluster with one vertex) of $\calP'$, if the vertex is already in $G$ before the update.
(2) $\calP'$ is a $\phi / \factor$-expander decomposition of $G$.~\footnote{The $\factor$ factor is obtained by applying deterministic expander decomposition from  \cite{chuzhoy2019deterministic} so that the number of new intercluster edges is bounded.}
(3) The number of distinct new intercluster edges is $O(|\updateseq|)$.


With the updated partition $\calP'$, we further update $\cc$ to $\cc'$, which is a union of 
$c$-cut containment sets of $G[P]$ for all $P \in \calP'$ by applying the cut containment set update algorithm on every cluster of $\calP'$ that is updated in the first step, making use of the condition that $\calP'$ is a refinement of $\calP$. 
Since $O(|\updateseq|)$ distinct inner edges of $\calP$ become intercluster edges of $\calP'$, 
by Lemma~\ref{lem:representation_update}, $\cc'$ is obtained by adding at most $O(|\updateseq|(10c)^{O(c)})$ edges. 
Consequently, the total number of vertex and edge updates that transform
 $\sparsifier(G, \calP,\cc)$ into $\sparsifier(G, \calP',\cc')$ is at most $O(|\updateseq|(10c)^{O(c)})$.  

In the end, we apply the update sequence $\updateseq$ to $G$ to obtain the updated graph $G'$. 
Since all vertices  involved in the update sequence $\updateseq$ are in singletons of $\calP'$,
applying $\updateseq$ to $\sparsifier(G, \calP', \cc')$ results in $\sparsifier(G', \calP', \cc')$.
 The length of overall updating sequence for the sparsifier is $|\updateseq| (10c)^{O(c)}$.



\paragraph{Multi-level $c$-edge connectivity sparsifier} 
The multi-level sparsifier is motivated by the trade-off of conductance parameter used for expander decomposition. 
On one hand, we want the conductance of the $c$-edge connectivity sparsifier of $G$ to be at least $1/n^{o(1)}$ so that $c$-edge connectivity can be efficiently determined. 
On the other hand, we want 
the resulting sparsifier to have conductance at least $1/n^{o(1)}$ so that $c$-edge connectivity in the sparsifier can be determined efficiently. 
However, these two requirements can not be satisfied at the same time by a one-level sparsifier.
Hence, we make use of multi-level sparsification, where the graph gets sparsified gradually so that at the end, the sparsifier has a good conductance. 

Our multi-level sparsifier is a set of tuples $\{(G^{(i)}, \calP^{(i)}, \cc^{(i)})\}$.
$G^{(0)}$ is same as the input graph, and $G^{(i)}$
is the sparsifier of $G^{(i-1)}$ with respect to $\calP^{(i-1)}$ and $\cc^{(i-1)}$. 
$\calP^{(i)}$ is an expander decomposition of $G^{(i)}$ and $\cc^{(i)}$ is a  cut-containment set of $G^{(i)}[\mathcal P^{(i)}]$. 
The sparsification terminates when the graph becomes a good expander.


\paragraph{Multi-level sparsifier update algorithm}
Our online-batch update algorithm applies the one-level update algorithm iteratively.
Let $\updateseq^{(0)} = \updateseq$.
We  run one-level update algorithm on tuple $(G^{(i)}, \calP^{(i)}, \cc^{(i)})$ with update sequence $\updateseq^{(i)}$, 
and use the returned update sequence as $\updateseq^{(i + 1)}$ to update the one-level sparsifier at  level $i + 1$.
Note that if the length of $\updateseq$ is large, e.g., a polynomial of $m$, 
some one-level sparsifier might contain too many edges after the update. 
In this case, we directly reconstruct the expander decompositions, cut containment sets of each cluster, and  the one-level sparsifiers for all following levels iteratively. 
Based on the result of the one-level sparsifier update algorithm, 
the running time of our multi-level update algorithm is still $O(|\updateseq| n^{o(1)})$ when parameters are chosen appropriately. 



\paragraph{Fully dynamic algorithm for \texorpdfstring{$c$}{c}-edge connectivity} 
Now we are ready to discuss our fully dynamic algorithm for the $s$-$t$ $c$-edge connectivity problem.
The goal of our update algorithm is to provide  access to a multi-level $c$-edge connectivity sparsifier of the up-to-date graph 
after each update. 
Using a general framework to turn update algorithms for update sequences into fully dynamic algorithms, which was implicitly used  in~\cite{nanongkai2017dynamic} (see Section~\ref{sec:online_batch}), we  can achieve this goal by maintaining a set of  $n^{o(1)}$ multi-level sparsifiers 
with $m^{1+o(1)}$ preprocessing time 
and $m^{o(1)}$ update time.

For a $c$-edge connectivity query on two vertices, 
we first obtain access to the multi-level  sparsifier of the up-to-date input graph.
%
%
Then we make the two queried vertices  boundary vertices for each level so that the sparsifier on the last level can be used to answer the $c$-edge connectivity query.
To achieve this, 
we generate an update sequence of constant length for the input graph such that the two queried vertices are involved and after applying the update sequence to the up-to-date input graph, the $c$-edge connectivity of the queried vertices does not change.
Then 
we make use of 
the one-level sparsifier update algorithm to update the multi-level sparsifier
from $(G^{(0)}, \calP^{(0)}, \cc^{(0)})$ to 
$(G^{(\ell)}, \calP^{(\ell)}, \cc^{(\ell)})$ one-by-one. 
Since the one-level sparsifier algorithm has the property that every vertex involved in the update sequence is in a  singleton in the vertex partition, 
after the update, 
the two queried vertices are boundary vertices in each one-level sparsifier, and the $c$-edge connectivity between the two queried vertices in each one-level sparsifier is the same as it is in the input graph.
In the end, we use $G^{(\ell)}$ to answer the query, as $G^{(\ell)}$ is a good expander.



\section{Warm-up: Existence of Small Cut Containment Repair Set}\label{sec:overview:cut_containment}

In this section, we answer the question of why there is a small cut containment repair set if the cut containment set is constructed iteratively from the $\IA$ sets (as described in Section~\ref{sec:review_cdl}). In particular, we prove the following lemma in this section. 


\begin{lemma}\label{lem:warmup}
If $c_0 \geq c^2 + 2c$, 
there is a cut containment repair set of size 
$(10c)^{O(c)}\cdot |S|$ for arbitrary $G_0$ and $G$. 
\end{lemma}


Our goal is to show that 
given a $(T_0, c_0)$-cut containment set $\cc_0$ of $G_0$ (defined in Definition~\ref{def:cut-containment-set}),
one can obtain a $(S\cup T, c)$-cut containment set of $G$ by adding $|S|(10c)^{O(c)}$ edges to $\cc_0 \cap E$ for some $c_0 > c$. 
Recall that the $(T_0, c_0)$-cut containment set is actually an $\IA_{G_0}(T_0, c_0', c_0)$ set for some $c_0' \geq c_0$, which is defined as a set of edges such that for each bipartition $(T', T_0\setminus T')$ with a minimum cut of size at most $c_0'$, 
there is a minimum $(T', T_0\setminus T')$-cut $C$ such that
each connected component of $G_0\setminus \IA_{G_0}(T_0, c_0', c_0)$ contains at most $\max\{mincut_{G_0}(T_0, T)-c_0, 0\}$ edges of $C$'s cut-set. 
Hence, in the decremental update scenario, when updating a cut containment set, we are essentially updating an $\IA$ set. 
Furthermore, since the $\IA$ set is recursively constructed as defined in Section~\ref{sec:review_cdl}, 
it is sufficient for us to prove the following lemma. 

\begin{lemma}
\label{lem:update_intersecting_all-1}
    In the decremental update scenario, given an $\IA_{G_0}(T_0, d, 2c+1)$ set,
    there is a set $F$ of $O(|S|\cdot c^2)$ edges from $G = (V, E)$, called  repair set, such that $F \cup (\IA_{G_0}(T_0, d, 2c+1)\cap E)$ is an $\IA_{G}(T\cup  S, c, 1)$ set.
\end{lemma}

Lemma~\ref{lem:warmup} is then obtained 
by recursively applying Lemma~\ref{lem:update_intersecting_all-1} 
in a way 
similar to the construction of the cut containment set.

\paragraph{Necessity of compression}
Before presenting the proof of Lemma~\ref{lem:update_intersecting_all-1}, 
we answer the question of why it is necessary to ``compress'' a $(T_0, c_0)$-cut containment set 
 to a $(T, c)$-cut containment set for some $c < c_0$.
We give an example such that the number of edges to be added is unbounded with respect to $|S|$ if one wants to obtain an $(T, c)$-cut containment set of $G$ from a $(T_0, c)$-cut containment set of $G_0$.
Consider the following construction of simple graph $G_0$ in Figure~\ref{fig:intro}:
$V_0$ is partitioned into $U_0, U_1, \dots, U_{\alpha- 1}$ such that each $U_i$ is a set of $2c+1$ vertices that contains a vertex $u_i$ in $T_0$.
The induced subgraph of $G_0$ on each $U_i$ is a clique of size $2c+1$.
For every $0 \leq i \leq \alpha - 1$, there are $c$ edges connecting $U_i$ and $U_{(i + 1) \bmod \alpha}$. 
$G_0$ does not have a cut of size less than or equal to $c$.
Hence, a $(T_0, c)$-cut containment set for graph $G_0$ can be an empty set.

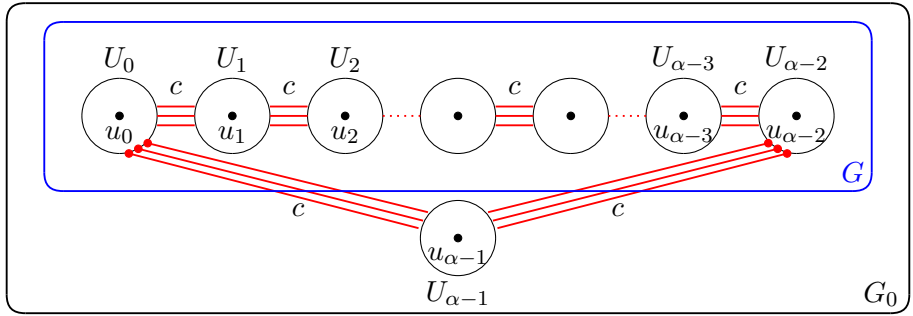
\begin{figure}[ht]
\begin{center}
      \input{fig1.tikz}
\end{center}
\caption{A construction of $G_0$ and $G$ such that every $(T, c)$-cut containment set of $G$ derived from a given $(T_0, c)$-cut containment set has  $\Omega(|T_0|c)$ new intercluster edges.}
\label{fig:intro}
\end{figure}

Assume $G = G_0[V_0\cut U_{\alpha - 1}]$.
Then $S$ contains $2c$ vertices, and $T = \{u_0, \dots, u_{\alpha - 2}\} \setminus S$. In graph $G$, for any $0 \leq i < \alpha - 2$,
the cut partitioning $\bigcup_{j = 0}^i U_j$ and $\bigcup_{j = i+1}^{\alpha - 2} U_j$ is of size $c$.
Thus, an arbitrary $(T, c)$-cut containment set of $G$ must contain 
the $c$ edges connecting $U_i$ and $U_{i+1}$ for any $0 \leq i < \alpha - 2$. 
Hence, any $(T, c)$-cut containment set of $G$ has at least $c(\alpha - 2) = \Omega(|T_0|c)$ edges, as opposed to linear in $|S|$ as needed.

\paragraph{Atomic cuts and algorithm framework.}
In the rest of this section, we  prove Lemma~\ref{lem:update_intersecting_all-1}.

We say a cut $C$ of graph $G$ is shattered by an edge set $R$ if no connected component of $G\setminus R$ contains all edges of $C$'s cut-set. We also say a terminal bipartition $(U, (S\cup T)\setminus U)$ is shattered by $R$ if there is a $(U, (S\cup T)\setminus U)$-minimum cut $C$ such that $R$ shatters $C$. 
Our high level intuition to prove Lemma~\ref{lem:update_intersecting_all-1} is to select a set of cuts such that the union of selected cuts' cut-sets shatters all $(S\cup T)$-bipartitions with minimum cut of size at most $c$ that are not shattered by $\IA_{G_0}(T_0, d, 2c+1)\cap E$. 
In order to do so, we first need to understand which bipartitions can be shattered by a cut's cut-set. 
The answer becomes clear when the cut is an atomic cut, i.e., the induced subgraphs on both sides of the cut are connected. We say two cuts $(V_1, V\cut V_1)$ and $(V_2, V\cut V_2)$ are parallel if $V_1$ or $V\cut V_1$ is a subset of either  $V_2$ or $V\cut V_2$. 
We observe that for two cuts $C_0$ and $C_1$ that are not parallel, if $C_0$ is atomic, then the cut-set of $C_0$ shatters $C_1$ (Figure~\ref{fig:shatter_example}). 
\begin{figure}[htb]
    \centering
    \includegraphics[width = .3\linewidth]{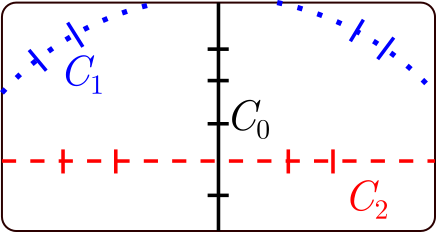}
    \caption{Cut $C_1$ (blue dotted line) and cut $C_2$ (red dashed line) are parallel. Therefore, the cut-set of neither cut shatters the other cut. Cut $C_1$ and cut $C_0$ (black line) are not parallel, since $C_0$ is atomic and $C_1$ is not, the cut-set of $C_0$ shatters $C_1$, but not vice versa.}
    \label{fig:shatter_example}
\end{figure}
Consequently, we can find a set of edges to shatter a given set of atomic cuts as follows: first find a \emph{maximal} set of pairwise parallel atomic cuts from the give atomic cuts, and then return the union of the cut-sets of selected cuts.


We generalize the procedure above to shattering terminal partitions: We say two terminal bipartitions $(S_1, S\cut S_1)$ and $(S_2, S\cut S_2)$ of $S$ are parallel if one of  $S_1$ or $S\cut S_1$ is a subset of either $S_2$ or $S\cut S_2$.
If two bipartitions of terminal set $(S_1, S\setminus S_1)$ and $(S_2, S\setminus S_2)$ are not parallel and there is an $(S_1, S\cut S_1)$-atomic cut,
the cut-set of an arbitrary atomic $(S_1, S\setminus S_1)$-cut shatters any $(S_2, S\setminus S_2)$-cut. 
Hence, we have 
the following algorithm to construct a repair set.

\def\bipartitionsys{\Gamma}
\def\shatterset{\mathsf{ToShatter}}

\begin{enumerate}\itemsep -3pt
    \item Find $\bipartitionsys$, a maximal set of pairwise parallel bipartitions of $S$ such that for each $(S', S\setminus S')\in \Gamma$, there is an atomic $(S', S\setminus S')$-cut. 
    \item Let $F_0$ be the union of cut-sets of one arbitrary atomic $(S', S\setminus S')$-cut for each $(S', S\setminus S')\in \bipartitionsys$.
    \item Find another set of edges $F_1$ that shatters all the bipartitions of $S\cup T$ that (a) induce bipartitions of $S$ that are in $\bipartitionsys$, and (b) are not shattered by $F_0 \cup (\IA_{G_0}(T_0, d, 2c+1)\cap E)$.
    \item Return $F_0\cup F_1$.
\end{enumerate}
$F_0 \cup F_1$ is a repair set, because $F_0$ shatters all $(S\cup T)$-bipartitions that induce bipartitions of $S$ that are not in $\bipartitionsys$ by the observation above;
$F_1$ shatters all $(S\cup T)$-bipartitions (that are not shattered by $\IA_{G_0}(T_0, d, 2c+1)\cap E$) that induce $S$-bipartitions in $\bipartitionsys$ by definition.

To this end, an obvious questions is how large $\bipartitionsys$ can be. We show that
the number of pairwise parallel bipartitions of $S$ is always upper bounded by $O(|S|)$. 
Consequently, $|F_0|$ is upper bounded by $O(|S|c)$.
\begin{lemma}\label{lem:linear_partition_overviewb3}
Let $\bipartitionsys$ be a maximal set of pairwise parallel bipartitions on $S$. 
$|\bipartitionsys| = O(|S|)$.
\end{lemma}
\begin{proof}
We show that $\Gamma \setminus \{(\emptyset, S)\}$ contains  at most $\abs S - 1$ bipartitions by induction. When $\abs S = 1$, $|\Gamma \setminus \{(\emptyset, S)\}| = 0$. 
Suppose the claim holds for $\abs S \leq k$. When $\abs S = k+1$, we pick one bipartition $(S_i, S\cut S_i)$ from $\Gamma \setminus \{(\emptyset, S)\}$.  For any $(S_j, S\cut S_j)\in \bipartitionsys \setminus \{(\emptyset, S), (S_i, S\cut S_i)\}$, by the definition of parallel bipartition, one of the following cases is true:
\begin{itemize}
    \item $\emptyset \subsetneq S_j\subsetneq S_i$. 
    \item $\emptyset \subsetneq S_j\subsetneq (S\cut S_i)$. 
    \item $\emptyset \subsetneq (S\cut S_j)\subsetneq (S\cut S_i)$. 
    \item  $\emptyset \subsetneq (S\cut S_j)\subsetneq S_i$. 
\end{itemize}
Let $\Gamma_1$ be the set of all the bipartitions in  $\bipartitionsys \setminus \{(\emptyset, S), (S_i, S\cut S_i)\}$ satisfying the first or the forth condition, and $\Gamma_2$ be the set of remaining bipartitions in  $\bipartitionsys \setminus \{(\emptyset, S), (S_i, S\cut S_i)\}$.
$\Gamma_1$ is a set of pairwise parallel bipartitions on $S_i$ that does not include $(\emptyset, S_i)$, and 
$\Gamma_2$ is a set of pairwise parallel bipartitions on $S\setminus S_i$ that does not include $(\emptyset, S\setminus S_i)$.
By induction, we have \[|\Gamma_1| \leq |S_i| - 1 \text{ and } |\Gamma_2| \leq |S\setminus S_i| - 1.\] Hence, we have 
\[|\Gamma \setminus \{\emptyset, S\}| = |\Gamma_1| + |\Gamma_2| + 1 \leq 
(\abs {S_i} -1) + (\abs {S\cut S_i} -1) + 2 = \abs{S} - 1.\]
Therefore, $|\Gamma| \leq |S|$. 
\end{proof}

In the rest of this section, we show that with an arbitrary $F_0$
and $\IA_{G_0}(T_0, d, 2c+1) \cap E$, there is always a $F_1$ of size $O(|S|c^2)$.

\paragraph{Characterizing bipartitions to be shattered}
Before we present our construction of $F_1$, we characterize the set of $(S\cup T)$-bipartitions that need to be shattered by the $S$-bipartitions they induce.

Our characterization  investigates the relation between cuts in $G$ and the $\IA_{G_0}(T_0, d, 2c)$ set such that the $\IA_{G_0}(T_0, d, 2c + 1)$ set is derived from the the $\IA_{G_0}(T_0, d, 2c)$ set by Equation~(\ref{equ:recursive_IA}), and thus also contains the $\IA_{G_0}(T_0, d, 2c)$ set.

Let $G'$ be $G \setminus (\IA_{G_0}(T_0, d, 2c) \cap E)$, and 
$\shatterset_{(S', S\setminus S')}$ be the set of $(S\cup T)$-bipartitions with minimum cut of  size at most $c$ that partition $S$ into $S'$ and $S\setminus S'$ satisfying the following conditions 
\begin{enumerate}\itemsep -3pt
\item All minimum cuts of bipartitions in $\shatterset_{(S', S\setminus S')}$ are atomic.
\item Every minimum cut of any bipartition in $\shatterset_{(S', S\setminus S')}$ has its cut-set contained  in one connected component of $G'$ that also contains at least one vertex of $S$. 
\end{enumerate}

We show that if a set of edges shatters $\shatterset_{(S', S\cut S')}$ for all $(S', S\setminus S')\in \bipartitionsys$, then 
this set of edges is a desirable $F_1$  (Lemma~\ref{overview_f0_f1_b4}). \\

To justify the first condition, we use the following claim to show that we can view the cut-set of a non-atomic minimum cut as a union of cut-sets of several atomic minimum cuts. 



\begin{claim}\label{claim:comprising_condition}
Let $(U, (S\cup T) \setminus U)$ be a bipartition of $S\cup T$ such that $\emptyset \subsetneq U \subsetneq (S\cup T)$. There is a minimum $(U, (S\cup T) \setminus U)$-cut $C$ satisfying the following conditions:
\begin{enumerate}
\item The cut-set of $C$ is a union of cut-sets of atomic cuts $C_1, \dots, C_t$ for some $t \geq 1$.
\item Let $(U_i, (S\cup T)\setminus U_i)$ be the bipartition induced by $C_i$. $U_i \neq \emptyset$ and $(S\cup T)\setminus U_i\neq \emptyset$. 
\item All the minimum $(U_i, (S\cup T)\setminus U_i)$-cuts are atomic cuts. 
\end{enumerate}
\end{claim}
\begin{proof}

Let $(V', V \setminus V')$ be a cut of $G$. Let $V^\dagger \subset V'$ be a vertex set that corresponds to a connected component of $G \setminus \partial_G(V')$, and 
$V^\diamond\subset (V\setminus V^\dagger)$ be a vertex set that corresponds to a connected component of $G \setminus \partial_G(V^\dagger)$. 
$(V^\diamond, V\setminus V^\diamond)$ is an atomic cut such that $\partial_G(V^\diamond)$ is a subset of $\partial_G(V')$. 
Hence, the cut-set of $(V', V \setminus V')$ can be partitioned into a few subsets such that each subset corresponds to the cut-set of an atomic cut. 


Fix a bipartition $(U, (S\cup T)\cut U)$ of $S\cup T$ that admits a minimum cut of size at most $c$. 
Let $C$ be an arbitrary minimum $(U, (S\cup T)\cut U)$-cut size at most $c$. 
As discussed above, the cut-set of $C$ can be partitioned into a few subsets  such that each subset corresponds to the cut-set of an atomic cut. 
Let $C_1, \dots, C_t$ be these atomic cuts, called comprising atomic  cuts. 
If any of these atomic cuts does not partition $S\cup T$ into two non-empty sets, then removing its cut-set from the cut-set of $C$
results a $(U, (S\cup T)\setminus U)$-cut of smaller size. 
In addition, let $(U_i, (S\cup T)\setminus U_i)$ denote the bipartition on $S\cup T$ induced by $C_i$. $C_i$ is a minimum $(U_i, (S\cup T)\setminus U_i)$-cut, otherwise $C$ is not a minimum $(U, (S\cup T)\setminus U)$-cut. 
Hence, each of $C_i$ corresponds to a minimum $(U_i, (S\cup T)\setminus U_i)$-cut.

In addition, if it is not the case that every minimum $(U_i, (S\cup T)\setminus U_i)$-cut is an atomic cut, and let $C_i'$ be a minimum $(U_i, (S\cup T)\setminus U_i)$-cut that is not an atomic cut. 
Then the union of cut-sets of all the cuts in $\{C_j : 1 \leq j \leq t, j\neq i\} \cup \{C'\}$ is also a minimum $(U, (S\cup T)\setminus U)$-cut. 
Hence, for bipartition $(U, (S\cup T)\cut U)$, there is a minimum cut $C$ such that all the comprising atomic cuts $C_1, \dots, C_t$ are minimum $(U_i, (S\cup T)\setminus U_i)$-cuts which are atomic. 
\end{proof}

 To justify the second condition, we observe that if the cut-set of a  cut is not contained in one connected component of $G'$, then this cut is already shattered by $\IA_{G_0}(T_0, d, 2c)$. Furthermore, if the cut-set of a cut is contained in a connected component of $G'$ that has no intersection with $S$, then this cut is shattered by $\IA_{G_0}(T_0, d, 2c+1)$. 
 Hence, we exclude bipartitions of $S\cup T$ which have one of these minimum cuts. 

\begin{lemma}\label{overview_f0_f1_b4}
Let 
$\bipartitionsys$ be a maximal set of pairwise parallel bipartitions of $S$ which admit atomic cuts.
Let $\shatterset_{(S', S\setminus S')}$ be the set of $(S\cup T)$-bipartitions that partition $S$ into $(S', S\setminus S')$ satisfying the following conditions 
\begin{enumerate}
\item For each $(U, (S\cup T)\setminus U) \in \shatterset_{(S', S\setminus S')}$, $mincut_G(U,  (S\cup T)\setminus U) \leq c$.
\item All minimum cuts of bipartitions in $\shatterset_{(S', S\setminus S')}$ are atomic.
\item Every minimum cut of any bipartition in $\shatterset_{(S', S\setminus S')}$ has its cut-set contained  in a connected component of $G'$ that also contains at least one vertex of $S$.
\end{enumerate}
Let $F_0$ be a set of edges that contains the cut-set of a $(S', S\setminus S')$-atomic cut for each $(S', S\setminus S')\in \bipartitionsys$, and 
$F_1$ be a set of edges that shatters all bipartition in $\shatterset_{(S', S\setminus S')}$ for all $(S', S\setminus S')\in \Gamma$. 
Then $F_0 \cup F_1$ is a repair set.
\end{lemma}
\begin{proof}
Let $(U, (S\cup T)\setminus U)$ be a bipartition of $S\cup T$ with minimum cut of size at most $c$, and $C$ be an arbitrary minimum cut satisfying the three conditions of Claim~\ref{claim:comprising_condition}, and $C_1, \dots, C_t$ be the atomic cuts with cut-set in the cut-set of $C$. 

If there is a $C_i$ that induces a bipartition on $S$ not in $\Gamma$, then by Lemma~\ref{lem:not_parallel}, $C_i$ is shattered by $F_0$. 
If $C_i$ that is a $(U_i, (S\cup T)\setminus U_i)$-cut satisfying one of the two conditions
\begin{enumerate}
    \item There is a minimum $(U_i, (S\cup T)\setminus U_i)$-cut $C'$ such that no connected component of $G'$ contains the cut-set of $C'$.
    \item There is a minimum $(U_i, (S\cup T)\setminus U_i)$-cut $C'$ such that the cut-set of $C'$ belongs to a connected component in $G'$ that does not contain any vertex of $S$. 
\end{enumerate}
then by the definition of $\IA$ set, there is a minimum $(U, (S\cup T)\setminus U)$-cut shattered by $\IA_{G_0}(T_0, d, 2c+1)\cap E$. 

Hence, if no minimum $(U, (S\cup T)\setminus U)$-cut is shattered by $F_0 \cap (\IA_{G_0}(T_0, d, 2c+1)\cap E)$, then each of $(U_i, (S\cup T)\setminus U_i)$ belongs to a $\shatterset_{(S', S\setminus S')}$ for some $(S', S\setminus S') \in \Gamma$.
By the definition of $F_1$, $F_1$ shatters a minimum $(U, (S\cup T)\setminus U)$-cut. 

Therefore, $F_0 \cup F_1 \cup (\IA_{G_0}(T_0, d, 2c+1)\cap E)$ shatters every bipartition on $S\cup T$ with minimum cut of size at most $c$. $F_0\cup F_1$ is a repair set. 
\end{proof}

The remaining problem is to find a set of edges
to shatter all bipartitions in $\shatterset_{(S', S\setminus S')}$. 
We show that 
for any $(S', S\setminus S')\neq (\emptyset, S)$, a set of $O(c^2)$ edges are sufficient (Lemma~\ref{lem:overview_size_each_type1}). 
If $(S', S\setminus S') = (\emptyset, S)$, 
a set of $O(|S|c^2)$ edges are sufficient (Lemma~\ref{lem:overview_size_type2}). 
Together with the fact that there are $O(\abs S)$ bipartitions in $\bipartitionsys$, we obtain Lemma~\ref{lem:update_intersecting_all-1}. 

\paragraph{Shattering all bipartitions of $S\cup T$ that partition $S$ non-trivially} Now we discuss how to find a set of $O(c^2)$ edges to shatter $\shatterset_{(S', S\cut S')}$ for $S'\not\in\set{\emptyset, S}$.

\begin{lemma}\label{lem:overview_size_each_type1}
For any $(S', S\setminus S')\in \shatterset_{(S', S\setminus S')}$ such that $S'\not\in\set{\emptyset, S}$, there is a set $F_{(S', S\cut S')}$ of $O(c^2)$ edges to shatter every bipartition of $\shatterset_{(S', S\setminus S')}$. 
\end{lemma}

To prove Lemma~\ref{lem:overview_size_each_type1}, we  show that fix an arbitrary $(S', S\setminus S') \in \Gamma$ that is not $ (\emptyset, S)$, 
\begin{enumerate}
    \item At most two connected components of $G'$ may contain cut-sets of minimum cuts inducing bipartitions in $\shatterset_{(S', S\cut S')}$ (Lemma~\ref{lem:overview_2cc_b6}).
    \item For any $V^\star\subset V$ such that $G'[V^\star]$ is a connected component of $G'$, there always exists a set of $O(c^2)$ edges that shatters all $(S\cup T)$-bipartitions in $\shatterset_{(S', S\setminus S')}$ admitting atomic minimum cuts whose cut-set is in $G'[V^\star]$ (Lemma~\ref{lem:overview_type1_b7}, \ref{lem:overview:elim_b5} and \ref{lem:find_parallel_cuts}). 
\end{enumerate}
As a result, it suffices to have $O(c^2)$ edges to shatter $\shatterset_{(S', S\cut S')}$.

\begin{lemma}\label{lem:overview_2cc_b6}
Let $S'$ be a subset of $S$ such that $S' \neq \emptyset$ and $S'\neq S$. 
There are at most two connected components of $G'$ containing the cut-set of an atomic $(S', S\setminus S')$-cut. 
\end{lemma}
\begin{proof}
Let $P_1, P_2 \subset V$ be two sets of vertices such that the following conditions hold:
\begin{enumerate} 
\item Both $G'[P_1]$ and $G'[P_2]$ are two different connected components of $G'$.
\item Both $G'[P_1]$ and $G'[P_2]$
 contain cut-sets of atomic cuts that partition $S$ into $(S', S\cut S')$.
 \item Both $P_1$ and $P_2$ have  non-empty intersection with $S$.
 \end{enumerate}
 
 Let $(V_1, V\cut V_1)$ be a cut partitioning $S$ into $(S', S\cut S')$ such that  $\partial_G(V_1)\subset G'[P_1]$ and $(V_2, V\cut V_2)$ be the cut partitioning $S$ into $(S', S\cut S')$ such that $\partial_G(V_2)\subset G'[P_2]$. We can without loss of generality assume that $S'\subset V_1$ and $S'\subset V_2$. Since $G'[P_1]$ and $G'[P_2]$ are disjoint connected components of $G'$, $\partial_G(V_1)$ does not shatter $(V_2, V\cut V_2)$ and $\partial_G(V_2)$ does not shatter $(V_1, V\cut V_1)$. By the contrapositive of Lemma~\ref{lem:not_parallel}, $(V_1, V\cut V_1)$ and $(V_2, V\cut V_2)$ are parallel. Since $\emptyset \neq S'\subset V_1\cap V_2$ and $\emptyset \neq (S\cut S')\subset (V\cut V_1)
\cap (V\cut V_2)$, either $V_1\subset V_2$ or $V_2\subset V_1$. We can without loss of generality assume $V_1\subset V_2$.

If there exists a third vertex set $P_3$ that corresponds to a connected component of $G'$ containing the cut-set of atomic cut $(V_3, V\cut V_3)$ that partitions $S$ into $(S', S\cut S')$, we can assume $V_1\subset V_2\subset V_3$ with no loss of generality. 

Since all three cuts induce the same partition of $S$, $(V_2\cut V_1)\cap S = \emptyset$ and $((V\cut V_2)\cut(V\cut V_3))\cap S = \emptyset$. Since $V_2\cap G'[P_2] \subset V_2\cut V_1$ and \[((V\cut V_2)\cap G'[P_2])\subset (V\cut V_2)\cap V_3 = ((V\cut V_2)\cut(V\cut V_3)),\] we have \[S\cap G'[P_2] \subset S\cap \wrap{(V_2\cut V_1)\cup \wrap{(V\cut V_2)\cap V_3}} = \emptyset.\] This is a contradiction to the assumption that $G'[P_3]\cap S\neq\emptyset$.
\end{proof}

Let $V^\star$ be a set of vertices corresponding to a connected component of $G'$. We select a maximal set of pairwise parallel atomic $(S', S\setminus S')$-cuts of size at most $c$ with cut-sets in $G'[V^\star]$
that partition $T$ differently. Lemma~\ref{lem:overview_type1_b7} shows that the union of the cut-sets of these cuts shatters all bipartitions in $\shatterset_{(S', S\setminus S')}$ that has a minimum atomic cut in $G[V^\star]$. Lemma~\ref{lem:overview:elim_b5} and \ref{lem:find_parallel_cuts} together show that there are at most $O(c)$ cuts in any maximal set of pairwise parallel atomic $(S', S\cut S')$ cuts. Since each cut-set consists of at most $c$ edges, it suffices to use $O(c^2)$ edges to shatter all biparitions in $\shatterset_{(S', S\cut S')}$.

\begin{lemma}\label{lem:overview_type1_b7}
Let $V^\star$ be a set of vertices corresponding to a connected component of $G'$.
Let $S'$ be a subset of $S$ such that $S' \neq \emptyset$ and $S'\neq S$. 
Let $\mathcal{C} = \set{(V_1, V\setminus V_1), \ldots, (V_t, V\cut V_t)}$ be a maximal set of pairwise parallel atomic $(S', S'\setminus S')$-cuts of size at most $c$  with cut-sets contained in $G'[V^\star]$ satisfying the following conditions:
\begin{enumerate}
    \item[(a)] $(V_i, V\setminus V_i)$ and $(V_j, V\setminus V_j)$ partition $T$ differently for each $1\leq i < j \leq t$. 
    \item[(b)] Let $U_i = V_i \cap (S\cup T)$. $(V_i, V\setminus V_i)$ is a minimum $(U_i, (S\cup T)\setminus U_i)$-cut.
\end{enumerate}
Then we have:
\begin{enumerate}
    \item $\bigcup_i\partial_G(V_i)$ shatters all partitions in $\shatterset_{(S', S\setminus S')}$ that has a minimum atomic cut whose cut-set is contained in $G'[V^\star]$.
    \item There exists a permutation $\sigma$ of $\set{1,\ldots, t}$ such that $V_{\sigma(1)} \subseteq V_{\sigma(2)} \cdots\subset V_{\sigma(t)}$.
\end{enumerate}
\end{lemma}
\begin{proof}
Fix a $(S'\cup T', (S\cup T)\cut (S'\cup T'))\in \shatterset_{(S', S\setminus S')}$ that admits a minimum atomic cut $C = (V_C, V\cut V_C)$ of size at most $c$, whose cut-set is contained in $G'[V^\star]$. If there exists a  $(V_j, V\cut V_j)\in \mathcal{C}$ that also partitions $S\cup T$ into $S'\cup T'$ and  $(S\cup T)\cut (S'\cup T'))$, then  $\partial_G(V_j)$ shatters partition $(S'\cup T', (S\cup T)\cut (S'\cup T'))$.

Otherwise, no $(V_i, V\cut V_i)$ in the $\mathcal{C}$ induces bipartition $(T', T\cut T')$ of $T$.  Since $\mathcal{C}$ is a maximal set and all cuts in $\mathcal{C}$  induce different bipartitions on $T$, there exists a $(V_j, V\cut V_j)$ in $\mathcal C$ such that the bipartition $(T', T\cut T')$ and the bipartition of $T$ induced by $(V_j, V\cut V_j)$ are not parallel. Therefore, $(V_j, V\cut V_j)$ and  $(V_C, V\cut V_C)$ are not parallel. Since $(V_j, V\cut V_j)$ is atomic, $\partial_G(V_j)$ shatters $(V_C, V\cut V_C)$ by Lemma~\ref{lem:not_parallel}.
Hence, $\bigcup_i\partial_G(V_i)$ shatters all partitions in $\shatterset_{(S', S\setminus S')}$ that has a minimum atomic cut whose cut-set is contained in $G'[V^\star]$.

Fix any two cuts $(V_i, V\cut V_i)$ and $(V_{j}, V\cut V_{j})$ in $\mathcal C$. Since they are parallel, by definition, one of the following cases holds:
\begin{itemize}\itemsep -3pt
    \item $V_i\subset V_{j}$.
    \item $V_i\subset V\cut V_{j}$. 
    \item $V\cut V_i\subset V_{j}$. 
    \item $V\cut V_i \subset V\cut V_{j}$. In this case, $V_j\subset V_i$.
\end{itemize}
The second and third cases are impossible, because we have the following conditions
\[S'\neq \emptyset, S\setminus S'\neq \emptyset, S'\subset V_i\cap V_{j}, \text{ and } (S\setminus S')\subset (V\setminus V_i)\cap (V\setminus V_{j}).\]
Therefore, we can define a permutation $\sigma$ as 
\begin{gather*}
    \sigma(i) \defeq k\in\set{1, \ldots, t}\text{ such that }V_k\subset V_j\text{ for all }j \in\set{1, \ldots, t}\cut \wrap{\bigcup_{\ell = 1}^{i-1} \set{\sigma(\ell)}}
\end{gather*}
Then we have $V_{\sigma(1)} \subseteq V_{\sigma(2)} \cdots\subset V_{\sigma(t)}$.
\end{proof}

The following lemmas bound the size of any such maximal set of pairwise parallel atomic cuts on a connected component $G'[V^\star]$ of $G'$. The intuition behind is that for two parallel  $(S', S\cut S')$-cuts $C_i$ and $C_j$ in $\mathcal C$, the union of their cut-sets forms a cut of size at most $2c$ such that one side does not have any terminal in $S$ but at least one terminal in $T$ (Figure~\ref{fig:fig2}(a)). We use $C_{i, j} = (V_{i, j}, V\setminus V_{i, j})$ to denote this cut such that $V_{i, j}$ is the side that dose not contain $S$,
and $T_{i, j }$ to denote the intersection of $T$ and $V_{i, j}$.

Since $G$ is an induced subgraph of $G_0$ and $V_{i, j}$ has not intersection with $S$, the cut-set of $C_{i, j}$ also induces a $(T_{i, j}, T_0\setminus T_{i, j})$-cut on graph $G_0$ with size at most $2c$.
By the definition of  $\IA_{G_0}(T_0, d, 2c)$, 
there must be a $(T_{i, j}, T_0 \setminus T_{i, j})$-cut $C'_{i, j}$ for graph $G_0$ of size at most $2c$ with cut-set as a subset of the $\IA_{G_0}(T_0, d, 2c)$ set.
At this point, taking the cut induced by the edges of $C'_{i, j}$'s cut-set in $G$, 
we obtain a cut $C^\diamond_{i, j}$ on graph $G$ of size at most $2c$ such that the cut-set of $C^\diamond_{i, j}$ is a subset of $\IA_{G_0}(T_0, d, 2c)\cap E$.

We now look at an arbitrary $C_{i,t}^\diamond$ (recall that $t$ is the index of last cut in $\mathcal C$). Since $V^\star$ corresponds to a connected component of $G'$, $V^\star$ lies on one side of $C_{i, t}^\diamond$. If it is on the same side as $T_{i, t}$, then for any $1\leq \ell\leq i - 1$, $T_{\ell, \ell+1}$ is on the other side of $C_{i,t}^\diamond$.  Since the cut-set of $C_{\ell, \ell+1}$ is a set of edges in $G'[V^\star]$, there is a path within $G[V_{\ell, \ell+1}]$ connecting a vertex of $T_{\ell, \ell+1}$  and  a vertex in $\endpoints\wrap{\partial_G(V_{\ell, \ell +1})}\cap V_{\ell, \ell+1}$ (the brown bold path in Figure~\ref{fig:fig2}(b)). This implies that there is an edge in the cut-set of $C_{i,t}^\diamond$ which has both endpoints in $V_{\ell, \ell+1}$. Since for different $\ell$ and $\ell'$, $V_{\ell, \ell+1}$ and $V_{\ell', \ell'+1}$ are disjoint and that there are at most $2c$ edges in the cut-set of $C_{i,k}^\diamond$, $i\leq 2c+1$. Similarly, if $V^\star$ and $T_{i, t}$ are on difference sides of $C_{i, t}^\diamond$, we can show that $t - i \leq 2c+1$. 

Hence, if $t \geq 4c+4$, then $V^\star$ cannot be on any side of $C_{2c+2, t}^\diamond$, contradiction. Thus, $t \leq 4c + 3$.
\begin{figure}[htb]
\centering
\begin{subfigure}[b]{.43\textwidth}
\includegraphics[width=\textwidth]{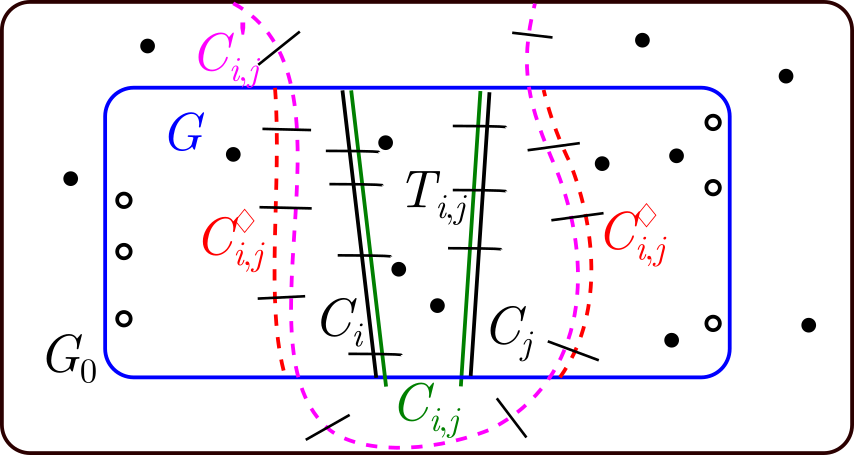}
\caption{}
\end{subfigure}
\hfill
\begin{subfigure}[b]{.43\textwidth}
\includegraphics[width=\textwidth]{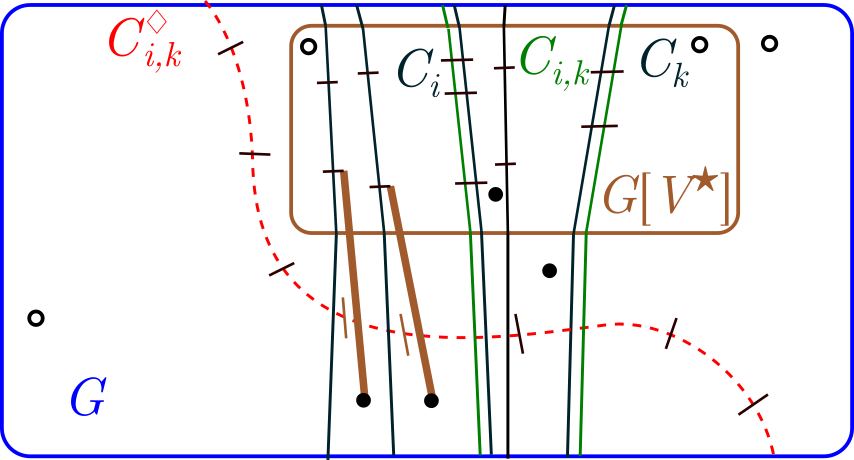}
\caption{}
\end{subfigure}
\caption{(a) The illustration of $C_{i, j}$ (green solid line), $C'_{i, j}$ (magenta dashed line) and $C^\diamond_{i, j}$ (red dashed line). Black dots represent vertices of $T_0$, and black circles represent vertices of $S$. (b) The illustration of the case that $V^\star$ is on the side of $C^\diamond_{i, k}$ containing $T_{i, k}$. The bold lines represent paths from some vertex of $T_{\ell, \ell+1}$ to some vertex of $\endpoints(\partial_G(V_{\ell, \ell+1})) \cap V_{\ell, \ell+1}$ for different $\ell$. 
}
\label{fig:fig2}
\end{figure}

\begin{lemma}\label{lem:overview:elim_b5}
Let $V^\star$ be a set of vertices corresponding to a connected component of $G'$, and 
$\mathcal{C} = \{C_1 = (V_1, V\cut V_1), C_2 = (V_2, V\cut V_2), \dots C_t = (V_t, V\cut V_t)\}$ be a set of atomic cuts satisfying the following conditions
\begin{enumerate}\itemsep -3pt
    \item The cut-sets of all cuts in $\mathcal{C}$ are in $G'[V^\star]$.
    \item $V_1\subset V_2\subset \cdots\subset V_t$.
    \item All cuts induce the same bipartition of $S$, but different bipartitions of $T$. 
\end{enumerate}
Then $\mathcal C$ has at most $4c+3$ cuts.
\end{lemma}

\begin{proof}

Let $T_i = V_i\cap T$. 
Each cut $(V_i, V\cut V_i)\in \mathcal C$ partitions $S\cup T$ into $(S'\cup T_i, (S\cut S')\cup (T\cut T_i))$. For $j>i$, let $T_{i,j} = T_j\cut T_i$. Since all cut in $\mathcal C$ induce different partitions of $T$, $T_{i,j}\neq \emptyset$.

Since all cuts in $\mathcal C$ induce the same partition of $S$, $S\cap(V_j\cut V_i) = \emptyset$. The union of the cut-sets of $(V_i, V\cut V_i)$ and $(V_j, V\cut V_j)$ forms a cut $C_{i, j} = (V_{i,j}, V\setminus V_{i, j})$ of size at most $2c$ that partitions $S\cup T$ into $T_{i,j}$ and $S\cup (T\cut T_{i,j})$, which is also a  cut of size at most $2c$ that partitions $T_0$ into $T_{i,j}$ and $T_0\cut T_{i,j}$ on $G_0$. Therefore, there is another cut on $G_0$ that partitions $T_0$ into $(T_{i,j}, T_0\cut T_{i,j})$, whose cut-set is a subset of $\IA_{G_0}(T_0, d, 2c)$. Let this cut be $C'_{i, j} = (V'_{i,j}, V_0\cut V'_{i,j})$. 
Furthermore, let $C^\diamond_{i,j} = (V^\diamond_{i, j}, V\setminus V^\diamond_{i, j})$ be the cut such that $V^\diamond_{i, j} = V'_{i, j} \cap V$. 
It is easy to verify that $\partial_G(V^\diamond_{i, j})$ is a subset of $\IA_{G_0}(T_0, d, 2c) \cap E$. 
Hence, $V^\diamond_{i, j}$ is a union of connected components of $G'$.

Now consider $V^{\diamond}_{i, t}$. 
If $V^\star$ is a subset of $V^\diamond_{i, t}$, then for each $1 \leq \ell < i$, $T_{\ell, \ell+1}$ is a subset of $V\setminus V^{\diamond}_{i, t}$.
Hence, there is a path from any vertex of $T_{\ell, \ell+1}$ to some vertex of 
$\endpoints(\partial_G(V_{\ell, \ell+1})) \cap V_{\ell, \ell+1}$ within $G[V_{\ell, \ell+1}]$. 
Since $\endpoints(\partial_G(V_{\ell, \ell+1})) \subseteq V^\star$,  $\endpoints(\partial_G(V_{\ell, \ell+1}))$
and $T_{\ell, \ell+1}$ are on different sides of $C^{\diamond}_{i, t}$, and thus the aforementioned path from a vertex of $T_{\ell, \ell+1}$ to some vertex of 
$\endpoints(\partial_G(V_{\ell, \ell+1})) \cap V_{\ell, \ell+1}$ within $G[V_{\ell, \ell+1}]$ connects two vertices on different sides of $C^{\diamond}_{i, t}$.
Hence, there exists an edge $e_{i}$ in the cut-set of $C^{\diamond}_{i, t}$ with both endpoints in $V_{\ell, \ell + 1}$.
Together with the fact that $V_{\ell, \ell + 1}$ and $V_{\ell', \ell'+1}$ are disjoint for any $\ell \neq \ell'$, $e_{\ell}$ is different to $e_{\ell'}$.
Since $C^{\diamond}_{i, t}$ is a cut of size at most $2c$, we have $i \leq 2c+1$. 

Similarly, if $V^\star$ is a subset of  $V\setminus V^{\diamond}_{i, t}$, one can show that $t - i \leq 2c+1$. 
Since our choice of $i$ is arbitrary, 
if $t \geq  4c+4$, then let $i = 2c+2$, $V^\star$ does not belong to any side of $C^{\diamond}_{i, t}$, contradiction. 
Hence, we have $t \leq 4c+3$. 
\end{proof}

Combining Lemma~\ref{lem:overview_type1_b7} and \ref{lem:overview:elim_b5}, we have
\begin{lemma}\label{lem:find_parallel_cuts}
Let $(S', S\setminus S')$ be a bipartition of $S$ such that $S' \neq \emptyset$ and $S\setminus S' \neq \emptyset$,
and $G[V^\star]$ be a connected component of $G'$. 
If $\mathcal{C} = \{C_1, C_2, \dots, C_k\}$ is a maximal set of pairwise parallel atomic $(S', S\setminus S')$-cuts of size at most $c$ 
satisfying the following two conditions, then $k \leq 4c+3$.
\begin{enumerate}\itemsep -3pt
    \item For any $i \neq j$, $C_i$ and $C_j$ induce different partitions on $S\cup T$.
    \item The cut-set of any $C_i$ is in $G[V^\star]$. 
\end{enumerate}
 \end{lemma}
\begin{proof}
Let $\mathcal C = \set{C_1 = (V_1, V\cut V_1), \ldots, C_k = (V_k, V\cut V_k)}$. By Lemma~\ref{lem:overview_type1_b7}, there exists a permutation $\sigma$ of $\set{1,\ldots, t}$ such that $V_{\sigma(1)} \subseteq V_{\sigma(2)} \cdots\subset V_{\sigma(t)}$. 
By Lemma~\ref{lem:overview:elim_b5}, there are at most $4c+3$ cuts in $\mathcal C$. 
\end{proof}

\begin{proof}[Proof of Lemma~\ref{lem:overview_size_each_type1}]
For each $(S', S\setminus S')\in\bipartitionsys$ such that $S'\neq \emptyset$ and $S'\neq S$, Lemma~\ref{lem:overview_2cc_b6} shows that there are at most 2 connected components, $G'[V^{(1)}]$ and $G'[V^{(2)}]$ of $G'$ that may contain cut-sets of cuts inducing $(S', S\cut S')$. 
For each $G'[V^{(i)}]$,
let \[\mathcal C_{V^{(i)}} = \set{\left(V_j^{(i)}, V\setminus V_j^{(i)}\right) \ldots, \left(V_t^{(i)}, V\cut V_t^{(i)}\right)}\] be a maximal set of pairwise parallel atomic $(S', S'\setminus S')$-cuts of size at most $c$  with cut-sets contained in $G'[V^\star]$ satisfying the following conditions:
\begin{enumerate}
    \item[(a)] $\left(V_j^{(i)}, V\setminus V_j^{(i)}\right)$ and $\left(V_k^{(i)}, V\setminus V_k^{(i)}\right)$ partition $T$ differently for each $1\leq j < k \leq t$. 
    \item[(b)] Let $U_j^{(i)} = V_j ^{(i)}\cap (S\cup T)$. $\left(V_j^{(i)}, V\setminus V_j^{(i)}\right)$ is a minimum $(U_j, (S\cup T)\setminus U_j)$-cut.
\end{enumerate}
Let $F_{(S', S\cut S')}$ be the union of the cut-sets of all cuts in $\mathcal C_{V^{(1)}}$ and $\mathcal C_{V^{(2)}}$. 
By Lemma~\ref{lem:overview_type1_b7}, $F_{(S', S\cut S')}$ shatters all bipartitions in $\shatterset_{(S', S\cut S')}$. By Lemma~\ref{lem:find_parallel_cuts}, each $\mathcal C_{V^{(i)}}$ contains $O(c)$ cuts. Therefore, $F_{(S', S\cut S')}$ contains $O(c^2)$ edges.
\end{proof}

\paragraph{Shattering all bipartitions of $S\cup T$ such that $S$ is on one side}
Now we discuss how to find a set of $O(\abs S c^2)$ edges to shatter $\shatterset_{(\emptyset, S)}$. 

The high level idea is similar to the previous case. By Lemma~\ref{overview_f0_f1_b4}, the bipartitions in $\shatterset_{(\emptyset, S)}$ have their cut-sets contained in one connected component of $G'$ that also contains a vertex of $S$. 
As an analogue of  Lemma~\ref{lem:overview_type1_b7}, one can show that for each vertex set $V^\star$ that corresponds to a connected component of $G'$ that contains a vertex of $S$,
let $\mathcal C_{V^\star}$ be a maximal set of pairwise parallel minimum atomic $(\emptyset, S)$-cuts of size at most $c$ that partition $T$ differently, and have cut-sets contained in $G'[V^\star]$.
The union of the cut-sets of cuts in $\mathcal C_{V^\star}$ shatters all partitions in $\shatterset_{(\emptyset, S)}$ that has a minimum atomic cut whose cut-set is contained in $G'[V^\star]$.
By lemma~\ref{lem:overview:elim_b5}, $\mathcal C_{V^\star}$ contains at most $4c+3$ cuts. 
Together with the fact that there are at most $|S|$ connected components that contain a vertex in $S$, we have the following lemma.

\begin{lemma}\label{lem:overview_size_type2} 
There is a set $F_{(\emptyset, S)}$ of $O(|S|c^2)$ edges to shatter every bipartition in $\shatterset_{(\emptyset, S)}$. 
\end{lemma}

Again, by Lemma~\ref{overview_f0_f1_b4}, the bipartitions in $\shatterset_{(\emptyset, S)}$ have their cut-sets contained in one connected component of $G'$ that also contains a vertex of $S$. There are at most $O(\abs S)$ such connected components.
Similar to the construction in the proof of Lemma~\ref{lem:overview_size_each_type1}, 
for each connected component $G[V^\star]$ of $G'$ that contains at least one vertex of $S$,
we construct a set of edges by choosing a maximal set of pairwise parallel $(\emptyset, S)$-cuts that partition $T$ differently  and show that this set of edges shatter all bipartitions in $\shatterset_{(\emptyset, S)}$ admitting atomic minimum cuts whose cut-set is in $G[V^\star]$ (Lemma~\ref{lem:overview_type2_b8}).  Then we utilize Lemma~\ref{lem:overview:elim_b5} to show that there are at most $O(c)$ such cuts in each of these connected components.

\begin{lemma}\label{lem:overview_type2_b8}
Let $V^\star$ be a set of vertices corresponding to a connected component of $G'$.
Let \\$\mathcal{C} = \set{(V_1, V\setminus V_1), \ldots, (V_t, V\cut V_t)}$ be a maximal set of pairwise parallel atomic $(\emptyset, S)$-cuts of size at most $c$  with cut-sets contained in $G'[V^\star]$ satisfying the following conditions:
\begin{enumerate}
    \item[(a)] $(V_i, V\setminus V_i)$ and $(V_j, V\setminus V_j)$ partition $T$ differently for each $1\leq i < j \leq t$. 
    \item[(b)] Let $U_i = V_i \cap (S\cup T)$. $(V_i, V\setminus V_i)$ is a minimum $(U_i, (S\cup T)\setminus U)$-cut.
\end{enumerate}
Then we have:
\begin{enumerate}
    \item $\bigcup_i\partial_G(V_i)$ shatters every bipartition in $\shatterset_{(\emptyset, S)}$ that admits a minimum atomic cut whose cut-set is contained in $G'[V^\star]$.
    \item There exists a permutation $\sigma$ of $\set{1,\ldots, t}$ such that $V_{\sigma(1)} \subseteq V_{\sigma(2)} \cdots\subset V_{\sigma(t)}$.
\end{enumerate}
\end{lemma}

To prove this lemma, we need  Claim~\ref{claim:type2_helper1} and Claim~\ref{claim:type2_helper2}.

\begin{claim}\label{claim:type2_helper1}
Let $V^\star$ be a set of vertices corresponding to a connected component of $G'$.
Fix a bipartition $(T', S'\cup (T\cut T'))\in \shatterset_{(\emptyset, S)}$ that admits a minimum atomic cut of size at most $c$ whose cut-set is contained in $G'[V^\star]$. 
There exists a $(T', T_0\cut T')$-cut $(V', V_0\cut V')$ of size $mincut_{G_0}(T', T_0\cut T')$ on $G_0$ satisfying $T'\subset V'$ and $\partial_{G_0}(V')\subset \IA_{G_0}(T_0, d, 2c)$, and $V^\star\subset V'$.
\end{claim}

\begin{proof}
Note that the cut-set of any $(T', S\cup (T\cut T'))$-cut also induces a $(T', T_0\cut T')$-cut on $G_0$. Hence, $mincut_G(T', S\cup (T\cut T')) \geq mincut_{G_0}(T', T_0\cut T')$.

By the definition of $\IA$ set, there exists a minimum $(T', T_0\cut T')$-cut $(V', V_0 \setminus V')$ with cut-set in $\IA_{G_0}(T_0, d, 2c+1)$ such that  $T'\subset V'$. 
If $V'\cap S = \emptyset$, then $(V' \cap V, V\setminus (V' \cap V))$ is a minimum $(T', S\cup (T\cut T'))$-cut with cut-set in $\IA_{G_0}(T_0, d, 2c+1)\cap E$, contradicting  $(T', S\cup (T\cut T'))\in \shatterset_{(\emptyset, S)}$. 
Hence $V' \cap S \neq \emptyset$.

Then there exists a connected component $V^\dagger$ of $G_0[V']$ that contains vertices from both $S$ and $T'$. Otherwise, a subset of the cut-set of $(V', V_0\cut V')$ induces a $(T', (S\cup T_0)\cut T'))$-cut on $G_0$ and thus a $(T', (S\cup T)\cut T')$-cut on $G$, which implies that there is a minimum $(T', S\cup (T\setminus T'))$-cut with cut-set in $\IA_{G_0}(T_0, d, 2c)$.  

Recall that there is a $(V^\dagger, V\cut V^\dagger)$ that is a  minimum atomic $(T', S\cup (T\cut T'))$-cut of size at most $c$ on $G$, and $\partial_G(V^\dagger)\subset G'[V^\star]$. Since $(V^\dagger, V\cut V^\dagger)$ is an atomic cut, any path from a vertex in $T'$ to  a vertex in $S$ has an edge that is in $\partial_{G}(V^\dagger)$, which is also in $G[V^\star]$. 
Note that within $G[V^\dagger]$, there is a path from some vertex of $S$ to some vertex of $T'$. 
Hence, $G[V^\dagger]$ contains an edge in $G[V^\star]$.
Since  $V^\star$ is contained in one connected component of $G_0\cut \IA_{G_0}(T_0, d, 2c)$,  $V^\star\subset V^\dagger\subset V'$. 
\end{proof}

\begin{claim}\label{claim:type2_helper2}
Let $(T_1, S\cup(T\cut T_1))$ and $(T_2, S\cup(T\cut T_2))$ be two partitions in $\shatterset_{(\emptyset,S)}$ that admit minimum atomic cuts $(V_1, V\cut V_1)$ and $(V_2, V\cut V_2)$ of size at most $c$ whose cut-sets are contained in  $G'[V^\star]$. Then $T_1\cap T_2\neq \emptyset$.
\end{claim}

\begin{proof}
$(V_1, V_0\cut V_1)$ and $(V_2, V_0\cut V_2)$ are minimum cuts on $G_0$ that form partitions $(T_1, T_0\cut T_1)$ and $(T_2, T_0\cut T_2)$. If $T_1\cap T_2 = \emptyset$, then $T\cap(V_1\cap V_2) = \emptyset$.  Let $V_1' = V_1\cut (V_1\cap V_2)$ and $V_2' = V_2\cut(V_1\cap V_2)$. $(V_1', V_0\cut V_1')$ and $(V_2', V_0\cut V_2')$ are also cuts on $G_0$ that form partitions $(T_1, T_0\cut T_1)$ and $(T_2, T_0\cut T_2)$. Since $(V_1,V_0\cut V_1)$ and $(V_2, V_0\cut V_2)$ are minimum cuts, $\abs{\partial_{G_0}(V_1) }\leq \abs{\partial_{G_0}(V_1')}$ and $\abs{\partial_G(V_2)}\leq \abs{\partial_G(V_2')}$. However, by Lemma~\ref{lem:swapping}, we have \[\abs{\partial_{G_0}(V_1')} + \abs{\partial_{G_0}(V_2')} = \abs{\partial_{G_0}(V_1) } + \abs{\partial_{G_0}(V_2)}.\] Therefore, $\abs{\partial_{G_0}(V_1) }= \abs{\partial_{G_0}(V_1')}$ and $\abs{\partial_G(V_2)} = \abs{\partial_G(V_2')}$. Note that $(V_1', V_0\cut V_1')$ and $(V_2', V_0\cut V_2')$  are also minimum cuts on $G_0$ that partition $T$ into $(T_1, T\cut T_1)$ and $(T_2, T\cut T_2)$ respectively and their cut-sets are also contained in $\IA_{G_0}(T_0, d, 2c)$. However, since $V^\star\in V_1\cap V_2$, $V^\star\not\subset V_1'$ and  $V^\star \not \subset V_2'$. This gives us a contradiction to Claim~\ref{claim:type2_helper1}.
\end{proof}

\begin{proof}[Proof of Lemma \ref{lem:overview_type2_b8}]

Fix a partition of $S\cup T$, $(T', S'\cup(T\cut T'))\in \shatterset_{(\emptyset, S)}$ that admits a minimum atomic cut $(V_C, V\cut V_C)$ whose cut-set is contained in $G'[V^\star]$. If there exists a $V_j$ such that $(V_j, V\cut V_j)$ also partitions $S\cup T$ into $(T', S'\cup(T\cut T'))$, then  $\partial_G(V_j)$ shatters partition $(T', S'\cup(T\cut T'))$.

Otherwise, no $(V_i, V\cut V_i)$ in the set induces partition $(T', T\cut T')$ of $T$.  Since \[\set{(V_1, V\setminus V_1), \ldots, (V_t, V\cut V_t)}\] is a maximal set, there exists a $(V_j, V\cut V_j)$ that is not parallel to $(V_C, V\cut V_C)$. Since $(V_j, V\cut V_j)$ is atomic, $\partial_G(V_j)$ shatters $(V_C, V\cut V_C)$.

 Fix $i, j$. By Claim~\ref{claim:type2_helper2}, $T_i\cap T_{j}\neq \emptyset$. Consider $(V_i, V\cut V_i)$ and $(V_{j}, V\cut V_{j})$. Since they are different and parallel, by definition, one of the following cases holds:
\begin{itemize}
    \item $V_i\subset V_{j}$.
    \item $V_i\subset V\cut V_{j}$. This is impossible because  $V_i\cap V_{j}\neq\emptyset$.
    \item $V\cut V_i\subset V_{j}$. This is impossible because $S \subset (V\cut V_i)\cap (V\cut V_{j})$. Therefore, $(V\cut V_i)\cap (V\cut V_{j})$ cannot be empty.
    \item $V\cut V_i \subset V\cut V_{j}$. In this case, $V_j\subset V_i$.
\end{itemize}
Therefore, we can define permutation $\sigma$ as 
\begin{gather*}
    \sigma(i) = \text{the unique } k\in\set{1, \ldots, t}\text{ such that }V_k\subset V_j\text{ for all }j \in\set{1, \ldots, t}\cut \wrap{\bigcup_{j = 1}^i \set{\sigma(j)}}
\end{gather*}
Then we have $V_{\sigma(1)} \subseteq V_{\sigma(2)} \cdots\subset V_{\sigma(t)}$.
\end{proof}

\begin{proof}[Proof of Lemma~\ref{lem:overview_size_type2}]
For $(\emptyset, S)\in \bipartitionsys$, on each connected component $G'[P_i]$ of $G'$ containing at least one vertex in $S$, Lemma~\ref{lem:overview_type2_b8} defines a set $\mathcal C_{P_i}$ of atomic cuts of size at most $c$. The union of the cut-sets of all cuts in all $\mathcal C_{P_i}$'s is $F_{(\emptyset, S)}$ that shatters all bipartitions in $\shatterset_{(\emptyset, S)}$. By Lemma~\ref{lem:overview_type2_b8}, for each $\mathcal C_{P_i} = \{(V_1, V\cut V_1), \ldots, $ $ (V_t, V\cut V_t)\}$, there exists a permutation $\sigma$ of $\set{1,\ldots, t}$ such that $V_{\sigma(1)} \subseteq V_{\sigma(2)} \cdots\subset V_{\sigma(t)}$. Since the ordering of cuts in $\mathcal C_{P_i}$ can be changed, we assume that $\sigma(i) = i$ for $i = 1,\ldots, t$ without loss of generality. By Lemma~\ref{lem:overview:elim_b5}, there are at most $O(c)$ cuts in each $\mathcal C_{P_i}$. Since there are at most $O(\abs S)$ connected components of $G'$ containing at least one vertex in $S$, $F_{(\emptyset, S)}$ contains at most $O(\abs S c^2)$ edges.
\end{proof}

\paragraph{Proof of Lemma~\ref{lem:update_intersecting_all-1}} Now we are ready to prove Lemma~\ref{lem:update_intersecting_all-1}. 
\begin{proof}[Proof of Lemma~\ref{lem:update_intersecting_all-1}]
Let $\bipartitionsys$ be a maximal set of pairwise parallel bipartitions of $S$ that admit atomic cuts. Let $\shatterset_{(S', S\cut S')}$ be as defined in Lemma~\ref{overview_f0_f1_b4}. $F_0$ be a set of edges that contains the cut-set of a $(S', S\cut S')$-atomic cut for each $(S', S\cut S')\in\bipartitionsys$ and $F_1$ be a set of edges such that $F_1$ shatters all bipartitions in $\shatterset_{(S', S\setminus S')}$ for each $(S', S\setminus S')\in \Gamma$. 
Lemma~\ref{overview_f0_f1_b4} states that $F_0 \cup F_1$ is a repair set. 

By Lemma~\ref{lem:linear_partition_overviewb3}, there are $O(\abs S)$ bipartitions in $\bipartitionsys$. Therefore, $\abs{F_0}$ is $O(\abs S c)$.

To bound the size of $F_1$, we consider $F_1$ as a union $\bigcup_{(S', S\cut S')\in\bipartitionsys}F_{(S', S\cut S')}$ such that $F_{(S', S\cut S')}$ shatters all bipartitions in $\shatterset_{(S', S\setminus S')}$.

For each $(S', S\setminus S')\in\bipartitionsys$ such that $S'\neq \emptyset$ and $S'\neq S$, Lemma~\ref{lem:overview_size_each_type1} shows that $\abs{F_{(S', S\cut S')}}$ has size $O(c^2)$. Since there are $O(\abs S)$ bipartitions in $\bipartitionsys$ such that $S'\neq \emptyset$ and $S'\neq S$, the union of these $F_{(S', S\cut S')}$ has size at most $O(\abs S c^2)$.

For $(\emptyset, S)\in \bipartitionsys$, Lemma~\ref{lem:overview_size_type2} shows that $F_{(\emptyset, S)}$ contains $O(\abs S c^2)$ edges.

Therefore, $\abs{F_1}$ is $O(\abs S c^2)$. The repair set $F_0\cup F_1$ contains at most $O(\abs S c^2)$ edges.
\end{proof}

\newpage

\def\dsinitialize{\textsc{DS-Initialize}}
\def\dsupdate{\textsc{DS-Update}}
\def\rti{t_{\mathrm{preprocess}}}
\def\rtu{t_{\mathrm{amortized}}}
\def\parametertimes{\xi}
\def\parameterlength{w}
\def\parametertimesub{\zeta}

\section{Decremental Cut Containment Set Update Algorithm}\label{sec:update_ia_set}

This section gives an efficient algorithm to update the cut containment set for graphs with good conductance ($n^{o(1)}$ conductance) in the  decremental update setting.


\paragraph{Cuts with one side of small volume} The main property we use for graphs with bounded conductance is that 
for a graph with conductance $\phi$, every cut of size at most $c$ has one side of volume at most $c / \phi$. When we update the cut containment set for graphs with good conductance, we only need to focus on cuts with one side of small volume.

For a set of terminals $T$ and a subset $T' \subset T$, a
cut $C$ for graph $G$ is a 
\emph{$(T', T \setminus T', t)$-cut} if 
$C$ partitions $T$ into $T'$ and $T\setminus T'$,
and the volume of the side that contains $T'$ is at most $t$. 
A cut $C$ is a \emph{minimum
$(T', T \setminus T', t)$-cut} if it has the smallest cut size among all the $(T', T \setminus T', t)$-cuts.
We use $mincut_G(T', T\setminus T', t)$ to denote the size of minimum $(T', T \setminus T', t)$-cut on graph $G$.

For a graph with conductance $\phi$ and a set of terminals $T$,
if a set of edges $F$ satisfies the condition that for every bipartition $(T', T\setminus T')$ of $T$ such that $mincut_G(T', T\setminus T, c / \phi) \leq c$, 
$F$ contains the cut-set of a minimum $(T', T\setminus T', c/\phi )$-cut, then $F$ is a \emph{$(T, c)$-cut containment set} of $G$.

Cut containment sets for graphs parameterized by conductance are also recursively constructed. 
In order to accommodate the change that we only consider cuts with one side of small volume,  we modify the definition of the $\IA$ set.

\paragraph{Redefine $\IA$ set} We define $\IA$ set for cuts with one side of small volume as follows.

\begin{definition}\label{def:ia_new}
Let $G = (V, E)$ be a connected graph, $T\subset V$ be a set of terminals, and $c, d, q, t$ be four integers such that $d \geq c > 0$ and $q \geq t > 0$. 
A set of edges $E'\subset E$ is an \emph{$\IA_G(T, t,q, d, c)$ set}
if 
for every 
subset $\emptyset \subsetneq T' \subset T$ such that $mincut_G(T', T\setminus T', t) \leq d$, 
there is a $(T', T\cut T', q)$-cut of size at most $mincut_G(T', T\setminus T', t)$ 
such that  every connected component of $G\setminus E'$ contains at most $\max\{mincut_G(T', T\setminus T', t) -  c, 0\}$ edges of the cut-set of $C$.

For a graph $G = (V, E)$ that is not necessarily connected, 
a set of edges $E' \subset E$ is an $\IA_G(T, t, q, d, c)$ set if 
for every $Q \subset V$ such that $Q$ forms a connected component of $G$, the edges of $E'$ in $G[Q]$ form an $\IA_{G[Q]}(T\cap V', t, q, d, c)$ set.

\end{definition}
For simplicity, we use $\IA_G(T, t, q, d, c)$ to denote an $\IA_G(T, t, q, d, c)$ set when there is no ambiguity.
Again, we ensure that $\IA$ sets under this new definition are intercluster edges of some vertex partitions of the graph so that each connected component of the graph after removing edges of an $\IA$ set is an induced subgraph on the same set of vertices.

By Definition~\ref{def:ia_new}, 
for a graph $G = (V, E)$ with conductance $\phi$ and a set of terminals $T$, 
$\IA_G(T, t, q, d, c)$ is a $(T, c)$-cut containment set of $G$ with respect to $T$ if $q \geq t \geq c /\phi$ and $d \geq c$.

We emphasize that this new definition is defined for arbitrary graphs, not just graphs parameterized by conductance, because when constructing the IA set recursively, the intermediate graph on which we want to construct $\IA$ set can be an arbitrary subgraph of the input graph, and thus the conductance is unbounded (see Figure~\ref{fig:size_construct_IA} for an example). 

\begin{figure}[htb]
\centering
\includegraphics[width=.4\linewidth]{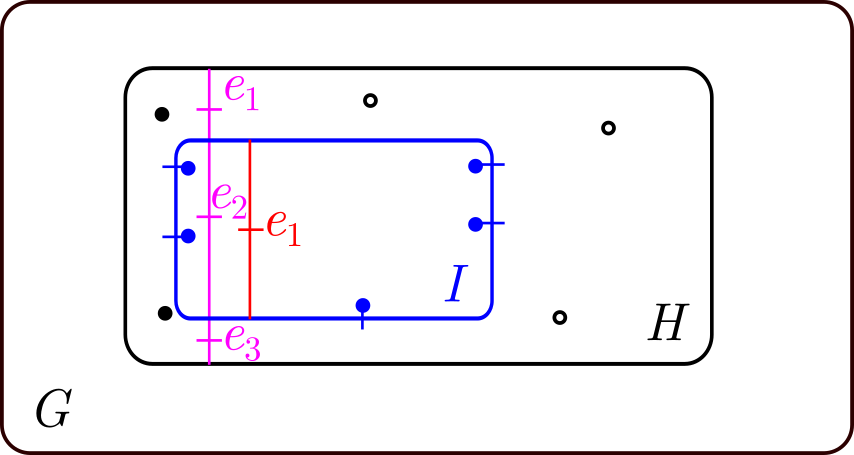}
\caption{$G$ is a graph with conductance $\phi$, and $H$ is an induced subgraph of $G$. The conductance of $H$ can be arbitrary. 
Suppose $T$ is the union of $T_1$ (black dots) and $T_2$ (black circles).
$C$ is a $(T_1, T_2, t)$-cut of size three on $H$ with cut-set $\{e_1, e_2, e_3\}$. 
$I$ is an induced subgraph of $H$. $\{e_2\}$ and $\{e_4\}$ induces cuts with left side of volume at most $t$ that partition terminals of $I$ in the same way. 
If an $\IA_H(\{t_1, t_2\}, t, q, d, 3)$ set is recursively constructed (via recursion on $I$), the $\IA$ set might contain edges $e_1, e_4, e_3$ but no $e_2$. Edges $\{e_1, e_4, e_3\}$ induce a cut with both sides having volume greater than $t$.
}
\label{fig:size_construct_IA}
\end{figure}

This is also the reason for having two additional parameters in the new IA set definition compared with the original definition of IA set (Definition~\ref{def:intersect-all}).
The new parameter $t$ is to capture the idea that we only consider cuts with one side of bounded volume (at most $t$). 
The purpose of the new parameter $q$ is for the recursive construction of IA set: 
Since our definition is for arbitrary graphs,
recursive construction of $\IA$ set 
cannot promise that 
for any terminal partition $(T', T\setminus T')$ such that $mincut(T', T\setminus T', t) \leq d$, 
there is a minimum $(T', T\setminus T', t)$-cut 
such that  every connected component of the graph after removing the IA set contains at most $\max\{mincut_G(T', T\setminus T', t) -  c, 0\}$ edges of the cut-set (see Figure~\ref{fig:size_construct_IA} for an example).
Instead, we promise that there is a $(T', T\setminus T', q)$-cut of cut size at most $mincut(T', T\setminus T', t)$ 
such that  every connected component of the graph after removing the IA set contains at most $\max\{mincut_G(T', T\setminus T', t) -  c, 0\}$ edges of the cut-set 
for some parameter $q$ that is still bounded as a function of $t$ and $c$.

The $\IA$ set in the new definition can also be constructed recursively if the parameters are set properly.
If $E_1, E_2, \dots, E_c$ are sets of edges satisfying the following ``recursive construction'' condition,
then $E_1 \cup \dots \cup E_c$ is an $\IA_G(T, t, q, d, c)$ set (Corollary~\ref{cor:ia_composition}).
\begin{framed}
\noindent \textbf{Recursive construction condition:}

\vspace{.2cm}\noindent Each $E_i$ is an 
\[\IA_{G\setminus (E_1 \cup \dots \cup E_{i-1})}(T \cup \endpoints(E_1 \cup \dots \cup E_{i-1}), t_i, q_i, d - i + 1, 1) 
\]
set for all $1 \leq i \leq c$, where $t_1, t_2, \dots, t_i, q_i$ satisfies \begin{equation}\label{equ:ia_composition_condition_old}t \leq t_1, q_c \cdot (d+1) \leq q, t_i \leq q_i \text{ for all } 1 \leq i \leq c, \text{ and  }  q_i \cdot (d+1) \leq t_{i+1} \text{ for all } 1 \leq i \leq c-1.\end{equation} 
\vspace{-.6cm}\end{framed}
\noindent 
Hence, to construct an $\IA_G(T, t, q, d, c)$ set, we can 
repeat the following procedure for $c$ iterations, we construct $E_i$ that is an 
\[\IA_{G\setminus (\cup_{j=1}^{i-1} E_j)}(T\cup \endpoints\wrap{\cup_{j=1}^{i-1}E_j}, t_i, q_i, d-i+1, 1)\] set. 
At the end, $E_1\cup E_2\cup \dots \cup E_c$ is a desirable $\IA_G(T, t, q, d, c)$ cut.

We say an $\IA(T, t, q, d, c)$ set is recursively constructed with respect to parameters $t_1, q_1, \dots, t_c, q_c$
if the $\IA$ set is partitioned into $E_1, \dots, E_c$ that satisfy the recursive construction condition. 

\paragraph{Decremental \text{IA} set update algorithm}
We consider the following decremental update scenario for $\IA$ set update: Let $G_0 = (V_0, E_0)$ be a connected graph, $T_0\subset V_0$ be a set of vertices,  and $V \subset V_0$ be 
a set of vertices  such that 
$G_0[V]$ is a connected graph.
We use $G= (V, E)$ to denote $G_0[V]$. Suppose $S = \endpoints(\partial_{G_0}(V)) \cap V$, and $T = (T_0\cap V) \setminus S$.
Given an $\IA_{G_0}(T_0, t, q, d, 2c + 1)$ set for $G_0$ which is recursively constructed, 
our goal is to obtain a small set of edges, called repair set, 
so that the union of the repair set and $\IA_{G_0}(T_0, t, q, d, 2c + 1)\cap E$ is an $\IA_{G}(S \cup T, t, q + t, c, 1)$ set for $G$.

As our main technical contribution, we give an algorithm to compute a repair set of size linear in $|S|$ with running time also linear in $|S|$.

\begin{lemma}\label{lem:repair_set_algorithm}
For an $\IA_{G_0}(T_0, t, q_0, d, 2c)$ set 
and an $\IA_{G_0}(T_0, t, q, d, 2c+1)$ set 
satisfying $q \geq 2t$ such that 
the $\IA_{G_0}(T_0, t, q, d, 2c+1)$ set
is derived from the $\IA_{G_0}(T_0, t, q_0, d, 2c)$ set,
given access to  $G, G\setminus (\IA_{G_0}(T_0, t, q_0, d, 2c)\cap E)$, and $G\setminus (\IA_{G_0}(T_0, t, q, d, 2c + 1)\cap E)$, 
there is an algorithm with running time $O(|S| (2q)^{c^2 + 3c+3}n^{o(1)})$ to find a repair set 
of size $|S|(16c^3 + 12c^2 + 2c)$ such that the union of \\$\IA_{G_0}(T_0, t, q, d, 2c+1) \cap E$ and  the  repair set is an $\IA_{G}(T \cup S, t, q + t, c, 1)$ set. 
\end{lemma}

To see the main challenge of proving Lemma~\ref{lem:repair_set_algorithm},
recall that one useful property to bound the number of maximal pairwise parallel atomic cuts Lemma~\ref{lem:find_parallel_cuts} of Section~\ref{sec:overview:cut_containment} is that for two parallel atomic cuts $C_1, C_2$  of size at most $c$ that have same partition on $S$, but different partition on $T$, 
the union of their cut-sets forms another cut $C$ of size at most $2c$ such that $S$ is on one side of the cut. Consequently, using the property of 
$\IA$ set, there is another cut $C^\diamond$ that partitions $T$ the same as $C$ such that the cut-set of $C^\diamond$ belongs to the $\IA$ set.

However, for two cuts $C_1$ and $C_2$ that both have one side with volume bounded by a parameter $t$, the union of $C_1$ and $C_2$'s cut-sets still forms a cut $C$. But $C$ might not have a side with volume bounded by $t$. See Figure~\ref{fig:simple_union} for an example. In this case we cannot guarantee the existence  of cut $C^\diamond$ that partitions $S$ the same as $C$ does with its cut-set contained in the $\IA$ set. Thus, the previous approach of finding maximal pairwise parallel atomic cuts in Section~\ref{sec:overview:cut_containment}  fails. 
\begin{figure}[htb]
    \centering
    \includegraphics[width=.25\linewidth]{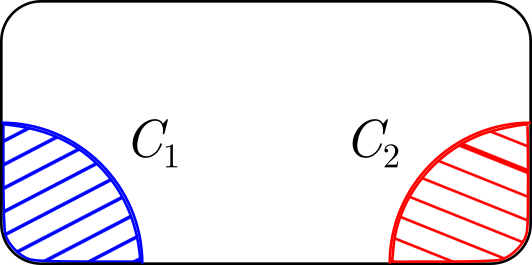}
    \caption{Both $C_1$ and $C_2$ have a side with volume bounded by $t$. The cut induced by union of $C_1$ and $C_2$ does not have a side with volume bounded by $t$.}
    \label{fig:simple_union}
\end{figure}


To overcome this barrier, 
we give an algorithm, called the \emph{\textbf{elimination procedure}}, which takes a set of cuts as input and selects a subset of these cuts such that the cut-set union of selected cuts shatters all the remaining given cuts (Section~\ref{sec:elimination}). 
We further show that with some additional restriction on the given set of cuts, 
the elimination procedure always selects a subset of bounded number cuts (Lemma~\ref{lem:procedure}).

Based on elimination procedure, we further show that we can carefully group the cuts that are not shattered by $\IA_{G_0}(T_0, t, q, d, 2c+1) \cap E$ into $O(|S|)$ groups so that the cuts in each group satisfy the additional restriction that enable efficient elimination procedure (Section~\ref{sec:ia_update}). 

\paragraph{Decremental cut containment set update algorithm}
Finally, we obtain a decremental cut containment set update algorithm. 
Given a cut containment set $\cc$ of $G_0$ with respect to a set of terminals $T_0$, 
our goal is to obtain a set $F$ of edges
such that the union of $F$ and the edges in $\cc\cap G$ is a cut containment set of $G$ with respect to $S\cup T$.

\begin{lemma}\label{lem:cut_containment_update}
Suppose in the decremental update setting, both $G_0 = (V_0, E_0)$ and $G = (V, E)$ have conductance at least $\phi$. 
Given three positive integers $c, t$ and $q$ such that $q \geq t \geq c/ \phi$,
a $(T_0, c^2 +2c)$-cut containment set $\cc$ of graph $G_0$ that is an $\IA_{G_0 }(T_0 ,t, q, c^2 + 2c, c^2+2c)$ recursively constructed from $E_1, E_2, \dots E_{c^2 + 2c}$ with respect to parameters $t_1, q_1, \dots, t_{c^2 + 2c},$ $ q_{c^2 + 2c}$. 
If parameters $t_1, q_1, \dots, t_{c^2 + 2c}, q_{c^2 + 2c}$ satisfy the following conditions, 
 \begin{equation}\begin{split}\label{equ:ia_composition_condition}& t \leq t_1,\quad
 q_{c^2 + 2c} ((c^2 + 2c)+2)^2 \leq q, \quad
  t_i \leq q_i \text{ for all } 1 \leq i \leq c^2 + 2c, \text{ and }\\ & \ \ \ \ \ \ \ q_i \cdot ((c^2 + 2c)+2)^2 \leq t_{i+1} \text{ for all } 1 \leq i \leq c^2 + 2c - 1,\end{split}\end{equation}
 then 
 given access to $G$ and $G_i = G\setminus (E_1 \cup \dots \cup E_i)$ for each $1 \leq i \leq c^2 + 2c$, 
there is an algorithm with running time 
\[O\wrap{\abs{S} (2q)^{O(c^2)} n^{o(1)}}\] that computes an edge set $F\subset E$ of size $|S|O(c)^{O(c)}$ such that $F \cup (\cc \cap E)$ is a recursively constructed $(S\cup T, c)$-cut containment set of $G$ with respect to parameters $t_1', q_1', \dots, t_{c}', q_c'$ satisfying 
\begin{equation}\label{equ:new_inequality}t_i' = t_{(i-1)\cdot(2c-i+3) + 1}, q_i' = q_{i\cdot (2c-i+2)} ((c^2 + 2c)+ 2) \text{ for all } 1 \leq i \leq c.\end{equation}
\end{lemma}

The cut containment set update algorithm is obtained by recursively applying the decremental $\IA$ set update algorithm. Given an $\IA_{G_0 }(T_0 ,t, q, c^2 + 2c, c^2+2c)$ set that is the union of $E_1, \ldots, E_{c^2 + 2c}$ satisfying the recursive construction condition, 
we can group these $c^2 + 2c$ sets into $c$ sets $E_i'$ such that $E_i'$ is an \[\IA_{G_0\cut \bigcup_{j=1}^{i-1} E_j'}\wrap{T_0\cup \endpoints\left(\cup_{j=1}^{i-1} E_j'\right), t_i', q_i' - t_i', c^2 + 2c - (i-1)\cdot (2c - i + 3), 2(c-i)+3 }\] set.

We then recursively compute a  repair set $F_i$ by Lemma~\ref{lem:repair_set_algorithm} for each $E_i'$ such that $F_i \cup (E_i' \cap E)$ is an 
\[\IA_{G_0\cut \bigcup_{j=1}^{i-1} \left(F_j \cup E_j'\right)}\wrap{T_0\cup \endpoints(\cup_{j=1}^{i-1} (E_j'\cup F_j)), t_i', q_i', c-i+1, 1}\] set. Then by Corollary~\ref{cor:ia_composition}, $\left(\cup_{i=1}^{c}(E_i' \cup F_i)\right)\cap E$ is a $(S\cup T, c)$-cut containment set. 
Note that when computing $F_i$, 
we only need to recurse on connected components of $G\setminus (\cup_{j=1}^{i-1}(F_j \cup E_j'))$ that contains at least one vertex of $S \cup \endpoints(\cup_{j=1}^{i-1} F_j)$. Thus $F_i$ is of size $|S|O(c)^{3i}$ by Lemma~\ref{lem:repair_set_algorithm}, and thus $F$ contains at most $|S|O(c)^{O(c)}$ edges.  \\



In the rest of this section, we prove some properties of new $\IA$ set definition in Section~\ref{sec:IA_set_property}. 
Section~\ref{sec:simple_cut} defines simple cut, a generalization of atomic cut, and characterizes the cuts that need to be shattered by the repair set. 
In Section~\ref{sec:elimination}, we give the elimination procedure. 
In Section~\ref{sec:ia_update}, we present our decremental $\IA$ set update algorithm. 
In Section~\ref{sec:subsec_cut_containment_udpate}, we present our  decremental cut containment set update algorithm.

\subsection{\texorpdfstring{$\IA$}{IA} Set Properties}\label{sec:IA_set_property}
We formally define cuts with one side of bounded volume. 
\begin{definition}[cut with one side of bounded volume]\label{def:intersect-lps}



For a set of terminals $T$ and a subset $T' \subset T$, a
cut $C$ for graph $G$ is a 
\emph{$(T', T \setminus T', t)$-cut} if 
$C$ partitions $T$ into $T'$ and $T\setminus T'$,
and the side that contains $T'$ has a volume of at most $t$. 

A cut $C$ is a \emph{minimum
$(T', T \setminus T', t)$-cut} if it has the smallest cut size among all the $(T', T \setminus T', t)$-cuts.
We use $mincut_G(T', T\setminus T', t)$ to denote the size of minimum $(T', T \setminus T', t)$-cut for graph $G$.
\end{definition}

\begin{remark}
We remark that 
this definition is not symmetric with respect to $T'$ and $T\setminus T'$, i.e., a $(T',  T\setminus T', t)$-cut might not be a $(T\setminus T', T', t)$-cut.
\end{remark}

By Definition~\ref{def:intersect-lps}, we have the following fact relating cut containment set to cuts with one side having small volume.

\begin{fact}\label{fact:ca_small_vol}
Let $G = (V, E)$ be a connected graph of conductance $\phi$, and $T\subset V$ be a set of terminals, and $c$ be a positive integer. 
If $F\subset E$ is a set of edges such that for every $\emptyset \subsetneq T' \subsetneq T$ such that $mincut_G(T', T\setminus T, c / \phi) \leq c$, there exists a $(T', T\setminus T')$-cut $C$ of size $mincut_G(T', T\setminus T, c / \phi)$ such that $C$'s cut-set is a subset of $F$, then $F$ is a cut containment set of $G$ with respect to $T$. 

\end{fact}

In the rest of this section, we prove the following properties of $\IA$ sets under the new definition.

\begin{enumerate}
    \item $\IA$ set with appropriate parameters implies cut containment set (Claim~\ref{claim:terimial_not_in_same_cc_new}). 
    \item $\IA$ set can be recursively constructed similar to \cite{chalermsook2020vertex} with appropriate parameters (Lemma~\ref{lem:ia_composition} and Corollary~\ref{cor:ia_composition}). 
\end{enumerate}

By Definition~\ref{def:cut-containment-set}, Fact~\ref{fact:ca_small_vol} and Definition~\ref{def:ia_new}, an $\IA$ set with appropriate parameters for a graph induces a cut containment set. 

\begin{claim}\label{claim:terimial_not_in_same_cc_new}
    Let $G = (V, E)$ be a connected graph with conductance $\phi$,  $T \subset V$ be a set of terminals, and $c$ be a positive integer. For any positive integers $q, t, d$ satisfying $q\geq t \geq c / \phi$ and $d \geq c$, an arbitrary $\IA_G(T, t,q, d, c)$ set is a $(T, c)$-cut containment set of $G$.


\end{claim}
\begin{proof}
Let $(T', T\setminus T')$ be an arbitrary bipartition of $T$ such that the minimum cut that partitions $T$ into $T'$ and $T\setminus T'$ is of size at most $c$. 
Let $C = (V', V\setminus V')$ be such a minimum cut, and $\alpha$ denote the size of $C$. 

Since the conductance of $G$ is $\phi$, we have 
\[\min\{\vol_G(V'), \vol_G(V\setminus V')\} \geq c / \phi.\]
Without loss of generality, we assume the side containing $T'$ has a volume of at most $c / \phi$. 
Since $\alpha$ is the size of minimum cut that partitions $T$ into $T'$ and $T\setminus T'$, we have 
\[ mincut_G(T', T\setminus T', c / \phi) \geq  mincut_G(T', T\setminus T',  t) \geq  mincut_G(T', T\setminus T', q) \geq \alpha.\]
On the other hand, $C$ is a $(T', T\setminus T', c/\phi)$-cut of size $\alpha$,  so we have  
\[mincut_G(T', T\setminus T', c / \phi) =  mincut_G(T', T\setminus T', t) =  mincut_G(T', T\setminus T', q) = \alpha.\]

By Definition~\ref{def:ia_new}, there is a $(T', T\setminus T', q)$-cut $C'$ of size at most  $\alpha$
such that  every connected component of  $G\setminus \IA_G(T, t, q, d, c)$ contains at most $\max\{\alpha -  c, 0\}$ edges of the cut-set of $C'$. 
Since $\alpha \leq c$,
it means that the cut-set of $C'$ is a subset of the $\IA_G(T, t, q, d, c)$ set.

Above argument works for any bipartition of $T$ that partitions $T$ into two non-empty subsets. By Definition~\ref{def:cut-containment-set} and Fact~\ref{fact:ca_small_vol}, the lemma follows. 
\end{proof}

Similar to the original $\IA$ set, we  show that the new $\IA$ set can also be constructed recursively by composing two $\IA$ sets with appropriate parameters. 

\begin{lemma}\label{lem:ia_composition}
Let $G = (V, E)$ be a graph, $T\subset V$ be a terminal set,
$E'$ be an $\IA_G(T, t_1, q_1, d, c_1)$ set, 
and $E''$ be a subset of $E\setminus E'$
such that $E''$
is an $\IA_{G \setminus E'}(T \cup \endpoints(E'), t_2, q_2, d - c_1, c_2)$ set. 
If $q_2 \geq t_2 \geq q_1 \geq t_1$, then 
 $E' \cup E''$ is an $\IA_G(T, t_1, q_2\cdot (d+1), d, c_1 + c_2)$ set.
\end{lemma}

\begin{proof}
Without loss of generality, we assme $G$ is a connected graph. 
For any $\emptyset \subsetneq T' \subset T$ 
such that there is a $(T', T\cut T', t_1)$-cut of size at most $d$ on $G$. We show that there exists a $(T', T\cut T', q_2\cdot (d+1))$-cut $(V^\dagger, V\cut V^\dagger)$ of size at most $mincut_G(T', T\setminus T', t_1)$
such that 
every connected component of $G\setminus (E' \cup E'')$ contains at most \[\max\{mincut_G(T', T\setminus T', t_1) - c_1 - c_2, 0\}\] edges of $\partial_G(V^\dagger)$.
Then by Definition~\ref{def:ia_new}, the lemma follows.


Let $\alpha = mincut_G(T', T\setminus T', t_1)$. We have $\alpha \leq d$.
If $\alpha \leq c_1$, then by definition, 
there is a $(T', T\setminus T', q_1)$-cut of size at most $\alpha$ whose cut-set is in $E'$, and we are done. 
In the rest of this proof,  we consider the case of $\alpha > c_1$.

By definition, there is a $(T', T\cut T', q_1)$-cut $(V', V\cut V')$ of size at most $\alpha$ such that in $G\cut E'$, every connected component contains at most $\max\{\alpha - c_1, 0\}$ edges of $\partial_G(V')$. 
Let $H = (V_H, E_H)$ be a connected component of $G\setminus E'$ that contains at least one edge of $\partial_G(V')$. 
Let 
\[F_H = \partial_G(V')\cap E_H, 
T_H = (T  \cup \endpoints(\partial_G(V_H)))\cap V_H, \text{ and } T'_H = T_H\cap V'.\] 
$(V'\cap V_H, (V\setminus V') \cap V_H)$ is a $(T_H', T_H \setminus T_H', q_1)$-cut on $H$ with cut-set $F_H$. 
Hence, $|F_H| \leq \alpha - c_1$.

Since $t_2 \geq q_1$, 
by definition of $E''$, there exists a $(T_H', T_H \setminus T_H', q_2)$-cut of size at most $|F_H|$ on graph $H$, denoted by $(V_H', V_H\cut V_H')$, 
such that every connected component of $H \cut (E''\cap E_H)$ contains at most $\max\{|F_H|-c_2, 0\}$ edges of $\partial_{H}(V_H')$. 

Now we look at $G\cut (E'\cup E'')$. 
Each connected component of $G\cut (E'\cup E'')$ is a connected component of $H\cut (E'' \cap E_H)$ for some $H$, which is a connected component of $G\setminus E'$. Let 
\begin{align*}V^\dagger = \left(\bigcup_{\substack{H = (V_H, E_H):V' \cap V_H \neq \emptyset, (V\setminus V') \cap V_H \neq \emptyset, \\ H \text{ is a connected component of } G \setminus E'}} V_H' \right)
\bigcup \left(\bigcup_{\substack{H = (V_H, E_H): V_H \subset V', \\ H \text{ is a connected component of } G \setminus E'}} V_H\right).\end{align*}
By Lemma~\ref{lem:cut_subgraph_replace}, $(V^\dagger, V\setminus V^\dagger)$ is a cut partitioning $T$ into $T'$ and $T\setminus T'$  such that 
each connected component of $G\cut (E'\cup E'')$  contains at most $\max\{\alpha -c_1-c_2, 0\}$ edges of $\partial_{G}(V^\dagger)$.
Since $\abs{\partial_G(V')}\leq \alpha$, there are at most $\alpha$ connected components of $G\setminus (E' \cup E'')$ that contain edges of $\partial_G(V')$. 
Since each $\vol_H(V_H')\leq q_2$ for any connected component $H$ of $G\setminus E'$
satisfying $V'\cap V_H\neq \emptyset$ and $(V\setminus V') \cap V_H \neq \emptyset$,
and 
$\vol_G(V') \leq q_1$
we have
\[\vol_G(V^\dagger)\leq 
\vol_G(V') + \left(\sum_{\substack{H:V' \cap V_H \neq \emptyset, (V\setminus V') \cap V_H \neq \emptyset, \\ H \text{ is a connected component of } G \setminus E'}} \vol_H(V_H') \right)\leq  
q_1 + q_2 \cdot \alpha \leq q_2 \cdot (d + 1).\]
Hence, $(V^\dagger, V\cut V^\dagger)$ is a $(T', T\cut T', q_2 \cdot (d+1))$-cut of size most $\alpha$ on $G$ 
such that 
each connected component of $G\cut (E'\cup E'')$  contains at most $\max\{\alpha -c_1-c_2, 0\}$ edges of $\partial_{G}(V^\dagger)$. 
\end{proof}

Applying Lemma~\ref{lem:ia_composition} recursively,
we have the following analogy of Equation~(\ref{equ:recursive_IA}) for our new IA definition. 

\begin{corollary}\label{cor:ia_composition}
Let $G = (V, E)$ be a connected graph and $T\subset V$ be a set of vertices.
For two integers $d \geq c$, 
let $t_1, q_1, t_2, t_2, \dots, t_c, q_c$ be $2c$ integers satisfying  
\[t_i \leq q_i \text{ for any } 1 \leq i \leq c, \text{ and } (d+1) \cdot q_i \leq t_{i+1} \text{ for any  } 1 \leq i \leq c-1,\]
and $E_1, \dots E_c$ be edge sets such that for any $1 \leq i \leq c$,
$E_i$ is an 
\[\IA_{G\setminus \wrap{\bigcup_{j=1}^{i-1} E_j}} \wrap{T \cup \endpoints\wrap{\textstyle{\bigcup}_{j=1}^{i-1}E_j}, t_i, q_i, d- i + 1, 1}\] set,
then $\bigcup_{j = 1}^c E_j$ is an $\IA_{G}(T, t_1, q_c \cdot (d+1), d, c)$ set.
\end{corollary}

We say an $\IA_G(T, t, q, d, c)$ set is \emph{derived from} an $\IA_G(T, t, q', d, c')$ set for some $q \geq q' (d+1), c > c'$ 
if
\[\IA_G(T, t, q', d, c') \subset \IA_G(T, t, q, d, c),\] and  $\IA_G(T, t, q, d, c) \setminus \IA_G(T, t, q', d, c')$ is an
\[\IA_{G\setminus \IA_G(T, t, q', d, c')}(T \cup \endpoints(\IA_G(T, t, q', d, c')), q', q / (d+1), d - c', c - c')\]
set.



\subsection{Simple Cut, Characterization of Cuts To Be Shattered}\label{sec:simple_cut}

Since a cut with one side having small volume might not be an atomic cut (see Figure \ref{fig:simplecut}), we define simple cut, an analogy of atomic cut for cuts with one side having bounded volume.

\begin{definition}\label{def:simple_cut_new}
A cut $(V', V\setminus V')$ is a \emph{simple} cut if $G[V']$ is connected.

\end{definition}

We remark that the definition of simple cut is not symmetric, i.e.,  $(V\setminus V', V')$ might not be a simple cut if $(V', V\setminus V')$ is a simple cut, though $(V', V\setminus V')$ and $(V\setminus V', V')$ have the same cut-set. 

A simple cut is not necessarily an atomic cut. But the cut-set of a simple cut can be partitioned into a few edge disjoint sets so that each set corresponds to the cut-set of an atomic cut, which is not necessarily a simple cut. 
Such a partition is unique. So, if the cut-set of atomic cut $C$ is in the cut-set of cut $C'$, then $C$ is called a composing atomic cut of $C'$.

\begin{figure}
    \centering
    \includegraphics[width=.3\linewidth]{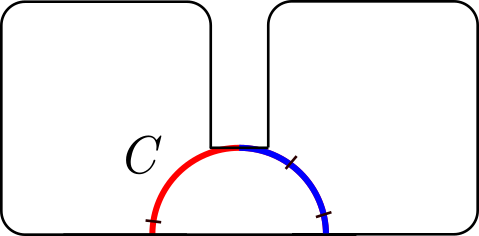}
    \caption{Cut $C$ is a non-atomic cut with one side  having small volume.
    The cut-set of $C$ is the union of cut-sets of two atomic cuts, both do not have a side with small volume. }
    \label{fig:simplecut}
\end{figure}

\begin{definition}
For a terminal set $T$ and a positive integer $t$, 
a cut $C$ for graph $G$ is a 
\emph{simple $(T', T \setminus T', t)$-cut} if $C$ is a simple cut and also a $(T', T \setminus T', t)$-cut. 
A cut $C$ is a \emph{simple minimum
$(T', T \setminus T', t)$-cut} if
$C$ is a simple
$(T', T \setminus T', t)$-cut of size $mincut_G(T', T\setminus T', t)$.


\end{definition}

The following lemma states that if a set of edges satisfy the $\IA$ set definition for bipartitions that have simple minimum cuts, then it is an IA set. Therefore, we can focus our attention on simple cuts.

\begin{lemma}\label{lem:simple_intercept}
Let $G = (V, E)$ be a connected graph, and $T\subset V$ be a set of vertices.
Let $E'$ be a set of edges such that for every $\emptyset \subsetneq T'\subset T$ satisfying the following two conditions:
\begin{enumerate}
    \item $mincut_G(T', T \setminus T', t) \leq c$ 
    \item there is a simple minimum $(T', T \setminus T', t)$-cut
\end{enumerate}
there is a $(T', T\setminus T', q)$-cut of size at most $mincut_G(T', T \setminus T', t)$ that is shattered by $E'$.
Then $E'$ is an $\IA_{G}(T, t, q +t, c, 1)$ set.
\end{lemma}
\begin{proof}
For a $\emptyset\subsetneq T_1 \subset T$ such that 
there is a $(T_1, T\cut T_1, t)$-cut of size at most $c$, 
let $(V_1, V\cut V_1)$ be a minimum $(T_1, T\cut T_1, t)$-cut. 

Let $V^\diamond$ be a subset of $V_1$ that forms a connected component of $G[V_1]$.
$(V^\diamond, V\setminus V^\diamond)$ is 
a  simple cut, and is also a minimum $(V^\diamond \cap T, T\setminus (V^\diamond \cap T), \vol_G(V^\diamond))$-cut satisfying $ V^\diamond \cap T \neq \emptyset$, otherwise, there is a $(T_1, T\cut T_1, t)$-cut with size smaller than $(V_1, V\setminus V_1)$.
On the other hand, since $(V^\diamond, V\setminus V^\diamond)$ is of size at most $c$,
we have 
\[c \geq mincut_G(V^\diamond\cap T, T\setminus (V^\diamond \cap T), \vol_G(V^\diamond)) \geq mincut_G(V^\diamond\cap T, T\setminus (V^\diamond \cap T), t).\]
By the definition of $E'$, 
there is a $(V^\diamond\cap T, T\setminus (V^\diamond \cap T), q)$-cut $(V^\dagger, V\setminus V^\dagger)$ of size at most $(V', V\setminus V')$ 
such that 
every connected component of $G\setminus E'$ contains at most $mincut_G(V^\diamond\cap T, T\setminus (V^\diamond \cap T), t)- 1$ edges of $\partial_G(V^\dagger)$.

Since it is possible for every connected component of $G\setminus \left((\partial_G(V_1) \setminus \partial_G(V^\diamond)) \cup \partial_G(V^\dagger)\right)$ to contain a vertex from $T_1$ and another vertex from $T\setminus T_1$, 
 there is a subset $E^\ddagger$ of $(\partial_G(V_1) \setminus \partial_G(V^\diamond)) \cup \partial_G(V^\dagger)$ corresponding to the cut-set of a $(T_1, T\setminus T_1, q + t)$-cut, whose size is at most $mincut_G(T_1, T\setminus T_1, t)$. 

If $|E^\ddagger| < mincut_G(T_1, T\setminus T_1, t)$, then
there is a $(T_1, T\setminus T_1, q+t)$-cut in graph $G$ such that every connected component of $G\setminus E'$ contains at most $mincut_G(T_1, T\setminus T_1, t) - 1$ edges of the cut-set. 
Otherwise, $|\partial_G(V^\diamond)| = mincut_G(T_1, T\setminus T_1, t)$ and $E^\ddagger$ contains all the edges of $(\partial_G(V_1) \setminus \partial_G(V^\diamond)) \cup \partial_G(V^\dagger)$.
Since no connected component of $G\setminus E'$ contains all the edges of $\partial_G(V^\dagger)$, 
no connected component of $G\setminus E'$ contains all the edges of $E^\ddagger$.
\end{proof}

The following lemma shows that we can focus our attention on the cuts that are contained in one connected component of $G\cut (\IA(T_0, t, q_0, d, 2c) \cap E)$ which contains at least one vertex from $S$. 


\begin{lemma}\label{lem:exclude_Case}
In the decremental $\IA$ set update scenario, suppose 
given 
an $\IA_{G_0}(T_0, t, q_0, d,  2c)$ set, an $\IA_{G_0}(T_0, t, q, d,  2c+1)$
set that is derived from the $\IA_{G_0}(T_0, t, q_0, d,  2c)$ set for some parameters $q \geq q_0 \geq t$ and $d \geq 2c+1$.
If $U$ is a subset of $S\cup T$ satisfying $U\neq \emptyset$ and $mincut_G(U, (S\cup T)\setminus U, t) \leq c$ such that 
there is a $(U, (S\cup T) \setminus U, q)$-cut $(V', V\setminus V')$ of size at most $mincut_G(U, (S\cup T)\setminus U, t)$ satisfying one of the following two conditions
\begin{enumerate}
    \item $(V', V\setminus V')$ is shattered by $\IA_{G_0}(T_0, t, q_0, d,  2c) \cap E$.
    \item  $\partial_G(V')$ is contained in a connected component of $G\setminus \wrap{\IA_{G_0}(T_0, t, q_0, d,  2c) \cap E}$ that does not have any vertex from $S$.
\end{enumerate}
then there is a $(U, (S\cup T) \setminus U, q)$-cut $C$ of size at most $mincut_G(U, (S\cup T)\setminus U, t)$ such that every connected component of $G\cut (\IA_{G_0}(T_0, t, q, d,  2c + 1)\cap E)$ contains at most $mincut_G(U, (S\cup T)\setminus U, t)-1$ edges of the cut-set of $C$.

\end{lemma}

%
%
%
%
%

\begin{proof}

Since the $\IA_{G_0}(T_0, t, q, d, 2c+1)$ set is derived from the $\IA_{G_0}(T_0, t, q_0, d, 2c)$ set, we have \[\IA_{G_0}(T_0, t, q_0, d, 2c) \subset \IA_{G_0}(T_0, t, q, d, 2c+1).\] 
The first condition implies that $\IA_{G_0}(T_0, t, q, d, 2c+1)\cap E$ shatters $(V', V\setminus V')$, and thus the lemma holds.

Now consider the second condition of the lemma.
Suppose $\partial_G(V')$ is contained in a connected component of  $G\cut (\IA_{G_0}(T_0, t, q_0, d, 2c) \cap E)$ satisfying $P\cap S = \emptyset$. 
Let $G[P] = (P, E_P)$ denote this connected component, where $E_P$ is the set of edges from $E$ that have both endpoints in $P$. 
$\partial_G(V')$ induces a cut of size at most $c$ with one side having volume at most $t$ on $G[P]$ that partition $T^\dagger =  (T \cup \endpoints(\partial_G(P)))\cap P$
into $T^\diamond = T^\dagger \cap (V'\cap P)$
and $T^\dagger \setminus T^\diamond = T^\dagger \cap ((V\setminus V') \cap P)$.
Note that $T^\diamond$ is not an empty set, otherwise $V' \cap (S\cup T) = \emptyset$.
Hence, $\partial_G(V')$ induces a $(T^\diamond, T^\dagger \setminus T^\diamond, t)$-cut for $G[P]$ of size at most $mincut_G(U, (S\cup T)\setminus U, t)$.

On the other hand, 
since the $\IA_{G_0}(T_0, t, q, d, 2c+1)$ set is derived from the $\IA_{G_0}(T_0, t, q_0, d, 2c)$ set, 
we have that 
$\IA_{G_0}(T_0, t, q, d, 2c+1) \cap E_P$  is an 
\[\IA_{G[P]}(T^\dagger, q_0, q / (d+1), d - 2c, 1)\] set for $G[P]$. 
By Definition~\ref{def:ia_new}, 
there is a $(T^\diamond, T^\dagger \setminus T^\diamond, q / (d+1))$-cut $C'$ of size at most $mincut_G(U, (S\cup T)\setminus U, t)$ 
in graph $G[P]$ such that 
every connected component of $G[P] \setminus (\IA_{G_0}(T_0, t, q, d, 2c+1)\cap E_P)$ contains at most $mincut_G(U, (S\cup T)\setminus U, t) - 1$ edges of the cut-set of $C'$. 

By Lemma~\ref{lem:cut_subgraph_replace}, 
the cut-set of $C'$ for graph $G[P]$
induces a $(U, (S\cup T) \setminus U, t + q / (d+1))$-cut $C^\dagger$ of size at most $mincut_G(U, (S\cup T)\setminus U, t)$ for graph $G$.
Since $t + q / (d+1) \leq q$, 
$C'$ is a $(U, (S\cup T) \setminus U, q)$-cut of size at most $mincut_G(U, (S\cup T)\setminus U, t)$ such that every connected component of $G\setminus (\IA_{G_0}(T_0, t, q, d, 2c+1) \cap E)$ contains at most $mincut_G(U, (S\cup T)\setminus U, t) - 1$ edges of the cut-set of $C'$.
\end{proof}



\subsection{Basic Graph Operations}
Before we present our algorithm, we discuss the basic operations supported by the data structure. 
The detailed discussion of graph data structure we used is in Appendix~\ref{sec:data_structure}. 

We assume the graph data structure maintains a spanning forest. Based on recent development on fully dynamic spanning forest~\cite{nanongkai2017dynamica, nanongkai2017dynamic, wulff2017fully, chuzhoy2019deterministic}, each of the following basic operations can be performed in $n^{o(1)}$ time.  

\begin{enumerate}
    \item insert an edge to the graph;
    \item delete an edge from the graph;
    \item given a vertex, return the id of the connected component that contains the given vertex;
    \item 
    if the graph is associated with a set of terminals, given a vertex, return the number of terminals in the connected component that contains the given vertex. 
\end{enumerate}
Based on these basic operations, we have some useful subroutines that will be used in our algorithm. 

\begin{lemma} \label{lem:subroutine_repair_set}
Given a graph $G$, we have the following subroutines:
\begin{enumerate}
    \item Given an edge set $E_0$ of graph $G$, there is an algorithm to determine if $E_0$ is the cut-set of an atomic cut with running time $O(|E_0| n^{o(1)})$. 
    \item Given vertex $x$ of $G$ and positive integers $c$ and $t$, 
    there is 
    an algorithm with running time $O(t^{c+2} n^{o(1)})$ to enumerate all the 
    simple cuts $(V', V\setminus V')$ of size at most $c$ satisfying $\vol_G(V') \leq t$ and $x \in V'$.
    Furthermore, the number of simple cuts  enumerated is at most $t^c$. 
    
    \item Given a set $U \subset S\cup T$ and two integers $t$ and $c$,
    there is 
    an algorithm with running time $O(t^{c^2 + c + 1} n^{o(1)})$ to enumerate all the 
    $(U, (S\cup T)\setminus U, t)$-cuts of size at most $c$. 
    Furthermore, the number of simple cuts  enumerated is at most $O(t^{c(c+1)})$. 

    \item Given two edge sets $E_0$ and $E_1$ each corresponding to the cut-set of an atomic cut, there is an algorithm to determine if the bipartitions on $S$ induced by $E_0$ is the same as the bipartition on $S$ induced by $E_1$ in $O((|E_0|+ |E_1|) n^{o(1)})$ time. 
    \item 
    Given two edge sets $E_0$ and $E_1$ each corresponding to the cut-set of an atomic cut, there is an algorithm to determine if the bipartitions on $S$ induced by $E_0$ is parallel with the bipartition on $S$ induced by $E_1$ in $O((|E_0|+ |E_1|) n^{o(1)})$ time. 
    
\end{enumerate}
\end{lemma}
The proof of Lemma~\ref{lem:subroutine_repair_set} and the detailed implementations of these subroutine are given in Appendix~\ref{sec:implementation}. \\

\subsection{Elimination Procedure}\label{sec:elimination}

Recall that  Section~\ref{sec:overview:cut_containment} showed that to shatter a given set of atomic cuts, 
it is sufficient to find a maximal subset of pairwise parallel atomic cuts from the given set of cuts, and take the cut-sets of the selected cuts. 
Furthermore, if the given atomic cuts satisfying the following conditions, then an arbitrary maximal set of pairwise parallel atomic cuts contains at most $O(c)$ cuts (Lemma~\ref{lem:find_parallel_cuts}). 
\begin{enumerate}
    \item[(S1)] The given cuts induce the same bipartition on $S$, but pairwise different bipartitions on $T$. 
    \item[(S2)] 
    There is a connected component of $G \cut (\IA_{G_0}(T_0, d, 2c)\cap E)$ contains the cut-sets of all the given atomic cuts, and also contains a vertex from $S$.
\end{enumerate}

In this section, new challenges arise since only cuts with one side of small volume are considered. The union of cut-sets of two atomic cuts both having one side with volume at most $t$ corresponds to a cut that might does not have a side with volume bounded by $t$.
To fix this problem, we show that with some additional properties, 
Lemma~\ref{lem:find_parallel_cuts} can be generalized to the new definition of $\IA$ set. 

First, we assume the cuts are given by realizable pairs (Definition~\ref{def:realizable_pair}), i.e. a pair that contains a simple cut, and a specific composing atomic cut of the simple cut.  
Both simple cuts and specific composing atomic cuts are crucial to our algorithm:
simple cuts define the cuts we want to shatter, 
and the atomic cuts are used to bound the size of cuts selected to shatter all the given realizable pairs. 

Second, instead of finding an arbitrary maximal set of pairwise parallel atomic cuts as Section~\ref{sec:overview:cut_containment},
we design a new algorithm, called the elimination procedure, to carefully select a subset of a given set of cuts to shatter the remaining cuts. 

We show that if the given cuts satisfy some additional properties, then the elimination procedure finds a subset of $O(c^2)$ cuts such whose cut-sets shatter all the given cuts (Lemma~\ref{lem:procedure}). \\

We define realizable pair as follows. 

\begin{definition}\label{def:realizable_pair}

For a connected graph $G = (V, E)$, 
a pair $(E', V')$ for some $V' \subset V$, $E' \subset E$ is a $(t, c)$-realizable pair if 
\begin{enumerate}
        \item 
        $(V', V \setminus V')$ is a simple $(t, c)$-cut. 
        \item 
        $E' \subset \partial_G(V')$
        is the cut-set of an atomic cut. 
\end{enumerate}




\end{definition}

We first give the elimination procedure. This will be the main tool to construct the repair set. 

\begin{framed}
\noindent \textbf{Elimination Procedure}

\noindent \textbf{Input:} Access to graph $G$, two vertex sets $S$, $T$, and a set $\Psi$ of $(t, c)$-realizable pairs satisfying the  condition that there is a vertex $x$ such that for every $(E', V') \in \Psi$, the atomic cut induced by $E'$ separates $V'$ and $x$.

\noindent     \textbf{Output:} A set of edges $W$.  
    
\begin{enumerate}    
\item Set $W = \emptyset$. 
\item Repeat the following process until $\Psi$ is an empty set and return $W$.
\begin{enumerate}
    \item Take a realizable pair $(E', V') \in \Psi$ such that there is no $(E'', V'') \in \Psi$ satisfying $\endpoints(E'')\subset V'$. 
    \item Put $\partial_G(V')$ into $W$.
    
    \item 
    Remove each $(E'', V'')$ from $\Psi$ if $W$ shatters a $$(V'' \cap (S\cup T), (V\setminus V'') \cap (S\cup T), \vol_G(V''))$$cut of size smaller than or equal to the size of $(V'',V\setminus V'')$. 
\end{enumerate}
\end{enumerate}
\vspace{-.5cm}\end{framed}

We show that the elimination procedure always terminates in $O(|\Psi|^2 t^{(c+1)c + 1} \poly(c)n^{o(1)})$ time for a given set of realizable pairs $\Psi$, and the output of the elimination procedure shatters a 
$(V' \cap (S\cup T), (V\setminus V')\cap (S\cup T), \vol_G(V'))$-cut of size smaller than or equal to the size of $(V', V\setminus V')$ for each $(E', V') \in \Psi$. 

\begin{lemma}\label{lem:elimination_procedure_correctness}
Let $G$ be a graph of $n$ vertices, and $\Psi$ be a set of $(t, c)$-realizable pairs for positive integers $t, c$ such that $c = (\log n)^{o(1)}$. If there is a vertex $x$ of $G$ such that for every $(E', V') \in \Psi$, the atomic cut induced by $E'$ separates $V'$ and $x$,
then the elimination procedure outputs a set of edges $W$ such that for each $(V', V\setminus V')\in \Psi$, $W$ shatter a $(V' \cap (S\cup T), (V\setminus V')\cap (S\cup T), \vol_G(V'))$-cut of size smaller than or equal to the size of $(V', V\setminus V')$.
The running time of the elimination procedure is $O(|\Psi|^2 t^{(c+1)c + 1} n^{o(1)})$. 
\end{lemma}
\begin{proof}
We first show that the algorithm always terminates. 
Let $(E', V'), (E'', V'') \in \Psi$ be two realizable pairs.
Let $(L', V\setminus L')$ and $(L'', V\setminus L'')$ be the atomic cuts induced by $E'$ and $E''$ respectively such that $V' \subseteq L'$ and $V''\subseteq L''$. 
Since $(L', V\setminus L')$ separates $V'$ and $x$, and $(L'', V\setminus L'')$ separates $V''$ and $x$,
if $E''$ is in $G[V']$, then $L''$ is a strict subset of $L'$. 
So Step (a) of the elimination procedure always finds a $(t, c)$-realizable pair in $\Psi$ if $\Psi$ is not an empty set. 
Since one iteration of Step (2) always removes at least one realizable pair from $\Psi$, the elimination procedure always terminates.

By the definition of the elimination procedure, a realizable pair $(E', V')$ is removed from $\Psi$ only if a $(V' \cap (S\cup T), (V\setminus V')\cap (S\cup T), \vol_G(V'))$-cut of size smaller than or equal to the size of $(V', V\setminus V'')$ is shattered by $W$. Hence, the output of the elimination procedure shatters a $(V' \cap (S\cup T), (V\setminus V')\cap (S\cup T), \vol_G(V'))$-cut of size smaller than or equal to the size of $(V', V\setminus V'')$ for each $(E', V')$ in $\Psi$ initially.

Now we bound the running time of the algorithm.
Note that for two $(t, c)$-realizable pairs $(E', V')$ and $(E'', V'')$, it takes $O(\poly(c)n^{o(1)})$ time to check if $\endpoints(E'') \subseteq V'$ by Lemma~\ref{lem:subroutine_repair_set}. Hence one execution of Step (2a) of the elimination procedure can be implemented in $O(|\Psi|\poly(c)n^{o(1)})$ time. 
By Lemma~\ref{lem:subroutine_repair_set},
for a $(t, c)$-realizable pairs $(E', V')$, $(V'\cap (S\cup T), (V\setminus V') \cap (S\cup T), \vol_G(V'))$-cuts can be enumerated in $O(t^{c(c+1) + 1}\poly(c)n^{o(1)})$ time, and thus one execution of Step (2c) of the elimination procedure can be implemented in $O(|\Psi|O(t^{c(c+1) + 1}n^{o(1)})$ time. 
Since Step (2) iterates at most $|\Psi|$ iterations, the overall running time of the elimination procedure is $O(|\Psi|^2 t^{(c+1)c + 1} n^{o(1)})$. 
\end{proof}

We show that if the given realizable pairs satisfying the following two properties in addition to (S1) and (S2), then the elimination procedure always selects a set of $O(c^2)$ pairs.
\begin{enumerate}
    \item[(S3)] There is a vertex $x \in V$ such that for each realizable pair, the cut induced by the composing atomic cut separates vertex $x$ and the simple cut's side with small volume. 
    \item[(S4)] For any two realizable pairs, their simple cuts have overlap on the side with small volume. 
\end{enumerate}



\begin{lemma}\label{lem:procedure}
In the decremental update setting, for any $\IA_{G_0}(T_0,t,q_0, d, 2c)$ set and a set of vertices  $V^\star \subset V$ that forms a connected component of $G \setminus (\IA_{G_0}(T,t,  q_0, d, 2c)\cap E)$,
let $\Psi$ be a set of 
$(t, c)$-realizable pairs $\{(E_1, V_1), \dots, (E_k, V_k)\}$ of graph $G$ 
satisfying the following conditions: 
(Let 
$(L_i, V\setminus L_i)$ be the atomic cut induced by 
    $E_i$   
    such that $V_i \subset L_i$ for every $1 \leq i \leq k$.)
\begin{enumerate}
    \item The cut-set of $(V_i, V\setminus V_i)$ is in $G[V^\star]$.
        
    \item One of the following two conditions hold:
    \begin{enumerate}
        \item $L_i \cap S = L_j \cap S$ for any $i, j$.
        \item $V_i \cap S = V_j \cap S$ for any $i, j$
    \end{enumerate}
    
    \item There is a vertex $x \in V$ such that for every $i$, $x \in V\setminus L_i$.
    \item $V_i \cap V_j \neq \emptyset$ for any $i \neq j$. 
\end{enumerate}
The elimination procedure on $\Psi$ outputs a set $W$ of at most $4c^3 + 3c^2$ edges. 
    
\end{lemma}

To prove Lemma~\ref{lem:procedure}, we first prove the following lemma.

\begin{lemma}\label{lem:small_number_cuts_general_new}
In the decremental update setting, 
    for any $\IA_{G_0}(T_0,t,q_0, d, 2c)$ set and a set of vertices  $V^\star \subset V$ that forms a connected component of $G \setminus (\IA_{G_0}(T_0,t, q_0, d, 2c)\cap E)$,
a set of $(t, c)$-realizable 
pairs $\{(E_1, V_1), \dots, (E_k, V_k)\}$ of graph $G$ 
satisfying the  conditions below    contains at most $4c^2 + 3c$ pairs.
(Let $(L_i, V\setminus L_i)$ denotes the cut induced by $E_i$ such that $V_i \subset L_i$ for any $1 \leq i \leq k$.)
\begin{enumerate}
    \item $V_1 \subset V_2 \subset \dots \subset V_{k}$.
    \item $L_1 \subset L_2 \subset \dots \subset L_{k}$.
    \item
    $(V_j \setminus L_i) \cap S = \emptyset$ for any $1 \leq i < j \leq k$.
    \item $\partial_G(V_i)$ are in $G[V^\star]$ for any $1 \leq i \leq k$.
        
    
    \item  
    For every $i > 1$, $\bigcup_{j =1 }^{i-1} \partial_G(V_j)$ does not shatter any  $(V_i \cap (S \cup T), (V\setminus V_i) \cap (S \cup T), \vol_G(V_i))$ -cut of size smaller than or equal the size of $(V_i, V\setminus V_i)$. 

    \end{enumerate}

\end{lemma}

\begin{proof}
For $j> i$, let $V_{i, j}$ denote $V_j\cut L_i$, and $T_{i, j}$ denote $ V_{i, j} \cap T$.
We show the following properties for $V_{i, j}$ and $T_{i, j}$:
\begin{enumerate}
    \item[(a)] $\vol_G(V_{i, j})\leq t$.
    \item[(b)] $\endpoints(\partial_G(V_i)) \subset V_j$ and $V_{i, j} \neq \emptyset$.
    \item[(c)]
    $\partial_G(V_{i, j}) = E_i \cup
    \{(x, y) \in \partial_G(V_j): x, y \in V \setminus L_i\}$.
    \item[(d)] $\abs{\partial_{G}(V_{i, j})}\leq 2c$. 
    \item[(e)] $S\cap V_{i, j} = \emptyset$. 
    \item[(f)] One of the following two conditions hold: (i) $T_{i, j} \neq \emptyset$; (ii) $T_{ i, j} = \emptyset$ and $|E_i| > |E_j|$.
    \item[(g)]
    For any $i \leq i' < j' \leq j$, $T_{i', j'} \subset T_{i, j}$. 
\end{enumerate}

Property (a) is obtained from the fact that $V_{i, j} \subset V_j$ and $\vol_G(V_j)\leq t$ by the fact that $(E_j, V_j)$ is a $(t, c)$-realizable pair.

Property (b) is obtained from the fact that $V_i \subset V_j$ by condition (1)  and $\partial_G(V_i)$ shatters $(V_j, V \setminus V_j)$ if $\partial_G(V_i) \cap \partial_G(V_j) \neq \emptyset$, violating condition (5).

For property (c), by property (b), $E_i$ is a subset of $\partial_G(V_{i, j})$, and by the definition of $V_{i, j}$, \[\{(x, y) \in \partial_G(V_j): x, y \in V \setminus L_i\}\] is a  subset of $\partial_G(V_{i, j})$. Now we show that every edge of $\partial_G(V_{i, j})$ is an edge in either $E_i$ or $\{(x, y) \in \partial_G(V_j): x, y \in V \setminus L_i\}$. 
Let $(x, y)$ be an edge in $\partial_G(V_{i,j}) = \partial_G(V_j\cut L_i)$. Without loss of generality, assume $x\in V_{i,j} = V_j\cut L_i$ and 
\[y\in V\cut V_{i, j} = V\cut \wrap{V_j\cut L_i} = (V\cut V_j)\cup  L_i = \wrap{\wrap{V\cut V_j}\cut L_i}\cup L_i.\] If $y\in (V\cut V_j)\cut L_i$, then $(x, y)\in \{(x, y) \in \partial_G(V_j): x, y \in V \setminus L_i\}$.  If $y\in L_i$, then $(x, y)\in E_i$. 

Property (d) is implied by property (c) and the fact that both $(E_i, V_i)$ and $(E_j, V_j)$ are $(t, c)$-realizable pairs. 

Property (e) is obtained from condition (3) of the lemma. 

Now we prove property (f). 
Suppose $T_{i, j} = \emptyset$.
We show $|E_i| > |E_j|$ by contradiction. 
If \[|E_i| \leq
|\{(x, y) \in \partial_G(V_j): x, y \in V \setminus L_i\}|,
\]
then cut $(V_j \setminus V_{i, j}, V\setminus (V_j \setminus  V_{i, j}))$ satisfies the following conditions 
\begin{itemize}
    \item $\vol_G(V_j \setminus V_{i, j}) < \vol_G(V_j) \leq t$.
    \item $(V_j \setminus V_{i, j}, V\setminus (V_j \setminus V_{i, j}))$ and $(V_j, V\setminus V_j)$ induce the same partition on $S\cup T$ by $T_{i, j} = V_{i, j} \cap T = \emptyset$ and property (e).
    \item The size of cut $(V_j \setminus V_{i, j}, V\setminus (V_j \setminus V_{i, j}))$ is smaller than or equal to the size of cut $(V_j, V\setminus V_j)$ by properties (b) and (c).
\end{itemize}
 Thus $\partial_G(V_i)$ shatters a  $(V_j \cap (S\cup T), (V\setminus V_j)\cap (S \cup T), \vol_G(V_j))$-cut of size at most $|\partial_G(V_j)|$ by Lemma~\ref{lem:replacement}, contradicting  condition (5) of the lemma. Hence, we have 
 \begin{equation}\label{equ:elimination}|E_i| > |\{(x, y) \in \partial_G(V_j): x, y \in V \setminus L_i\}|.\end{equation}
 Note that $E_j\subset \partial_G(V_j)$. Fix a $(x, y)\in E_j$. 
 One of $x$ and $y$ is in $V\setminus L_j$.
 Suppose $x\in (V\cut L_j)\subset (V\cut L_i)$. Since cut $(L_i, V\cut L_i)$ is atomic, $y\in V\cut L_i$.  Therefore,
 we have 
 \[E_j\subset \{(x, y) \in \partial_G(V_j): x, y \in V \setminus L_i\}, \text{ and } 
 \abs {E_j}\leq |\{(x, y) \in |\partial_G(V_j): x, y \in V \setminus L_i\}|< \abs{E_i},\] which contradicts to Equation~(\ref{equ:elimination}). So property (f) holds.

Property (g) is obtained by the definition of $V_{i, j}$ and the conditions (1), (2) of the lemma.

Let $w_1 < w_2 < \dots < w_\ell$ be all the integers in $[k]$ such that
$w_i = 1$ or $T_{w_i - 1, w_{i}} \neq \emptyset$. 
In the rest of this proof, we prove $\ell \leq 4c+3$ by contradiction. Then together with properties (f) and (g), $k \leq 4c^2 + 3c$.

Properties (a), (d), (e) and the construction of $w_1, \dots, w_\ell$ imply that there is a $(T_{w_i, w_j}, T_0 \setminus T_{w_i, w_j}, t)$-cut of size at most $2c$ in $G_0$ for any $1 \leq i < j \leq \ell$. 
    By Claim~\ref{claim:terimial_not_in_same_cc_new}, 
    there is a $(T_{w_i,w_j}, T_0 \setminus T_{w_i,w_j}, q)$-cut of size at most $2c$ in $G_0$ whose cut-set is a subset of the $\IA_{G_0}(T_0, t,q_0,  d, 2c)$ set. 
    Let $(Q'_{w_i, w_j}, V_0 \cut  Q'_{w_i, w_j})$ be such a cut for graph $G_0$.
    By Definition~\ref{def:ia_new}, 
    Both $Q'_{w_i, w_j}$ and $ V_0 \cut  Q'_{w_i, w_j}$ are union of some connected components of $G_0 \setminus \IA_{G_0}(T_0, t, q_0, d, 2c)$.
    Since $G$ is a induced subgraph of $G_0$, 
    $(Q_{w_i, w_j}, V \cut  Q_{w_i, w_j})$
 is a cut of size at most $2c$ in $G$ with cut-set in $\IA_{G_0}(T_0, t, q_0, d, 2c) \cap E$, where $Q_{w_i, w_j} = Q'_{w_i, w_j} \cap V$. 

Consider cut $(Q_{w_i, w_j}, V\setminus Q_{w_i, w_j})$. 
Since $V^\star$ corresponds to a connected component of \[G\setminus (\IA_{G_0}(T_0, t, q_0, d, 2c)\cap E),\]
$V^\star$ is either a subset of $Q_{w_i, w_j}$ or a subset of $V\setminus Q_{w_i, w_j}$.

We show that 
if $V^\star$ is a subset of $Q_{w_i, w_j}$, 
then 
for any $1 \leq h < i$, 
$G[V_{w_h, w_{h+1}}]$ contains an edge of $\partial_G(Q_{w_i, w_j})$. 
    Note that $T_{w_h, w_{h+1}}$ is not an empty set, and $(V_{w_h, w_{h+1}}, V\setminus V_{w_h, w_{h+1}})$ is a cut partitioning $T$ into $T_{w_h, w_{h+1}}$  and $T\setminus T_{w_h, w_{h+1}}$ such that $T_{w_h, w_{h+1}}$ is a subset of $V_{w_h, w_{h+1}}$.
    Hence, for each vertex $u \in T_{w_h, w_{h+1}}$, there exists a path from some vertex in \[\endpoints(\partial_G(V_{w_h, w_{h+1}})) \cap V_{w_h, w_{h+1}}\] to $u$ within the induced subgraph $G[V_{w_h, w_{h+1}}]$.
    On the other hand, by property (c), \[\partial_{G}(V_{w_h, w_{h+1}}) \subset \partial_{G}(V_{w_h}) \cup  \partial_G(V_{w_{h+1}}),\] and thus $\partial_{G}(V_{w_h, w_{h+1}})$ is a set of edges within $G[V^\star]$ by condition (4).  
    Since $V^\star \subseteq Q_{w_h, w_{h+1}}$ and $T_{w_h, w_{h+1}} \subseteq V\setminus Q_{w_h, w_{h+1}}$,
    there is an edge of $\partial_G(Q_{w_h, w_{h+1}})$ which belongs to $G[V_{w_h, w_{h+1}}]$.
    
    Similarly, if $V^\star$ is a subset of $V\setminus Q_{w_i, w_j}$, 
then 
for any $i \leq h < j$, 
$G[V_{w_h, w_{h+1}}]$ contains at least one edge of $\partial_G(Q_{w_i, w_j})$.

Now we prove $\ell \leq 4c+3$ by contradiction. 
Assume $\ell \geq 4c+4$. 
Consider the cut $(Q_{w_{2c+2}, w_\ell}, V\setminus Q_{w_{2c+2}, w_\ell})$. 
Note that for any $1 \leq i < j \leq \ell$, $V_{w_i, w_{i+1}} \cap V_{w_j, w_{j+1}} = \emptyset$ by the definition of $V_{w_i, w_{i+1}}$ and $V_{w_j, w_{j+1}}$.
$V^\star$ cannot be a subset of $Q_{w_{2c+2}, w_\ell}$ because otherwise
$|\partial_G(Q_{w_{2c+2}, w_\ell})| > 2c$.
Similarly, $V^\star$ also cannot be a subset of $V\setminus Q_{w_{2c+2}, w_\ell}$.
This contradicts to the fact that $V^\star$ is either a subset of $Q_{w_{2c+2}, w_\ell}$ or a subset of $Q_{w_{2c+2}, w_\ell}$. 
Hence $\ell \leq 4c + 3$.   
\end{proof}

Now we are ready to prove Lemma~\ref{lem:procedure}.

\begin{proof}[Proof of Lemma~\ref{lem:procedure}]
We show that 
the cuts selected to construct $W$ satisfy the five conditions of Lemma \ref{lem:small_number_cuts_general_new}, and then by Lemma~\ref{lem:small_number_cuts_general_new}, the current lemma follows. 
Let $(E_i, V_i)$ be the realizable pair selected in the $i$-th iteration of the elimination procedure to construct $W$.

For the first condition of Lemma~\ref{lem:small_number_cuts_general_new}, 
we show that $V_i \subset V_j$ for any $i < j$.
By the description of the elimination procedure, $\partial_G(V_i)$ does not shatter cut $(V_j, V \setminus V_j)$.
Hence, $\partial_G(V_i) \cap \partial_G(V_j) = \emptyset$, and 
it is impossible that there is an edge of $\partial_G(V_j)$ in $G[V_i]$,
and another edge of $\partial_G(V_j)$ in $G[V\setminus V_i]$, otherwise $\partial_G(V_i)$ shatters $(V_j, V\setminus V_j)$. 
Meanwhile by the elimination procedure, it is also impossible that $\partial_G(V_j)$ are in $G[V_i]$, because otherwise $E_j$ are edges of $G[V_i]$ and thus $(E_j, V_j)$ will be selected before $(E_i, V_i)$. 
Hence, $\partial_G(V_j)$ are in $G[V\setminus V_i]$.
With the condition (4) of the Lemma~\ref{lem:procedure} and the fact that $(V_i, V\setminus V_i)$ and $(V_j, V \setminus V_j)$ are simple cuts,  $V_i \subset V_j$.

For the second condition of Lemma~\ref{lem:small_number_cuts_general_new}, 
we show that $L_i \subset  L_j$ for any $i < j$.
By the third condition of the Lemma~\ref{lem:procedure}, we have $x \in V\setminus L_i$ as well as $x \in V\setminus L_j$. 
If $E_i = \partial_G(L_i)$ shatters $(L_j, V\setminus L_j)$, then $(E_j, V_j)$ is not selected to construct $W$, hence $E_i$ does not shatter $(L_j, V\setminus L_j)$.
By the first condition of Lemma~\ref{lem:small_number_cuts_general_new} and the elimination procedure, $E_j$ does not shatter $(L_i, V\setminus L_i)$.
By Lemma~\ref{lem:intercept_condition_2}, 
$(L_i, V\setminus L_i)$ and $(L_i, V\setminus L_i)$  are parallel, and thus $V \setminus L_i$ is a subset or a superset of $V\setminus L_j$. 
Since 
$V_i \subset V_j$,
$E_j$ is not in $G[V_i]$,  and thus $L_i \subset L_j$.

Now we prove the third condition of  Lemma~\ref{lem:small_number_cuts_general_new}.
If $L_i \cap S = L_j \cap S$ for any $i, j$,
then for any $i < j$,
since $V_j \subset L_j$  
we have $(V_j \setminus L_i) \cap S \subset (L_j \setminus L_i) \cap S = \emptyset$.
If $V_i \cap S = V_j\cap S$ for any $i, j$, 
we have $(V_j \setminus L_i) \cap S \subset (V_j \setminus V_i) \cap S = \emptyset$.
Hence, for either cases of the second condition of the current lemma, 
the third condition of Lemma~\ref{lem:small_number_cuts_general_new} holds.

The forth condition of Lemma~\ref{lem:small_number_cuts_general_new} is implied by the first condition of the current lemma, and the fifth condition of Lemma~\ref{lem:small_number_cuts_general_new} is implied by 
the definition of elimination procedure. 
\end{proof}

\subsection{\texorpdfstring{$\IA$}{IA} Set Decremental Update Algorithm}\label{sec:ia_update}
In this section, we present the $\IA$ set update algorithm and prove Lemma~\ref{lem:repair_set_algorithm}. 
Recall that our goal is to prove that in the decremental update scenario, there is an efficient algorithm to find a subset of edges $F$, called  repair set,  such that the union of the  repair set and
$\IA_{G_0}(T_0, t, q, d, 2c+1)\cap E$ is an $\IA$ set of $G$ with respect to $T\cup S$.  
We formally define the  repair set as follows.

\begin{definition}\label{def:repair_set}

For a fixed 
$\IA_{G_0}(T_0, t, q, d, 2c+1)$ set,
a set of edges $F \subset E$ is a \emph{repair set} of 
$\IA_{G_0}(T_0, t, q,$ $ d,2c+1)$ with respect to $G$
if \[\wrap{\IA_{G_0}(T_0, t, q, d, 2c+1) \cap E} \cup F\] is an $\IA_{G}(
S\cup T
, t, q + t, c, 1)$ set.
\end{definition}

We divide all the bipartitions of $S\cup T$ into three types: 
for $U \subset S\cup T$ satisfying $U\neq \emptyset$, 
a
bipartition $(U, (S\cup T) \setminus U)$ is

\begin{enumerate}
    \item a \emph{type 1 bipartition} if $\emptyset \subsetneq (U \cap S) \subsetneq S$, 
    \item a \emph{type 2 bipartition} if $U \cap S = \emptyset$,
    \item a \emph{type 3 bipartition} if $U \cap S = S$.
\end{enumerate}
The underlying reason of these three type bipartitions is that different types imply different properties so that we need to use different strategies to shatter their corresponding minimum cuts. Hence, we construct repair set for each type as defined in Definition~\ref{def:repair_set_type}. 
By Lemma~\ref{lem:simple_intercept}, the repair set only focuses on bipartitions that has a simple minimum cut. 


\begin{definition}\label{def:repair_set_type}

In the $\IA$ set decremental update setting, for a fixed 
$\IA_{G_0}(T_0, t, q, d, 2c+1)$ set 
and an integer $\alpha \in \{1, 2, 3\}$,
a set of edges $E' \subset E$ is a \emph{type $\alpha$ repair set} of 
$\IA_{G_0}(T_0, t, q, d,$ $ 2c+1)$ with respect to $G$
if for each type $\alpha$ bipartition $(U, (S\cup T) \setminus U)$ of $G$ such that $mincut_G(U, (S\cup T) \setminus U, t) \leq c$ and there is a simple minimum $(U, (S\cup T) \setminus U, t)$-cut, 
there is a $(U, (S\cup T) \setminus U, q)$-cut of size at most $mincut_G(U, (S\cup T) \setminus U, t)$
whose cut-set does not belong to any  connected component of $E' \cup (\IA_{G_0}(T_0, t, q, d, 2c+1) \cap E)$.

\end{definition}

By Definition~\ref{def:repair_set} and Definition~\ref{def:repair_set_type}, 
a repair set can be obtained by the union of a type 1 repair set, a type 2 repair set, and a type 3  repair set. 

\begin{corollary}
\label{corollary:type_combine}
In the $\IA$ set decremental update setting, 
for any $\IA_{G_0}(T_0, t, q, d, 2c+1)$ set satisfying the conditions $q \geq t, d \geq 2c+1$, 
the union of a type 1 repair set, a type 2 repair set and a type 3 repair set for the $\IA_{G_0}(T_0, t, q, d, 2c+1)$ set with respect to $G$ is a repair set of the $\IA_{G_0}(T_0, t, q, d, 2c+1)$ set with respect to $G$.
\end{corollary}

Hence, in our $\IA$ set update algorithm, we find repair sets for each type, and output their union. 

\begin{framed}
\noindent \textbf{Decremental $\IA$ Set Update Algorithm}

    \vspace{.2cm}\noindent  \textbf{Input:} Access to graphs 
    \[G, G \setminus (\IA_{G_0}(T_0, t, q_0, d, 2c) \cap E), \text{ and } G \setminus (\IA_{G_0}(T_0, t, q, d, 2c+1) \cap E)\] such that 
    the $\IA_{G_0}(T_0, t, q, d, 2c+1)$ set  is derived from the 
      $\IA_{G_0}(T_0, t, q_0, d, 2c)$ set for parameters satisfying $t\leq q_0\leq q$ and $d \geq 2c+1$, and two vertex sets $S, T$.
    

    \vspace{.2cm} \noindent  \textbf{Output:} A set of edges $W$.  
    
\begin{enumerate}    
\item Find a type $1$ repair set $W_1$. 
\item Find a type $2$ repair set $W_2$. 
\item Find a type $3$ repair set $W_3$. 
\item Return $W_1 \cup W_2 \cup W_3$. 
\end{enumerate}
\vspace{-.3cm}\end{framed}

\paragraph{Type One Repair Set Algorithm}
We give an algorithm to construct a type one repair set, and prove the following lemma. 
\begin{lemma}\label{lem:type-one-set-construct}
There is an algorithm to compute 
a type 1 repair set of size at most $|S|(8c^3 + 6c^2)$
in time $O(|S|(2t)^{c^2 +3c + 3} n^{o(1)})$. 
\end{lemma}

Recall that in Section~\ref{sec:overview:cut_containment}, we select a maximal set of pairwise parallel bipartitions of $S$ to construct repair sets. 
Here, since we only consider cuts with one side having volume bounded by $t$, the pairwise parallel bipartitions on $S$ are given by a bipartition system.




Two $(t, c)$-realizable pairs $(E', V')$ and  $(E'', V'')$ are called to be \emph{atomic cut equivalent} if  $L' \cap S = L'' \cap S$, where $(L', V\setminus L')$ is the cut induced by $E'$  such that $V' \subset L'$, and $(L'', V\setminus L'')$ is the cut induced by $E''$  such that $V'' \subset L''$.

\begin{definition}\label{def:bipartition_system}

     An \emph{$(S, t, c)$-bipartition system} $\mathcal B$ of graph $G$ is a set of $(t, c)$-realizable pairs 
     \[\{(E_1, V_1), (E_2, 
     V_2), \dots, (E_k, V_k)\}\] 
     satisfying the following conditions:  (Suppose $(L_i, V\setminus L_i)$ is the atomic cut induced by $E_i$ such that $V_i \subset L_i$, and $S_i$ is $L_i \cap S$.)
     \begin{enumerate}
         
        \item For every $(E_i, V_i) \in \mathcal{B}$, 
        $(V_i, V\setminus V_i)$ is a simple $(t, c)$-cut such that $\emptyset \subsetneq V_i \cap S \subsetneq S$.
        
        
        \item For any $(E_i, V_i)$ and $(E_j, V_j)$ in $\mathcal B$, 
        $(E_i, V_i)$ and $(E_j, V_j)$ are not atomic cut equivalent, 
        satisfy one of the following two conditions:
        \begin{enumerate}
            \item $S_i = S \cut  S_j$
            \item One of $S_i$ and $S\setminus S_i$ is a strict subset of either $S_j$ or $S\setminus S_j$. 
        \end{enumerate}
     \end{enumerate} 
     
        
        
 \end{definition}
 
 In our algorithm, we first find a maximal bipartite system $\mathcal{B}$. 
 Since $\mathcal{B}$ is maximal, if a simple cut $C$ has a composing atomic cut that is not parallel with the atomic cut induced by $E_i$ for some $(E_i, V_i) \in \mathcal{B}$, then $C$ is already shattered by any atomic cut that partition $S$ the same as $(E_i, V_i)$. Hence, we only need to consider simple cuts that can form an atomic cut equivalent realizable pair with some realizable pair in $\mathcal{B}$.
 
 For each $(E_i, V_i)$ in $\mathcal{B}$, we further show that all the $(E_i, V_i)$'s atomic cut equivalent 
 realizable pairs whose 
 simple cuts are
 not shattered by $\IA_{G_0}(T_0, t, q, d, 2c+1)\cap E$
 can be partitioned into two groups that satisfying the conditions of Lemma~\ref{lem:procedure}. Thus, by running elimination procedure on these two groups of realizable pairs, we obtain a set of $O(c^2)$ edges such that together with $\IA_{G_0}(T_0, t, q, d, 2c+1)\cap E$, it is sufficient to shatter all the simple cuts of the realizable pairs that are atomic cut equivalent to $(E_i, V_i)$. 

 \begin{framed}
\noindent \textbf{Type One Repair Set Algorithm}

    \vspace{.2cm} \noindent  \textbf{Input:} 
    Access to graphs 
    $G$ and $G \setminus (\IA_{G_0}(T_0, t, q_0, d, 2c)\cap E)$, and two vertex sets $S, T$.

    

\vspace{.2cm}  \noindent  \textbf{Output:} A set of edges $W$.  
    
\begin{enumerate}    
\item Find a maximal $(S, t, c)$-bipartition system $\mathcal{B}$. 
\item For each realizable pair $(E', V') \in \mathcal{B}$, 
\begin{enumerate}
\item Collect all the $(t, c)$-realizable pairs $(E'', V'')$ such that $(E'', V'')$ is atomic cut equivalent to $(E', V')$ on $S$, and the cut-set of cut $(V'', V\setminus V'')$ belongs to a connected component of $G\setminus (\IA_{G_0}(T_0, t, q_0, d, 2c) \cap E)$ that contains a vertex of $S$.

\item Group all the realizable pairs $(E'', V'')$ collected in Step (a) by the connected component of $G\setminus (\IA_{G_0}(T_0, t, q_0, d, 2c) \cap E)$ that contains the cut-set of $(V'', V\setminus V'')$.
\item Run elimination procedure on each group of realizable pairs obtained in Step (b). 
\end{enumerate}
\item 
Return the union of results returned by each execution of elimination procedure in Step (2c).
\end{enumerate}
\vspace{-.3cm}\end{framed}

 We show that the number of realizaible pairs in a $(S, t, c)$-bipartition system is at most $O(\abs S)$.
   \begin{claim}\label{claim:num-bipart}
      Let $G= (V, E)$ be a connected graph, and $S \subset V$ be a set of terminals. 
      Any $(S, t, c)$-bipartition system 
       contains at most $2(|S|- 1)$ $(t, c)$-realizable pairs. 
      
 
 \end{claim}
  \begin{proof}
  We first prove that if $Q = \{(S_1, S\setminus S_1), (S_2, S\setminus S_2), \dots\}$ is a set of bipartitions  of $S$ 
  satisfying the following conditions:
  \begin{itemize}
  	\item[(a)] $S_i \cap S \neq \emptyset$ and $(S\setminus S_i) \cap S \neq \emptyset$ hold for any $(S_i, S\setminus S_i)\in Q$;
	\item[(b)] for any $(S_i, S\setminus S_i), (S_j, S\setminus S_j) \in Q$, one of $S_j$ and $S\setminus S_j$ is a subset of $S_i$ or $S\setminus S_i$,
  \end{itemize}
  assuming $(S', S\setminus S')$ and $(S\setminus S', S)$ are same for any $\emptyset \subsetneq S' \subsetneq S$,  then $|Q| \leq |S|-1$. 
  We prove by induction. 
   If $|S| = 1$, then $|Q| = 0$ because each $(S_i, S\setminus S_i)$ partitions $S$ into two non-empty subsets.
   Assume $|Q| \leq |S| - 1$ holds for any $|S|$ of size at most $k$,
    we prove $|Q| \leq |S| - 1$ for any $S$ of size $k + 1$. 
   Fix an arbitrary $(S_1, S\setminus S_1)$ in $Q$, any $(S_i, S\setminus S_i)$ for $i > 1$ satisfies one of the following two conditions: 
  \begin{enumerate}
  	\item one of $S_i$ and $S\setminus S_i$ is a subset of $S_1$,
  	\item one of $S_i$ and $S\setminus S_i$ is a subset of $S\setminus S_1$.
  \end{enumerate}   
  Let $Q_1$ be the bipartitions of $Q$ satisfying the first condition (excluding $(S_1, S\setminus S_1)$), and $Q_2$ be the bipartitions of $Q$ satisfying the second condition (excluding $(S_1, S\setminus S_1)$). 
  Then set \[\{(S_i \cap S_1, (S\setminus S_i) \cap S_1) : (S_i, S \setminus S_i)\in Q_1\}\] is a set of bipartitions on $S_1$ satisfying condition (a) and (b), and set \[\{(S_i \cap (S\setminus S_1), (S\setminus S_i) \cap (S\setminus S_1)) : (S_i, S \setminus S_i)\in Q_1\}\] is a set of bipartitions on $S\setminus S_1$ satisfying condition (a) and (b).
  Hence we have $|Q_1| \leq |S_1| - 1$ and $|Q_2| \leq |S\setminus S_1| - 1$ by induction. 
  Thus we have \[|S| = 1 + |Q_1| + |Q_2| \leq 1 + ( |S_1| - 1) +  (|S\setminus S_1| - 1) \leq |S| - 1.\]

  
  Let $\mathcal{B}$ be an $(S, t, c)$-bipartition system. 
  For each $(E_i, V_i) \in \mathcal{B}$, let $S_i$ be the subset of $S$ defined as   Definition~\ref{def:bipartition_system}. 
  By Definition~\ref{def:bipartition_system},
   for any $(E_i, V_i)\in \mathcal{B}$, another $(E_j, V_j)$ in $\mathcal{B}$ satisfies one of the following three conditions:
  \begin{enumerate}
      \item $S_j = S \setminus S_i$\;
      \item One of $S_{j}$ and $S \cut  S_j$ is a strict subset of $S_i$;
      \item One of $S_{j}$ and $S \cut  S_j$ is a strict subset of $S \cut  S_i$.
  \end{enumerate}
  The set $\{(S_i, S\setminus S_i): (E_i, V_i) \in \mathcal{B}\}$ satisfies conditions (a) and (b) if $(S_i, S\setminus S_i)$ and $(S\setminus S_i, S_i)$ are treated as the same pair. 
  Hence $|\mathcal{B}| \leq 2(|S| - 1)$. 
 \end{proof}
 The following lemma shows that for any bipartition of terminal set $S$ which partitions $S$ into two non-empty sets and any vertex partition, at most two connected components induced by the vertex partition can contain the cut-set of a cut that induces the given $S$ bipartition.
 \begin{lemma}\label{lem:number_cc_partition}
 In the decremental graph update setting, 
let $S'$ be a  subset of $S$ such that $\emptyset \subsetneq S' \subsetneq S$. 
There are at most two connected components of $G\setminus (\IA_{G_0}(T_0, t, q_0, d, 2c) \cap E)$ satisfying the following two conditions:
\begin{enumerate}
    \item the connected component contains at least one vertex of $S$.
    \item there is a cut $(V', V \setminus V)$ of graph $G$ 
    that partitions $S$ into $S'$ and $S \cut  S'$ such that the cut-set of $(V', V \setminus V)$ are edges of the connected component.
\end{enumerate}
\end{lemma}
\begin{proof}
    Let $P_1$ and $P_2$ be two connected components of $G\setminus (\IA_{G_0}(T_0, t, q_0, d, 2c) \cap E)$ satisfying the two conditions. 
    By the second condition, there exist cuts $C_1 = (V_1, V\cut V_1)$ and $C_2 = (V_2, V\cut V_2)$ such that 
    the cut-set of $C_1$ is in $P_1$,and the cut-set of $C_2$ is in $P_2$. 
    Without loss of generality, assume $S'\subset V_1$ and $S' \subset V_2$.
    
    Since $P_1$ and $P_2$ are connected graphs, the cut-set of $C_1$ does not shatter $C_2$ and the cut-set of $C_2$ does not shatter $C_1$. By the contrapositive of Lemma~\ref{lem:intercept_condition_2}, $C_1$ and $C_2$ are parallel. 
    
    Therefore, if there exists another  connected component $P_3$ satisfying the two conditions with cut $C_3 = (V_3, V\cut V_3)$ such that the cut-set of $C_3$ is in $P_3$ and $S'\subset V_3$, then $C_1$, $C_2$, $C_3$ are pairwise parallel.
    
    Without loss of generality, assume $V_1\subset V_2\subset V_3$. Since $C_1$, $C_2$ and $C_3$ induce the same partition on $S$, $P_2$ does not contain any vertex of $S\cap V_2$ as well as $S \cap (V\setminus V_2)$. This contradicts the first condition.
    \end{proof}
    
    


\begin{claim}\label{claim:nontrivial-intersection}
Let $G = (V, E)$ be a connected graph, $S \subset V$ be a set of terminals. 
If $(E_1, V_1)$ and $(E_2, V_2)$ are atomic cut equivalent realizable pairs,  
then $V_1 \cap V_2 \neq \emptyset$.
    
    

  \end{claim}
 \begin{proof}
 We prove by contradiction.
 Assume $V_1 \cap V_2 = \emptyset$. 
 Let $(V_i', V \cut  V_i')$ be the cut induced by $E_i$ satisfying $V_i \subseteq V_i'$ for $i \in \{1, 2\}$.
 Since $G[V_1]$ and $G[V_2]$ are connected, $V_2$ belongs to one connected component of $G[ V \setminus V_1]$, and this connected component is a subset of $V_1'$ because of \[\emptyset \subsetneq V_2 \cap S \subset V_2' \cap S = V_1' \cap S.\]
 Hence, in $G[V  \setminus V_2]$, $V_1$ and $V \cut  V_1'$ are in the same connected component. 
 Since $V_1\cap S \neq \emptyset$ and $(V\setminus V_1') \cap S \neq \emptyset$, no subset of $(V_2, V\setminus V_2)$'s cut-set forms a cut that partitions $S$ the same as that of $(V_1', V\setminus V_1')$.  Contradiction.
 \end{proof}

\begin{proof}[Proof of Lemma~\ref{lem:type-one-set-construct}]
    
    We first prove the correctness of the algorithm. 
    Let $W_1$ be the edge set returned by the algorithm.  
    We first show that for every type 1 bipartition $(U, (S\cup T) \setminus U)$ that has a simple minimum $(U, (S\cup T) \setminus U, t)$-cut of size at most $c$, there is a  $(U, (S\cup T) \setminus U, q)$-cut $C$ of size at most $mincut_G(U, (S\cup T) \setminus U, t)$ shattered by \[ (\IA_{G_0}(T_0, t, q, d, 2c+1) \cap E) \cup W_1.\] Then $W_1$ is a type one repair set. 
    
    Let $C' = (V', V\setminus V')$ be  a simple minimum  $(U, (S\cup T) \setminus U), t)$-cut.
    We assume the cut-set of $C'$ belongs to a connected component of $G\setminus (\IA_{G_0}(T_0, t, q_0, d, 2c)\cap E)$ containing a vertex from $S$, otherwise we are done by letting $C = C'$ by Lemma~\ref{lem:exclude_Case}.
    Since $U \cap S \neq \emptyset$ and $U \cap S \neq S$, there must be a $E' \subset \partial_G(V')$ such that $E'$ induces an atomic cut that partitions $S$ into two non-empty sets. Thus, $(E', V')$ is a $(t, c)$-realizable pair such that 
    $E'$ partitions $S$ into two non-empty sets.
    
    If there is a $(E^\dagger, V^\dagger)\in\mathcal B$ such that $(E', V')$ is atomic cut equivalent to  $(E^\dagger, V^\dagger)$, 
    when run Step (2a) with respect to $(E^\dagger, V^\dagger)$, 
    $(E', V')$ is collected. By Lemma~\ref{lem:elimination_procedure_correctness}, 
    $W_1$ shatters a $(U, (S\cup T) \setminus U, \vol_G(V'))$-cut of size at most $mincut_G(U, (S\cup T) \setminus U, t)$.
    If no realizable pair in $\mathcal B$ is atomic cut equivalent to  $(E', V')$, then there is a realizable pair $(E^\dagger, V^\dagger) \in \mathcal{B}$ such that the bipartition on $S$ induced by $E^\dagger$ is not parallel with the bipartition on $S$ induced by $E'$. 
    By the algorithm, there is a realizable pair $(E^\diamond, V^\diamond)$ atomic cut equivalent to $(E^\dagger, V^\dagger)$ such that 
    $W_1$ contains the cut-set of cut $(V^\diamond, V\setminus V^\diamond)$. Hence $W_1$ shatters cut $C'$. 
      Hence, for either case, $(\IA_{G_0}(T_0, t, q, d, 2c+1) \cap E)\cup W_1$ shatters a $(U, (S\cup T) \setminus U,t)$-cut of size at most $mincut_G(U, (S\cup T) \setminus U, t)$.

    

    Now we bound the size of $W_1$. By Claim~\ref{claim:num-bipart}, $|\mathcal{B}| \leq 2(|S| - 1)$. By Lemma~\ref{lem:number_cc_partition}, for each $(E^\dagger, V^\dagger) \in \mathcal{B}$, the realizable pairs collected at Step (2a) are partitioned into two groups.  And thus for each $(E^\dagger, V^\dagger)\in \mathcal{B}$, the elimination procedure is executed twice. Notice that each group of realizable pairs satisfies the conditions of Lemma~\ref{lem:procedure}. Hence, by Lemma~\ref{lem:procedure}, Step (2) for each realizable pair of $\mathcal{B}$ adds $4c^3 + 3c^2$ edges to $W_1$. 
    Hence, the total number of edges in $W_1$ is at most \[2(|S| - 1)(4c^3 + 3c^2) \leq |S|(8c^3 + 6c^2).\] 
    
At the end, we bound the running time of the algorithm. 
By Lemma~\ref{lem:subroutine_repair_set}, given a vertex $x$, we can enumerate all the $(t, c)$-realizable pairs $(E', V')$ such that $(V', V\setminus V')$ is a simple cut of size at most $c$ satisfying that $\vol_G(V') \leq t$ and $x \in V'$ in time $O((2c)^{c+2} n^{o(1)})$.
Moreover, the number of realizable pairs enumerated is at most $t^{c}$.

To find a maximal bipartite system, we first enumerate all the $(t, c)$-realizable pairs $(E', V')$ such that both $(V', V\setminus V')$ and the cut induced by $E'$ partition $S$ into two non-empty sets in $O(|S|(2t)^{c+2} n^{o(1)})$ time.
The number of pairs enumerated is at most $|S|t^c$. 
Then we set $\mathcal{B}$ to be an empty set initially, and try to add each enumerated realizable pair to $\mathcal{B}$ if not violating the definition of bipartite system. 
Note that if two realizable pairs cannot be both contained in a bipartition system, then 
they are atomic cut equivalent or their atomic cuts induce bipartitions on $S$ that are not parallel. 
Hence, to determine if a realizable pair $(E', V')$ can be added to $\mathcal{B}$, we can enumerate all the realizable pairs that 
cannot be contained in any bipartition system that contains $(E', V')$, and then check if $\mathcal{B}$ contains any of them. 
By Claim~\ref{claim:nontrivial-intersection}, if two realizable pairs $(E_1, V_1)$ and $(E_2, V_2)$ are atomic cut equivalent, then $V_1$ and $V_2$ have overlap. 
If the atomic cuts of two realizable pairs $(E_1, V_1)$ and $(E_2, V_2)$ induce parallel bipartitions on $S$ that are not parallel, then $V_1$ and $V_2$ also have overlap. 
Hence, by Lemma~\ref{lem:subroutine_repair_set}, for any $(E', V')$, all the realizable pairs that 
cannot be contained in any bipartition system containing $(E', V')$ can be enumerated in $O((2t)^{c+3}n^{o(1)})$ time. 
Hence, Step (1) of the algorihtm can be implemented in $O(|S| (2t)^{2c+3} n^{o(1)})$ time. 

Fix a $(E', V') \in \mathcal{B}$.
By Claim~\ref{claim:nontrivial-intersection}, 
Step (2a) can be implemented in $O((2t)^{c+3}n^{o(1)})$ time, and the pair enumerated is at most $(2t)^{c+1}$ 
By Lemma~\ref{lem:subroutine_repair_set}, Step (2b) can be implemented in $O((2t)^{c+1}n^{o(1)})$ time.
By Lemma~\ref{lem:elimination_procedure_correctness}, Step (2c) can be implemented in $O((2t)^{2c+2} t^{(c+1)c + 1} n^{o(1)})$ time. 
Since $\mathcal{B}$ contains $O(|S|)$ realizable pairs, the running time for Step (2) (for all the pairs in $\mathcal{B}$) is $O(|S| (2t)^{c^2 + 3c + 3} n^{o(1)})$. 

Thus, the overall running time of the algorithm is $O(|S| (2t)^{c^2 + 3c + 3} n^{o(1)})$.
%
\end{proof}


\paragraph{Type Two Repair Set Algorithm}
We give an algorithm to construct a type two repair set, and prove the following lemma.

\begin{lemma}\label{lem:type-two-set-construct}
There is an algorithm to compute 
a type 2 repair set of size at most $|S|(4c^3 + 3c^2)$
in time $O(|S|(2q)^{c^2 + 3c + 3}n^{o(1)})$. 
\end{lemma}

    

Our type two bipartition update algorithm collects type two bipartitions whose minimum cuts are not shattered by the $\IA$ set, and constructs their realizable pairs. The algorithm then partitions the realizable pairs into a few groups such that the elimination procedure is efficient on each group. 
 \begin{framed}
\noindent \textbf{Type Two Bipartition Update Algorithm}

    \vspace{.2cm}\noindent \textbf{Input:} Access to 
    $G$, $G \setminus (\IA_{G_0}(T_0, t, q_0, d, 2c)\cap E)$ and $G \setminus (\IA_{G_0}(T_0, t, q, d, 2c + 1)\cap E)$, and two vertex sets $S, T$.


    \noindent \textbf{Output:} A set of edges $W$.  
    
\begin{enumerate}    

\item Enumerate all the bipartitions $(T', (T\setminus T')\cup S)$ satisfying the following conditions: 
(a) $mincut_G(T', (T\setminus T')\cup S, t) \leq c$,
(b)
there is a minimum $(T', (T\setminus T')\cup S, t)$-cut that is a simple cut, and 
(c) $\IA_{G_0}(T_0, t, q, d, 2c+1)\cap E$ does not shatter any $(T', (T\setminus T')\cup S, q)$-cut of size at most 
$mincut_G(T', (T\setminus T')\cup S, t)$.


\item Let $s$ be an arbitrary vertex of $S$. For each bipartition $(T', (T\setminus T')\cup S)$ enumerated in Step (1), collect a realizable pair $(E', V')$ such that 
$(V', V\setminus V')$ is a 
simple minimum $(T', (T\setminus T')\cup S, t)$-cut, and $E'$ is a subset of $(V', V\setminus V')$'s cut-set that forms an atomic cut separating $V'$ and $s$.

\item Partition all the realizable pairs obtained in Step (2) into groups so that $(E', V')$ and $(E'', V'')$ are in the same group if and only if the cut-sets of $(V', V\setminus V')$ and $(V'', V\setminus V'')$ are in the same connected component of $G\setminus (\IA_{G_0}(T_0, t, q_0, d, 2c)\cap E)$. Run elimination procedure on each group. 

\item Return the union of results returned by executions of elimination procedure in Step (3).
\end{enumerate}
\vspace{-.5cm}
\end{framed}

Before we prove Lemma~\ref{lem:type-two-set-construct}, we first prove a property of the bipartitions that need to be shattered by a type two repair set. 

\begin{lemma}\label{lem:basic_type_2}

In the decremental update scenario, 
let $V^\star \subset V$ be a set of vertices that forms a connected component of $G \setminus \left(\IA_{G_0}(T_0,t, q_0, d, 2c)\cap E\right)$ satisfying $V^\star \cap S \neq \emptyset$. 
Let $(T_1, S\cup (T\setminus T_1))$ and $(T_2, S\cup (T\setminus T_2))$ be two type two bipartitions on $S\cup T$ satisfying the following conditions: for any $i \in \{1, 2\}$
\begin{enumerate}
\item $mincut_G(T_i, (T\setminus T_i)\cup S, t) \leq c$;
\item
there is a simple minimum $(T_i, (T\setminus T_i)\cup S, t)$-cut with cut-set in $G[V^\star]$.
\item $\IA_{G_0}(T_0, t, q, d, 2c+1)\cap E$ does not shatter any $(T_i, (T\setminus T_i)\cup S, q)$-cut of size at most \\
$mincut_G(T', (T\setminus T')\cup S, t)$.
\end{enumerate}
Then $T_1 \cap T_2 \neq \emptyset$.


\end{lemma}
\begin{proof}
For each $i \in\{1,2\}$,
let $(V_i, V \setminus V_i)$ be a simple minimum $(T_i, (T\setminus T_i) \cup S, t)$-cut of $G$ with cut-set in $G[V^\star]$. Without loss of generality, assume $V_i \cap T = T_i$. 
$(V_i, V_0 \setminus V_i)$ is a simple $(T_i, T_0\setminus T_i, t)$-cut of graph $G_0$ with size the same as $(V_i, V\setminus V_i)$ for graph $G$.
By the definition of $\IA$ set,  there is a $(T_i, T_0\setminus T_i, q_0)$-cut in graph $G_0$ of size at most $mincut_G(T_i, (T\setminus T_i) \cup S, t)$ with cut-set in the $\IA_{G_0}(T_0, t, q_0, d, 2c)$ set. 
Let $(V_i', V_0 \setminus V_i')$ be such a cut with minimum size. Without loss of geniality, assume $T_i \subseteq V_i'$.

We show that $V^\star$ is a subset of $V_i'$. 
There must be a connected component of $G_0[V_i']$, denoted by $V_i^\dagger$, that contains vertices from both $S$ and $T_i$, otherwise, a subset of $(V_i', V_0, \setminus V_i')$'s cut-set can induce a $(T_i, (S \cup T_0)\setminus T_i, q_0)$-cut on graph $G_0$ whose cut-set is in $\IA_{G_0}(T_0, t, q_0, d, 2c)$, and thus a $(T_i, (S\cup T) \setminus T_i, q_0)$-cut on graph $G$ whose cut-set is in $\IA_{G_0}(T_0, t, q_0, d, 2c) \cap E$, contradicting the third condition of the current lemma. 
On the other hand, since $(V_i, V_0 \setminus V_i)$ is a simple $(T_i, T_0\setminus T_i, t)$-cut on graph $G_0$ whose cut-set is in $G[V^\star]$,
any path between a vertex from $S$ and a vertex from $T_i$ in $G_0$ contains an edge of the cut-set of $(V_i, V_0 \setminus V_i)$. 
Since all vertices of $V^\star$ are in the same connected component of $G_0\setminus \IA_{G_0}(T_0, t, q_0, d, 2c)$, we have $V^\star \subset V^\dagger$. Hence $V^\star$ is a subset of $V_i'$. 

Now we prove $T_1 \cap T_2 \neq \emptyset$ by contradiction. Suppose  $T_1 \cap T_2 = \emptyset$.
Since $V^\star$ is a subset of both $V_1'$ and $V_2'$,  then by Lemma~\ref{lem:swapping},
one of the following conditions hold:
\begin{itemize}
    \item $(V_1', V_0\setminus V_1')$ is not a cut with size minimized among all  
 $(T_1, T_0 \setminus T_1, q_0)$-cuts of $G_0$ with cut-set in $\IA_{G_0}( T_0, t, q_0, d, 2c)$.
    \item $(V_2', V_0\setminus V_2')$ is not a cut with size minimized among all  
$(T_2, T_0 \setminus T_2, q_0)$-cut of $G_0$ with cut-set in $\IA_{G_0}( T_0, t, q_0, d, 2c)$.
\item There is a $(T_1, T\setminus T_1, \vol_{G_0}(V_1'))$-cut of size at most $(V_1', V_0 \setminus V_1')$
for $G_0$ with cut-set in $\IA_{G_0}( T_0, t, $ $ q_0, d, 2c)$ such that the side containing $T_1$ does not contain $V^\star$, and a $(T_2, T_0\setminus T_2, \vol_{G_0}(V_2'))$ of size at most $(V_2', V_0 \setminus V_2')$ for $G_0$ with cut-set in $\IA_{G_0}( T_0, t, q_0, d, 2c)$ such that the side containing $T_2$ does not containing $V^\star$.
\end{itemize}
The first two cases contradict the definitions of $(V_1', V\setminus V_1')$ and $(V_2', V\setminus V_2')$.
The third case contradicts the fact that $V^\star$ is a subset of both $V_1'$ and $V_2'$. Hence $T_1 \cap T_2 \neq \emptyset$.
\end{proof}

\begin{proof}[Proof of Lemma~\ref{lem:type-two-set-construct}]
    We show that the Type Two Bipartition Update Algorithm gives the desirable repair set and has the desirable running time. 
     
    Let $W_2$ be the edge set returned by the algorithm.  
    Let bipartition $(T', (S\cup T) \setminus T')$ be a type two bipartition that has a simple minimum $(T', (S\cup T) \setminus T', t)$-cut of size at most $c$. 
    If there does not exist a $(T', (T\setminus T') \cup S, q)$-cut of size at most $mincut_G(T', (T\setminus T')\cup S, t)$ shattered by $\IA_{G_0}(T_0, t,q, d, 2c+1)\cap E$, then $(T', (S\cup T) \setminus T')$ is enumerated in Step (1) of the algorithm. 
    Hence, a realizable pair $(E', V')$ such that $(V', V\setminus V')$ is a simple $(T', (S\cup T) \setminus T', t)$-cut of size $mincut_G(T', (T\setminus T')\cup S, t)$ is constructed in Step (2) of the algorithm. By Lemma~\ref{lem:elimination_procedure_correctness}, $W_2$ shatters a $(T', (S\cup T) \setminus T', t)$-cut of size at most $mincut_G(T', (T\setminus T')\cup S, t)$.
    Hence $W_2$ is a type two repair set.

    To bound the size of $W_2$, 
    by Lemma~\ref{lem:exclude_Case}, if a cut is not shattered by the $\IA_{G_0}(T_0, t, q, d, 2c+1)\cap E$ set, then the cut-set of that cut must belongs to a connected component of \[G\setminus (\IA_{G_0}(T_0, t, q_0, d, 2c)\cap E)\] that contains a vertex of $S$. 
    Since $G\setminus (\IA_{G_0}(T_0, t, q_0, d, 2c)\cap E)$ has 
    at most $|S|$  connected components 
    that contain a vertex of $S$, all the realizable pairs constructed are partitioned into at most $|S|$ groups in Step (3), and each group satisfies  the conditions of Lemma~\ref{lem:procedure}. Hence, by Lemma~\ref{lem:procedure}, $|W_2| \leq |S|(4c^3 + 3c^2)$. 

    Now we bound the running time of the algorithm. To implement Step (1), we observe that if a bipartition $(T', (T\setminus T') \cup S)$ is enumerated in Step (1), then there must be a simple  $(T', (T\setminus T') \cup S, q)$-cut $(V', V\setminus V')$ with cut-set in $\IA_{G_0}(T_0, t, q, d, 2c+1)\cap E$ such that $G[V']$ is connected, $V' \cap S \neq \emptyset $ and $V' \cap T \subseteq T'$. 
    Hence, we can first enumerate all the simple cuts $(V', V\setminus V')$ of size at most $c$ with one side having volume at most $q$ and containing a vertex from $S$, and take $Q$ to be the union of $V'$ for all the enumerated $(V', V\setminus V')$. By Lemma~\ref{lem:subroutine_repair_set}, this step can be done in $O(|S|q^{c+2} n^{o(1)})$ time, and the size of $Q$ is also upper bounded by $O(|S|q^{c})$. 
    Next we enumerate all the simple cuts that contains at least one vertex of $Q$, keep the bipartitions that satisfying the three conditions of Step (1). By Lemma~\ref{lem:subroutine_repair_set}, this can be done in $O(|S|q^{2c+2} n^{o(1)})$ time. 
    Hence, Step (1) of the algorithm can be implemented in $O(|S|q^{2c+2} n^{o(1)})$ time.
    
    By Lemma~\ref{lem:basic_type_2} and Lemma~\ref{lem:subroutine_repair_set}, all the realizable pairs are partitioned into $|S|$ groups and each group contains at most $(2t)^{c+1}$ realizable pairs. 
    Hence, Step (2) of the algorithm can be implemented in $O(|S|(2t)^{c+3} n^{o(1)})$ time. 
    By Lemma~\ref{lem:elimination_procedure_correctness}, Step (3) can be implemented in $O(|S|(2t)^{c^2 + 3c + 3}n^{o(1)})$ time. 
    
    Hence the overall running time is $O(|S|(2q)^{c^2 + 3c + 3}n^{o(1)})$.
\end{proof}


\paragraph{Type Three Repair Set}
We give an algorithm to construct a type three repair set, and prove the following lemma.

\begin{lemma}\label{lem:type-three-set-construct}
There is an algorithm 
with running time $O(|S|(2t)^{c^2 + 3c + 1} n^{o(1)})$
to compute a set of $|S|(4c^3 +3c^2 + 2c)$ edges whose union with an arbitrary type two repair set is a type three repair set. 
\end{lemma}

We explain why a type two repair set is useful while dealing with type three bipartitions. 
For each type three bipartition $(S\cup T', T\cut T')$ and a simple minimum $(S\cup T', T\cut T', t)$-cut $(V', V\cut V')$, we look at the connected components of $G[V\cut V']$. 
    If there is a connected component that is of volume at most $t$ and some vertex from $T$, then there is a simple $(T\setminus T', S\cup T', t)$-cut of size at most $c$, and thus 
    a $(T\setminus T', S\cup T', q)$-cut of size at most $mincut_G(T\setminus T', S\cup T', t)$ is shattered by $W_2\cup (\IA_{G_0}(T_0, t, q, d, 2c+1)\cap E)$. We can use this cut shattered by $W_2\cup (\IA_{G_0}(T_0, t, q, d, 2c+1)\cap E)$ to obtain a $(S\cup T', T\setminus T', 2t)$-cut of size at most $mincut_G(S\cup T', T\setminus T', t)$ shattered by $W_2\cup (\IA_{G_0}(T_0, t, q, d, 2c+1)\cap E)$. 
    This observation gives us a useful characterization of type three bipartitions whose minimum cuts are not shattered by $W_2\cup (\IA_{G_0}(T_0, t, q, d, 2c+1)\cap E)$ so that we have the following algorithm.
    

 \begin{framed}
\noindent \textbf{Type Three Bipartition Update Algorithm}

    \vspace{.2cm}\noindent \textbf{Input:} Access to 
    $G$, and $G \setminus (\IA_{G_0}(T_0, t, q_0, d, 2c)\cap E)$, and two vertex sets $S, T$. 
    
    

    \noindent \textbf{Output:} A set of edges $W$.  
    
\begin{enumerate}    
\item For each $V^\star$ such that $G[V^\star]$ corresponds to a connected component containing 
a vertex from $S$
of $G\setminus (\IA_{G_0}(T_0, t, q_0, d, 2c)\cap E)$ 
\begin{enumerate}
\item
Find a $V^\dagger$ with volume minimized satisfying the following conditions: (i) $V^\dagger \cap T \neq \emptyset$, (ii) $\vol_G(V^\dagger) > t$, and 
(iii) there is a simple cut $(V', V\setminus V')$ of size at most $c$ with $S\subseteq V'$, $\vol_G(V')\leq t$ and $\endpoints(\partial_G(V'))\subset V^\star$ such that $V^\dagger$ corresponds to a connected component of $G[V\setminus V']$.


\item Let $x$ be an arbitrary vertex in $T \cap V^\dagger$. Collect all $(t, c)$-realizable pairs $(E', V')$ satisfying the following conditions: (i) $V' \cap S = S$, (ii) $E' \subset\partial_G(V')$ induces an atomic cut separating $x$ and $V'$, and (iii) $\endpoints(\partial_G(V')) \subset V^\star$.

\item Run elimination procedure on realizable pairs collected in Step (b). 
\end{enumerate}
\item Return the union of the cut-set of an arbitrary minimum $(S\cup T, \emptyset, t)$-cut, 
$\partial_G(V^\dagger)$ obtained for each execution of Step (1a), and 
the edge sets obtained in  each execution of Step (1c).
\end{enumerate}
\vspace{-.5cm}
\end{framed}






\begin{lemma}\label{lem:type_three_easy}
In the decremental $\IA$ set update scenario, 
for an $\IA_{G_0}(T_0, t, q, d, 2c+1)$ set, let $W_2$ be a type two repair set of the $\IA$ set, and $(V', V\setminus V')$ be a simple $(S \cup T', T\setminus T', t)$-cut of size at most $c$ for some $T'\subsetneq T$ such that one connected component of $G[V\setminus V']$ contains at least one vertex in $T$ and is of volume at most $t$ with respect to $G$.  

There is a $(S \cup T', T\setminus T', 2t)$-cut $(Q, V\setminus Q)$ of size 
at most $mincut_G(S\cup T', T\setminus T', t)$ such that
every connected component of 
\[G\setminus (W_2 \cup (\IA_{G_0}(T_0, t, q, d, 2c + 1) \cap E))\] contains at most $mincut_G(S\cup T', T\setminus T', t) - 1$ edges of $(Q, V\setminus Q)$'s cut-set.
\end{lemma}
\begin{proof}
Let $G'$ denote $G\setminus (W_2 \cup (\IA_{G_0}(T_0, t, q, d, 2c + 1) \cap E))$. 
Let $V^\dagger$ be a set of vertices corresponding to a connected component of $G[V\setminus V']$ satisfying $\vol_G(V^\dagger) \leq t$ and $V^\dagger \cap T \neq \emptyset$. Let $T^\dagger = V^\dagger \cap T$.
There is a $(T^\dagger, S \cup (T\setminus T^\dagger), t)$-cut of size at most $c$ for $G$.
By the definition of type two repair set, there is a $(T^\dagger, (S \cup T)\setminus T^\dagger, q)$-cut  $(V^\diamond, V\setminus V^\diamond)$ of size at most 
$|\partial_G(V^\dagger)|$
such that every connected component of $G'$
contains at most $|\partial_G(V^\dagger)| - 1$ edges of $(V^\diamond, V\setminus V^\diamond)$'s cut-set.
Let $Q$ be $(V' \cup V^\dagger) \setminus V^\diamond$. We have $\vol_G(Q) \leq 2t$ and $Q \cap (S\cup T) = S \cup T'$.
Furthermore, by Lemma~\ref{lem:replacement}, we have  \[\partial_G(Q) \subset (\partial_G(V') \setminus \partial_G(V^\dagger)) \cup \partial_G(V^\diamond).\] 
Note that $|\partial_G(V^\diamond)| \leq |\partial_G(V')|$
and $\partial_G(V^\dagger) \subset \partial_G(V')$, 
hence \[|(\partial_G(V') \setminus \partial_G(V^\dagger)) \cup \partial_G(V^\diamond)| \leq |\partial_G(V')|.\]
And thus
$(Q, V\setminus Q)$ is a $(S\cup T', T\setminus T', 2t)$-cut of size at most $|\partial_G(V')|$.

If $|\partial_G(Q)| <|\partial_G(V')|$, then every connected component of $G'$ contains at most $|\partial_G(V')| - 1$ edges of $\partial_G(Q)$.
If $|\partial_G(Q)| = |\partial_G(V')|$, then
$|\partial_G(V^\diamond)| = |\partial_G(V^\dagger)|$, and \[\partial_G(Q) = \partial_G(V') \setminus \partial_G(V^\dagger)) \cup \partial_G(V^\diamond),\] thus 
$\partial_G(V^\diamond)$ is a subset of $\partial_G(Q)$.
Since every connected component of 
$G'$
contains at most $|\partial_G(V^\dagger)| - 1$ edges of $\partial_G(V^\diamond)$, 
every connected component of $G'$
contains at most $|\partial_G(V')| - 1$ edges of $\partial_G(Q)$.
\end{proof}

\begin{proof}[Proof of Lemma~\ref{lem:type-three-set-construct}]
    We show that the Type Three Bipartition Update Algorithm gives the desired repair set in the desired running time. 
     
    Let $W_3$ be the edge set returned by the algorithm, and $W_2$ be an arbitrary type two repair set with respect to $\IA_{G_0}(T_0, t, q, d, 2c+1)$. 
    Let bipartition $(S\cup T', T \setminus T')$ be a type three bipartition that has a simple minimum $(T', (S\cup T) \setminus T', t)$-cut of size at most $c$.
    Let $C' = (V', V\setminus V')$ be  a simple minimum  $(U, (S\cup T) \setminus U), t)$-cut.
    We only consider the case that the cut-set of $C'$ belongs to a connected component of $G\setminus (\IA_{G_0}(T_0, t, q_0, d, 2c)\cap E)$ containing a vertex from $S$, otherwise we are done by letting $C = C'$ by Lemma~\ref{lem:exclude_Case}.

Let $V^\star$ be the connected component of $G\setminus (\IA_{G_0}(T_0, t, q_0, d, 2c)\cap E))$ that contains the edge set of $C'$, and $V^\dagger$ be the vertex set selected in Step 1(a) when run Step 1 with respect to $V^\star$. 
We consider the following three cases:
\begin{itemize}
\item Case 1. There is a connected component of $G[V\setminus V']$ with volume at most $t$ and containing at lease one vertex of $T$. 
By  Lemma~\ref{lem:type_three_easy},
there is a $(S \cup T', T\setminus T', 2t)$-cut of size at most $mincut_G(S\cup T', T\setminus T', t)$ such that every connected component of $G\setminus (W_2 \cup (\IA_{G_0}(T_0, t, q, d, 2c + 1)\cap E))$ contains at most $mincut_G(S\cup T', T\setminus T', t) - 1$ edges of the cut-set of the cut. 

\item Case 2. $x \notin V'$. By the algorithm and Lemma~\ref{lem:elimination_procedure_correctness},  $W_3$ shatters a $(S\cup T', T\setminus T', \vol_G(V'))$-cut of size at most $mincut_G(S\cup T', T\setminus T', t)$. 

\item Case 3. The first two cases does not hold. We show that $\partial_G(V^\dagger)$ shatters $(V', V\setminus V')$ for this case. 
Since $\vol_G(V^\dagger) > t$, 
$V'$ does not contain all the vertices of $V^\dagger$.
But since $x \in V'$, there must be a subset of $\partial_G(V')$ inducing an atomic cut separating $V'$ and some vertex in $V^\dagger$.
If $V^\dagger$ shatters such an atomic cut-set, then we are done. 
Otherwise, 
all the subsets of $\partial_G(V')$ that induce atomic cuts separating $V'$ and some vertex in $V^\dagger$ are in $G[V^\dagger]$.
Since Case 1 does not hold, 
by our choice of $V^\dagger$, 
$S\cup T$ is on one side of these induced cuts. Thus, there must be an $E'' \subset \partial_G(V')$ inducing a nontrivial bipartition on $S\cup T$ and belonging to $G[V\setminus V^\dagger]$. Hence, $\partial_G(V^\dagger)$ also  shatters $(V', V\setminus V')$.
\end{itemize}
Hence, $W_2 \cup W_3$ is a type 3 repair set.

    Now we bound the size of $W_3$. 
    Since there are at most $|S|$  connected components of \\$G\setminus (\IA_{G_0}(T_0, t, q_0, d, 2c)\cap E)$ that contain a vertex of $S$, and in each execution of Step 1(c), the realizable pairs satisfies  the conditions of Lemma~\ref{lem:procedure}. By Lemma~\ref{lem:procedure}, we have \[|W_3| \leq |S|(4c^3 + 3c^2 + c) + c \leq |S|(4c^3 + 3c^2 + 2c).\]

Now we bound the running time of the algorithm. 
By Lemma~\ref{lem:subroutine_repair_set}, all simple cuts of size at most $c$ with one side containing $S$ and volume at most $t$ can be enumerated in $O(t^{c+2}n^{o(1)})$ time. The total number of such cuts is bounded by $t^c$. 
Hence, by Lemma~\ref{lem:subroutine_repair_set}, for each execution of Step (1a), the running time is $O(t^{c+2}n^{o(1)})$. 
The running time of each execution of Step (1b) is $O((2t)^c n^{o(1)})$. 
By Lemma~\ref{lem:elimination_procedure_correctness}, the running time of each execution of Step (1c) is $O((2t)^{c^2 + 3c + 1} n^{o(1)})$. 
Since we have at most $|S|$ connected components of $G\setminus (\IA_{G_0}(T_0, t, q_0, d, 2c)\cap E)$ that contain vertex from $S$, the overall running time is $O(|S|(2t)^{c^2 + 3c + 1} n^{o(1)})$. 
\end{proof}

\paragraph{Proof of Lemma~\ref{lem:repair_set_algorithm}}
Lemma~\ref{lem:repair_set_algorithm} is then obtained from repair set algorithm for each type. 


\begin{proof}[Proof of Lemma~\ref{lem:repair_set_algorithm}]
By Definition~\ref{def:repair_set}, Definition~\ref{def:repair_set_type} and Lemma~\ref{lem:simple_intercept}, the union of a type 1 repair set, a type 2 repair set, and a type 3 repair set is a repair set. 
Then the lemma follows by 
Lemma~\ref{lem:type-one-set-construct}, Lemma~\ref{lem:type-two-set-construct}, and Lemma~\ref{lem:type-three-set-construct}. 
\end{proof}

Lemma~\ref{lem:repair_set_algorithm} also implies an algorithm to construct an $\IA_G(T, t, 3t, d, 1)$ set for a given graph $G$ and a set of vertices $T \subset V$, because an empty set is an $\IA_{G}(\emptyset, t, t, d, c)$ set for any parameters $t, d$ by Definition~\ref{def:ia_new}.
Since an $\IA_G(T, t, 3t, d, 1)$ set is also an $\IA_G(T, t, q, d, 1)$ set for any $q \geq 3t$ by definition, hence we can construct $\IA_G(T, t, q, d, 1)$ set for any $q \geq 3t$.

\begin{corollary}\label{cor:initial_cut_partition_basic}
Let $G = (V, E)$ be a connected graph of $n$ vertices, $t, q, d$ be three positive integers such that $d \geq 1$, $q \geq 3t$. 
Given access to  graph $G$ and a set of vertices $T \subset V$, there is an algorithm to construct  an $\IA_{G}(T, t, q, d, 1)$ set in time \[O(|T|\cdot (2q)^{d^2 + 3d+3} n^{o(1)})\] such that the $\IA_{G}(T, t, q, d, 1)$ set contains at most $|T| (16d^3 + 12d^2  + 2d)$ edges.
\end{corollary}
\subsection{Decremental Cut Containment Set Update Algorithm}\label{sec:subsec_cut_containment_udpate}

We present our decremental cut containment set update algorithm, and prove Lemma~\ref{lem:cut_containment_update}.

\begin{framed}
\noindent \textbf{Cut Containment Set Update Algorithm}

\noindent \textbf{Input:} Access to 
$G = (V, E)$, 
$G_i = G \setminus ((\cup_{j=1}^i E_j) \cap E)$ for each $1 \leq i\leq c^2 + 2c$, and vertex sets  $T$, $S$, where $E_1, \dots E_{c^2 + 2c}$ forms a recursively constructed $(T_0, c^2+2c)$-cut containment set for graph $G_0$.

\noindent \textbf{Output:} A set of edge $F \subset E$. 


\begin{enumerate}
    \item Set $k_0 \gets 0$, $S_1\gets S$.
    \begin{enumerate}
        \item Set $k_i \gets k_{i-1} + 2(c-i) + 3$.
        
        \item For each connected component $H = (V_H, E_H)$ of $G$ such that $V_H \cap S_i \neq \emptyset$, run Decremental $\IA$ Set Update (DISU) Algorithm with:
        \begin{itemize}
            \item $H$ as $G$ of DISU Algorithm;
            \item $S_i\cap V_H$, $c-i+1, t_{k_{i-1}}, q_{k_i-1} \cdot ((c^2 + 2c)+1), q_{k}\cdot ((c^2 + 2c)+1)$ as $S$, $c, t, q_0, q$ of the DISU Algorithm  respectively;
            \item the induced subgraph of $G_{k_{i} - 1}$ on $V_H$ as $G\setminus (\IA_{G_0}(T_0, t, q_0, 2c+1, 2c) \cap E)$ of DISU Algorithm;
            \item the induced subgraph of $G_{k_{i}}$ on $V_H$ as $G\setminus (\IA_{G_0}(T_0, t, q, 2c+1, 2c+1) \cap E)$ of DISU Algorithm;
            
            
        \end{itemize}
        \item Let $F_i$ be the union of edge sets returned by executions of DISU Algorithm in Step (b).
        
        \item Set $G \gets G_{k_i} \setminus  \cup_{j=1}^{i} F_i$, $S_{i+1} \gets S_i \cup  \endpoints(\cup_{j=1}^i F_i)$. 

    \end{enumerate}
    \item Return $F = \bigcup_{i=1}^{c} F_i$.
\end{enumerate}
\end{framed}

\begin{proof}[Proof of Lemma~\ref{lem:cut_containment_update}]
We first show that \[F\cup \left(\left(\cup_{i=1}^{c^2 + 2c} E_i\right) \cap E\right)\] is a recursively constructed $(T\cup S, c)$-cut containment set of $G$.
By the recursive construction condition and Corollary~\ref{cor:ia_composition}, 
$\cup_{k_{i-1}+1}^{k_i-1}E_i$ is an \[\IA_{G_0\setminus \left(\cup_{j=1}^{k_{i-1}} E_i\right)}\left(T_i, t_{k_{i-1}+1}, q_{k_i - 1} \cdot ((c^2 + 2c)+1), 2(c-i) + 3, 2(c-i) + 2\right)\] set, and 
$\cup_{k_{i-1}+1}^{k_i}E_i$ is an \[\IA_{G_0\setminus \left(\cup_{j=1}^{k_{i-1}} E_i\right)}\left(T_i, t_{k_{i-1}+1}, q_{k_i} \cdot((c^2 + 2c)+1), 2(c-i) + 3, 2(c-i) + 3\right)\] set.
Since we run the Decremental IA Set Update Algorithm for each connected component of \[G\setminus \left(\left(\cup_{j=1}^{k_{i-1}} E_i  \cap E\right) \bigcup\left(\cup_{j=1}^{i-1}F_{j}\right)\right)\] that contains a vertex of $S_i$,
by induction and Lemma~\ref{lem:repair_set_algorithm},
$F_i \cup \left(\cup_{j=k_{i-1}+1}^{k_i} E_j\right)$ is an

\begin{equation*}
\IA_{F_i \cup \left(\cup_{j=k_{i-1}+1}^{k_i} E_j\right)}\wrap{T_i\cup S_i,t_{k_{i-1}+1}, q_{k_i} \cdot ((c^2 + 2c)+2), c - i + 1, 1}\end{equation*}
set. By Corollary~\ref{cor:ia_composition}, 
\[F\cup \left(\left(\cup_{i=1}^{c^2 + 2c} E_i\right) \cap E\right)\] is a recursively constructed cut containment set.


Note that 
the Decremental IA Set Update Algorithm 
is invoked only at Step (1b) on connected component that contains vertices of $S \cup \endpoints(\cup_{j=1}^{i-1}F_j)$ for $i$-th iteration of the for loop. 
By Lemma~\ref{lem:repair_set_algorithm} and induction, 
\[|F_i| = O(|S \cup \endpoints(\cup_{j=1}^{i-1} F_j)|(32c^3 + 24c^2 + 4c)^{i}).\]
Hence $F$ is of size at most $O(|S|(10c)^{O(c)})$. 
The total running time is bounded by 
\[O\wrap{\abs{S} (2q)^{O(c^2)} n^{o(1)} }.\]
\end{proof}

\section{\texorpdfstring{$c$}{c}-Edge Connectivity Sparsifier}
\label{sec:sparsifier}
We formally define our sparsifiers and prove some useful properties. 
Our definition makes use of tree contraction~developed in~\cite{henzinger1997fully,holm2001poly}.
Below we give the definition of tree contraction. The properties of tree contraction are given in Appendix~\ref{sec:spanning_forest_contraction}. 

\begin{definition}
Given a forest $F = (V, E)$ and a set of key vertices $K$. 
The contraction of $F$ with respect to $K$, denoted by $\contract_S(F)$, is a minimal forest (in terms of number of vertices and edges) satisfying the following two conditions:
\begin{enumerate}
    \item Edges of $\contract_S(F)$ correspond to edge disjoint paths of $F$.
    \item Each edge of $\contract_S(F)$ corresponds to a path of $F$ that is part of a path connecting two vertices of $S$. 
\end{enumerate}
\end{definition}

Based on the definition of tree contraction, we define sparsifier with respect to terminals as follows:

\begin{definition}
Let $G = (V, E)$ be a connected graph, 
$T$ be a set of terminals, 
$\cc$ be the a cut containment sets of $G$ with respect to $T$, 
$F$ be a spanning forest of $G\cut \cc$.  The \textit{sparsifier of $G$ with respect to $T$, $\cc$, and a parameter $\gamma > c$}, denoted by $\sparsifier(G, T, \cc, \gamma)$, is a multigraph with the vertex set of $\contract_{T}(F)$ as vertices, 
 the union of $\cc$ and the edge set of $\contract_{T}(F)$ as edges, with multiplicities:

 \begin{itemize}
     \item For every edge $(x, y)$ of $\sparsifier(G, T, \cc, \gamma)$ such that $x$ and $y$ are in different connected component of $G\cut \cc$, the multiplicity of $(x, y)$ in $\sparsifier(G, T, \cc, \gamma)$ is equal to the multiplicity of $(x, y)$ in $G$.
     \item For every edge $(x, y)$ of $\sparsifier(G, T, \cc, \gamma)$ such that $x$ and $y$ are in the same vertex set of $G \cc$, the multiplicity of $(x, y)$ in $\sparsifier(G, T, \cc, \gamma)$ is $\gamma$.
 \end{itemize}
\end{definition}

We define our one-level sparsifier with respect to vertex partition as follows.


\begin{definition}[One-level sparsifier]\label{def:one_level_sparsifier}
Let $G = (V, E)$ be a graph, $\mathcal P $ be a vertex partition of $G$ such that $G[ P]$ is connected for every $P\in\mathcal P$, $\cc$ be the union of the cut containment sets of $G[P]$ for each $P\in \mathcal P$, 
$F$ be a spanning forest of $G[\mathcal P]\cut \cc$.  The \textit{sparsifier of $G$ with respect to $\mathcal P$, $\cc$, and a parameter $\gamma > c$}, denoted by $\sparsifier(G, \mathcal P, \cc, \gamma)$, is a multigraph with the vertex set of $\contract_{\endpoints(\partial_G(\calP) \cup \cc)}(F)$ as vertices, 
 the union of $\partial_G(\calP) \cup \cc$ and the edge set of $\contract_{\endpoints(\partial_G(\calP) \cup \cc)}(F)$ as edges, with multiplicities:

 \begin{itemize}
     \item For every edge $(x, y)$ of $\sparsifier(G, \mathcal P, \cc, \gamma)$ such that $x$ and $y$ are in different connected component of $G[\mathcal P]\cut \cc$, the multiplicity of $(x, y)$ in $\sparsifier(G, \calP, \cc, \gamma)$ is equal to the multiplicity of $(x, y)$ in $G$.
     \item For every edge $(x, y)$ of $\sparsifier(G, \mathcal P, \cc, \gamma)$ such that $x$ and $y$ are in the same vertex set of $G[\mathcal P]\cut \cc$, the multiplicity of $(x, y)$ in $\sparsifier(G, \calP, \cc, \gamma)$ is $\gamma$.
 \end{itemize}
 Equivalently, $\sparsifier(G, \calP, \cc, \gamma)$ can be obtained as follows: For each $P \in \calP$, replace $G[P] = (P, E_P)$ in $G$ by $\sparsifier(G[P], \endpoints(\partial_G(P)) \cap P, \cc \cap E_P, \gamma)$.
\end{definition}

With the definition of one-level sparsifier, we prove the following property of one-level sparsifier. 


\begin{lemma}\label{lem:sparsifier_basic}
Given a graph $G = (V, E)$, 
 a vertex partition $\calP$ of $V$ such that $G[P]$ is connected for every $P \in \calP$,
a cut containment set $\cc$ for $G[\mathcal P]$, 
and parameter $\gamma > c$, we have 

\begin{enumerate}
    \item $\simple(\sparsifier(G, \calP, \cc, \gamma))$ contains at most $3 \wrap{|\partial_G(\calP)| + \abs{\cc}}$  edges. 
    \item $\sparsifier(G, \calP, \cc, \gamma)$ is a $c$-edge connectivity equivalent graph of $G$ with respect to $\endpoints(\partial_G(\calP))$.
\end{enumerate}
\end{lemma}

\begin{proof}
The first property is implied by Claim~\ref{claim:contracted_graph_basic}. We prove the second property in the rest of this proof.

Let $x,y$ be two distinct vertices in $\endpoints\wrap{\partial_G(\mathcal P)}$. 
By Definition~\ref{def:one_level_sparsifier}, 
$x$ and $y$ are in the same connected component of $G$ if and only if $x$ and $y$ are in the same connected component of $\sparsifier(G, \calP, \cc, \gamma)$.
In the rest of this proof, we assume $x$ and $y$ are in the same connected component of $G$. 

Let $H = (V_H, E_H)$ be $\sparsifier(G, \calP, \cc, \gamma)$.
Let $V'$ be the vertex set of the connected component containing $x$ and $y$ in $G$, and $V^\dagger$ be the vertex set of the connected component containing $x$ and $y$ in $H$. 

We show that 
for any $\alpha \leq c$,
the size of minimum cut separating $x$ and $y$ in $G[V']$ is $\alpha$  if and only if the size of minimum cut separating $x$ and $y$ in $H[V^\dagger]$ is $\alpha$.

Consider the case that the size of minimum cut  separating $x$ and $y$ in $G[V']$ is $\alpha$. 
Let $C_1 = (V_1, V' \setminus V_1)$ be such a minimum cut.
$C_1$ must be an atomic cut, otherwise it is not minimum.
Let $F = \partial_{G[V']}(V_1) \cap \partial_G(\calP)$, 
\[\calP' = \{P \in \calP: P \cap V_1 \neq \emptyset \text{ and } P \cap (V'\setminus V_1) \neq \emptyset \}, \]
and $F_P$ be the edges of $\partial_{G[V']}(V_1)$ that are also in ${G[P]}$ for every $P \in \calP'$.



For any $P \in \calP'$, 
let $T_P = \endpoints(\partial_G(\mathcal P))\cap P$. 
Since $C_1$ is an atomic cut for $G[V']$, 
both $T_P \cap V_1$ and $T_P \cap (V' \setminus V_1)$ are not empty sets.
Hence, the size of the minimum cut that partitions $T_P$ into $T_P \cap V_1$ and $T_P \cap (V' \setminus V_1)$ is $|F_P|$, otherwise $C_1$ is not the minimum cut separating $x$ and $y$. 
By Definition~\ref{def:one_level_sparsifier}, there is an edge set  $F'_P$  of size $|F_P|$  within $G[P]$
such that $F'_P$ is a subset of $\cc$, and is the cut-set of a minimum cut on $G[P]$  partitioning $T_P$ into   $T_P \cap V_1$ and $T_P \cap (V' \setminus V_1)$.
Hence, $F \cup \wrap{\bigcup_{P \in \calP'} F_P'}$ induces a cut separating $x$ and $y$ for $G[V']$ of size $\alpha$, and $F \cup (\bigcup_{P \in \calP'} F_P')$ is a subset of $\partial_G(\calP)\cup \cc$. 
By the definition of $\sparsifier(G, \calP, \cc, \gamma)$, $F \cup (\bigcup_{P \in \calP'} F_P')$ induces a cut separating $x$ and $y$ in $H[V^\dagger]$ of size $\alpha$.

On the other hand, by Definition~\ref{def:one_level_sparsifier},
for any cut of size at most $c$ separating $x$ and $y$ in $H[V^\dagger]$, 
its cut-set also induces a cut of same size for $G[V']$ separating $x$ and $y$.

Hence, the size of minimum cut separating $x$ and $y$ in $G[V']$ is $\alpha$  if and only if the size of minimum cut separating $x$ and $y$ in $H[V^\dagger]$ is $\alpha$ for any $\alpha \leq c$.
Then the second property holds by Definition~\ref{def:c-edge-equivalent-graph-intro}. 
\end{proof}

We define multi-level sparsifier as follows. 

\begin{definition}\label{def:multi_sparifier_intro}
A $(\phi, \eta, \gamma)$ \emph{multi-level 
$c$-edge connectivity sparsifier} for a  graph $G$ of at most  $m$ distinct edges and parameters $0 < \phi, \eta \leq 1$, $\gamma \geq c + 1$ is a set of tuples $\{(G^{(i)},  \calP^{(i)}, \cc^{(i)})\}_{i = 0}^\ell$ such that the following conditions hold:
\begin{enumerate}\itemsep -3pt
\item $G^{(0)}=G$; for $i > 0$, $G^{(i)} = \sparsifier(G^{(i - 1)}, \calP^{(i - 1)}, \cc^{(i - 1)}, \gamma)$ such that $G^{(i)}$ contains at most $m \eta^i$ distinct edges. 
\item 
$\calP^{(i)}$ is a $\phi$-expander decomposition of $\simple(G^{(i)})$ for each $0 \leq i \leq \ell$, and $\calP^{(\ell)}$ contains only one cluster.  
\item $\cc^{(i)}$ is the union of $(\endpoints(\partial_{G^{(i)}}(P)) \cap P, c)$-cut containment sets  for all $P\in \calP^{(i)}$.
 \end{enumerate}
\end{definition}

By Lemma~\ref{lem:sparsifier_basic} and Definition~\ref{def:multi_sparifier_intro}, we have the following corollary. 
\begin{corollary}
For a $(\phi, \eta, \gamma)$ multi-level $c$-edge connectivity sparsifier $\{(G^{(i)},  \calP^{(i)}, \cc^{(i)})\}_{i = 0}^\ell$ for $\gamma \geq c + 1$, if $x, y$ are vertices in $\endpoints(\partial_{G^{(i)}}(\calP^{(i)}))$ for every $0 \leq i < \ell$, then the $c$-edge connectivity between $x$ and $y$ in $G$ is the same as that of $G^{(\ell)}$.
\end{corollary}

\section{One-Level Connectivity Sparsifier Update Algorithm}\label{sec:one_level_update}
Based on the decremental cut containment set update algorithm obtained in Section~\ref{sec:update_ia_set},
we present our one-level sparsifier update algorithm and prove Lemma~\ref{lem:one-level-update-intro}. 
Starting from this section, we assume every vertex of our multiple graph has a constant number of neighbors. In Section~\ref{sec:fully}, we will show that this is no loss of generality in assuming so. 

\vspace{-.15cm}\begin{lemma}\label{lem:one-level-update-intro}
Given access to a mutligraph $G$, a $\phi$-expander decomposition $\calP$ of $\simple(G)$, a 
set of edges $\cc$ that is the union of $(\endpoints(\partial_G(P)) \cap P, c^2 + 2c)$-cut containment sets of $G[P]$ for each $P\in \calP$
 recursively constructed with respect to parameters $t_1, q_1, \dots, t_{c^2 + 2c},$ $ q_{c^2 + 2c}$ satisfying Equation~(\ref{equ:ia_composition_condition}),
and a multigraph update sequence $\updateseq$, 
there is a deterministic algorithm with running time $\widehat O(|\updateseq|m^{o(1)} / \phi^3 )$ to update $G$ to $G'$, $\calP$ to $\calP'$, $\cc$ to $\cc'$~\footnote{We use $\widetilde O(f)$ to denote $O(f \cdot \text{polylog}(f))$ and $\widehat O(f)$ to denote $O(f^{1 + o(1)})$.},
and output 
a multigraph update sequence $\newupdateseq$ of length  $|\updateseq|(10c)^{O(c)}$ 
such that the following conditions hold
\vspace{-.1cm}\begin{enumerate}\itemsep -3pt
\item $G'$ is the resulting graph of applying $\updateseq$ to $G$.
\item $\calP'$ is a $\phi / 2^{O(\log^{1/3} m \log^{2/3}\log m)}$-expander decomposition of $\simple(G')$.

\item  $\cc'$ is a union of $(\endpoints(\partial_{G'}(P)) \cap P, c)$-cut containment sets of $G'[P]$ for each $P\in \calP'$
with respect to parameters $t_1', q_1', \dots, t_{c}', q_c'$ satisfying 
Equation~(\ref{equ:new_inequality}).
\item $\newupdateseq$ updates
 $\sparsifier(G, \calP, \cc, \gamma)$ to  $\sparsifier(G', \calP', \cc', \gamma)$ for any $\gamma > c^2 + 2c$.
 \item  
Every vertex involved in the update sequence $\updateseq$ becomes a singleton of $\calP'$ if the vertex is in $G'$, where a singleton is a cluster of one vertex.
\end{enumerate}
\end{lemma}

\subsection{Expander Decomposition Update}
We first give our algorithm to update expander decomposition with respect to edge deletions.
Given a graph $G = (V, E)$, a $\phi$-expander decomposition $\calP$ of $\simple(G)$, and a set of vertex $S \subset V$, 
our goal is to update $\calP$ to a $\phi/\factor$-expander decomposition $\calP'$ of $\simple(G)$ such that every vertex of $S$ is a singleton in $\calP'$.

\begin{lemma}\label{lem:expander_decompoisition_update}
Let $G = (V, E)$ be a multigraph of $n$ vertices such that every vertex has at most $\Delta$ different neighbors,
$\calP$ be a $\phi$-expander decomposition of $\simple(G)$, and 
$S \subset V$ be a set of vertices. 
Given access to $G$, $\calP$, $G[\calP]$, and $S$, 
there is a deterministic algorithm with running time $\widehat O(|S|\Delta \factor / \phi^3)$ 
to 
update $\calP$ to $\calP'$, and $G[\calP]$ to $G[\calP']$ satisfying the following conditions
\begin{enumerate}
    \item $\calP'$ is a refinement of $\calP$, and is a $(\phi / \factor)$ expander decomposition of $G$. 
    \item $\partial_G(\calP) \setminus \partial_G(\calP')$ contains at most $O(k\Delta)$ distinct edges. 
    \item Every vertex in $S$ is singleton of $\calP'$. 
\end{enumerate}


\end{lemma}


Lemma~\ref{lem:expander_decompoisition_update} uses the deterministic expander decomposition algorithm by  
Chuzhoy et al.~\cite{chuzhoy2019deterministic}, and dynamic expander pruning algorithm 
by  Saranurak and Wang in \cite{saranurak2019expander} in an offline way.

\begin{theorem}[cf. Corollary 7.1 of ~\cite{chuzhoy2019deterministic}]\label{thm:expander_decomposition}
Given a simple graph $G = (V, E)$ of $m$ edges and a parameter $\phi$,
there is 
a constant $\delta > 0$ and 
a deterministic algorithm  $\expanderdecomposition$ to compute a $(\phi, \phi/ \truefactor )$-expander decomposition of $G$ in time $\widehat O(m/ \phi^2)$.
\end{theorem}

%


\begin{theorem}[rephrased, cf. Theorem 1.3 of \cite{saranurak2019expander}]\label{thm:pruning}
Let $G = (V, E)$ be a simple $\phi$-expander with $m$ edges. 
Given access to adjacency lists of $G$ and a set $D$ of $k \leq \phi m / 10 $ edges, 
there is a deterministic algorithm $\pruning$ to find a pruned set $P \subseteq V$ in time $O(k \log m/\phi^2)$ such that all of the following conditions hold:
\begin{enumerate}
\item $\vol_G(P) = 8 k  / \phi$.
\item $|E_G(P, V \cut P)| \leq 4k$.
\item $G'[V \cut P]$ is a $\phi / 6$ expander, where $G' = (V, E \cut D)$.
\end{enumerate}
\end{theorem}

Our algorithm first collects all the edges incident to vertices of $S$ and run the expander pruning algorithm with collected edges. For pruned vertices, we further run a expander decomposition to make sure the number of new intercluster edges is linear with respect to the size of $S$ (assuming every vertex only has a constant number of neighbors).

\begin{figure}[htb]
 \begin{framed}
\noindent \textbf{Expander Decomposition Update Algorithm}

    \vspace{.2cm}\noindent \textbf{Input:} Graph $G$, $G[\calP]$, vertex partition $\calP$,  vertex set $S$ and conductance parameter $\phi$. 
    

    \noindent \textbf{Output:} $\calP'$ updated from $\calP$, and $G[\calP']$ updated from $G[\calP]$. 
    
\begin{enumerate}
\item For each $P \in \calP$ such that $P \cap S \neq \emptyset$, 
\begin{enumerate}    
\item Let $m_P$ denote the number of edges in $\simple(G[P])$, and $D_P$ be the set of edges in $\simple(G[P])$ that have at least  one endpoint in $S$.
\item If $|D_P| > m_P \phi / 10$, then 
\begin{enumerate}
    \item Run expander decomposition on each connected component of $\simple(G[P])\setminus D_P$ with conductance parameter $\phi / \factor$.
    \item Replace $P$ in $\calP$ by the union of obtained clusters in Step i. 
    \end{enumerate}
\item If $|D_P| \leq m_P \phi / 10$, then 
\begin{enumerate}
\item 
Run expander pruning on $\simple(G)$ with edge set $D_P$, and let $Q_P$ denote the resulting pruned vertex set. 
\item 
Run expander decomposition on each connected component of $\simple(G[Q_P]) \setminus D_P$ with parameter $\phi / \truefactor$. 
\item 
Replace $P$ in $\calP$ by the union of $\{P\setminus Q_P\}$ and the clusters obtained in Step ii. 
\end{enumerate}
\end{enumerate}
\item $\calP' \gets \calP$, and update $G[\calP']$ by removing all the edges that connects two vertices not in the same vertex set of $\calP'$. 
\end{enumerate}
\vspace{-.5cm}\end{framed}
\end{figure}

\begin{proof}[Proof of Lemma~\ref{lem:expander_decompoisition_update}]
The running time is obtained by Theorem~\ref{thm:expander_decomposition} and Theorem~\ref{thm:pruning}.
In the rest of this proof, we show that the vertex partition satisfies desirable properties by showing that for each $P \in \calP$ that contains at least one vertex of $S$, $P$ is partitioned into a few vertex sets $\mathcal{R}_P = \{P' \in \calP': P' \cap S \neq \emptyset \}$ in $\calP'$ such that the following conditions hold
\begin{enumerate}
    \item Every vertex of $S\cap P$ is a singleton in $\calP'$. 
    \item $\mathcal{R}_P$ is a $\phi / \factor$-expander decomposition of $\simple(G[P])$.
    \item $\partial_{\simple(G[P])}(\mathcal{R}_P)$ contains at most $O(|S\cap P|\Delta)$ distinct edges. 
\end{enumerate}
The first property is obtained by the fact that every vertex of $S\cap P$ is an isolated vertex in $\simple(G[P]) \setminus D_P$. The second property is obtained by Theorem~\ref{thm:expander_decomposition} and Theorem~\ref{thm:pruning}. 
So in the rest, we prove the third property. 

Let $G' = \simple(G[P] \setminus D_P)$. If $|D_P| \leq m_P \phi / 10$, by Theorem~\ref{thm:expander_decomposition} and Theorem~\ref{thm:pruning}, the following properties hold
\begin{enumerate}
\item[(a)]  Every connected component of $G'[P\cut Q_P]$ is a $\phi / 6$ expander;
\item[(b)] $|\partial_{G'}(Q_P)| \leq 8|D_P|$; 
\item[(c)] $\vol_{G'}(Q_P) \leq 8|D_P| / \phi$.
\end{enumerate}
If $|D_P|> m_P\phi / 10$, let $Q_P = P$. We have  $\vol_{G'}(Q_P) = \vol(G') \leq 2m_P < 20 |D_P| / \phi$.
Let $H$ denote $G'[Q_P]$ for either case. We have $\vol_H(Q_P) \leq 20|D_P| / \phi$. 

Note that we run expander decomposition on each connected component of $H$ with conductance parameter $\phi / \truefactor$ for both cases, where $\delta$ is the constant in Theorem~\ref{thm:expander_decomposition}.
For a vertex set $V'$ that forms a connected component of $H$, let $\alpha$ denote the number of edges in $H[V']$. 
The output of expander decomposition on $H[V']$
is a $\phi / \truefactor$-expander decomposition of $H[V']$ with $O(\phi \alpha)$ intercluster edges  
  by Theorem~\ref{thm:expander_decomposition}. 
Since the total number of edges in $H$ is at most $20|D_P| / \phi$, 
sum over all the connected components of $H$, we obtain a $(\phi / \truefactor, \phi)$-expander decomposition of $H$ with $O(|D_P|)$ intercluster edges.

Hence, for both cases, 
$\mathcal{R}_P$ is a $(\phi / \truefactor)$-expander decomposition of  $\simple(G[P] \setminus D_P)$ with $O(|D_P|) = O(|S\cap P| \Delta)$ intercluster edges. 
Since $D_P$ are intercluster edges of $\calR_P$ with respect to $\simple(G[P])$ by the first property, $\mathcal{R}_P$ is a $(\phi / \truefactor)$-expander decomposition of  $\simple(G[P])$ with $O(|S\cap P| \Delta)$ intercluster edges. 
\end{proof}

\subsection{One-Level Sparsifier Update Algorithm}

In the rest of this section, we present 
our one level sparsifier algorithm and prove Lemma~\ref{lem:one-level-update-intro}. 

We first give an overview of the data structure used.  The detailed data structure is discussed in Appendix~\ref{sec:data_structure}.
Recall that $\cc$ corresponds to the union of cut containment sets of $G[P]$ with $\endpoints(\partial_G(P))$ as terminals for each $P\in \calP$. 
Each cut containment set is recursively constructed, i.e., the edges in $\cc$ are partitioned into $E_1, E_2, \dots, E_{c^2 + 2c}$ so that for each $G[P] = (P, E_P)$ of $P\in\calP$, $E_i \cap E_P$ for all $1 \leq i \leq c^2 + 2c$ satisfy the recursive construction condition (see Section~\ref{sec:update_ia_set}) for $G[P]$ with  $\endpoints(\partial_G(P))\cap P$ as terminals. 
We also assume the graph data structure contains $G\setminus (\cup_{j=1}^i E_j)$ for each $1 \leq i \leq c^2 + 2c$ so that we can run Cut Containment Set Update Algorithm. 
For notation convenience, in this section, all the algorithm descriptions omit $G\setminus (\cup_{j=1}^i E_j)$ in the input. 

In addition, we assume that we maintain graph $G[\calP] \setminus \cc$ and a spanning forest $F$ of graph $G[\calP] \setminus \cc$ with $\endpoints(\partial_G(\calP) \cup \cc)$ as terminals, denoted by $F_{G[\calP] \setminus \cc}$, so that for each update (vertex insertion/deletion, edge insertion/deletion or terminal insertion/deletion) to the graph, 
there is an algorithm to update the spanning forest, and to output an update sequence of $O(1)$  length that updates the contraction of the spanning forest with respect to the terminals. 

\begin{lemma}\label{lem:data_structure_pre_new}
For a dynamic multigraph $G$ and a terminal set $T$ both undergoing updates,
there is an algorithm to construct and maintain a data structure of $G$ that contains a simple spanning forest $F$ of $G$ such that for each update (vertex/edge/terminal insertion/deletion), the algorithm output an update sequence that update $\contract_T(F)$ respectively.

Assume $G$ contains at most $n$ vertices and $m$ edges throughout the updates. 
The algorithm takes  $\widehat O(n+m)$ time to construct the data structure,
and $n^{o(1)}$ update time for each update. 

\end{lemma}

Lemma~\ref{lem:data_structure_pre_new} is discussed in detail in Appendix~\ref{sec:data_structure}.\\


Now we are ready to give our one-level sparsifier algorithm. 
Our algorithm first updates $\phi$-expander decomposition $\calP$ to a refined partition 
$\calP'$  of $G$ satisfying the following three properties:
\begin{enumerate}
\item Any vertex involved in the multigraph update sequence $\updateseq$ is a singleton of $\calP'$ if the vertex is already in $G$.
\item $\calP'$ is a $\phi / \factor$-expander decomposition of $\simple(G)$.
\item The number of distinct new intercluster edges is no more than $O(|\updateseq|)$. 
\end{enumerate}
To achieve this goal, we 
collect all the incident edges to the vertices involved in $\updateseq$, and run the Expander Decomposition Update Algorithm with the collected edges. 
By Lemma~\ref{lem:expander_decompoisition_update}, 
the new partition is a $\phi / \factor$ expander decomposition of $\simple(G)$, and all the vertices involved in the update sequence are  singletons in the updated vertex partition. 


With the updated partition $\calP'$, we further update $\cc$ to $\cc'$ that is the union of $c$-cut containment sets for each connected component of $G[\calP']$ by  applying Lemma~\ref{lem:cut_containment_update} on every affected cluster of $\calP$, making use of the condition that $\calP'$ is a refinement of $\calP$. 
By Lemma~\ref{lem:cut_containment_update}, $\cc' \setminus \cc$ contains at most $|\updateseq|(10c)^{O(c)}$  distinct edges. By the properties of contracted graph, 
the total number of vertex and edge insertions and deletions that transform
 $\sparsifier(G, \calP,\cc, \gamma)$ to $\sparsifier(G, \calP',\cc', \gamma)$ is at most $O(|\updateseq|(10c)^{O(c)})$.  

At the end, we apply the update sequence $\updateseq$ to $G$ to obtain the updated graph $G'$. Since all the vertices  involved in the update sequence $\updateseq$ are singletons in $\calP'$,
to update the one-level sparsifier from $\sparsifier(G, \calP', \cc', \gamma)$ to $\sparsifier(G', \calP', \cc', \gamma)$,
we only need to apply $\updateseq$ to  $\sparsifier(G, \calP', \cc', \gamma)$.
 So the length of update sequence for the sparsifier is $|\updateseq| (10c)^{O(c)}$.

 \begin{framed}
\begin{flushleft}
\textbf{One-Level Sparsifier Update Algorithm}

    \vspace{.2cm}\textbf{Input:} Graph $G$, $G[\calP]$, vertex partition $\calP$, edge set $\cc$, spanning forest $F$ of $G[\calP] \setminus \cc$, update sequence $\updateseq$, and conductance parameter $\phi$. 
    

    \textbf{Output:} $\calP'$ updated from $\calP$, and $G[\calP']$ updated from $G[\calP]$, $\cc'$ updated from $\cc$, $F'$ updated from $F$, and $\newupdateseq$. 
    
\begin{enumerate}
\item Let $S$ be all the vertices involved in $\updateseq$. Run the Expander Decomposition Update Algorithm with respect to $G, \calP$ and $S$ to update $\calP$ to $\calP'$. 
\item 
For each newly added cluster $P \in \calP'$, 
let $P_0$ be the cluster in $\calP$ that contains $P$, and 
run Cut Containment Update Algorithm with
$G[P]$ (as $G$), $G[P_0]$ (as $G_0$),
$\endpoints(\partial_G(P_0)) \cap P_0$ as $T_0$, and $\endpoints(\partial_G(P)) \setminus \endpoints(\partial_G(P_0))$ as $S$. \item 
Add all the edges returned by Cut Containment Update Algorithm to $\cc$, and remove all the edges not in $G[\calP']$ from $\cc$. Denote the result edge set as $\cc'$. 
\item Update $F$ to $F'$ that is the spanning forest of $G[\calP']\setminus \cc'$ with $\endpoints(\partial_G(\calP')\cup \cc')$ as terminals, and get $\newupdateseq$ that updates contraction of spanning forest of $G[\calP]\setminus \cc$ w.r.t $\endpoints(\partial_G(\calP)\cup \cc)$ to that of $G[\calP']\setminus \cc'$ w.r.t. $\endpoints(\partial_G(\calP')\cup \cc')$ (with edge multiplicity to be $\gamma$ for each edge insertion in $\updateseq$).
\item Append edge insertions of  all the edges $(\partial_G(\calP') \setminus \partial_G(\calP')) \cup (\cc' \setminus \cc)$ to the end of $\newupdateseq$.
\item Apply $\updateseq$ to $G$, and add/delete singletons in $\calP'$ for vertex insertions/deletions. 
Append $\updateseq$ to the end of $\newupdateseq$ and return $\newupdateseq$
\end{enumerate}
\end{flushleft}
\end{framed}
\begin{proof}[Proof of Lemma~\ref{lem:one-level-update-intro}]

By Lemma~\ref{lem:expander_decompoisition_update}, after Step (1), $\calP'$ is a refined partition of $\calP$ such that every vertex involved in $\updateseq$ (if exists in $G$) is a singleton in $\calP'$ such that $\partial_G(\calP') \setminus \partial_G(\calP)$ contains at most $O(|\updateseq|\Delta)$ distinct edges.

By Lemma~\ref{lem:cut_containment_update}, after Step (3), $\cc'$ is the union of $c$-cut containment sets for each $P \in \calP'$ such that $\cc' \setminus \cc$ contains at most $O(|\updateseq|(10c)^{O(c)})$ distinct edges. Hence, by Lemma~\ref{lem:data_structure_pre_new}, the $
\newupdateseq$ obtained in Step (4) 
is of length $O(|\updateseq|(10c)^{O(c)})$.

By the definition of one level sparsifier, the $\newupdateseq$ obtained after Step (5) is of length \\$O(|\updateseq|(10c)^{O(c)})$, and updates $\sparsifier(G, \calP, \cc, \gamma)$ to $\sparsifier(G, \calP', \cc', \gamma)$.

Since each vertex involved in $\updateseq$ is an isolated singletons of $\calP'$ if the vertex is in $G$,
$\updateseq$ can also be applied to the sparsifier. 
Hence, applying $\updateseq$ to $\sparsifier(G, \calP', \cc', \gamma)$ updates $\sparsifier(G, \calP', \cc', \gamma)$ to $\sparsifier(G', \calP', \cc', \gamma)$, where $G'$ is the result of applying $\updateseq$ on $G$.

The running time is obtained by Lemma~\ref{lem:cut_containment_update}, Lemma~\ref{lem:expander_decompoisition_update}, and Lemma~\ref{lem:data_structure_pre_new}.
\end{proof}

\section{Multi-Level Connectivity Sparsifier Update algorithm}\label{sec:multilevel_update}

In this section, we present our algorithm to construct and update multi-level sparsifier. Our goal is to maintain a multi-level sparsifier so that for some $\parametertimesub = O(\log \log \log n)$,
after applying $\parametertimesub$ times of the multi-level sparsifier update algorithm to the multi-level sparsifier with respect to $\parametertimes$ update sequences $\updateseq_1, \updateseq_2, \dots, \updateseq_\parametertimes$, the resulting sparsifier is a multi-level $c$-edge connectivity sparsifier.

To achieve this goal, we carefully set the parameters so that after initialization, the multi-level sparsifier is a $(\phi_0, \phi_0 \factor, \gamma)$-sparsifier such that the cut containment set at each level is a $c_0$ cut containment set for some $\phi_0 = n^{o(1)}$, $c_0 > c$, and $\gamma > c_0$. 
After applying $i$-th update algorithm, the sparsifier is a \[(\phi_0 / 2^{O(i \log^{1/3} n \log^{2/3} \log n)}, \phi_0 2^{O(i \log^{1/3} n \log^{2/3} \log n)}, \gamma)\]
sparsifier such that the cut containment set at each level is a $c_i$ cut containment set such that $c_{i-1} = c_i^2 + 2c_i$. 
We make sure that $c_\parametertimes = c$. 
Using the fact that $c'$-cut containment set is also a $c$-cut containment set for any $c' > c$, 
after applying update algorithm on $\updateseq_1, \updateseq_2, \dots, \updateseq_\parametertimes$, 
the multi-level sparsifier is a $c$-edge connecticity sparsifier of the graph after applying the $\parametertimes$ update sequences to the graph. 

\subsection{Parameter Setting}

We set the parameters as Figure~\ref{fig:multi-level} for a 
given $c = (\log m)^{o(1)}$.
We say a $c$-edge connectivity multi-level sparsifier $\{G^{(i)}, \calP^{(i)}, \cc^{(i)}\}_{i=0}^\ell$ is with respect to parameters $t_1, t_2, \dots, t_c, q_c$ if each $\cc^{(i)}$ is a union of $(\endpoints(\partial_{G^{(i)}}(P)) \cap P, c)$-cut containment set for $G^{(i)}[P]$ for each $P \in \calP^{(i)}$ such that the cut containment sets correspond to $\IA$ sets recursively constructed with respect to parameters $t_1, q_1, \dots, t_c, q_c$.
We have the following properties of multi-level sparsifier.
\begin{enumerate}
    \item After initialization, the multi-level sparsifier is a $(\phi_0, \eta_0, \gamma)$ multi-level sparsifier for $c_0$-edge connectivity of the input graph with respect to parameters $t_{0, 1}, q_{0, 1}, \dots, t_{0, c_0}, q_{0, c_0}$.
    \item After applying update algorithm with respect to update sequences \[\updateseq_1, \updateseq_2, \dots, \updateseq_i,\] the multi-level sparsifier is a $(\phi_i, \eta_i, \gamma)$ multi-level sparsifier for $c_i$-edge connectivity of the input graph with respect to parameters $t_{i, 1}, q_{i, 1}, \dots, t_{i, c_i}, q_{i, c_i}$.
\end{enumerate}

We prove the following properties about parameter setting.
\begin{figure}[ht]
\begin{framed}
\begin{align*}
& \phi \eqdef \frac{1 }{2^{\log^{3/4} n}}, 
\parametertimesub \eqdef \left \lfloor \log \left( \frac{\log \log n / 10}{\log (4c)} \right)\right \rfloor  - 1\\
& \phi_0 \eqdef \phi, \phi_i \eqdef \frac{\phi_{i-1} }{ \truefactor} \text{ for } 1 \leq i \leq \parametertimesub + 1, \text{where } \delta \text{ is the constant of Theorem~\ref{thm:expander_decomposition}}
\\
& c_{\parametertimesub+1} \eqdef c \text{ and } c_i \eqdef c_{i+1} (c_{i+1} + 2)  \text{ for } 0 \leq i \leq \parametertimesub \\
& \gamma \eqdef c_0 + 1\\
& \eta_i \eqdef 4 \phi (\truefactor)^{i + 1} (10c_0)^{3c_0 } \text{ for } 0 \leq i \leq \parametertimesub+ 1 \\
& t \eqdef \left\lceil \frac{c_0 \gamma}{ \phi_{\parametertimesub + 1} }\right\rceil \\
& t_{0, 1} \eqdef t, t_{0, j} \eqdef t_{0, j - 1}\cdot (c_0 + 2)^{2(\parametertimesub + 3)} \text{ for } 1 \leq j \leq c_0, \text{ and } q_{0, j} \eqdef t_{0, j} \cdot (c_0+2)^{\parametertimesub + 3} \text{ for } 0 \leq j \leq c_0\\
& t_{i, j} \eqdef t_{i-1, w_{i, j- 1} + 1}, q_{i, j} \eqdef q_{i - 1, w_{i, j} }\cdot (c_0 + 2) \text{ for } 0 \leq i \leq \parametertimesub+1, 0 \leq j \leq c_i, \text{ where } w_{i, 0}\eqdef 0, \\ & w_{i, j} \eqdef w_{i, j - 1} + 2(c_i - j ) + 3
\end{align*}
\end{framed}
\caption{Parameter setting}\label{fig:multi-level}
\end{figure}

\begin{claim}\label{claim:parameter_property}
Given $c = (\log n)^{o(1)}$, and parameters defined as Figure~\ref{fig:multi-level}, the following conditions hold:
\begin{enumerate}
    \item $\parametertimesub = \omega(1)$ and $\parametertimesub = O(\log \log \log n)$.
    \item $c_i < (4c)^{2^{\parametertimesub+1}} = \log^{1/10} n$ for any $0 \leq i \leq \parametertimesub + 1$.
    \item $\phi^{1.5} \leq \phi_i \leq \phi$ for any $0 \leq i \leq \parametertimesub + 1$.
    \item $\phi \leq \eta_i \leq \sqrt{\phi}$ for any $0 \leq i \leq \parametertimesub + 1$.
    \item $t_{i, j} (c_0 + 2)\leq q_{i, j}$ and $q_{i, j}(c_0 + 2)\leq t_{i, j + 1}$ for any $i, j$.
    \item $t, t_{i,j}^{O(c_{i}^2)}, q_{i, j}^{O(c_{i}^2)} = n^{o(1)}$ for any $i, j$
\end{enumerate}
\end{claim}
\begin{proof}
For the first condition, since $c = (\log n)^{o(1)}$, we have 
\[\log \left(\frac{\log \log n / 10}{\log (4c)}\right) =  \log \left(\frac{\log \log n / 10}{o(1) \log \log n + 2}\right) = \log (\omega(1)) = \omega(1), \]
and thus $\parametertimesub = \omega(1)$. On the other hand, 
\[\parametertimesub = O\left(\log \left(\frac{\log \log n / 10}{\log (4c)}\right)\right) = O(\log (\log \log n / 10)) = O(\log \log \log n).\]

For the second condition, we have 
$c_i = c_{i+1}(c_{i+1} + 2)<  4c_{i+1}^2$ for any $i \leq \parametertimesub$. By induction, \[c_i \leq 4^{2^{\parametertimesub + 1 - i} - 1}c^{2^{\parametertimesub + 1 - i}} < (4c)^{2^{\parametertimesub + 1}} \leq (4c)^{\left(\frac{\log \log n / 10}{\log (4c)}\right)} = 2^{\log \log n / 10}= \log^{1/10} n.\]

The third and forth condition hold by \[\left(\truefactor\right)^{\parametertimesub + 1} = 2^{O(\log^{1/3} n\cdot  \log^{2/3 }\log n\cdot \log \log \log n)} = (1/\phi)^{o(1)},\]
and 
$(10c_0)^{2c_0} < (10\log^{1/10} n)^{2(\log^{1/10} n)} = (1/\phi)^{o(1)}$.

The fifth condition holds for $i=0$ by definition.
For $i > 0$, 
$t_{i, j} (c_0 + 2) \leq q_{i, j}$ is obtained by the definition of $t_{i, j}, q_{i,j}$ and induction.
$q_{i, j} (c_0 + 2) \leq t_{i, j + 1}$ is proved by the induction of $t_{i, j + 1} = q_{i, j} (c_0 + 2)^{\parametertimesub + 3 - i}$, and the fact that $i \leq \parametertimesub + 1$.

For the six condition, we have 
\[t \leq \frac{c_0}{\phi^{1.5}} + 1 = O\left(\frac{\log^{1/10} n}{\left(1 / 2^{\log^{3/4}m}\right)^{1.5}}\right) \leq 2^{1.6 \log^{3/4} n} = n^{o(1)}. \]
By the definition of $t_{i,j}$ and $q_{i, j}$, we have 
\[\begin{split}
t_{i, j}, q_{i, j} \leq & 2^{1.6 \log^{3/4} n} \cdot (c_0 + 2)^{2(\parametertimesub + 3)\cdot c_0} \cdot (c_0 + 2)^{2(\parametertimesub + 3)} \\
= & 2^{1.6 \log^{3/4} n} \cdot (\log^{1/10}m)^{O(\log \log \log n \cdot \log^{1/10} n)} \\
< & 2^{1.7 \log^{3/4} n},
\end{split}
\]
and thus 
\[
t_{i, j}^{O(c_i^2)}, q_{i, j}^{O(c_i^2)} = 2^{1.7 \log^{3/4} n \cdot O(\log^{1/5} n)} = 2^{O(\log ^{0.95} n)} = n^{o(1)}. 
\]
\end{proof}

\subsection{Initialization and Update Algorithm}
We give our algorithm to initialize a multi-level sparsifier for a give graph. 

 \begin{framed}
\begin{flushleft}
\textbf{Multi-Level Sparsifier Initialization Algorithm}

    \vspace{.2cm}\textbf{Input:} Graph $G$, $\phi_0$, $c_0$, $\gamma$, and parameters $t_1, q_1, \dots, t_{c_0}, q_{c_0}$. 
    
    \textbf{Output:} $\{G^{(i)}, \calP^{(i)}, \cc^{(i)}\}_{i=0}^\ell$ and additional data structures. 
    
\begin{enumerate}
\item $\ell \gets 0$, $G^{(0)} \gets G$.

\item Repeat until the expander decomposition of $G^{(\ell)}$ only contains one cluster:
\begin{enumerate}
    \item Run expander decomposition on $\simple(G^{(\ell)})$ with conductance parameter $\phi$ and denote the result as $\calP^{(\ell)}$.
    \item Construct $c_0$-cut containment set on each connected component of $G^{(\ell)}[\calP]$ with parameters $t_1, q_1, \dots, t_{c_0}, q_{c_0}$, and denote the union of resulted edges as $\cc^{(\ell)}$.
    \item $\ell \gets \ell + 1, G^{(\ell)}\gets \sparsifier(G^{(\ell - 1)}, \calP^{(\ell - 1)}, \cc^{(\ell - 1)}, \gamma)$.
\end{enumerate}
\item Return $\{G^{(i)}, \calP^{(i)}, \cc^{(i)}\}_{i = 0}^\ell$
\end{enumerate}
\end{flushleft}
\end{framed}
By By Theorem~\ref{thm:expander_decomposition} and Corollary~\ref{cor:initial_cut_partition_basic},  we have the following lemma to initialize the data structure for a multi-level $c_0$-edge connectivity sparsifier data structure for a given multigraph $G$
with respect to parameters $\phi_0, \gamma, t, t_{0, 1}, q_{0, 1}, \dots,$ $ t_{0, c_0}, q_{0, c_0}$.


\begin{lemma}\label{lem:preprocess_dynamic}
Given a graph $G$ of $n$ vertices such that every vertex has at most $\Delta$ distinct neighbors for $\Delta = O(1)$, 
a positive integer $c = \log^{o(1)} n$, 
there is an algorithm 
algorithm to output a $(\phi_0, \eta_0, \gamma)$ multi-level $c_0$-edge connectivity sparsifier
$\{G^{(i)}, \calP^{(i)}, \cc^{(i)}\}_{i=0}^\ell$
with $\ell =O( \log_{1/\eta_0} m ) $
 for graph $G$  with respect to parameters $\phi_0, \gamma, t, t_{0, 1}, q_{0, 1},\dots, t_{0, c_0}, q_{0, c_0}$ defined as Figure~\ref{fig:multi-level}
 in   
 $\widehat O(n)$ time.
%
\end{lemma}

 \begin{framed}
\begin{flushleft}
\textbf{Multi-Level Sparsifier Update Algorithm}

    \vspace{.2cm}\textbf{Input:} $\{G^{(i)}, \calP^{(i)}, \cc^{(i)}\}_{i=0}^\ell$ and $\updateseq$. 
    
    \textbf{Output:} Updated $\{G^{(i)}, \calP^{(i)}, \cc^{(i)}\}_{i=0}^{\ell'}$. 
    
\begin{enumerate}
\item $i \gets 0$, $\updateseq^{(0)}\gets \updateseq$.

\item Repeat until $\updateseq^{(i)}$ contains more than $m \phi^{i+1}$ updates
\begin{enumerate}
    \item Run One-Level Sparsifier Update Algorithm on $G^{(i)}, P^{(i)}, \cc^{(i)}$ with update sequence $\updateseq^{(i)}$ and get update sequence $\updateseq^{(i + 1)}$.
    \item $i\gets i + 1$.
\end{enumerate}
\item Run Multi-Level Sparsifier Initialization Algorithm 
to get $\{G^{(j)}, \calP^{(j)}, \cc^{(j)}\}_{j = i+1}^{\ell'}$
\item Return $\{G^{(i)}, \calP^{(i)}, \cc^{(i)}\}_{i=0}^{\ell'}$.
\end{enumerate}
\end{flushleft}
\end{framed}

\begin{lemma}\label{lem:update_online_batch_sparsifier}
Let $G$ be a graph undergoing multigraph update operations (vertex and edge insertions and deletions) 
such that 
throughout the updates, every vertex of the graph has at most $\Delta = O(1)$ distinct neighbors.
Suppose $G$ is initially connected, and has $n$ vertices and $m$ edges. 
Given $c = (\log n)^{o(1)}$, let parameters be set as Figure~\ref{fig:multi-level}.
Assume the multigraph updates are partitioned into $\parametertimesub$ multigraph update sequences $\updateseq_1, \updateseq_2, \dots,$ $\updateseq_{\parametertimesub}$ 
such that $|\updateseq_j| \leq \frac{\phi m}{\Delta \log n}$ for every $1 \leq j \leq \parametertimesub$. 

Given a $(\phi_0, \eta_0, \gamma)$ multi-level sparsifier for $c_0$-edge connectivity, 
there is an update algorithm for the sparsifier such that 
running the algorithm with  update sequences  \[\updateseq_1, \updateseq_2, \dots, \updateseq_{\parametertimesub}\] sequentially, then the following conditions hold

\begin{enumerate}
\item After $j$-th execution of the update algorithm, 
the sparsifier is a $(\phi_j, \eta_j, \gamma)$ multi-level $c_j$-edge connectivity sparsifier with respect to parameters $\phi_j, \gamma, t, t_{j, 1},$ $q_{j, 1}, \dots, t_{j, c_j}, q_{j, c_j}$.
\item The running time of the update algorithm on update sequence $\updateseq_i$ is  $O(|\updateseq_i|\cdot n^{o(1)})$ for each updated sequence. 

\end{enumerate}
\end{lemma}

\begin{proof}
Let $\{G^{(i)}_0, \calP^{(i)}_0, \cc^{(i)}_0\}_{i=0}^{\ell_0}$ be the initial multi-level sparsifier for $c_0$-edge connectivity, and 
\\$\{G^{(i)}_j, \calP^{(i)}_j, \cc^{(i)}_j\}_{i=0}^{\ell_j}$
be the initial multi-level sparsifier after applying update algorithm with update sequences $\updateseq_1, \dots, \updateseq_j$. 
Let $\updateseq_{j}^{(i)}$ be the multigraph update sequence that is used to update $G^{(i)}_{j-1}, \calP^{(i)}_{j-1}$ and $ \cc^{(i)}_{j-1}$
to  $G^{(i)}_{j}, \calP^{(i)}_{j}$, and  $\cc^{(i)}_{j}$ respectively
(if exists).
Let $m_j$ denote the number of edges of graph $\simple(G^{(0)}_j)$. 
Since the update sequences have a total of at most $\phi m / \Delta \log n$ updates, 
$m_j \geq m / 2$ for all the $0 \leq j \leq \parametertimes$.

By the definition of the update algorithm, Lemma~\ref{lem:one-level-update-intro}, and induction, 
the following conditions hold: for any $1 \leq j \leq \parametertimesub$ and $0 \leq i \leq \ell_j$,

\begin{enumerate}
\item $G^{(0)}_j$ is the graph resulting from applying $\updateseq_1, \dots \updateseq_j$ sequentially to $G$. 
$G^{(i)}_j$ is $\sparsifier(G^{(i-1)}_j, \calP^{(i-1)}_j, \cc^{(i-1)}_j,  \gamma)$ for any $1 \leq i \leq \ell_j$.
\item $\calP_j^{(i)}$ 
is a $\phi_j$-expander decomposition of $\simple(G_j^{(i)})$. 
\item $\cc_j^{(i)}$ 
is a $c_j$-cut containment set of $G_j^{(i)}[\calP_j^{(i)}]$. 
\item $\sparsifier(G^{(\ell_j)}_j, \calP^{(\ell_j)}_j, \cc^{(\ell_j)}_j,  \gamma)$ is an empty graph. 
 \end{enumerate}
 

Since $\{G^{(i)}_0, \calP^{(i)}_0, \cc^{(i)}_0\}_{i=0}^{\ell_0}$ is a $(\phi_0, \eta_0, \gamma)$ sparsifier, for $j=0$ and $0 \leq i \leq \ell_0$, $\simple(G_0^{(i)})$ has at most $m_0\eta_0^i$ edges.  Consider the case of $j > 0$.  
For $i > 0$, assume $G_{j}^{(i-1)}$ has at most $m_j\eta_{j}^{i - 1}$ distinct edges, and $G_{j-1}^{(i)}$ has at most
$m_{j-1}\eta_{j - 1}^{i}$ distinct edges.
If $G_{j}^{(i)}$  is obtained by running
Multi-Level Sparsifier Initialization Algorithm, then $G_{j}^{(i)}$ has at most 
$m_j\eta_j^{i}$
distinct edges by Lemma~\ref{lem:preprocess_dynamic} the parameter setting. 
If $G_{j}^{(i)}$  is obtained by running 
One-Level Sparsifier Update Algorithm on $G_{j - 1}^{(i)}$, then $G_{j}^{(i)}$ has at most 
\[m_j\eta_{j - 1}^i + m \phi^{i + 1}  < m_j\eta_{j}^i\]
distinct edges. 
By induction, $\{G^{(i)}_j, \calP^{(i)}_j, \cc^{(i)}_j\}_{i=0}^{\ell_j}$ is a $(\phi_j, \eta_j, \gamma)$ multi-level $c_j$-edge connectivity sparsifier with respect to parameters $\phi_j, \gamma, t, t_{j, 1}, q_{j,1},\dots, t_{j, c_j}, q_{j, c_j}$ for every $0 \leq j \leq \parametertimesub$. 
By Definition~\ref{def:multi_sparifier_intro}, $\ell_j = O(\log_{1/\eta_j} m_j)$.

Now we bound the running time. 
By Claim~\ref{claim:parameter_property}, Lemma~\ref{lem:decomposition_phase_one} and induction, 
if $\updateseq^{(i)}_j$ is computed, 
then
there exists some constant $c' > 1$ such that 
 \begin{equation}\label{equ:multi_level_proof_one}\begin{split} |\updateseq^{(i)}_j| = & O( |\updateseq_j| (c' \cdot \Delta \cdot (10c_0)^{3c_0})^i ) \\
 = & |\updateseq| (10c_0)^{O(c_0\log^{1/4} n)} \\
 = & O(|\updateseq_j| / \phi_\parametertimesub).\end{split}\end{equation}

Let $\alpha_j$ be the integer such that for $i\leq \alpha_j$, $G_{j}^{(i)}$ is obtained by One-Level Sparsifier Update Algorithm,
and for $j > \alpha_j$,  $G_{j}^{(i)}$ is obtained by Multi-Level Sparsifier Initialization Algorithm.
By Lemma~\ref{lem:decomposition_phase_one} and Equation~(\ref{equ:multi_level_proof_one}), for $i\leq \alpha_j$, the running time to obtain $G_{j}^{(i)}, \calP_{j}^{(i)}$ and   $\cc_{j}^{(i)}$ is
\begin{equation}\label{equ:multi_level_main_one}\begin{split} & O(|\updateseq_j^{(i)}|  \Delta (20c_jq_{j, c_j})^{c_j^2 + 5c_j+3} \poly(c_j)\polylog(n)) +  \widehat O(|\updateseq_j^{(i)}| \Delta  \log n/ \phi_j^3)\\
= & \widehat O(|\updateseq_j^{(i)}| (q_{j, c_j} / \phi_\parametertimesub)^{c_j^2 + 5c_j + 3}) \\
= & \widehat O(|\updateseq_j| (q_{j, c_j} / \phi_\parametertimesub)^{c_j^2 + 5c_j + 4}).
\end{split}\end{equation}


For $i > \alpha_j$,
since $|\updateseq^{(\alpha_j + 1)}_j|  > m \phi^{\alpha_j + 2}$ according to the algorithm, by Equation~(\ref{equ:multi_level_proof_one}),
we have 
\[
     m \phi^{\alpha_j + 2} = O(|\updateseq_j| / \phi_\parametertimesub),
     \]
and thus $n = O(|\updateseq_j|  / \phi_\parametertimesub^{\alpha_j + 3})$.
On the other hand, $G_{j}^{(\alpha_j)}$ contains  at most 
\begin{equation}\label{equ:seq_length_bound}
\begin{split}
m_j \eta_j^{\alpha_j}  = & O(|\updateseq_j|  \eta_j^{\alpha_j}  / \phi_\parametertimesub^{\alpha_j + 3}) \\
= & O(|\updateseq_j| (\truefactor (10c_0)^{3c_0})^{O(\log^{1/4} n)}/ \phi_\parametertimesub^{3})
\end{split}
\end{equation}
distinct edges. 

Hence,  by Theorem~\ref{thm:expander_decomposition}, Corollary~\ref{cor:initial_cut_partition_basic}, and Equation~(\ref{equ:seq_length_bound}),  the running time to obtain obtain $G_{j}^{(i)}, \calP_{j}^{(i)}$ and   $\cc_{j}^{(i)}$  for all the $i \geq \alpha_j + 1$ is 
\begin{equation}\label{equ:multi_level_main_two}\begin{split}
    & 
\widehat O(m_j\eta_j^{\alpha_j} / \phi_j^2) + O(m_j\eta_j^{\alpha_j} \phi_j \truefactor  (20c_jq_{j, c_j})^{c_j^2 + 5c_j+3} \poly(c)\polylog(n)) \\
= & \widehat O(m \eta_j^{\alpha_j} (20c_jq_{j, c_j} / \phi_j)^{c_j^2 + 5c_j + 3}) \\
= & \widehat O(|\updateseq_j| (\truefactor (10c_0)^{3c_0})^{O(\log^{1/4}n)}
(20c_0 q_{j, c_j} / \phi_\parametertimesub)^{c_0^2 + 5c_0 + 3}). 
\end{split}\end{equation}
Combining Equation~(\ref{equ:multi_level_main_one}) and Equation~(\ref{equ:multi_level_main_two}), the running time of the algorithm is $O(|\updateseq_j|n^{o(1)})$ for $j$-th execution by Claim~\ref{claim:parameter_property}.
%
%
%
%
%
%
\end{proof}

\section{Fully Dynamic Algorithms For \texorpdfstring{$s$}{s}-\texorpdfstring{$t$}{t} \texorpdfstring{$c$}{c}-edge Connectivity }\label{sec:fully}

In this section, we present our fully dynamic algorithm for $c$-edge connectivity and prove Theorem~\ref{thm:2-edge-connectivity}.
Our algorithm makes use of a general framework by~\cite{nanongkai2017dynamic} to turn offline multi-level sparsifier update algorithm to fully dynamic algorithm.
This general framework is given in Section~\ref{sec:online_batch}. 
Section~\ref{sec:sub_fully_dynamic} gives our fully dynamic update algorithm and query algorithm. 

\subsection{An Online-Batch Dynamic to Fully Dynamic Framework}\label{sec:online_batch}

The online-batch dynamic graph setting, motivated by parallelization,  has been studied in  parallel dynamic algorithms~\cite{acar2011parallelism, simsiri2016work, acar2017brief, acar2019parallel, tseng2019batch, dhulipala2019parallel}.
We consider the online-batch setting for sequential dynamic graph problems.

The setting has two parameters: batch number $\parametertimesub$ and sensitivity $\parameterlength$.
Let $\mathfrak{D}$  be a data structure we would like to maintain for an input graph undergoing updates. 
The dynamic algorithm for data structure $\mathfrak{D}$ in  online-batch setting 
has  $\parametertimesub + 1$ phases: a preprocessing phase and $\parametertimesub$ update phases.

In the preprocessing phase, 
the preprocessing algorithm initializes an instance of data structure $\mathfrak{D}$ for the given  input graph.
In the $i$-th update phase  for $1 \leq i \leq \parametertimesub$, 
the update algorithm is given
the $i$-th update batch, which is a sequence of at most  $\parameterlength$ updates.
The goal of the update algorithm is to maintain the data structure $\mathfrak{D}$ according to the $i$-th update batch  such that 
at the end of this phase, $\mathfrak{D}$ is updated for the resulted graph after applying the  first $i$ update batches sequentially to  graph $G$.
We say an update algorithm in the online-batch setting has amortized update time $t$
if for each update batch of size at most $\parameterlength$,  the running time of the update algorithm is upper bounded by $t$ times the batch size. 
The goal of the online-batch setting   is to minimize the amortized update time. 

The online-batch setting  is a generalization of emergency planning (a.k.a sensitivity setting and  fault-tolerant). The difference is that the update algorithm in the online-batch setting needs to handle multiple update batches, but only one update batch in  emergency planning.


The lemma below, which was implicitly implied in~\cite{nanongkai2017dynamic} without being formally stated, shows that online-batch dynamic algorithm with bounded amortized update time can be turned into a dynamic algorithm with worst case update time in a blackbox way. Therefore, we only need to design an online-batch algorithm with good amortized complexity.


\begin{lemma}[Section 5, \cite{nanongkai2017dynamic}]\label{lem:fully_framework}
Let $G$ be a graph undergoing batch updates. 
Assume for two parameters $\parametertimesub$ and $ \parameterlength$,
there is a preprocessing algorithm with preprocessing time $\rti$
and 
 an update algorithm with amortized update time $\rtu$
for a data structure $\mathfrak{D}$
in the online-batch setting with batch number $\parametertimesub$  and  sensitivity $\parameterlength$,
where $\rti$ and $\rtu$ are functions that map the upper bounds of some graph measures throughout the update, e.g. maximum number of edges, to non-negative numbers.

Then for any $\parametertimes \leq \parametertimesub$ satisfying $\parameterlength \geq 2\cdot 6^\parametertimes$, 
there is a fully dynamic algorithm  with preprocessing time $O(2^\parametertimes \cdot \rti)$
and worst case update time $O\left(4^\parametertimes \cdot \left( \rti  /w  + w^{(1/\parametertimes)}\rtu\right)\right)$
to maintain a set of $O(2^\parametertimes)$ instances of the data structure $\mathfrak{D}$
such that after each update, 
the update algorithm specifies one of the maintained data structure instances satisfying the following conditions
\begin{enumerate}
\item The specified data structure  instance is for 
the up-to-date graph.
\item The online-batch update algorithm is executed for at most $\parametertimes$ times on the specified data structure instance with each update batch of size at most $\parameterlength$.
\end{enumerate}
\end{lemma}

For completeness, we give a Proof of Lemma~\ref{lem:fully_framework}
 in  Appendix~\ref{sec:online_batch_general}.\\

By Lemma~\ref{lem:preprocess_dynamic} and Lemma~\ref{lem:update_online_batch_sparsifier}, 
we maintain a multi-level $c$-edge connectivity sparsifier with preprocessing time $\widehat O(m)$ and amortized update time $n^{o(1)}$ for batch number $O(\log \log \log n)$ and batch size $O(m \phi_0 / \log n)$ for graphs with bounded number of distinct neighbors.

\subsection{Fully Dynamic Algorithm for \texorpdfstring{$c$}{c}-edge Connectivity}\label{sec:sub_fully_dynamic}
Let $\overline{G} = (\overline{V}, \overline{E})$ be the original dynamic simple graph with arbitrary degree. 
We use the degree reduction technique~\cite{harary6graph} to transform $\overline{G}$ to a multigraph $G  = (V, E)$ 
such that every vertex has at most constant number of distinct neighbors
as follows:
\begin{enumerate}
\item The vertex set $V$ of $G$ is \[V = \{v_{u, w} : (u, w) \in \overline{E}\} \cup \{v_{u, u} : u \in \overline{V}\}.\] 
\item 
For any edge $(u, w) \in \overline{E}$, add edge $(v_{u, w}, v_{w, u})$ to $E$ with edge multiplicity 1.
\item For every vertex $u$ of $\overline{G}$ with degree at least $1$, 
let $w_{u, 0} = u$ and  $w_{u, 1}, w_{u, 2}, \dots, w_{u, \deg(u)}$ be the neighbors of $u$ in $\overline{G}$.
Add edge $(v_{u, w_{u, i}}, v_{u, w_{u, i+1}})$ to $E$ with edge multiplicity $c+1$ for every $0 \leq i < \deg(u)$.
\end{enumerate}

To maintain the correspondence between $\bar G$ and $G$, for every $u \in \bar V$, we maintain the list of $w_{u, 0}, \dots, $ $w_{u, \deg(u)}$.

It is easy to verify that for any $u, w \in \overline{V}$, the $c$-edge connectivity for $u$ and $w$ in $\overline{G}$ is the same as the $c$-edge connectivity of $v_{u, u}$ and $v_{w, w}$ in $G$.

\paragraph{Update algorithm}

To insert or delete an edge for $\overline{G}$, we can generate a multigraph update sequence for $G$ of $O(1)$ length to maintain the following invariant:
\begin{enumerate}
\item Vertex $v_{u, w}$ is in graph $G$ iff $u = w$ or $(u, w)$ is an edge of graph $\overline{G}$. 
\item For every edge $(u, w)$ of $\overline{G}$, there is an edge $(v_{u, w}, v_{w, u})$ with multiplicity 1 in $G$.
\item For any vertex $u$ of $\overline{G}$, 
the induced subgraph of $G$ on 
all the $v_{u, w}$ that are in $G$ forms a path with edge multiplicity $c+1$ for each edge on the path. 
\end{enumerate}

By Lemma~\ref{lem:fully_framework}, Lemma~\ref{lem:preprocess_dynamic} and Lemma~\ref{lem:update_online_batch_sparsifier}, 
we have the following corollary. 
\begin{corollary}\label{cor:fully_dynamic_update}
Let $\overline{G}$ be a simple graph of $n$ vertices undergoing edge insertions and deletions. 
For any $c = (\log n)^{o(1)}$,  
there is a fully dynamic algorithm with $m^{1+o(1)}$ preprocessing time and $n^{o(1)}$ worst case update time to maintain a set of $n^{o(1)}$ multi-level sparsifier data structures for the corresponding multigraph
such that after obtaining each update of $\overline{G}$,
the update algorithm outputs the access to one maintained multi-level sparsifier data structure that is  a $(\phi_\parametertimesub, \eta_\parametertimesub, \gamma)$ multi-level $c_\parametertimesub$-edge connectivity sparsifier data structure for the multigraph that corresponds to the up-to-date $\overline{G}$, where $c_\parametertimesub, \phi_\parametertimesub, \eta_\parametertimesub, \gamma$ are defined as Figure~\ref{fig:multi-level}. 
\end{corollary}

\paragraph{Query algorithm}
We use the following query algorithm to answer $c$-edge connectivity queries.

\begin{framed}
\noindent \textbf{Algorithm} \textsc{Fully-Dynamic-Query}

\noindent \textbf{Input:} Two vertices $u, w$ of $\overline G$

\noindent \textbf{Output:} $\mathsf{true}$ or $\mathsf{false}$ 
\begin{enumerate}
    \item Take the multi-level $c$-edge connectivity sparsifier $\{G^{(i)}, \calP^{(i)}, \cc^{(i)}\}_{i=1}^\ell$ such that $G^{(0)}$ is the multigraph corresponding to the up-to-date $\overline G$.
    \item  Let  $\updateseq^{(0)}$ be the following multigraph update sequence \[(\insertt(v'), \insertt(v''), \insertt(v_{u, u}, v', c+1), \insertt(v_{w, w}, v'', c+1)),\]where $v', v''$ are two new vertices of $G$.
    \item For $i=0$ to $\ell$, 
run one-level sparsifier algorithm on $G^{(i)}, \calP^{(i)}, \cc^{(i)}$ with update sequence $\updateseq^{(i)}$, and parameters $\phi_{\parametertimesub}, c_{\parametertimesub + 1}, t, \gamma, t_{\parametertimesub, 1}, q_{\parametertimesub, 1}, \dots, t_{\parametertimesub, c_{\parametertimesub}}, q_{\parametertimesub, c_{\parametertimesub}}$ defined as Figure~\ref{fig:multi-level} with respect to $c$, and denote the resulted sequence by $\updateseq^{(i + 1)}$. 
\item Use $\updateseq^{(\ell + 1)}$ to construct a  multigraph $H$ containing $v_{u, u}$ and $v_{w, w}$.
\item Run standard offline $c$-edge connectivity algorithm for $v_{u, u}$ and $v_{w, w}$ on graph $H$, and return the result (reverse the change made on $\multileveldatastructure$ at line 3 before return).
\end{enumerate}
\end{framed}

\begin{lemma}\label{lem:query_dynamic}
Let $c$ be a positive integer such that $c = (\log n)^{o(1)}$. 
Given access to $(\phi_\parametertimesub, \eta_\parametertimesub, \gamma)$  multi-level $c_\parametertimesub$-edge connectivity sparsifier $\{G^{(i)}, \calP^{(i)}, \cc^{(i)}\}_{i=1}^\ell$ for multigraph $G$ corresponding to $\overline{G}$, where $\phi_\parametertimesub, \eta_\parametertimesub, c_\parametertimesub, \gamma$ are defined as Figure~\ref{fig:multi-level},
for a query $\query(u, w)$ for any two vertices $u, w \in \overline{G}$, 
there is an algorithm to return $\mathtt{true}$ if and only if $u$ and $w$ are $c$-edge connected with running time 
$n^{o(1)}$.

\end{lemma}
\begin{proof}
Since applying $\updateseq^{(0)}$ to $G^{(0)}$
does not change the $c$-edge connectivity between $v_{u, u}$ and $v_{w, w}$,  the $c$-edge connectivity between $v_{u, u}$ and $v_{w, w}$ for $G^{(0)}$ is the same as the $c$-edge connectivity of $u$ and $w$ in $\overline{G}$.


By Lemma~\ref{lem:decomposition_phase_one}, there exists a constant $c'$ such that for $0 \leq i \leq \ell + 1$
\[|\updateseq^{(i)}| = |\updateseq|(c'\Delta (10c)^{3c})^i  .\]
Since $\ell = O(\log_{1/\eta_{\parametertimesub}} n)$, $H$ is a multigraph of $n^{o(1)}$ vertices and edges.
By Lemma~\ref{lem:decomposition_phase_one}, for every $0 \leq i \leq \ell$, 
$v_{u, u}$ and $v_{w, w}$ are in $\endpoints(\partial_{G^{(i)}}(\calP^{(i)}))$.
Hence  
$H$ contains $v_{u, u}$ and $v_{w, w}$, 
and the $c$-edge connectivity between $v_{u, u}$ and $v_{w, w}$ in $H$ is same to the $c$-edge connectivity for $v_{u, u}$ and $v_{w, w}$ in $G^{(0)}$.
Hence, the algorithm outputs $c$-edge connectivity of $u$ and $w$ correctly. 

The running time of the algorithm is obtained by Lemma~\ref{lem:decomposition_phase_one} and Claim~\ref{claim:parameter_property}.
\end{proof}

\bibliography{dynamic}
\bibliographystyle{alpha}
\clearpage
\appendix

\section{Useful Properties of Cuts}\label{sec:useful_lemmas_cuts}

We give formal definitions related to cut. 

\begin{definition}
For a connected graph $G = (V, E)$, 
a \emph{cut} on $G$ is a bipartition  $(V_1, V \cut V_1)$ of vertices of the graph. 
The \emph{cut-set} of a cut is the multiset of edges that have one endpoint in each side of the bipartition.
For cut $(V_1, V \cut V_1)$, 
let $\partial_G(V_1)$ denote the cut-set of the cut, i.e. 
\[\partial_G(V_1) = \set{e = (u, v)\in E: u\in V_1, v\in V \cut V_1}.\]
The \emph{size} of a cut is the number of edges in the multiset $\partial_G(V_1)$, i.e. $|\partial_G(V_1)|$. 
A set of edges $E'$ induces a cut if $E'$ is the cut-set of a cut, i.e.,
there exists $V' \subset V$ such that $E' = \partial_G(V')$. 



\end{definition}

We define atomic cut as follows. 
\begin{definition}\label{def:simple_cut_old}
For a connected graph $G = (V, E)$,
a cut $(V_1, V\cut V_1)$ is an \emph{atomic cut} if both $G[V_1]$ and $G[V \cut  V_1]$ are connected, otherwise, it is \emph{non-atomic}. 
\end{definition}


For a set  $T$, $T'$ is a nontrivial subset of $T$ if $\emptyset \subsetneq T' \subsetneq T$. 
A bipartition $(T', T \setminus T')$ of $T$ is a nontrivial bipartition if $T'$ is a nontrivial subset of $T$.

\begin{definition}

Let $T\subset V$ be a set of vertices. 
A cut $(V_1, V \cut  V_1)$ 
is a $(T', T\setminus T')$-cut if $V_1 \cap T = T'$. 
Cut $(V_1, V \cut  V_1)$ is a \emph{minimum $(T_1, T \cut  T_1)$-cut} if $(V_1, V \cut  V_1)$ partitions $T$ into $T_1$ and $T \cut  T_1$, and the number of edges in the cut-set of $(V_1, V \cut  V_1)$ is minimum among all the cuts that partition $T$ into $T_1$ and $T \cut  T_1$. 
\end{definition}

Shattering a cut is formally defined as follows.
\begin{definition}\label{def:intersect}
For a connected graph $G = (V, E)$,
a set of edges $F\subset E$
\emph{shatters} a cut $C = (V', V \cut  V')$  if 
no connected component of $G \cut F$ contains all edges of $\partial_G(V')$. 
\end{definition}
We define parallel cuts as follows.

\begin{definition}\label{def:parallel_cut}
Let $C_1 = (V_1, V \cut V_1)$ and $C_2 = (V_2, V \cut  V_2)$ be two cuts.  for a connected graph $G = (V, E)$. 
Two cuts $C_1 = (V_1, V \cut  V_1)$ and $C_2 = (V_2, V \cut  V_2)$ are \emph{parallel} if one of $V_1$ and $V  \cut  V_1$ is a subset of one of $V_2$ and $V \cut  V_2$.


\end{definition}

Now we prove some useful properties of cuts.

\begin{lemma}\label{lem:not_parallel}
Let $C_1 = (V_1, V \cut V_1)$ and $C_2 = (V_2, V \cut V_2)$ be two cuts for a connected graph $G = (V, E)$ that are not parallel. 
If $C_1$ is an atomic cut, then 
the cut-set of $C_1$ shatters $C_2$.
\end{lemma}
\begin{proof}

    By the definition of parallel cuts, 
    both $V_1$ and $V \cut V_1$ have nontrivial intersections with both  $V_2$ and $V \cut V_2$. 
    Since $C_1$ is an atomic cut, 
    both $G[V_1]$ and $G[V \cut V_1]$ are connected.
    On the other hand, $C_1$ is a cut, so there is an edge from the cut-set of $C_2$ in $G[V_1]$, and another edge from the cut-set of $C_2$ in $G[V \cut V_1]$. Hence the cut-set of $C_1$ shatters $C_2$.
\end{proof}

\begin{lemma}\label{lem:intercept_condition_2}
    Let $G = (V, E)$ be a  connected graph. Suppose $C_1 = (V_1, V\cut V_1)$ and $C_2 = (V_2, V\cut V_2)$ are two cuts on $G$ that are not parallel. Either $\partial_G(V_1)$ shatters $C_2$ or $\partial_G(V_2)$ shatters $C_1$.
\end{lemma}
\begin{proof}
    If $\partial_G(V_1)\cap \partial_G(V_2) \neq \emptyset$,
    then we have $\partial_G(V_1)$ shatters $C_2$, and $\partial_G(V_2)$ shatters $C_1$ by Definition~\ref{def:intersect}. In the rest of the proof, we assume $\partial_G(V_1)\cap \partial_G(V_2) = \emptyset$.

    Suppose $\partial_G(V_2)$ does not shatter $C_1$. By Lemma~\ref{lem:not_parallel}, $C_2$ is a non-atomic cut. 
    On the other hand, $\partial_G(V_1)$ is in a connected component $V'$ of $G \setminus \partial_G(V_2)$, and 
    we have $V' \subset V_2$ or $V' \subset (V \setminus V_2)$. 
    Since $\endpoints(\partial_G(V_1)) \subset V'$, 
    for every edge $(x, y) \in \partial_G(V') \subset \partial_G(V_2)$, 
    $\{x, y\}$ is a subset of $V_1$ or is a subset of $V \setminus V_1$.
    By Definition~\ref{def:parallel_cut} $V_1\not\subset V'$ and $(V\cut V_1)\not\subset V'$.
    Hence, there are edges $(x, y), (x', y') \in \partial_G(V')$ such that $x, y \in V_1$ and $x', y' \in V \setminus V_1$. 
    Therefore, $\partial_G(V_1)$ shatters $C_2$.
\end{proof}




\begin{lemma}\label{lem:swapping}
Let $G = (V, E)$ be a connected graph and $T$ be a set of terminals. 
Let $C_1 = (V_1, V \cut V_1)$ be a cut that partitions $T$ into $T_1$ and $T \cut T_1$,
and $C_2 = (V_2, V \cut V_2)$ be a cut that partitions $T$ into $T_2$ and $T \cut  T_2$.
If $T_1 \cap T_1' = \emptyset$, then 
let $V_1' = V_1 \setminus (V_1 \cap V_2)$, and $V_2' = V_2 \setminus (V_1 \cap V_2)$, 
$(V_1', V \setminus V_1')$ is a cut partitions $T$ into $T_1$ and $T \cut T_1$,
and $(V_2', V \setminus V_2')$ is a cut partitions $T$ into $T_2$ and $T \cut T_2$,
such that $\partial_G(V_1'), \partial_G(V_2') \subset \partial_G(V_1) \cup \partial_G(V_2)$, and one of the following conditions hold:
\begin{enumerate}
    \item the size of cut $(V_1', V \setminus V_1')$ is smaller than the size of cut $(V_1, V \setminus V_1)$.
    \item the size of cut $(V_2', V \setminus V_2')$ is smaller than the size of cut  $(V_2, V \setminus V_2)$.
    \item  the size of cut $(V_1', V \setminus V_1')$ is the same to that of  $(V_1, V \setminus V_1)$, and the size of cut $(V_2', V \setminus V_2')$ is the same to that of $(V_2, V \setminus V_2)$.

\end{enumerate}

\end{lemma}
\begin{proof}

    
    Since $T_1\cap T_2 = \emptyset$, we have 
    \begin{equation}\label{equ:no_overlap}T\cap (V_1\cap V_2) = \emptyset.\end{equation} 
        If $V_1\cap V_2 = \emptyset$, then we have $V_1' = V_1$, $V_2' = V_2$ and thus
        the third condition holds. 
        In the rest of this proof, we consider the case of
        $V_1\cap V_2\neq \emptyset$. 
        
        
        
        By the definition of $V_1'$ and $V_2'$, 
        for each $(x, y) \in \partial_G(V_1 \cap V_2)$, 
        the following conditions hold:
        \begin{enumerate}
            \item[(a)] If $(x, y)$ is in both $\partial_G(V_1)$ and $\partial_G(V_2)$, then $(x, y)$ is in both $\partial_G(V_1')$ and $\partial_G(V_2')$.
            \item[(b)] If $(x, y)$ is  $\partial_G(V_1) \setminus \partial_G(V_2)$, then $(x, y)$ is in $\partial_G(V_2')$ but not in  $\partial_G(V_1')$.
            \item[(c)] If $(x, y)$ is  $\partial_G(V_2) \setminus \partial_G(V_1)$, then $(x, y)$ is in $\partial_G(V_1')$ but not in  $\partial_G(V_2')$.
        \end{enumerate}

        Hence, we have \[\abs{\partial_G(V_1')} + \abs{\partial_G(V_2')} = \abs{\partial_G(V_1)}+ \abs{\partial_G(V_2)}.\]
        So, one of the three conditions holds. 
\end{proof}

\begin{claim}\label{claim:cut_cc_number}
Let $G = (V, E)$ be a connected graph and $(V', V \cut  V')$ be a cut of graph $G$ with size $c$. 
Then $G[V']$ contains at most $c$ connected components. 
\end{claim}
\begin{proof}
    Since the size of cut $(V', V \cut  V')$ is $c$, $\partial_G(V')$ contains at most $c$ edges, and thus $\endpoints(\partial_G(V')) \cap V' \leq c$.
    Any connected component of $G[V']$ must contains one vertex in $\endpoints(\partial_G(V')) \cap V'$, otherwise $G$ is not connected.
    Hence, $G[V']$ contains at most $c$ connected components. 
\end{proof}


\begin{lemma}\label{lem:replacement}
Let $G = (V, E)$ be a connected graph, and $T\subset T$ be a set of vertices.
Let 
$(V', V \setminus V')$ be a simple cut of graph $G$,
$E'$ be a subset of $\partial_G(V')$ such that $E'$ induces a cut $(V^\dagger, V\setminus V^\dagger)$ of graph $G$ separating $T$ nontrivially such that $V'\subset V^\dagger$. Let $E''$ be another set of edges such that $E''$ induces a cut that partitions $T$ the same as the cut induced by $E'$ does. 
Then there exists a subset of $(\partial_G(V') \setminus E')\cup E''$ induces a cut $(Q, V \setminus Q)$ that partitions $T$
the same as $(V', V\setminus V')$ does satisfying $Q \subset V' \cup (V\setminus V^\dagger)$.

\end{lemma}
\begin{proof}
Since $G[V']$ is connected, and $E' \subset \partial_G(V')$ induces a cut, 
$V\setminus V^\dagger$ forms the union of several connected components of $G[V\setminus V']$. 
Let $L = V' \cup (V\setminus V^\dagger)$.
$G[L]$ is connected, and $\partial_G(L)$ is $\partial_G(V') \setminus E'$.

Let $(V^\ddagger, V\setminus V^\ddagger)$ be the cut induced by $E''$ satisfying $V^\dagger \cap T = V^\ddagger \cap T$.
Since a cut-set restricted on an induced subgraph induces a cut on the induced subgraph, 
$(L \cap V^\ddagger, L\cap (V\setminus V^\ddagger))$ is a cut of $G[L]$ with cut-set in $E''$. Since 
\[L \cap T = (V' \cup (V\setminus V^\dagger))\cap T = (V' \cap T) \cup ((V\setminus V^\dagger)\cap T) = (V' \cap T) \cup ((V\setminus V^\ddagger) \cap T),\] by the definition of $E''$,  
$(L \cap V^\ddagger) \cap T = V' \cap T$. 
Hence $(L \cap V^\ddagger, V\setminus (L \cap V^\dagger))$ is a cut of $G$ with cut-set in $(\partial_G(V')\setminus E')\cup E''$ that partitions $T$ the same as $(V', V\setminus V')$.
\end{proof}

\begin{lemma}\label{lem:cut_subgraph_replace}
Let $G = (V, E)$ be a connected graph, and $T\subset V$ be a set of vertices.
Let 
$(V', V \setminus V)$ be a cut of graph $G$,
and $V^\star$ be a set of vertices such that $G[V^\star]$ is a connected graph
with $E' = \partial_G(V')\big{|}_{G[V^\star]} \neq \emptyset$. 
Let $E''$ be another set of edges on $G[V^\star]$ such that $E''$ induces a cut that partitions $(T\cup \endpoints(\partial_G(V^\star)))\cap V^\star$ the same way that the cut induced by $E'$ on $G[V^\star]$ does. 
Then $(\partial_G(V') \setminus E') \cup E''$ induces a cut that partitions $T$
same as $(V', V\setminus V')$ does.

\end{lemma}
\begin{proof}
We know $E'$ induces cut $(V'\cap V^\star, (V\cut V')\cap V^\star)$ on $G[V^\star]$. Suppose $E''$ induces cut $(V^\dagger, V^\star\cut V^\dagger)$ on $G[V^\star]$.

$\partial_G(V')\cut E'$ induces a cut on $G[V\cut V^\star]$ because a cut-set restricted on any induced subgraph induces a cut on that induced subgraph that partitions $\endpoints(\partial_G(V\setminus V^\star)) \cap (V \setminus V^\star)$ into $\endpoints(\partial_G(V\setminus V^\star)) \cap (V \setminus V^\star) \cap V'$ and $\endpoints(\partial_G(V\setminus V^\star)) \cap (V \setminus V^\star) \cap (V\setminus V')$. 
Since both $E''$ and $E'$ induce cuts on $G[V^\star]$ 
that partition $\endpoints(\partial_G(V\setminus V^\star)) \cap  V^\star$ into $\endpoints(\partial_G(V\setminus V^\star)) \cap  V^\star \cap V'$ and $\endpoints(\partial_G(V\setminus V^\star)) \cap V^\star \cap (V\setminus V')$,
$\partial_G(V')\cut E'\cup E''$ induces a cut $((V'\cut V^\star) \cup V^\dagger, ((V\cut V')\cut V^\star)\cup(V^\star\cut V^\dagger))$ on $G$.

Then it suffices show that the cut induced by $\partial_G(V')\cut E'\cup E''$ partitions $T$ the same way that $(V', V\cut V')$ does. 

For any $t\in T$, one of the following holds:
\begin{itemize}
    \item $t\in V'\cut V^\star$. 
    \item $t\in (V\cut V')\cut V^\star$.
    \item $t\in V'\cap V^\star$. Then $t \in V^\dagger$.
    \item $t\in (V\cut V')\cap V^\star$. Then $t\in V^\star\cut V^\dagger$.
\end{itemize}
Therefore, $(V'\cut V^\star\cup V^\dagger, (V\cut V')\cut V^\star\cup(V^\star\cut V^\dagger))$ and $(V', V\cut V')$ partition $T$ the same way.
\end{proof}

\section{Fully Dynamic Spanning Forest And Contracted Graph}\label{sec:spanning_forest_contraction}
We review the results on fully dynamic spanning forest~\cite{nanongkai2017dynamica, nanongkai2017dynamic, wulff2017fully} and 
contracted graph \cite{henzinger1997fully, holm2001poly} used in this paper. \\

Combining the development on fully dynamic minimum spanning forest~\cite{nanongkai2017dynamica, nanongkai2017dynamic, wulff2017fully}, 
recently Chuzhoy et al.~\cite{chuzhoy2019deterministic} gave a deterministic fully dynamic algorithm for minimum
spanning forest in subpolynomial update time. 

\begin{theorem}[cf. Corollary 7.2 of ~\cite{chuzhoy2019deterministic}]
There is a deterministic fully dynamic minimum spanning forest algorithm on a $n$-vertex $m$-edge graph 
with 
$\widehat O(m)$ preprocessing time 
and $2^{O(\log n \log \log n)^{2/3}}$ worst case update time. 
\end{theorem}

As a corollary, there is a deterministic fully dynamic spanning forest algorithm for a simple graph such that 
every edge deletion results at most one edge insertion in the spanning forest. 

\begin{corollary}\label{corollary:fully_dynamic_spanning_tree}
There is a deterministic fully dynamic spanning forest algorithm on a $n$-vertex $m$-edge simple graph with $\widehat O(m)$ preprocessing time 
and $2^{O((\log n \log \log n)^{2/3})}$ worst case update time 
such that for 
every update,
the update algorithm makes a change of at most two edges to the spanning forest, and the update algorithm returns a simple graph update sequence of length $O(1)$ regarding the change of the spanning forest. 
\end{corollary}

The contraction technique 
developed by 
Henzinger and King~\cite{henzinger1997fully} and Holm et al.~\cite{holm2001poly} are widely used in dynamic graph algorithms. 
\begin{definition}\label{def:superedge}
Given a simple forest $F = (V, E)$ and a set of terminals $S\subseteq V$, 
the set of connecting paths of $F$ with respect to $S$, denoted by $\path_S(F)$, is defined as follows:
\begin{enumerate}
\item $\path_S(F)$ is a set of edge disjoint paths of $F$.
\item For any two terminals $u, v\in S$ that belong to the same tree of $F$, 
the path between $u$ and $v$ in $F$ is partitioned into several paths in $\path_S(F)$.
\item For any terminal $v \in S$ such that the tree of $F$ containing $v$ has at least two terminals, $v$ is an endpoint of some path in $ \path_S(F)$.
\item For any endpoint $v$ of some path in $\path_S(F)$, 
either $v \in S$, or $v$ is an end point of at least three paths in $\path_S(F)$.
\end{enumerate}
It is easy to verify that for fixed $F$ and $S$, $\path_S(F)$ is unique.

A pair of vertices $(u, v)$ is a \emph{superedge} of $F$ with respect to $S$ if $u$ and $v$ are two endpoints of some path in $\path_S(F)$.
We denote the set of superedges of $F$ with respect to $S$ by $\superedge_S(F)$.
Let $\contract_S(F)$ be the simple graph with the endpoints of $\superedge_S(F)$ as vertices and 
$\superedge_S(F)$ as edges.

For an edge $(x, y)$ in $F$ and an edge $(u, v) \in \contract_S(F)$,
we say $(u, v)$
is the superedge in \\$\contract_S(F)$ covering $(x, y)$ 
if the path between $u$ and $v$ in $\path_S(F)$
contains edge $(x, y)$.



\end{definition}

The contracted graph is defined by a simple graph $G$, a vertex partition $\calP$ of $G$, and a simple spanning forest of $G[\calP]$.
\begin{definition}[Contracted graph]\label{def:contraction}
Let $G = (V, E)$ be a graph, $\calP$ be a partition of $V$ such that $G[P]$ is connected for every $P \in \calP$, and $F$ be a spanning forest of $G[\calP]$.
The contracted graph of $G$ with respect to $\calP$ and $ F$, denoted by $\contract_{\calP, F}(G)$,
is a simple graph with the endpoints of $\superedge_{\endpoints\wrap{\partial_G(\calP)}}(F)$ as vertices,  and 
$\partial_G(\calP) \cup \superedge_{\endpoints\wrap{\partial_G(\calP)}}(F)$ as edges.
\end{definition}


The following properties hold for a contracted graph.

\def\terminal{\mathtt{terminal}}

\def\contractioninsert{\mathsf{Contraction}\textsf{-}\mathsf{Insert}}
\def\contractiondelete{\mathsf{Contraction}\textsf{-}\mathsf{Delete}}
\def\contractioninsertterminal{\mathsf{Contraction}\textsf{-}\mathsf{Insert}\textsf{-}\mathsf{Terminal}}
\def\contractiondeleteterminal{\mathsf{Contraction}\textsf{-}\mathsf{Delete}\textsf{-}\mathsf{Terminal}}
\def\contractionfindsuperedge{\mathsf{Contraction}\textsf{-}\mathsf{Covering}\textsf{-}\mathsf{Edge}}

\begin{claim}\label{claim:contracted_graph_basic}
For any simple graph $G = (V, E)$, vertex partition $\calP$, and spanning forest $F$ of $G[\calP]$,
the following conditions hold
\begin{enumerate}
\item For any two vertices $u$ and $v$ of $\contract_{\calP, F}(G)$, 
$u$ and $v$ are connected in $\contract_{\calP, F}(G)$
if and only if they are connected in $G$. 
\item 
The number of vertices and edges of $\contract_{\calP, F}(G)$ is at most $3|\partial_G( \calP)|$.
\end{enumerate}
\end{claim}



The following lemma  is implicit in Section 5 of \cite{holm2001poly} by Holm, de Lichtenberg and  Thorup. 

\def\updatecontraction{\textsc{Update-Contraction}}
\def\preprocesscontraction{\textsc{Preprocess-Contraction}}
\begin{lemma}\label{lem:maintain_super_edge}
Let $F = (V, E)$ be a simple forest and $S \subseteq V$. 
Assume $F$ and $S$ undergo the following updates
\begin{itemize}
\item $\contractioninsert(v)$: insert a vertex $v$ to $F$
\item $\contractiondelete(v)$: delete an isolated vertex $v$ from $F$
\item $\contractioninsert(u, v)$: insert edge $(u, v)$ to $F$ (if $u$ and $v$ are in different trees) 
\item $\contractiondelete(u, v)$: delete edge $(u, v)$ from $F$ (if edge $(u, v)$ is in  $F$) 
\item $\contractioninsertterminal(v)$: add $v$ to terminal set $S$ (if $v$ is not in $S$)
\item $\contractioninsertterminal(v)$: remove $v$ from terminal set $S$ (if $v$ is in $S$)
\end{itemize}
and query 
\begin{itemize}
\item $\contractionfindsuperedge(u, v)$:   return the superedge of $\contract_S(F)$  that covers edge $(u, v)$ of $F$ if such a superedge exists
\end{itemize} 
 such that after every update,  $F$ contains at most $n$ vertices, and $S$ is a subset of $V$.

There exist a preprocessing algorithm 
with $O(n \polylog (n))$ running time and another  algorithm to handle update and query operations with $O(\polylog (n))$ worst case update  and query time 
such that 
for each update, the update-and-query algorithm 
maintains $F, S$ accordingly, and outputs an $O(1)$-length simple graph update sequence 
 that updates $\contract_S(F)$ accordingly.
\end{lemma}

\section{\texorpdfstring{$c$}{c}-Edge Connectivity Sparsifier Data Structures}\label{sec:data_structure}
We introduce the data structures used to store graphs, vertex partitions, cut containment sets, and $c$-edge connectivity sparsifiers, and give the operations supported by the these data structures. 

\paragraph{Dynamic Graph Data Structure}


Let $G$ be a graph and $T$ be a set of vertex terminals in $G$.
We use $\datastructure(G, T)$ to represent a data structure that contains 
\begin{enumerate}
    \item a copy of the graph $G = (V, E)$ represented by adjacency list such that multiple edges connecting same pair of vertices are represented by an edge and its multiplicity
    \item  a terminal set $T$
    \item  a simple spanning forest $F$ of graph $G$
    \item  the simple contracted graph $\contract_T(F)$
\end{enumerate}
and supports the query operations and update operations as defined in Figure~\ref{fig:data_structure_query_updates} such that such that after each update, the data structure also returns an update sequence for $\contract_T(F)$, where $\contract_T(F)$ is the contraction of $F$ with respect to terminal set $T$ as Definition~\ref{def:contraction}.

\begin{figure}[ht]
\begin{framed}
\begin{flushleft}
Query operations: 
\begin{itemize}
    \item $\textsf{VertexNumber}(x)$: return the number of vertices  
    of the connected component  containing vertex $x$ in $G$
    \item $\textsf{Volume}(x)$: return the volume
    of the connected component  containing vertex $x$ in $G$
    
    \item  $\textsf{DistinctEdgeNumber}(x)$: return the number of distinct edges 
    of the connected component  containing vertex $x$ in $G$ (parallel edges are only counted once)
    \item  $\textsf{TerminalNumber}(x)$: return the number of terminals 
    of the connected component  containing vertex $x$ in $G$
    \item  $\textsf{ID}(x)$: return the id of the connected component containing vertex $x$ in $G$
    \item  $\textsf{OneTerminal}(x)$: return an arbitrary terminal in the connected component  containing vertex $x$ in $G$ if exists
\end{itemize}
Update operations:
\begin{itemize}
    \item  $\textsf{Insert}(x)$: insert  vertex $x$ to the graph
    \item  $\textsf{Delete}(x)$: delete  isolated vertex $x$ from the graph (if the vertex is a terminal of $T$, also delete it from the $T$)
    \item  $\textsf{Insert}(x, y, \alpha)$: insert edge $(x, y)$ with multiplicity $\alpha$ to the graph
    \item  $\textsf{Delete}(x, y)$: delete all the $(x, y)$ multiple edges  from the graph
    \item  $\textsf{Insert-Terminal}(x)$: insert vertex $x$ to the terminal set
    \item  $\textsf{Delete-Terminal}(x)$: delete vertex $x$ from the terminal set
\end{itemize}
\end{flushleft}
\end{framed}
\caption{Graph data structure query and update operations}\label{fig:data_structure_query_updates}
\end{figure}

When the terminal set and the corresponding contracted graph are irrelevant, we use $\datastructure(G)$ to denote a data structure for an arbitrary terminal set.

By Corollary~\ref{corollary:fully_dynamic_spanning_tree} and Lemma~\ref{lem:maintain_super_edge}, we have the following lemma for the graph data structure.  
\begin{lemma}\label{lem:data_structure_pre}
For a dynamic multigraph $G = (V, E)$ with at most $n$ vertices and $m$ distinct edges through the updates, there is a preprocessing algorithm to construct $\datastructure(G, T)$ for an arbitrary $T \subset V$ in time $\widehat O(n+m)$
and an algorithm supporting all the operations defined in Figure~\ref{fig:data_structure_query_updates} in time $O(n^{o(1)})$ such that for each update operation, the  algorithm outputs an update sequence of length $O(1)$ to update the corresponding $\contract_{T}(F)$ in $\datastructure(G, T)$.
\end{lemma}

\paragraph{One-Level Connectivity Sparsifier Data Structure}
We define the one-level connectivity sparsifier data structure based on the graph data structure.

Let $G$ be a graph, $\mathcal P$ be a vertex partition of $G$, 
and $\cc$ be 
the union of 
$(\partial_G(P), c)$-cut containment sets on graph $G[P]$ for each $P \in \calP$, which
is an $\IA_{G[\mathcal P]}\wrap{\partial_G(\calP)), t, q, c, c}$ set of $G[\mathcal P]$ for some $q \geq t > 0$.
Suppose $\cc$ is recursively constructed, and is partitioned into $E_1, \dots, E_c$ satisfying the recursive construction condition with respect to parameters $t_1, q_1, \dots, t_c, q_c$ as defined in Section~\ref{sec:update_ia_set}. 
Data structure $\oneleveldatastructure$ with respect to $G$, $\calP$ and $\cc$ contains the following graph data structures.  
\begin{enumerate}
    \item $\datastructure = \datastructure(G,\emptyset)$,
    \item  $\datastructure_0 = \datastructure(G[\mathcal P], \endpoints\wrap{\partial_G(\mathcal P)})$ and
    $\datastructure_0' = \datastructure(G[\mathcal P],\emptyset)$,
    \item $\datastructure_i = \datastructure\wrap{G[\mathcal P]\cut \wrap{\bigcup_{j=1}^{i} E_j}, \endpoints\wrap{\partial_G(\mathcal P)}\cup \endpoints\wrap{\bigcup_{j=1}^{i} E_j}}$ for all $1 \leq i \leq c$,
    \item $\datastructure_i' = \datastructure\wrap{G[\mathcal P]\cut \wrap{\bigcup_{j=1}^{i} E_j}, \emptyset}$
    for all $1 \leq i \leq c$.
\end{enumerate}

We use $\sparsifier(\oneleveldatastructure, \gamma)$ to denote the $\sparsifier(G, \calP, \cc, \gamma)$ for $G, \calP, \cc$ defined by $\oneleveldatastructure$.

\paragraph{Multi-Level Connectivity Sparsifier Data Structure}


A multi-level $c$-edge connectivity sparsifier data structure of graph $G$ 
is a collection of one-level connectivity sparsifier data structures, denoted by 
\[\multileveldatastructure = \set{\oneleveldatastructure^{(i)}}_{i = 0}^{\ell}.\]

For each $\oneleveldatastructure^{(i)} \in \multileveldatastructure{}$, let $G^{(i)}$ be the graph of $\datastructure$ in $\oneleveldatastructure^{(i)}$, $\calP^{(i)}$ be the partition induced by the connected components of $\datastructure_0$ of $\oneleveldatastructure^{(i)}$, and $\calQ^{(i)}$ be the partition induced by the connected components of $\datastructure_c$ of $\oneleveldatastructure^{(i)}$.
We say a multi-level $c$-edge connectivity sparsifier data structure $\multileveldatastructure$ is a $(\phi, \eta, \gamma)$ multi-level $c$-edge connectivity sparsifier data structure if $\{(G^{(i)}, \calP^{(i)}, \calQ^{(i)})\}_{i = 0}^\ell$ is a $(\phi, \eta, \gamma)$ multi-level $c$-edge connectivity sparsifier as defined in Definition~\ref{def:multi_sparifier_intro}.

\section{Implementation of Decremental Cut Containment Set Update Algorithm}\label{sec:implementation}
With the supporting operations of the graph data structure, we give the detailed implementation of the algorithms in Section~\ref{sec:update_ia_set}. For an input graph $G = (V, E)$, we assume  $\abs V = n$ and $\abs E = m$.

\paragraph{Basic Operations}
We give the detailed implementation of the operations in Lemma~\ref{lem:subroutine_repair_set} and prove Lemma~\ref{lem:subroutine_repair_set}. \\


Algorithm~\ref{alg:atomic_cut_verification} determines if a subset of edges of the graph is the cut-set of an atomic cut. The algorithm makes use of the the following properties of atomic cuts:
If a set of edges $E_0$ is the cut-set of an atomic cut, then 
\begin{itemize}
    \item All the edges in $E_0$ are in the same connected component. 
    \item If we remove $E_0$ from the graph, the connected component containing $E_0$ is split into two connected components, and the endpoints of each edge in $E_0$ are in different connected components.
\end{itemize}


\begin{algorithm}[H]
\caption{\atomiccutverification($\datastructure(G), E_0$)}\label{alg:atomic_cut_verification}
\SetKwProg{Fn}{Function}{:}{}
\SetKw{Continue}{continue}
\SetKw{Break}{break}
\SetKwFunction{enumcuts}{Enumerate-Cuts}
\SetKwFunction{stcut}{Find-$S$-$T$-Cut}

\Input{$\datastructure(G)$: dynamic data structure of graph $G$ \newline
$E_0$: a non-empty set of edges}
\Output{$\mathtt{true}$ if $E_0$ induces an atomic cut; $\mathtt{false}$ otherwise}
    
$U\gets \emptyset, W\gets \emptyset$\;
\For {every edge $(x, y) \in E_0$} {
    $id_x \gets \datastructure(G).\textsf{ID}(x)$,
    $id_y \gets \datastructure(G).\textsf{ID}(y)$,
    $W \gets W \cup \{id_x, id_y\}$\;
}
\lIf {$|W| > 1$} {\Return $\mathtt{false}$}

         \lFor{every edge $(x, y) \in E_0$}{ run $\datastructure(G).\textsf{Delete}(x, y)$}


\For {every edge $(x, y) \in E_0$} {
    $id_x \gets \datastructure(G).\textsf{ID}(x)$, 
    $id_y \gets \datastructure(G).\textsf{ID}(y)$\;
    \lIf {$id_x = id_y$} 
    {reverse all the changes made on $\datastructure(G)$ on line 5 and  \Return $\mathtt{false}$}
    $U \gets U \cup \{id_x, id_y\}$\;
}
reverse all the changes made on $\datastructure(G)$ on line 5\;
\Return $\mathtt{true}$ if $|U| = 2$, otherwise \Return $\mathtt{false}$
\end{algorithm}

\begin{lemma}\label{lem:cut_verification}
Given access to the graph data structure for graph $G = (V, E)$ with at most $n$ vertices, a set of edges $E_0 \subset E$, 
algorithm  $\atomiccutverification$  determines whether $E_0$ forms an atomic cut of some connected component of $G$ with running time 
$O(\abs{E_0}n^{o(1)})$.
\end{lemma}
\begin{proof}
By the definition of atomic cuts, $E_0$ forms an atomic cut for some connected component of $G$ if and only if the following conditions hold
\begin{enumerate}
    \item The endpoints of edges in $E_0$ are in same connected component of $G$.
    \item After removing $E_0$ from $G$, for any edge $(x, y) \in E_0$, $x$ and $y$ are in different connected components.
    \item After removing $E_0$ from $G$, all the endpoints of $E_0$ belong to two connected components.
\end{enumerate}
By the description of the algorithm and Lemma~\ref{lem:data_structure_pre}, the lemma holds.
\end{proof}


Algorithm~\ref{alg:enumerate_simple} enumerates all the simple cuts of size at most $c$ such that the side containing vertex $x$ is of volume at most $t$ for a given vertex $x$.
    To enumerate all the desirable simple cuts, we repeat the following process for at most $c$ times:
    \begin{itemize}
        \item Perform a BFS starting at $x$.
        \item Pick an edge from the BFS tree, put it in the cut-set, remove the edge from the graph, and start another BFS at $x$.
    \end{itemize}
    Since the side containing $x$ is of volume at most $t$, the BFS stops if the volume of vertices visited is greater than $t$. \\
    

\def\findcut{\textsc{Find-Cut}}

\begin{algorithm}[H]
\caption{\enumeratecuts$(\datastructure(G), x, c, t)$}\label{alg:enumerate_simple}

\Input{$\datastructure(G)$: dynamic data structure of graph $G$ \newline $x$: a vertex of $G$ \newline $c, t$: parameters}
\Output{$H = \set{V'\subset V}$: a set of vertex sets of $G$ such that
for each $V' \in H$, $x\in V'$, $\vol_G(V') \leq t$, and $(V', V\cut V')$ is simple cut of size at most $c$}

\SetKwProg{Fn}{Function}{:}{}
$H\gets \emptyset$\;
run \findcut($\datastructure(G), x, c, t, \emptyset$)\;
\Return $H$\;
\SetKwProg{Fn}{Function}{:}{}

\Fn{\findcut($\datastructure(G), x, i, t, F$)}{
    $V'\gets \emptyset$\;
    run BFS from vertex $x$ on $G \setminus F$ and put all the visited vertices to $V'$, until BFS stops or $\vol(V')> t$\; 
    \If {$\vol(V') \leq t$ and $\partial_G(V') = F$} {
        $H   \gets \set{V'}\cup H$ and 
        \Return\;
    }
    \If {$i = 0$} {\Return\;}
    $T\gets$ a BFS tree of $G \setminus F$ rooted at $x$ until BFS stops or the volume of the vertices in the BFS tree is at least $t + 1$\;
    \For{every edge $(x, y)$ of T} {
        \findcut($\datastructure(G), x, i-1, t, F \cup \{(x, y)\}$)\;
    }
}

\end{algorithm}

\begin{lemma}\label{lem:enumerate_cut}
For a connected graph $G = (V, E)$,  
two integers $c, t$, 
and a vertex $x$ of $G$, 
there are at most $t^c$ simple cuts $(V', V \setminus V')$ of size at most $c$ satisfying $x \in V'$, $\vol_G(V')\leq t$.

Moreover, given access to the graph data structure of $G$,
algorithm $\enumeratecuts$ outputs all these cuts that are represented by vertex sets of the side containing $x$ with $O(t^{c + 2} \poly(c))$ running time.

\end{lemma}
\begin{proof}
Let $(V', V \cut  V')$ be a simple cut of size at most $c$ such that $x \in V'$, $\vol_G(V') \leq t$. 
Let $U$ be a proper subset of edges of $\partial_{G}(V')$. If we run BFS on $x$ for graph $G \setminus U$ until BFS stops or the volume of the visited vertices is at least $t + 1$, 
at least one of the edge in the corresponding BFS tree belong to $\partial_G(V') \setminus U$. 
Hence, all required cuts are enumerated. 
The number of cuts and the running time of the algorithm are obtained by the algorithm.
\end{proof}

Now we present our algorithm such that given three sets $T_1, T_2$, and $T'$ satisfying $T' \subset T_1\cup T_2$, 
the algorithm enumerates all the cuts of size at most $c$ that partitions $T_1\cup T_2$ into $T'$ and $(T_1 \cup T_2) \setminus T'$. 

\vspace{.3cm}

\begin{algorithm}[H]
\caption{\enumeratenonsimplecuts$(\datastructure(G, T_1), \datastructure(G, T_2), T', c, t)$}\label{alg:enumerate_cut}

\Input{$\datastructure(G, T_1), \datastructure(G, T_2)$: dynamic data structures of graph $G$ with $T_1$ and $T_2$ as terminals\newline
$T'$: vertex set such that $T' \subset T_1 \cup T_2$ \newline $c, t$: parameters}
\Output{$H = \set{V'\subset V}$: a set of vertex sets of $G$ such that for every $V'$, $(V', V\cut V')$ is a $(T', (T_1\cup T_2) \setminus T', t)$-cut of size at most $c$ satisfying that every connected component of $G[V']$ contains a vertex from $T'$}
\If {$\vol(G[T']) > t$ } {\Return $\emptyset$\;}
$H\gets \emptyset$, $U\gets \emptyset$\;
\For {every $x \in T'$}
{
$ U\gets U \cup \enumeratecuts(\datastructure(G, T_1), x, c, t)$\;
}
\SetKwProg{Fn}{Function}{:}{}


\For {every $k\leq c$}{
\If{there exist $V_1, \ldots, V_k\in U$ such that $\wrap{\bigcup_{i=1}^k V_i} \cap (T_1\cup T_2)= T'$ and $\vol\wrap{\bigcup_{i=1}^k V_i} \leq t$ and $V_i \cap V_j = \emptyset$  for every $1 \leq i < j \leq k$ }
{
$H\gets H\cup \set{\bigcup _{i=1}^k V_i}$
}
}
\Return $H$\;
\end{algorithm}

\begin{lemma}\label{lem:enumerate_non_simple_cut}
For a connected graph $G = (V, E)$, 
two integers $c, t$,
and three vertex sets $T_1, T_2, T' \subset V$ satisfying $T' \subset T_1 \cup T_2$, 
there are at most \[O(t^{c(c+1)})\] $(T', (T_1\cup T_2) \setminus T', t)$-cuts $(V', V \setminus V')$ of size at most $c$ such that every connected component of $G[V']$ contains at least one vertex from $T'$. 

Moreover, given access to the graph data structures $\datastructure(G, T_1)$, $\datastructure(G, T_2)$, and $c, t, T'$, there is an 
algorithm $\enumeratenonsimplecuts$ to output all these cuts that are represented by vertex sets of the side containing $T'$ in \[O(t^{c(c+1) + 1} \poly(c))\] time. 

\end{lemma}

\begin{proof}
If $\vol_G\wrap{T'} > t$, there is no $(T', (T_1\cup T_2)\setminus T', t)$-cut. 
It is easy to verify that 
for a $V'\in H$, $(V', V\setminus V')$  is a 
$(T', (T_1\cup T_2)\setminus T', t)$-cut of size at most $c$  satisfying that every  connected component of $V'$ contains at least one vertex from $T'$.

Let $(V', V\setminus V')$ be a 
$(T', (T_1\cup T_2)\setminus T', t)$-cut of size at most $c$ such that every  connected component of $V'$ contains at least one vertex from $T'$.
For every $V^\star$ that forms a connected component of $G[V']$,
$(V^\star, V\setminus V^\star)$ is a simple cut of size at most $c$ satisfying \[V^\star\cap T' \neq \emptyset and  \vol_G(V^\star)\leq t.\]
Hence, $V'$ is in $H$. So $H$ contains all  $V'$ such that $(V', V\setminus V')$ is a $(T', (T_1\cup T_2) \setminus T', t)$-cut of size at most $c$ satisfying that every connected component of $G[V']$ contains at least one vertex from $T'$.

By Lemma~\ref{lem:enumerate_cut}, the number of vertex sets in $H$ is at most $O(t^{c(c+1)})$, and running time of the algorithm is $O\wrap{t^{c(c+1) + 1}\poly(c)}$. 
\end{proof}

We use the following algorithm to determine if two given atomic cuts are parallel. 
\vspace{.3cm}

\begin{algorithm}[H]
\caption{\textsc{Parallel}$(\datastructure(G, S), E_1, E_2, c)$}\label{alg:parallel}

\Input{$\datastructure(G, S)$: dynamic data structure of graph $G$ with $S$ as terminals \newline 
$E_1, E_2$: set of edges that form atomic cuts of size at most $c$ \newline $c$: parameter}
\Output{$\mathtt{true}$ if the partitions of $S$ induced by removing $E_1$ and $E_2$ are parallel, or $\mathtt{false}$ if not}

    $P,  P_1,  P_2\gets\emptyset$\;
    \For{every $(x, y)\in E_1\cup E_2$}{
    $\datastructure(G, S).\textsf{Delete}(x, y)$\label{alg4:line1}\;
    $P \gets P\cup \set{\datastructure(G, S).\mathsf{OneTerminal}(y), \datastructure{(G, S)}.\mathsf{OneTerminal}(z)}$\;
    }
    
    reverse all deletions done in line~\ref{alg4:line1}\;
    \For{every $(x, y)\in E_1$} {
        $\datastructure(G, S).\textsf{Delete}(x, y)$\label{alg4:line2}\;

    }
    $z\gets $ an arbitrary endpoint of an arbitrary edge in $E_1$, $P_1 \gets \{z\}$\;
    \For{every $s\in P$} {
        \If{$\datastructure{(G, S)}.\mathsf{ID}(s) = \datastructure{(G, S)}.\mathsf{ID}(z)$} {
            $P_1 \gets P_1\cup\set s$
        }
    }
    reverse all deletions done in line~\ref{alg4:line2}\;

    \For{every $(x, y)\in E_2$} {
        $\datastructure(G, S).\textsf{Delete}(x, y)$\label{alg4:line3}\;
        
    }
    $z\gets $ an arbitrary endpoint of an arbitrary edge in $E_2$, $P_2 \gets \{z\}$\;
    \For{every $s\in  P$}{ 
        \If{$\datastructure{(G, S)}.\mathsf{ID}(s) = \datastructure{(G, S)}.\mathsf{ID}(z)$} {
            $ P_2 \gets  P_2\cup\set s$
        }
    }
    reverse all deletions done in line~\ref{alg4:line3}\;
    \If{$P_1\subset P_2$ or $(P\setminus P_1)\subset P_2$
    or $P_1\subset (P\setminus P_2)$ or
    $(P\setminus P_1)\subset (P\setminus P_2)$} {\label{alg4:line20}
        \Return \texttt{true}\;
    }
    \Return\texttt{false}
\end{algorithm}
\begin{lemma}\label{lem:algo_parallel_s}
Given access to $\datastructure{(G, S)}$ for a connected graph $G$ of at most $n$ vertices, a parameter $c$ and two sets of edges $E_1$, $E_2$ that form atomic cuts of size at most $c$, algorithm \textsc{Parallel} determines if the bipartitions of $S$ induced by $E_1$ is parallel with the bipartitions of $S$ induced by $E_2$ in $O(n^{o(1)} (|E_1| + |E_2|))$ time.
\end{lemma}
\begin{proof}
Let $C_1$ denote the cut induced by $E_1$ and $C_2$ denote the cut induced by $E_2$. 
Since vertices of every connected component of $G\setminus (E_1 \cup E_2)$ are on the same side of $C_1$ as well as $C_2$,
$E_1$ and $E_2$ induce parallel bipartitions on $S$ if and only if they induce parallel bipartitions on $P$. 
Since $\abs {E_1}\leq c$ and $\abs {E_2}\leq c$, there are at most $O(c)$ connected components in $G \setminus (E_1\cup E_2)$, and thus $\abs{P} = O(c)$.

Since both $E_1$ and $E_2$ form atomic cuts, $E_1$ partitions $P$ into $(P_1, P\cut P_1)$, and $E_2$ partitions $P$ into $(P_2, P\cut P_2)$. By definition, we can determine whether these two partitions are parallel by checking if one of $P_1$ and $P \cut P_1$ is a subset of either $P_2$ and $P \cut P_2$ (line~\ref{alg4:line20}). If so, then the two partitions are parallel.

The running time is obtained by Lemma~\ref{lem:data_structure_pre}.
\end{proof}

\begin{algorithm}[H]
\caption{\sequivalent$(\datastructure(G, S), E_1,  E_2, c)$}\label{alg:equivalent}
\SetKwProg{Fn}{Function}{:}{}
\SetKw{Continue}{continue}
\SetKw{Break}{break}
\SetKwFunction{enumcuts}{Enumerate-Cuts}
\SetKwFunction{stcut}{Find-$S$-$T$-Cut}
\SetKwInOut{Requirement}{Requirement}

\Input{$\datastructure(G, S)$: dynamic data structure of graph $G$ with $S$ as terminals \newline 
$E_1,  E_2$: edge sets that form atomic cuts of size at most $c$ \newline $c$: parameter}
\Output{$\mathtt{true}$ if the two atomic cuts formed by $E_1$ and $E_2$ induce the same partition of $S$, or $\mathtt{false}$ if not}
$A_1 \gets \emptyset, B_1 \gets \emptyset, A_2 \gets \emptyset, B_2 \gets \emptyset$, $b \gets \mathtt{true}$, $v\gets \datastructure(G, S).\mathsf{OneTerminal}(x)$ for an arbitrary $x \in \endpoints(E_1)$\; 


    run $\datastructure(G, S).\textsf{Delete}(x, y)$ for every $(x, y) \in E_1$\;
\For {every $(x, y) \in E_1$} {
    $id_1 \gets \datastructure(G, S).\textsf{ID}(x)$,
    $id_2 \gets \datastructure(G, S).\textsf{ID}(y)$, $id \gets \datastructure(G, S).\textsf{ID}(v)$\;

    \If {$id = id_1$} {
        $A_1 \gets A_1 \cup \{x\}, B_1 \gets B_1 \cup \{y\}$ 
    }
    \Else {
            $A_1 \gets A_1 \cup \{y\}, B_1 \gets B_1 \cup \{x\}$ 
    }
}
reverse all the change made on $\datastructure(G, S)$ in line 2 \;

    run $\datastructure(G, S).\textsf{Delete}(x, y)$ for every $(x, y) \in E_2$\;
\For {every $(x, y) \in E_2$} {
    $id_1 \gets \datastructure(G, S).\textsf{ID}(x)$,
    $id_2 \gets \datastructure(G, S).\textsf{ID}(y)$, $id \gets \datastructure(G, S).\textsf{ID}(v)$\;

    \If {$id = id_1$} {
        $A_2 \gets A_2 \cup \{x\}, B_2 \gets B_2 \cup \{y\}$ 
    }
    \Else {
            $A_2 \gets A_2 \cup \{y\}, B_2 \gets B_2 \cup \{x\}$ 
    }
}

    run $\datastructure(G, S).\textsf{Delete}(x, y)$ for every $(x, y) \in E_1$\;

\For {every $x\in A_1, y \in B_2$ 
or $x \in B_1, y \in A_2$
such that $\datastructure(G, S).\mathsf{ID}(x) = \datastructure(G, S).\mathsf{ID}(y) $} {
\If {$\datastructure(G, S).\mathsf{TerminalNumber}(x) > 0$} {$b \gets \mathtt{false}$\;}
}

reverse all the change made on $\datastructure(G, S)$ at line 10 and 17\;

\Return $b$\;
\end{algorithm}
\begin{lemma}\label{lem:algo_equivalent_atomic}
Given access to $\datastructure{(G, S)}$ for a connected graph $G$ of at most $n$ vertices, a parameter $c$ and two sets of edges $E_1$, $E_2$ that form atomic cuts of size at most $c$, algorithm $\sequivalent$ determines if the bipartitions of $S$ induced by $E_1$ is the same as the bipartitions of $S$ induced by $E_2$ in $O(n^{(1)} (|E_1|+|E_2|))$ time.
\end{lemma}
\begin{proof}
Let $C_1$ denote the cut induced by $E_1$ and $C_2$ denote the cut induced by $E_2$. 
By algorithm $\sequivalent$,
$A_1$ contains all the vertices in $\endpoints(E_1)$ that are on the same side of cut $C_1$ with $v$, and 
$B_1$ contains all the vertices in $\endpoints(E_1)$ that are on the different side of $C_1$ with $v$. 
Similarly, $A_2$ contains all the vertices in $\endpoints(E_2)$ that are on the same side of cut $C_2$ with $v$, and 
$B_2$ contains all the vertices in $\endpoints(E_2)$ that are on the different side of $C_2$ with $v$. 
Hence, $C_1$ and $C_2$ have the same bipartitions on $S$ if and only if no connected component of $G\setminus (E_1 \cup E_2)$ contains vertices from $S$, $A_1$, and $B_2$, and no connected component of $G\setminus (E_1 \cup E_2)$ contains vertices from $S$, $A_2$, $B_1$. 
The running time is obtained by Lemma~\ref{lem:data_structure_pre}.
\end{proof}

Lemma~\ref{lem:subroutine_repair_set} is obtained by Lemma~\ref{lem:cut_verification}, Lemma~\ref{lem:enumerate_cut}, 
Lemma~\ref{lem:enumerate_non_simple_cut},
Lemma~\ref{lem:algo_parallel_s}, and Lemma~\ref{lem:algo_equivalent_atomic}.

\begin{algorithm}[H]
\caption{\sequivalentsimplecut$(\datastructure(G, S), E_1, V_1, E_2, V_2, c, t)$}
\SetKwProg{Fn}{Function}{:}{}
\SetKw{Continue}{continue}
\SetKw{Break}{break}
\SetKwFunction{enumcuts}{Enumerate-Cuts}
\SetKwFunction{stcut}{Find-$S$-$T$-Cut}
\SetKwInOut{Requirement}{Requirement}

\Input{$\datastructure(G, S)$: dynamic data structure of graph $G$ with $S$ as terminals \newline 
$E_1, V_1, E_2, V_2$: $(t, c)$-realizable pairs $(E_1, V_1)$ and $(E_2, V_2)$ \newline $c, t$: parameters}
\Output{$\mathtt{true}$ if $(E_1, V_1)$ and $(E_2, V_2)$ are $S$-equivalent, or $\mathtt{false}$ if not}
$A_1 \gets \emptyset, B_1 \gets \emptyset, A_2 \gets \emptyset, B_2 \gets \emptyset$, $b \gets \mathtt{true}$\;


    run $\datastructure(G, S).\textsf{Delete}(x, y)$ for every $(x, y) \in E_1$\label{alg6:line2}\;
\For {every $(x, y) \in E_1$} {
    $id_1 \gets \datastructure(G, S).\textsf{ID}(x)$,
    $id_2 \gets \datastructure(G, S).\textsf{ID}(y)$, $id \gets \datastructure(G, S).\textsf{ID}(z)$ for an arbitrary $z \in V_1$\;

    \If {$id = id_1$} {
        $A_1 \gets A_1 \cup \{x\}, B_1 \gets B_1 \cup \{y\}$ 
    }
    \Else {
            $A_1 \gets A_1 \cup \{y\}, B_1 \gets B_1 \cup \{x\}$ 
    }
}
reverse all the change made on $\datastructure(G, S)$ in line~\ref{alg6:line2} \;

    run $\datastructure(G, S).\textsf{Delete}(x, y)$ for every $(x, y) \in E_2$\label{alg6:line10}\;

\For {every $(x, y) \in E_2$} {
    $id_1 \gets \datastructure(G, S).\textsf{ID}(x)$,
    $id_2 \gets \datastructure(G, S).\textsf{ID}(y)$, $id \gets \datastructure(G, S).\textsf{ID}(z)$ for an arbitrary $z \in V_2$\;

    \If {$id = id_1$} {
        $A_2 \gets A_2 \cup \{x\}, B_2 \gets B_2 \cup \{y\}$ 
    }
    \Else {
            $A_2 \gets A_2 \cup \{y\}, B_2 \gets B_2 \cup \{x\}$ 
    }
}

    run $\datastructure(G, S).\textsf{Delete}(x, y)$ for every $(x, y) \in E_1$\label{alg6:line17}\;

\For {every $x\in A_1, y \in B_2$ 
or $x \in B_1, y \in A_2$
such that $\datastructure(G, S).\mathsf{ID}(x) = \datastructure(G, S).\mathsf{ID}(y) $} {
\If {$\datastructure(G, S).\mathsf{TerminalNumber}(x) > 0$} {$b \gets \mathtt{false}$\;}
}

reverse all the change made on $\datastructure(G, S)$ at line \ref{alg6:line10} and \ref{alg6:line17}\;

\Return $b$\;
\end{algorithm}

\begin{lemma}\label{lem:sequivalent_algo}
Given access to $\datastructure(G, S)$
for a connected graph $G$ of at most $n$, two parameters $t, c$, 
and two $(t, c)$-realizable pairs $(E_1, V_1), (E_2, V_2)$,
algorithm $\sequivalentsimplecut$ determines if 
$(E_1, V_1)$ and  $(E_2, V_2)$ are atomic cut equivalent in 
$O\wrap{\poly(c) \cdot n^{o(1)}}$
time. 
\end{lemma}
\begin{proof}
By Definition~\ref{def:realizable_pair}, 
$(E_1, V_1)$ and $(E_2, V_2)$ are atomic cut equivalent if and only if the following two conditions hold: (Let $(L_1, V\setminus L_1)$ denote the cut induced by $E_1$ with $V_1 \subset L_1$, and $(L_2, V\setminus L_2)$ denote the cut induced by $E_2$ with $V_2 \subset L_2$.)
\begin{enumerate}
    \item Every connected component of $G\setminus (E_1 \cup E_2)$ that  contains vertices from both $\endpoints(\partial_G(E_1))\cap L_1$ and $\endpoints(\partial_G(E_2))\cap (V\setminus L_2)$ does not have vertices from $S$.
    \item Every connected component of $G\setminus (E_1 \cup E_2)$ that  contains vertices from both $\endpoints(\partial_G(E_1))\cap (V
    \setminus L_1)$ and $\endpoints(\partial_G(E_2))\cap L_2$ does not have vertices from $S$.
\end{enumerate}
Hence, algorithm $\sequivalent$ outputs $\mathtt{true}$ if and only if $(E_1, V_1)$ and $(E_2, V_2)$ are $S$-equivalent. 
The running time is obtained by Lemma~\ref{lem:data_structure_pre} and the description of the algorithm.
\end{proof}






\paragraph{Elimination Procedure} We present our implementation for the elimination procedure. 
\vspace{.3cm}

\begin{algorithm}[H]
\caption{\enumerateparallelcuts$(\datastructure(G, T), \datastructure(G, S), \Gamma, c, t)$}
\SetKwProg{Fn}{Function}{:}{}
\SetKw{Continue}{continue}
\SetKw{Break}{break}
\SetKwFunction{enumcuts}{Enumerate-Cuts}
\SetKwFunction{stcut}{Find-$S$-$T$-Cut}
\SetKwInOut{Requirement}{Requirement}

\Input{$\datastructure(G, T), \datastructure(G, S)$: dynamic data structures of graph $G$ with $T$ and $S$ as terminals  \newline $\Gamma$: a set of $(t, c)$-realizable pairs \newline  $c, t$: parameters}
\Output{$W$: A set of edges}
$W \gets \emptyset$\;

\For {every  $(E^\dagger, V^\dagger) \in \Gamma$} {
    $H_{V^\dagger} \gets \enumeratenonsimplecuts(\datastructure(G, T), \datastructure(G, S), V^\dagger \cap (T\cup S), \vol_G(V^\dagger), |\partial_G(V^\dagger)|)$\;
}

\Repeat {$\Gamma = \emptyset$} {
let $(E', V')\in \Gamma$ be an arbitrary pair such that there is no $(E'', V'') \in \Gamma$ satisfying $\endpoints(E'') \subset V'$ \;
$W \gets W \cup E'$\;
\For {every $(x, y) \in E'$} {$\datastructure(G, T).\textsf{Delete}(x, y)$\label{alg7:line8}\;}
\For {every  $(E^\dagger, V^\dagger) \in \Gamma$} {
    \If {there is a $V^\diamond \in H_{V^\dagger}$  s.t.  it is not the case that all of $\endpoints(\partial_G(V^\diamond))$ belong to the same connected component of $\datastructure(G, T)$} {
    remove $(E^\dagger, V^\dagger)$ from $\Gamma$\;
    }
    }
}
reverse all the changes made in line \ref{alg7:line8}\;

\Return $W$\;
\end{algorithm}

\begin{lemma}\label{lem:elimination_algo}
Let $G= (V, E)$ be a connected graph with at most $n$ vertices, and a set $\Gamma$ of $(t, c)$-realizable pair. 
Given access to $\datastructure(G, T)$, $\datastructure(G, S)$, $\Gamma$, $t, c$,
there is an algorithm $\textsc{Elimination}$ to implement the elimination procedure with running time 

$$O\wrap{|\Gamma|^2 \cdot t^{c(c+1) + 1}\cdot \poly(c)\cdot n^{o(1)}}$$
\end{lemma}

\begin{proof}
By the definition of elimination procedure and the description of algorithm $\enumeratenonsimplecuts$, algorithm $\enumerateparallelcuts$ implements the elimination procedure. By Lemma~\ref{lem:data_structure_pre} and  Lemma~\ref{lem:enumerate_non_simple_cut}, the running time of algorithm $\enumerateparallelcuts$ is 
$O\wrap{|\Gamma|^2 \cdot t^{c(c+1) + 1}\poly(c) n^{o(1)}}$. 
\end{proof}

\paragraph{Type one repair set}
We present our implementation of type one repair set algorithm. 
We start by an algorithm to find a maximal
 $(S, t, c)$-bipartition system for a given vertex set $S$. 
\vspace{.3cm}

\begin{algorithm}[H]\label{alg:bipart-system}
\caption{\bipartitionsystem($\datastructure(G, S), c, t$)}
\SetKwProg{Fn}{Function}{:}{}
\SetKw{Continue}{continue}
\SetKw{Break}{break}
\SetKwFunction{enumcuts}{Enumerate-Cuts}
\SetKwFunction{stcut}{Find-$S$-$T$-Cut}
\Input{$\datastructure(G, S)$: dynamic data structure of graph $G$ with $S$ as terminals \newline 
$c, t$: parameters}
\Output{$\mathcal{B}$: a $(S, t, c)$-bipartition system as defined in Definition~\ref{def:bipartition_system}}
$\mathcal B\gets \emptyset$ \;

\end{algorithm}
\clearpage
{\setlength{\interspacetitleruled}{0pt}%
\setlength{\algotitleheightrule}{0pt}%
\begin{algorithm}[H]
\setcounter{AlgoLine}{1}


\For{each $x\in S$, every $V'\in \enumeratecuts(\datastructure(G, S), x, c, t)$ and every $E' \subset \partial_G(V')$ s.t. \atomiccutverification$\wrap{\datastructure(G, S), E'}=\mathtt{true}$, $V'\cap S\neq S$, and the cut induced by $E''$ partitions  $S$ into two non-empty sets} {
        $b\gets \mathtt{true}$\;
        \For{every $y \in V'$, $V''\in \enumeratecuts(\datastructure(G, S), y, c, t)$ and $E'' \subset \partial_G(V'')$ s.t. \atomiccutverification$\wrap{\datastructure(G, S), E''}=\mathtt{true}$, $V''\cap S\neq \emptyset$, $V''\cap S\neq S$, and the cut induced by $E''$ partitions  $S$ into two non-empty sets} {
            \If {$(E'', V'') \in \mathcal{B}$ and  $\sequivalent(\datastructure{(G, S)}, E', E'', c, t) = \mathtt{true}$ and $V'\cap S$ and $V''\cap S$ are in the same connected component of $G\setminus E''$} {
                $b\gets \mathtt{false}$;
            }
            \If {$(E'', V'') \in \mathcal{B}$ and \textsc{Parallel}$(\datastructure{(G, S)}, E', E'', c, t) = \mathtt{false}$} {
                $b\gets \mathtt{false}$;
            }            
        }
        \If{$b=\mathtt{true}$}{ 
            $\mathcal  B\gets \mathcal B\cup\set{(E', V')}$\;
        
        }

}



\Return $\mathcal{B}$\;
\end{algorithm}}

\begin{lemma}\label{lem:algo_bipartite_system}
Given access to $\datastructure(G, S)$ for a  connected graph $G = (V, E)$ of at most $n$ vertices, a set of vertices $S$, and two integers $c, t$, algorithm $\bipartitionsystem$ computes a maximal $(S, t, c)$-bipartition system as  Definition~\ref{def:bipartition_system}
in 
$O(|S| (2t)^{2c+3}n^{o(1)})$ time.
\end{lemma}

\begin{proof}
    We show that the output of the algorithm is a maximal $(S, t, c)$-bipartition system as Definition~\ref{def:bipartition_system}.
    By the description of the algorithm and Lemma~\ref{lem:enumerate_cut}, 
    line 2 enumerated all the $(t, c)$-realizable pairs $(E', V')$ such that both cut $(V', V\setminus V')$ and the cut induced by $E'$ partition $S$ into two non-empty sets. Hence any bipartition system is a subset of the realizable pairs enumerated in line 2. 
    
    Note that a $(t, c)$-realizable pair $(E', V')$ cannot be added into a $(S, t, c)$-bipartition system if and only if there is a $(E'', V'')$ in the bipartition system such that $(E'', V'')$ and $(E', V')$ are atomic cut equivalent, or the cut induced by $E''$ is not parallel with the cut induced by $E'$.
    For either case, $V'' \cap V' \neq \emptyset$. 
    Hence, if $(E', V')$ cannot be added into $\mathcal{B}$, then Boolean variable $b$ is set to $\mathtt{false}$ at line 6 or 8. 
    Thus, there does not exist another realizable pair can be add into $\mathcal{B}$ at the end of the algorithm.


        

    
    Now we analyze the running time of the algorithm. 
    By Lemma~\ref{lem:enumerate_cut}, at most $|S|(2t)^c$ realizable pairs are enumerated at line 2, and for each iteration of the outer for loop, at most $(2t)^{c+1}$ cuts are enumerated. 
    By Lemma~\ref{lem:cut_verification}, Lemma~\ref{lem:enumerate_cut},  Lemma~\ref{lem:algo_parallel_s}, and Lemma~\ref{lem:algo_equivalent_atomic}, 
%
%
%
    the overall running time of the algorithm is $O(|S|(2t)^{2c+3} n^{o(1)})$.
    %
    %
    %
\end{proof}

We give our type one repair set algorithm. 
\vspace{.3cm}

\begin{algorithm}[H]
\caption{\typeonerepairset($\datastructure_1,
\datastructure_2, \datastructure_3, c, t$)}
\SetKwProg{Fn}{Function}{:}{}
\SetKw{Continue}{continue}
\SetKw{Break}{break}
\SetKwFunction{enumcuts}{Enumerate-Cuts}
\SetKwFunction{stcut}{Find-$S$-$T$-Cut}
\SetKwInOut{Requirement}{Requirement}

\Input{$\datastructure_1 = \datastructure(G, S),
\datastructure_2 = \datastructure(G, T),
\datastructure_3 = \datastructure(G \setminus (\IA_{G_0}(T_0, t, q_0, d, 2c)\big{|}_G), S)$ for some $q_0\geq t$, $d \geq 2c$\newline
$c, t$: parameters}
\Output{$W_1$: a set of edges}
$\mathcal{B}\gets \bipartitionsystem(\datastructure_1, c, t)$

 $W_1 \gets \emptyset$\;

\end{algorithm}
\clearpage
{\setlength{\interspacetitleruled}{0pt}%
\setlength{\algotitleheightrule}{0pt}%
\begin{algorithm}[H]
\setcounter{AlgoLine}{2}
\For{each $(E^\dagger, V^\dagger) \in \mathcal{B}$}{
$U[]\gets \emptyset$, 
$D \gets \emptyset$\;

\For {every $x \in V^\dagger$}{
    \For{every $V'$ in $\enumeratecuts(\datastructure_1, x, c, t)$ and every $E'$ subset of $\partial_G(V')$ s.t. \atomiccutverification$\wrap{\datastructure(G, S), E'}=\mathtt{true}$} {
    \If {$V' \cap S \neq \emptyset$ and $V'\cap S \neq S$
    and 
    $\sequivalent(\datastructure_1, E', V', E^\dagger, V^\dagger)=\mathtt{true}$
    and
    $\endpoints(\partial_G(V'))$ are in the same connected component of $\datastructure_3$ that contains a vertex from $S$
    } 
        {     
        put $\datastructure_3.\textsf{ID}(x)$ for any $x \in \endpoints(E')$ to $D$, and
        put $(V', E')$ into $U[\datastructure_3.\textsf{ID}(x)]$\;
        }

    }

}

$W_1\gets W_1 \cup  \partial_G(V^\dagger)$\;
\For {every $id \in D$} {
        
    $W_1 \gets W_1 \cup  \enumerateparallelcuts(\datastructure_2, \datastructure_1,   U[id], c, t)$\;

}
}

\Return $W_1$\;
\end{algorithm}}

\begin{lemma}\label{lem:type_one_repair_set_algo}
Let $G = (V, E)$ be a connected graph of at most $n$ vertices,  four positive integers $c, d, t, q_0$ such that $d \geq 2c + 1$ and $q_0 \geq t$, and two vertex sets $T, S \subset V$. 

Given access to $\datastructure_1 = \datastructure(G, S)$, 
$\datastructure_2 = \datastructure(G, T)$,
$\datastructure_3 = \datastructure(G \setminus (\IA_{G_0}(T_0, t, q_0, d, 2c) \cap E), S)$ for an $\IA_{G_0}(T_0, t, q_0, d, 2c)$ set, 
there is a $\typeonerepairset$ algorithm with 

$$O(|S|(2t)^{c^2 + 3c + 3}\poly(c)n^{o(1)})$$
running time to output a type one repair set $W_1$ of $|S|(8c^3 + 6c^2)$ edges for any $\IA_{G_0}(T_0, t, q, d, 2c+1)$ set such that the $\IA_{G_0}(T_0, t, q, d, 2c+1)$ set is derived from the $\IA_{G_0}(T_0, t, q_0, d, 2c)$ set.

\end{lemma}

\begin{proof}
By the definition of algorithm $\typeonerepairset$ and Lemma~\ref{lem:type-one-set-construct}, the output of algorithm $\typeonerepairset$ is a repair set of size $|S|(8c^3 + 6c^2)$ for any $\IA_{G_0}(T_0, t, q, d, 2c+1)$ set such that the $\IA_{G_0}(T_0, t, q, d, 2c+1)$ set is derived from the $\IA_{G_0}(T_0, t, q_0, d, 2c)$ set.

To bound the running time, by Lemma~\ref{lem:algo_bipartite_system}, the running time of line 1 is $O(|S| (2t)^{2c+3}n^{o(1)})$. By Claim~\ref{claim:num-bipart}, $\mathcal{B}$ contains less than $2|S|$ $(t, c)$-realizable pairs.
For each $(E^\dagger, V^\dagger) \in \mathcal{B}$, 
the running time of line 5-9 is $O((2t)^{c+1}n^{o(1)})$
by Lemma~\ref{lem:cut_verification} and Lemma~\ref{lem:sequivalent_algo}, 
and $|D| \leq 2$, $|U[id]| \leq (2t)^{c+1}$ for each $id\in CC$. 
So, the running time of line 11-12 is  $O((2t)^{c^2 + 3c + 3}\poly(c)n^{o(1)})$ 
by Lemma~\ref{lem:elimination_algo}. 

Hence, the overall running time of the algorithm is $O(|S|(2t)^{c^2 + 3c + 3}\poly(c)n^{o(1)})$.
%
%
%
%
%
\end{proof}

\paragraph{Type two repair set}
We give the implementation of our type two repair set algorithm (Algorithm~\ref{alg:type2}).

\begin{algorithm}[htb]
\caption{\typetworepairset($\datastructure_1,
\datastructure_2, \datastructure_3, \datastructure_4, c, t, q$)\label{alg:type2}}
\SetKwProg{Fn}{Function}{:}{}
\SetKw{Continue}{continue}
\SetKw{Break}{break}
\SetKwFunction{enumcuts}{Enumerate-Cuts}
\SetKwFunction{stcut}{Find-$S$-$T$-Cut}
\SetKwInOut{Requirement}{Requirement}

\Input{$\datastructure_1 = \datastructure(G, S),
\datastructure_2 = \datastructure(G, T),
\datastructure_3 = \datastructure(G \setminus (\IA_{G_0}(T_0, t, q_0, d, 2c), S),$ 
$\datastructure_4 = \datastructure(G \setminus (\IA_{G_0}(T_0, t, q, d, 2c+1), S)$ for some $q\geq q_0 \geq t$, $d \geq 2c$ \newline
$c, t, q$: parameters}
\Output{$W_2$: a set of edges}

 $W_2 \gets \emptyset$, 
 $Q\gets \emptyset$\, $\mathcal B = \emptyset$, $U\gets\emptyset$, $IDS\gets\emptyset$, $y\gets$ arbitary vertex of $S$\;
 
\For{every $s\in S$} {\label{alg10:line2}
$IDS\gets IDS\cup\set{\datastructure_3.\textsf{ID}(s)}$\;
    \For{every $V'$ in $\enumeratecuts{(\datastructure_1, s, c, q)}$} {
        $U\gets U\cup (V'\cap T)$\label{alg10:line5}
    }
   
}
 \For{every $u\in U$, every $V' \in \enumeratecuts{(\datastructure_1, u, c, q)}$ s.t. $V'\cap S = \emptyset$} {\label{alg10:line6}
             $T'\gets V'\cap T$, 
    $\alpha\gets \abs{\partial_G(V')}$, 
    $b\gets \mathtt{true}$\;
    \For {$V' \in \textsc{Enumerate-Cuts}{(\datastructure_1, \datastructure_2, T',\alpha, q)}$} {
    \If {$(V', V\setminus V')$ is of size smaller than $\alpha$ or $(V', V\setminus V')$ is of size $\alpha$ with cut-set not in a connected component of $\datastructure_4$} {
     $b\gets \mathtt{false}$\label{alg10:line10}
    }
    }
    \If {$T' \notin Q$ and $b = \mathtt{true}$ } {\label{alg10:line11}
        $Q\gets Q\cup \{T'\}$\;
        Find $E' \subset \partial_G(V')$ s.t. 
        $\atomiccutverification{(\datastructure_1, E')} = \mathtt{true}$ and $E'$ seperates $y$ and $T'$\;
        $x\gets$ an arbitrary vertex of $\endpoints(\partial_G(V'))$ \;
        $\mathcal{B}[\datastructure_3.\mathsf{ID}(x)] \gets \mathcal{B}[\datastructure_3.\mathsf{ID}(x)] \cup \{(E', V')\}$\;\label{alg10:line15}
    }

            }

\For{each $id\in IDS$} {\label{alg10:line16}
    $W_2\gets W_2\cup\textsc{Elimination}(\datastructure_1, \mathcal B[id],c, t)$\;\label{alg10:line17}
}


    


\Return $W_2$\;
\end{algorithm}

\begin{lemma}\label{lem:type_two_repair_set_algo}
Let $G = (V, E)$ be a connected graph of at most $n$, $T, S \subset V$ be two vertex sets, and $c, d, t, q, q_0$  be five positive integers such that $d \geq 2c + 1$ and $q \geq q_0 \geq t$.
Let $\datastructure_1, \datastructure_2, \datastructure_3$ and $\datastructure_4$ be graph data structures of $\datastructure(G, S)$, $\datastructure(G, T)$, 
$\datastructure(G \setminus (\IA_{G_0}(T_0, t, q_0, d, 2c)\big{|}_G), S)$ for an $\IA_{G_0}(T_0, t, q_0, d, 2c)$ set,
and $\datastructure(G \setminus (\IA_{G_0}(T_0, t, q, d, 2c+1)\big{|}_G), S)$ for an $\IA_{G_0}(T_0, t, q, d, 2c + 1)$ set that is derived from the $\IA_{G_0}(T_0, t, q_0, d, 2c)$ set.

Given access to $\datastructure_1, \datastructure_2, \datastructure_3$, $\datastructure_4$,
and parameters $c, t, q$,
 $\typetworepairset$ algorithm outputs a type two repair set $W_2$ of $|S|(4c^3 + 3c^2)$ edges for the $\IA_{G_0}(T_0, t, q, d, 2c+1)$ set with 
$$O(|S|(2q)^{c^2 + 3c+3} \poly(c)n^{o(1)})$$
running time. 

\end{lemma}

\begin{proof}
Let $(T', (S\cup T)\setminus T')$ be a type two bipartition satisfying the following conditions:
\begin{enumerate}
    \item 
$mincut_G(T', (S\cup T)\cut T', t)\leq c$,
\item there is a simple minimum $(T', (S\cup T)\cut T', t)$ cut, 
\item no simple $(T', (S\cup T)\cut T', q)$-cut of size at most $mincut_G(T', (S\cup T)\cut T', t)$ shattered by the $\IA_{G_0}(T_0, t, q, d, 2c+1) \cap E$ set. 
\end{enumerate}
Let $C$ be an arbitrary simple minimum $(T', (S\cup T)\cut T', t)$ cut. 
The cut-set of $C$ also forms a $(T', T_0 \cut T', t)$-cut of size $mincut_G(T', (S\cup T)\cut T', t)$. 
Hence, there is a $(T', T_0\setminus T', q)$-cut of size at most $mincut_G(T', (S\cup T)\cut T', t)\leq c$ with cut-set in $\IA_{G_0}(T_0, t, q, d, 2c+1)$. 
So there is a $(T', T\setminus T', q)$-cut for graph $G$ of size at most $mincut_G(T', T\cut T', t)\leq c$ with cut-set in $\IA_{G_0}(T_0, t, q, d, 2c+1) \cap E$. 
Thus, before the execution of line 6, there is a simple $(T', T\setminus T', q)$-cut $(V^\dagger, V\setminus V^\dagger)$ of size at most $c$ such that $V^\dagger \cap U \neq \emptyset$.
Hence, before line 16, 
there is a $(t, c)$-realizable pair $(E', V')$ belonging to some $\mathcal{B}[id]$ such that $(V', V\setminus V')$ is a simple minimum $(T', (S\cup T)\setminus T', t)$ cut. 
Hence, $W_2$ is a type two repair set by Lemma~\ref{lem:elimination_procedure_correctness}.




Now we bound the running time. 
By Lemma~\ref{lem:enumerate_cut}, $U$ can be found in $O\wrap{\abs S \cdot t^{c+2}\text{poly}(c)}$ time and contains at most $O(\abs S \cdot t^c)$ vertices. 
Therefore, the block from line~\ref{alg10:line6} to \ref{alg10:line15} takes at most 
\[O\wrap{\abs S \cdot t^c \cdot q^{c+2}\cdot \text{poly}(c)\cdot q^{c^2 + c + 1}\cdot \text{poly}(c)\cdot 2^c \cdot n^{o(1)}} = O\wrap{\abs S \cdot (2q)^{c^2 + 3c + 3}\text{poly}(c)} \cdot n^{o(1)}\] time. 
By Lemma~\ref{lem:basic_type_2}, for each connected component in $IDS$, line~\ref{alg10:line11} to \ref{alg10:line16} finds a set of $O(q^{c+1})$ realizable pairs. Since there are at most $\abs S$ connected components in $IDS$, the block from line~\ref{alg10:line16} to \ref{alg10:line17} takes at most 
\[O\wrap{\abs S\cdot q^{2c+2}\cdot t^{c^2 + c + 1}\cdot \text{poly}(c)n^{o(1)}} = O\wrap{\abs S\cdot q^{c^2 + 3c + 3}\cdot \text{poly}(c)n^{o(1)}}\] time.
Therefore, the overall running time of the algorithm is 
\[O\wrap{\abs S \cdot (2q)^{c^2 + 3c + 3}\text{poly}(c) n^{o(1)}}.\]
\end{proof}

\paragraph{Type three repair set}

The following algorithm finds a type 3 repair set.


        


\begin{algorithm}[H]\label{alg:second-type-s-cut}
\caption{\typethreerepairset($\datastructure_1, \datastructure_2, \datastructure_3, c, t$)}
\SetKwProg{Fn}{Function}{:}{}
\SetKw{Continue}{continue}
\SetKw{Break}{break}
\SetKwFunction{enumcuts}{Enumerate-Cuts}
\SetKwFunction{stcut}{Find-$S$-$T$-Cut}
\Input{$\datastructure_1 = \datastructure(G, S),
\datastructure_2 = \datastructure(G, T),
\datastructure_3 = \datastructure(G \setminus (\IA_{G_0}(T_0, t, q_0, d, 2c), S)$ for some $q_0\geq t$, $d \geq 2c$ \newline
$c, t$: parameters}

\Output{$W_3$: a set of edges}

$W_3 \gets \emptyset$\; 

\If {$|S| + |T| \leq t$} {
    find the cut with smallest size in $\enumeratenonsimplecuts(\datastructure_1, \datastructure_2, S\cup T, c, t)$, and put the cut-set in $W_3$ if exists\;
}

\For {every $s\in S$} {
$id\gets \datastructure_3.\textsf{ID}(s)$,         $vol\gets +\infty$, $V^\dagger\gets\emptyset$\;

    \For{every $V'$ in $\enumeratecuts(\datastructure_1, s, c, t)$ s.t. $V' \cap S = S$ and $\datastructure_3.\mathsf{ID}(y) =id$ for any $y \in \endpoints(\partial_G(V'))$ } {
        
        \For{each connected component $V^\star$ of $G[V\cut V']$} {
            \If {$V^\star\cap T\neq\emptyset$ and $vol > \vol_G\wrap{V^\star}> t$ }{
                $vol\gets\vol_G(V^\star)$\;
                $V^\dagger\gets V^\star$\;
            }
        }
        }
        $x\gets$ an arbitrary vertex in $T\cap V^\dagger$, $U\gets \emptyset$\;
        \For{every $V'$ in $\enumeratecuts(\datastructure_1, s, c, t)$ s.t. $V' \cap S = S$ and $x\notin V'$ and $\datastructure_3.\mathsf{ID}(y) =id$ for any $y \in \endpoints(\partial_G(V'))$ } {
            Find $E'\subset \partial_G(V')$ s.t. $\atomiccutverification{(\datastructure_1, E')}$ and $E'$ separates $x$ and $V'$\;
            $U\gets U\cup \{(E', V')\}$\;

        
    }
        $W_3\gets W_3\cup \partial_G(V^\dagger) \cup \textsc{Elimination}(\datastructure_1, U[id], c, t)$\;    
    }
\Return $W_3$\;
\end{algorithm}

\begin{lemma}\label{lem:algo_type_three_repair_set}
Let $G = (V, E)$ be a connected graph of at most $n$ vertices,  
$c, d, t, q_0$ be  
four positive integers such that $d \geq 2c + 1$ and $q_0 \geq t$, and $T, S \subset V$ be two sets.

Given access to $\datastructure_1 = \datastructure(G, S)$, 
$\datastructure_2 = \datastructure(G, T)$ 
$\datastructure_3 = \datastructure(G \setminus (\IA_{G_0}(T_0, t, q_0, d, 2c)\big{|}_G), S)$ for an $\IA_{G_0}(T_0, t, q_0, d, 2c)$ set, 
there is a $\typethreerepairset$ algorithm with 
$$O(|S|(2t)^{c^2 + 3c + 1} \poly(c)n^{o(1)})$$
running time to output a set $W_3$ of at most $|S|(4c^3 + 3c^2 + 2c)$ edges such that for any $\IA_{G_0}(T_0, t, q, d,$  $2c+1)$ set
that is derived from the $\IA_{G_0}(T_0, t, q_0, d, 2c)$ set, and a type two repair set $W_2$ of $\IA_{G_0}(T_0, t, q, d, 2c+1)$ with respect to $G$,
$W_3 \cup W_2$ is a type three repair set of $\IA_{G_0}(T_0, t, q, d, 2c+1)$ with respect to $G$.

\end{lemma}
\begin{proof}
The correctness is obtained by Lemma~\ref{lem:type-three-set-construct}. 
We only bound the running time. 
By Lemma~\ref{lem:enumerate_non_simple_cut}, the running time of line 5 is $O(t^{c(c+1) + 1}\poly(c))$.
By Lemma~\ref{lem:enumerate_cut}, 
$U$ contains at most $(2t)^c$ realizable pairs for each $s \in S$. 
By Lemma~\ref{lem:data_structure_pre},  
Lemma~\ref{lem:cut_verification}, Lemma~\ref{lem:enumerate_cut} and Lemma~\ref{lem:elimination_algo}, the total running time of line 5-15 is \[O((2t)^{c^2 + 3c + 1} \poly(c)n^{o(1)})\]  for each $s \in S$.
Hence, the overall running time of algorithm $\typethreerepairset$ is \[O(|S|(2t)^{c^2 + 3c + 1} \poly(c)n^{o(1)}).\]
\end{proof}

\paragraph{Algorithm for Lemma~\ref{lem:repair_set_algorithm}}

Finally, we present the algorithm for Lemma~\ref{lem:repair_set_algorithm}.

\begin{algorithm}[H]
\caption{\repairsetalgo($\datastructure_1, \datastructure_2, \datastructure_3,\datastructure_4, S, c, t , q$)}
\SetKwProg{Fn}{Function}{:}{}
\SetKw{Continue}{continue}
\SetKw{Break}{break}
\SetKwFunction{enumcuts}{Enumerate-Cuts}
\SetKwFunction{stcut}{Find-$S$-$T$-Cut}

\Input{$\datastructure_1 = \datastructure(G, \emptyset),
\datastructure_2 = \datastructure(G, T_0 \cap V),\newline
\datastructure_3 = \datastructure(G \setminus (\IA_{G_0}(T_0,t, q_0, d, 2c) \big{|}_G), \emptyset)$, \newline
$\datastructure_4 = \datastructure(G \setminus (\IA_{G_0}(T_0, t, q, d, 2c+1)\big{|}_G), \emptyset)$, \newline
$S$: a set of terminals \newline $c, t, q$: parameters}
\Output{$W$: a set of edges that is a cut repair set}

\For{$s \in S$} {
    $\datastructure_1.\mathsf{InsertTerminal}(s)$\;
    $\datastructure_3.\mathsf{InsertTerminal}(s)$\;
    $\datastructure_4.\mathsf{InsertTerminal}(s)$\;
    $\datastructure_2.\textsf{Delete-Terminal}(s)$\;
}

$W \gets  \typeonerepairset(\datastructure_1, \datastructure_2, \datastructure_3,  c, t)$\;
$W \gets W \cup \typetworepairset(\datastructure_1, \datastructure_2, \datastructure_3, \datastructure_4, c, t, q)$\;
$W \gets W \cup \typethreerepairset(\datastructure_1, \datastructure_2, \datastructure_3, c, t)$\;
reverse all the changes made at line 2, 3, 4 and 5\;
\Return $W$\;
\end{algorithm}


\begin{lemma}
Let $G = (V, E)$ be a connected graph of at most $n$ vertices,  
$c, d, t, q_0, q$ be  
four positive integers such that $d \geq 2c + 1$ and $q \geq q_0 \geq t$, and $T, S \subset V$ be two sets.

Given access to $\datastructure_1 = \datastructure(G, \emptyset)$, 
$\datastructure_2 = \datastructure(G, T_0 \cap V)$, 
$\datastructure_3 = \datastructure(G \setminus (\IA_{G_0}(T_0, t, q_0, d, 2c)\big{|}_G), \emptyset)$ and
$\datastructure_4 = \datastructure(G \setminus (\IA_{G_0}(T_0, t, q, d, 2c + 1)\big{|}_G), \emptyset)$
such that the $\IA_{G_0}(T_0, t, q, d,$  $2c+1)$ set
that is derived from the $\IA_{G_0}(T_0, t, q_0, d, 2c)$ set, there is a $\repairsetalgo$ algorithm with 
$$O(|S|(2q)^{c^2 + 3c + 3} \poly(c)n^{o(1)})$$
running time to output a repair set of at most $|S|(16c^3 + 12c^2 + 2c)$ edges.

\end{lemma}

\begin{proof}
By the description of the algorithm, before execution of line 6, $\datastructure_1 = \datastructure(G, S)$, $\datastructure_2 = \datastructure(G, T)$, 
$\datastructure_3 = \datastructure(G \setminus (\IA_{G_0}(T_0,t, q_0, d, 2c) \big{|}_G), S)$, and $\datastructure_4 = \datastructure(G \setminus (\IA_{G_0}(T_0,t, q, d, 2c+1) \big{|}_G), S)$. 
    By  Lemma~\ref{lem:type_one_repair_set_algo},
    Lemma~\ref{lem:type_two_repair_set_algo} and  Lemma~\ref{lem:algo_type_three_repair_set}    
    the algorithm returns a repair set of size at most $|S|(16c^3 + 12c^2 + 2c)$.
    
    The running time of the algorithm is obtained by Lemma~\ref{lem:data_structure_pre}, Lemma~\ref{lem:type_one_repair_set_algo},
    Lemma~\ref{lem:type_two_repair_set_algo} and  Lemma~\ref{lem:algo_type_three_repair_set}.
\end{proof}

\paragraph{Cut Containment Set Update Algorithm}
With the implementation of the $\IA$ set update  algorithm, we obtain the following cut containment set update algorithm. 

\begin{algorithm}[H]
\caption{\textsc{Containment-Set-Update}($c, t_1, q_1, \ldots , t_{c^2 + 2c}, q_{c^2 + 2c}, S, \newline \datastructure_0,\datastructure_0', \datastructure_1, \datastructure_1', \ldots, \datastructure_{c^2 + 2c}, \datastructure_{c^2 + 2c}'$)}\label{alg:containment_set_update}
\SetKwProg{Fn}{Function}{:}{}
\SetKw{Continue}{continue}
\SetKw{Break}{break}
\SetKwFunction{enumcuts}{Enumerate-Cuts}
\SetKwFunction{stcut}{Find-$S$-$T$-Cut}
\SetKwProg{Begin}{for}{ do}{}

 \Input{

$\datastructure_0$, $\datastructure_1, \datastructure_1', \dots, \datastructure_{c^2 + 2c}, \datastructure_{c^2 + 2c}'$: dynamic graph data structures inside a $\oneleveldatastructure{}$\newline
$S$: a set of vertices\newline
$c, t$: parameters
}
\Output{Edge set $F_1, \dots, F_c$ }

$k_0\gets 0$, $S_1 \gets S$\;

\For {$i \gets 1$ to $c$ } {
    $k_i\gets k_{i-1} + 2(c-i+1) + 1$,
    $F_i\gets\emptyset$,
    $ID\gets \emptyset$\;
    \For{each $s\in S_i$} {
        $ID\gets ID\cup \set{ \datastructure_{k_{i-1}}.\textsf{ID}(s)}$\;
    }
    \For{each $id\in ID$} {
        $N\gets \{s\in S_i: \datastructure_{k_{i-1}}.\textsf{ID}(s) = id\}$\;
        $F_i\gets F_i\cup \repairsetalgo\wrap{\datastructure_{k_{i-1}}', \datastructure_{k_{i-1}}, \datastructure_{k_{i} - 1}', \datastructure_{k_{i}}', N, c-i+1, t_{k_{i-1}}, q_{k_i}\cdot (c^2 + 2c + 1)}$ ~\footnote{Although the data structures passed in are data structures for $G_0$, only the in that contains terminals in $N$ are involved in the executions of $\repairsetalgo$. Therefore, we assume that the input data structures are for graph $G$.}\label{alg13:line8}
    }
    $S_{i+1}\gets S_i$\;
    \For{$j = 1$ to $i$} {
        \For{each $(x, y)\in F_j$} {
            
            $\datastructure_{k_i}.\mathsf{Delete}((x, y))$, $\datastructure_{k_i}'.\mathsf{Delete}((x, y))$\label{alg13:line12}\;
            $S_{i+1}\gets S_{i+1}\cup \{x, y\}$\;
        }
    }
}
    Reverse all the changes made in line~\ref{alg13:line12}\;
\Return $F_1, \cdots F_c$
    \end{algorithm}
\begin{lemma}
Let $G = (V, E)$ be a graph.    Given access to \[\datastructure_0 ,\datastructure_0', \datastructure_1, \datastructure_1', \ldots, \datastructure_{c^2 + c}, \datastructure_{c^2 + c}'\]
in a $\oneleveldatastructure{}$
such that the cut containment set is recursively constructed with respect to parameters $t_1, q_1, \ldots,$ $ t_{c^2 + 2c}, q_{c^2 + 2c}$ satisfying Equation~\ref{equ:ia_composition_condition}, and a set $S$ of vertices, Algorithm~\textsc{Containment-Set-Update} can return $F_1, \ldots, F_c$ such that their union $F$ has $\abs S O(c)^{O(c)}$ edges and the union of 
$F$ and the cut containment set used to construct the $\oneleveldatastructure{}$ is the new cut containment set of graph $G$ after adding vertices in $S$ to the terminal set. The running time of Algorithm~\textsc{Containment-Set-Update} is
$$O(\abs S \wrap{q\cdot \text{poly}(c)}^{O(c^2)}n^{o(1)} )$$
\end{lemma}
\begin{proof}
 Since $\datastructure_0 ,\datastructure_0', \datastructure_1, \datastructure_1', \ldots, \datastructure_{c^2 + c}, \datastructure_{c^2 + c}'$ are in a $\oneleveldatastructure{}$, there exists a cut containment set that is recursively constructed as $\bigcup_{j = 1}^{c^2 + c} E_j$ such that the graph in $\datastructure_i$ and $\datastructure_i'$ is
$G\cut \wrap{\bigcup_{j=1}^i E_j}$. Let $T$ be the set of terminals in $G$ in $\oneleveldatastructure{}$. Line~\ref{alg13:line8} computes a repair set $F_i$ for $\bigcup_{j = k_{i-1} +1}^{k_i}E_{j}$. By Lemma~\ref{lem:cut_containment_update}, $F_i\cup \wrap{\bigcup_{j = k_{i-1} + 1}^{k_i} E_j}$ is a
$$\IA_{G\cut \wrap{\wrap{\bigcup_{j = 1}^{k_{i-1}} E_j}\cup \wrap{\bigcup_{j = 1}^{i-1} F_j}}}\wrap{T\cup \wrap{\endpoints\wrap{\bigcup_{j = 1}^{k_{i-1}} E_j}}\cup \wrap{S_i}, t_{k_{i-1}}, q_{k_i}\cdot (c^2 + 2c + 1), c-i+1, 1}$$
set, and the union of $F$ and $\bigcup_{j = 1}^{c^2 + c} E_j$ is a $(T\cup S, c)$-cut containment set. $\abs F$ is $\abs S O(c)^{O(c)}$.

To bound the running time, we observe that in total there are at most $O(\text{poly}(c))^{O(c)}$ calls to $\repairsetalgo$. Therefore, the running time of the algorithm is at most 
\[O(\abs S \wrap{q\cdot \text{poly}(c)}^{O(c^2)}n^{o(1)} )\qedhere\]
\end{proof}

\section{One-Level Data Structure Update Algorithm Implementation}\label{sec:online_batch_update_cut_partition}

In this section, we provide implementations of the algorithms in Section~\ref{sec:one_level_update}. Let $\oneleveldatastructure{}$ be a one level data structure with respect to graph $G$, partition $\mathcal P$, and cut containment set $\cc$. We use $\sparsifier\wrap{\oneleveldatastructure, \gamma}$ to denote $\sparsifier\wrap{G, \mathcal P, \cc, \gamma}$, which is defined in Definition~\ref{def:one_level_sparsifier}.

Since we need to update the graph and maintain the $c$-connectivity between terminal vertices, we want to design an efficient way to generate the update sequence for the sparsifier. To make this happen, we need to ensure the following properties for the update:
\begin{itemize}
\item Before the update, $\mathcal{P}$ is a $\phi$-expander decomposition, we want the partition $\mathcal P'$ after the update to be a good expander decomposition of the updated graph. 
\item Before the update, we have a $(\endpoints(\partial_G(\mathcal P)), c^2+2c)$-cut containment set for the graph. After the update, we want to have a $(\endpoints(\partial_G(\mathcal P')), c)$-cut containment set for the updated graph.
\item The update sequence for the corresponding sparsifier has bounded length.
\end{itemize}

We first update $\phi$-expander decomposition $\calP$ to a refined partition 
$\calP^\star$  of $G$ satisfying the following three properties:
\begin{enumerate}
\item Any vertex involved in the multigraph update sequence $\updateseq$ is a singleton of $\calP^\star$ if the vertex is already in $G$.
\item $\calP^\star$ is a $\phi / \factor$-expander decomposition of $\simple(G)$.
\item The number of distinct new intercluster edges is no more than $O(|\updateseq|)$.
\end{enumerate}
To achieve this goal, we first remove from $G$ all the incident edges to the vertices involved in $\updateseq$, 
and then  update the expander decomposition for the the resulting graph. 
Because vertices involved in the multigraph update sequence become isolated vertices in the resulting graph, 
 these vertices are  singletons in the expander decomposition. 
In the end, we add all the removed edges back. 
Since every removed edge becomes an intercluster edge with respect to the new partition, 
adding edges back does not affect the conductance of each cluster. 

We show that 
if $k$ edges are removed from a constant degree simple $\phi$-expander, then there is a 
$\phi/\factor$-expander decomposition of the resulted graph with $O(k)$ intercluster edges. 
The idea is to 
first use the expander pruning algorithm~\cite{saranurak2019expander} to 
remove a set of vertices with volume at most $O(k /\phi)$ from the expander such that 
the remaining vertices form an $O(\phi)$-expander, and 
the number of edges between the removed vertices and the remaining vertices is $O(k)$.
Then we run the deterministic expander decomposition~\cite{chuzhoy2019deterministic} with conductance parameter $\phi / \factor$  on each connected component of the induced subgraph of the removed vertices. 

With the updated partition $\calP^\star$, we further update $\cc$ to $\cc^\star$ that  is a $(\endpoints(\partial_G(\mathcal P^\star)), c)$-cut containment set by applying Lemma~\ref{lem:repair_set_algorithm} on every cluster of $\calP$ affected by the update, making use of the condition that $\calP^\star$ is a refinement of $\calP$. 
By Lemma~\ref{lem:repair_set_algorithm}, $\cc^\star \setminus \cc$ contains at most $|\updateseq|(10c)^{O(c)}$  distinct edges. By the properties of contracted graph, 
the total number of vertex and edge insertions and deletions that transform
 $\sparsifier(G, \calP,\cc, \gamma)$ into $\sparsifier(G, \calP^\star,\cc^\star, \gamma)$ is at most $O(|\updateseq|(10c)^{O(c)})$.  

At the end, we apply the update sequence $\updateseq$ to $G$ to obtain the updated graph $G'$, and
 update $\calP^\star$ to $\calP'$.


\begin{lemma}\label{lem:decomposition_phase_one}
Let $G = (V, E)$ be a graph and $\updateseq$ be a multigraph update sequence  that updates $G$ to 
$G' = (V', E')$ such that
both $G$ and $G'$ contains at most $m$ vertices and distinct edges, and 
every vertex of $G$ or $G'$ has at most $\Delta$ distinct neighbors.

Given access to data structure $\oneleveldatastructure$ for graph $G$
with respect to partition $\mathcal P$ and parameters $t_1, q_1, \dots,$ $ t_{c^2 + 2c}, q_{c^2 + 2c}$
such that $\calP$ is a $\phi$-expander decomposition of $\simple(G)$
and parameters $\gamma > c$, $t_1, q_1, \dots,$ $ t_{c^2 + 2c}, q_{c^2 + 2c}$ satisfying
 \begin{equation}t \leq t_1, t_i \leq q_i \text{ for all } 1 \leq i \leq c^2 + 2c, \text{ and } q_i \cdot ((c^2 + 2c)+2)^2 \leq t_{i+1} \text{ for all } 1 \leq i \leq c^2 + 2c - 1,\end{equation}
and  $\updateseq$,
there is a deterministic algorithm \oneshotupdate\ 
with running time 
\[O(|\updateseq| \Delta (20cq_{c^2 + 2c})^{c^2 + 5c+3} \poly(c)\polylog(m)) +  \widehat O(|\updateseq| \Delta  \log m/ \phi^3)\]
to update $\oneleveldatastructure$
to  $\oneleveldatastructure'$ for graph $G'$ with respect to  partition $\mathcal P'$ and parameters $t_1', q_1',\ldots, t_{c}', q_{c}'$.

\oneshotupdate\  also outputs an update sequence $\updateseq'$ of  \[O(\Delta |\updateseq|(10c)^{3c})\] multigraph update operations.

The following properties hold for $\mathcal P'$ and $\updateseq'$: 
\begin{enumerate}[itemsep=2pt]
\item 
$\calP'$ is a $(\phi / \truefactor)$-expander decomposition of $\simple(G')$
such that
\begin{enumerate}[itemsep=2pt]
\item Every vertex of $G'$ that is involved in the update sequence $\updateseq$ is a singleton in $\calP'$
\item Every $P \in \calP'$ which contains at least two vertices is a subset of some vertex set in $\calP$
\item
$|\partial_{\simple(G')}(\calP')| \leq|\partial_{\simple(G)}( \calP )| + O(|\updateseq| \Delta)$.
\end{enumerate}
\item $\updateseq'$ updates $\sparsifier(\oneleveldatastructure, \gamma)$ to
$\sparsifier(\oneleveldatastructure', \gamma)$.
\item 
$t_i' = t_{w_{i-1} + 1}, q_i' = q_{w_i} ((c^2 + 2c)+2)$, where 
$w_0 = 0, w_i = w_{i - 1} + 2(c - i) + 3$ for $1 \leq i \leq c$.
\end{enumerate}

\end{lemma}

We first give an update algorithm in the following decremental update setting: 
Given access to $\oneleveldatastructure$  for graph $G = (V, E)$ with respect to a partition $\mathcal P$ and a $(\endpoints(\partial_G(\mathcal P)), c^2 + 2c)$-cut containment set $\cc$ of $G[\mathcal P]$, and
a set of edges $R \subset E$
such that $\partial_G(\calP) \cup R=\partial_G(\calP^\star)$ for a vertex partition $\calP^\star$ that is a refinement of $\calP$, 
we want to update $\oneleveldatastructure$ to  $\oneleveldatastructure'$ for $G$ with respect to partition $\calP^\star$ and a $(\endpoints(\partial_G(\mathcal P^\star)), c)$-cut containment set $\cc'$ of $G[\mathcal P']$. \newpage

\begin{algorithm}[H]\label{alg:updatepartition}
\caption{\updatepartition($\oneleveldatastructure,R, t, c, \gamma, t_1, q_1, \dots t_{c^2 + 2c}, q_{c^2 + 2c}$)}
\SetKwProg{Fn}{Function}{:}{}
\SetKw{Continue}{continue}
\SetKw{Break}{break}
\SetKwFunction{enumcuts}{Enumerate-Cuts}
\SetKwFunction{stcut}{Find-$S$-$T$-Cut}
\SetKwProg{Begin}{for}{ do}{}

\Input{$\oneleveldatastructure = \{\datastructure,
\datastructure_0, \datastructure_0' \dots, \datastructure_{c^2 + 2c}, \datastructure_{c^2 + 2c}'\}$ with respect to parameters $t_i, q_i$ 
for $1 \leq i \leq c^2 + 2c$\newline 
$R$: a set of edges\newline
 $c, t, \gamma, t_i, q_i$: parameters
}
\Output{$\oneleveldatastructure'$: updated from $\oneleveldatastructure$\newline 
$\newupdateseq$: multigraph update sequence}
$h[0]\gets 0$\;
\For{$i = 1$ to $c$}{
    $h[i + 1] \gets h[i] + 2i+1$\;
}
$S\gets \partial_G(R)$\;
$F_1, \ldots, F_c \gets\textsf{Containment-Set-Update}(c, t, t_1, q_1, \cdots, t_{c^2 + 2c}, q_{c^2 + 2c}, S, \newline\phantom{fsdfsdsfsfsfsfssfdsfsfsfdfsf}\datastructure_0, \datastructure_0', \datastructure_1,\datastructure_1' \ldots, \datastructure_{c^2 + 2c},\datastructure_{c^2 + 2c}')$\;
    
    $\newupdateseq \gets \emptyset$, $\seqone\gets \emptyset$ \;
            \For { $i\gets 1$ to $c$} {
    \For {$(x, y)\in F_i$ } {
        $\datastructure_{h[c]}'.\mathsf{Delete}(x,y )$\;
    
        $\seqone \gets \seqone \circ \datastructure_{h[c]}.\textsf{Insert-Terminal}(x) \circ \datastructure_{h[c]}.\textsf{Insert-Terminal}(y) \circ \datastructure_{h[c]}.\mathsf{Delete}(x, y)$ \tcp*{$\circ$ denotes sequence concatenation}
    }
    }

    
   \For {every operation $\mathtt{op} \in \seqone$} {
            \If {$\mathtt{op}$ is an edge insertion}{
                append $\mathtt{op}$ to the end of $\newupdateseq$ with edge multiplicity $\gamma$\;
            }
            \Else {append $\mathtt{op}$ to the end of $\newupdateseq$}
            }
\For {every edge $(x, y)\in R$} {
    append $\textsf{insert}(x, y, \alpha)$ to the end of $\newupdateseq$, where $\alpha$ is the edge multiplicity of $(x, y)$ in the graph of $\datastructure$\;
}
\Return $\oneleveldatastructure' = \{\datastructure, \datastructure_{h[0]},\datastructure_{h[0]}', \dots, \datastructure_{h[c]}, \datastructure_{h[c]}'\}, \newupdateseq$\;
\end{algorithm}

\begin{lemma}\label{lem:update_partition_main}
Let $G = (V, E)$ be a graph with at most $m$ vertices and distinct edges
, and  $\oneleveldatastructure$ be a one-level data structure of $G$ with respect to partition $\mathcal P$ and a $(\endpoints(\partial_G(\mathcal P)), c^2 + 2c)$-cut containment set $\cc$ given by parameters $t_1, q_1, \dots, t_{c^2 + 2c}, q_{c^2 + 2c}$ satisfying 
Inequality~\ref{equ:ia_composition_condition}. 
Given access to $\oneleveldatastructure$, 
a set  $R \subset E$ of $k$ distinct edges
such that $\partial_G(\calP) \cup R$  is $\partial_G(\calP^\star)$ for a vertex partition $\calP^\star$ that is a refinement of $\calP$, and a parameter $\gamma > c$, 
algorithm \updatepartition, with running time $O(k(20\cdot c \cdot q_{c^2 + 2c})^{c^2 + 5c + 3} \poly(c)\polylog(m))$, updates $\oneleveldatastructure$ to 
another one-level  data structure $\oneleveldatastructure'$ of $G$ with respect partition $\mathcal P^\star$ and $(\endpoints(\partial_G(\mathcal P^\star)), c)$-cut containment set $\cc'$ given by parameters
\[t_i' = t_{w_{i-1} + 1}, q_i' = q_{w_i} ((c^2 + 2c)+2) \text{ for all } 1 \leq i \leq c,\]
where $w_0 = 0, w_i = w_{i - 1} + 2(c - i) + 3$ for $1 \leq i \leq c$,
and outputs an update sequence of length $O(k (10c)^{3c})$ to update
$\sparsifier(\oneleveldatastructure, \gamma)$
to 
$\sparsifier(\oneleveldatastructure', \gamma)$.
\end{lemma}
\begin{proof}
Let $\calP^\star_i$ be the vertex partition induced by connected components of $\datastructure_i$ of $\oneleveldatastructure'$.
By 
Lemma~\ref{lem:repair_set_algorithm}, 
in $\oneleveldatastructure'$, 
$\calP^\star_0 = \calP^\star$,
$\datastructure_0 = \datastructure(G[\calP^\star_0], \endpoints(\partial_G(\calP^\star_0)))$ and $\datastructure_0' = \datastructure(G[\calP^\star_0], \emptyset)$.
For $i \geq 1$,
\[\begin{split}
\datastructure_i = \datastructure(& G[\calP^\star_{i-1}] \setminus \IA_{G[\calP^\star_{i-1}]}(\endpoints(\partial_G(\calP^\star_{i-1})), t_i', q_i', c- i + 1, 1), \\
& \endpoints(\IA_{G[\calP^\star_{i-1}]}(\endpoints(\partial_G(\calP^\star_{i-1})), t_i', q_i', c- i + 1, 1))
)
\end{split}\]
and 
\[\datastructure_i' = \datastructure(G[\calP^\star_{i-1}] \setminus \IA_{G[\calP^\star_{i-1}]}(\endpoints(\partial_G(\calP^\star_{i-1})), t_i', q_i', c- i + 1, 1), \emptyset).\]
By Corollary~\ref{cor:ia_composition}, $\oneleveldatastructure'$ is a one-level data structure of $G$ with respect to partition $\mathcal P^\star$ and $(\endpoints(\partial_G(\mathcal P)), c)$-cut containment set $\cc'$ given by parameters $t_1', q_1', \dots, t_c', q_c'$.


Let $k_i$ denote the number of distinct edges of $R[i]$ for $0 \leq i \leq c$.
By Lemma~\ref{lem:repair_set_algorithm} and the algorithm, 
$k_c = k$, 
and for $0 \leq i < c$, 
\[k_i \leq |k_{i + 1}| + 2 \cdot |k_{i + 1}|\cdot (16c^3 + 12c^2 + 2c) < |k_{i + 1}| \cdot (10c)^3\] edges.
By induction, $R[0]$ contains at most  $k (10c)^{3c} $.

By Lemma~\ref{lem:data_structure_pre}, $\newupdateseq$ updates $\sparsifier(\oneleveldatastructure, \gamma)$ to $\sparsifier(\oneleveldatastructure', \gamma)$, and the number of operations of  $\newupdateseq$ is $O(|R| (10c)^{3c})$.

The running time of the algorithm is obtained by Lemma~\ref{lem:repair_set_algorithm} and  Lemma~\ref{lem:data_structure_pre}. 
\end{proof}

\begin{algorithm}[H]
\caption{\decsingleexpander($\datastructure(G), \phi, D$)}
\SetKwProg{Fn}{Function}{:}{}
\SetKw{Continue}{continue}
\SetKw{Break}{break}
\SetKwFunction{enumcuts}{Enumerate-Cuts}
\SetKwFunction{stcut}{Find-$S$-$T$-Cut}

\Input{$\datastructure(G)$: dynamic data structure for simple graph $G$ \newline
$\phi$: parameter
\newline
$D$: a set of edges
}
\Output{ $R$: a set of edges}

$R \gets \emptyset, \calP \gets \emptyset$\;
$m\gets \datastructure(G).\textsf{DistinctEdgeNumber}(x)$ for an arbitrary vertex $x$ in $G$\;
 \If {$|D| \leq m \phi / 10$} {
 $P \gets \pruning(G, D)$, $R \gets   \partial_G(P) \setminus D$\;
 }
 \Else {$P\gets V$}
 compute $H = G[P] \setminus D$\;
$\calP \gets \calP \cup  \expanderdecomposition(H[V'], \phi / 2^{O(\log^{1/3}m \log^{2/3} \log m)})$ for  every $V'$ that forms a connected component of $H$\;
\Return $R \cup \partial_H(\calP)$
\end{algorithm}

%
%
%


\begin{lemma}\label{lem:decomposition_single_graph_appendix}
Given access to $\datastructure(G)$ for a simple graph $G = (V, E)$ such that $G$ is a $\phi$-expander with $m$ edges, and the maximum degree of $G$ is at most $\Delta$, 
and an edge set $D \subseteq E$ of $k$ edges,
there is a deterministic algorithm $\decsingleexpander$ with running time $\widehat O(|D|\log m / \phi^3)$ 
to output an edge set $R$ from graph $G \setminus D$ such that $|R| = O(k)$ and $R$ is the set of intercluster edges of a vertex  partition $\calR$ for graph $G \setminus D$ such that $\calR$ is a $(\phi /\truefactor)$-expander decomposition of $G \setminus D$, where $\delta$ is the constant as Theorem~\ref{thm:expander_decomposition}.

\end{lemma}
\begin{proof}
Let $G' = (V, E \cut D)$. If $k \leq m \phi / 10$, by Theorem~\ref{thm:pruning}, we have 
\begin{enumerate}
\item  Every connected component of $G'[V\cut P]$ is a $\phi / 6$ expander;
\item $|\partial_G(P)| \leq 8k$; 
\item $\vol_G(P) \leq 8k / \phi$.
\end{enumerate}
If $k > m\phi / 10$, we have $P = V$, and $\vol_G(P) = \vol(G) \leq 2m < 20 k / \phi$.
Since $H = G'[P]$, for either case, $\vol(H) \leq 20k / \phi$.

For a $V'$ that forms a connected component of $H$, let $\alpha$ denote the number of edges in $H[V']$. 
The output of 
\[\expanderdecomposition(H[V'], \phi / \truefactor)\]
is a  $\phi / \truefactor$-expander decomposition of $H[V']$ with $O(\phi \alpha)$ intercluster edges  
  by Theorem~\ref{thm:expander_decomposition}. 
Since the total number of edges in $H$ is at most $20k / \phi$, 
sum over all the connected components of $H$, we obtain a $(\phi / \truefactor)$-expander decomposition of $H$ with $O(k)$ intercluster edges.

By the definition of algorithm $\decsingleexpander$, $R$ contains all the intercluster edges for a $(\phi / \truefactor)$-expander decomposition of  $G \setminus D$ with $O(k)$ intercluster edges. 

The running time is obtained by Theorem~\ref{thm:expander_decomposition} and Theorem~\ref{thm:pruning}.
\end{proof}

\begin{algorithm}[H]
\caption{\oneshotupdate$(\oneleveldatastructure, \updateseq, \phi, c, t, \gamma, t_1, q_1, \dots t_{c^2 + 2c}, q_{c^2 + 2c})$}
\SetKwProg{Fn}{Function}{:}{}
\SetKw{Continue}{continue}
\SetKw{Break}{break}
\SetKwFunction{enumcuts}{Enumerate-Cuts}
\SetKwFunction{stcut}{Find-$S$-$T$-Cut}

\Input{$\oneleveldatastructure = \{\datastructure,
\datastructure_0, \datastructure_0' \dots, \datastructure_{c^2 + 2c}, \datastructure_{c^2 + 2c}'\}$ with respect to parameters $t_i, q_i$ 
\newline
$\updateseq$: a multigraph update sequence \newline
$\phi, c, t, \gamma, t_1, q_1, \dots, t_{c^2 + 2c}, q_{c^2 + 2c}$: parameters
}
\Output{ $\oneleveldatastructure'$: updated from $\oneleveldatastructure$\newline
$\newupdateseq$: multigraph update sequence}


$R \gets \emptyset, I \gets \emptyset$\;
\For {vertex $x$ involved in $\updateseq$} {
    $I \gets I \cup \{\datastructure_0.\mathsf{ID}(x)\}$\;
}
\For {every $id \in I$} {
$W_{id} \gets \{v: v \text{ is involved in  }\updateseq \text{ and } \datastructure_0.\mathsf{ID}(v) = id\}$\;
let $D_{id}$ be the distinct edges of $ \{(u, v) \text{ in graph of } \datastructure_0 : u \in W_{id} \text{ and } \datastructure_0.\mathsf{ID}(v) = id\}$\;
$R_{id} \gets \decsingleexpander\left(
\simple\left(\datastructure_0 \big{|}_{G'}\right),  \phi, D_{id}\right)$, where $G'$ is the connected component with identification $id$ in $\datastructure_0$, and
$\simple\left(\datastructure_0 \big{|}_{G'}\right)$ means edge multiplicities are ignored in the subroutine for $\datastructure_0 \big{|}_{G'}$ \;
$R\gets R \cup \{(x, y) \text{ in graph of } \datastructure_0: R_i \cup D_i\}$\;
}
$\oneleveldatastructure', \newupdateseq \gets \updatepartition(\oneleveldatastructure, R, c, t, \gamma, t_1, q_1, \dots, t_{c^2 + 2c}, q_{c^2 + 2c})$\label{alg16:line9}\;
\For {every vertex $x$ involved in $\updateseq$ that is an isolated vertex in the graph of $\oneleveldatastructure'$} {
$\newupdateseq \gets \newupdateseq \circ \insertt(x)$\;
run $\datastructure_i.\textsf{Insert-Terminal}(x)$ for every $\datastructure_i$ in $\oneleveldatastructure'$\;
}

 $\newupdateseq \gets \newupdateseq\circ \updateseq$\;

 apply $\updateseq$ to $\datastructure$ in $\oneleveldatastructure'$\;

\For {every vertex $x$ involved in $\updateseq$ that is an isolated vertex in the graph of $\oneleveldatastructure'$} {
$\newupdateseq \gets \newupdateseq \circ \delete(x)$\;
run $\datastructure_i.\textsf{Delete-Terminal}(x)$ for every $\datastructure_i$ in $\oneleveldatastructure'$\;
}

 \Return $\oneleveldatastructure'$ and $\newupdateseq$\;

\end{algorithm}

\begin{proof}[Proof of 
Lemma~\ref{lem:decomposition_phase_one}]
    We make the following observations for each $id \in I$:
    (Let $G_{id}$ denote the connected component for the graph of $\datastructure_0$ with connected component identification $id$.)
    
\begin{itemize}
\item[(a).] The number of distinct edges of $D_{id}$ is at most $\Delta |W_{id}|$. 
\item[(b).] 
$R_{id}$ is the set of intercluster edges for a 
$(\phi / \truefactor)$-expander decomposition of\\ $\simple(G_{id})$ with at most $O(|D_{id}|)$ intercluster edges.
\item[(c).] For every vertex $v \in W_{id}$,
$v$ is an isolated vertex in  $G_{id}$ if all the edges of $R|_{G_{id}}$ are removed.
\end{itemize}
Property (a) is obtained by the definition of the algorithm and the fact that the maximum degree of $\simple(G)$ is at most $\Delta$.
Property (b) is obtained by Lemma~\ref{lem:decomposition_single_graph_appendix}.
Property (c) is obtained by the fact that $D_{id}$ contains all the distinct edges incident to $v$ in $G_{id}$ for every $v \in W_{id}$.

Hence, $\partial_G(\calP) \cup R$ is the set of intercluster edges of a $(\phi / \truefactor)$-expander decomposition for graph $G$ such that every vertex involved in the update sequence is a singleton in the new partition. Let $\calP^\star$ denote this new partition.

By Lemma~\ref{lem:update_partition_main}, after line \ref{alg16:line9},
$\oneleveldatastructure$  is updated to $\oneleveldatastructure'$ for graph $G$, and 
update sequence $\newupdateseq$ updates $\sparsifier(\oneleveldatastructure, \gamma)$ to $\sparsifier(\oneleveldatastructure', \gamma)$ for $\oneleveldatastructure'$. 
The length of $\newupdateseq$ is at most $O(\Delta |\updateseq| (10c)^{3c})$.

Since every vertex involved in the update sequence $\updateseq$ is a singleton in $\calP^\star$, 
applying update sequence $\updateseq$ to $\datastructure$ of $\oneleveldatastructure'$ updates the $(\calP^\star, t, c)$-edge connectivity  partition data structures for $G'$.
After adding necessary isolated vertices to the sparsifier (by 
appending vertex insertions for these vertices to the end of $\newupdateseq$), 
$\updateseq$ can also be applied to the sparsifier to obtain the sparsifier of graph $G'$.

The running time is obtained by Lemma~\ref{lem:update_partition_main} and Lemma~\ref{lem:decomposition_single_graph_appendix}.
\end{proof}

\section{Fully Dynamic Algorithm From Online-Batch Dynamic Algorithm}
\label{sec:online_batch_general}

Lemma~\ref{lem:fully_framework} was implicitly given in~\cite{nanongkai2017dynamic}.
For completeness, 
we prove 
Lemma~\ref{lem:fully_framework} in this section.
Assume we have the following algorithms for data structure $\mathfrak{D}$:
\begin{enumerate}
\item $\dsinitialize(G)$: a preprocessing algorithm with running time $\rti$ that initializes an instance of data structure $\mathfrak{D}$ for the input graph $G$,
\item $\dsupdate(G,  \mathfrak{D}, \updateseq)$: an 
update algorithm with 
 amortized update time  $\rtu$
with batch number $\parametertimes$ and sensitivity $\parameterlength$.
The input of the algorithm is the access to graph $G$ and data structure $\mathfrak{D}$ for graph $G$, 
and an update batch, which is a sequence of updates $\updateseq$ for $G$.
The algorithm updates $\mathfrak{D}$ so that 
the resulted data structure is for 
the resulted graph after applying $\updateseq$ on  $G$. 
\end{enumerate}
We assume $\rti$ and $\rtu$ are functions that map the upper bounds of some graph measures throughout the update, e.g. maximum number of edges, to non-negative numbers.
%

We define a multi-level structure. 
The total number of levels is $\parametertimes + 1$.
Let $s =  \lfloor (\parameterlength / 2)^{1/\parametertimes} \rfloor$.
For $0 \leq i \leq \parametertimes$ and $j \geq 0$, 
let 
\[d_i = s^{\parametertimes - i} \text{ and } t_{i, j} = j \cdot d_i.\] 
Without loss of generality, assume $t_{i, j} = 0$ if $j < 0$.




\begin{definition}\label{def:data_strcuture_online_batch}
Let $G_{i, j}$ denote the resulted graph of $G$ after applying first $j \cdot d_i$ updates. 
For a dynamic data structure $\mathfrak{D}$, let $\mathfrak{D}_{i, j}$ for every $0\leq i \leq \parametertimes$ and $j \geq 0$ as follows:
\begin{itemize}
\item For $i = 0$ or $j = 0$,  
$\mathfrak{D}_{i, j} = \dsinitialize(G_{0, j})$.
\item
For $1 \leq i < \parametertimes $, $j > 0$
\[\begin{split}
\mathfrak{D}_{i, j} = 
 \dsupdate\left( G_{i-1, \lceil j / s \rceil - 2}, \mathfrak{D}_{i-1, \lceil j / s \rceil - 2},\updateseq_{[t_{i-1, \lceil j / s \rceil - 2} + 1, t_{i, j}]}\right),
\end{split}
\]
where $\updateseq_{[a, b]}$ denotes the update sequence that contains $a $-th, $(a+1)$-th, $\dots$, $b$-th update of the entire update sequence. 
\end{itemize}


We say \emph{$\mathfrak{D}_{i, j}$ depends on $\mathfrak{D}_{i', j'}$} for some $i' < i$ if there exist $j_{i'}, \dots, j_{i}$ such that 
$j_{i'} = j'$, $j_{i} = j$ and 
$\mathfrak{D}_{h, j_h}$ is obtained by running $\dsupdate$ on $\mathfrak{D}_{h-1, j_{h-1}}$ for all the $i' < h \leq i$. 

\end{definition}

By Definition~\ref{def:data_strcuture_online_batch} and induction,
we have the following observation.
 \begin{claim}\label{lem:short_distance}
 Assume $s \geq 6$.
If $\mathfrak{D}_{i, j}$ depends on $\mathfrak{D}_{i', j'}$, then $t_{i, j} \leq t_{i', j' + 2} + d_{i'} / 2$. 
\end{claim}
The goal of the update algorithm is to provide  access to $\mathfrak{D}_{\parametertimes, j}$ at time $j$ (after obtaining the $j$-th update). 
Ideally, the  algorithm below is enough for our purpose, because 
if we only consider the running time of $\dsinitialize$ and $\dsupdate$ of the algorithm, then the overall running time is desirable.
However, we 
cannot afford to have multiple instances of each $\mathfrak{D}_{i, j}$. 


\begin{framed}
\begin{enumerate}
\item 
For $0 \leq i \leq \parametertimes - 1, j > 0$, compute $\mathfrak{D}_{i, j}$ by Definition~\ref{def:data_strcuture_online_batch} with work evenly distributed within time interval $[t_{i, j} + (d_i / 2) + 1, t_{i, j + 1}]$. 
     
\item 
Compute $\mathfrak{D}_{\parametertimes, j}$ by Definition~\ref{def:data_strcuture_online_batch} at time $j$.
\end{enumerate}
\end{framed}

Let $\mathfrak{P}$ denote the data structure that contains a graph $G$ and its corresponding data structure $\mathfrak{D}$. 
The fully dynamic algorithm maintains $O(2^{\parametertimes})$ copies of  $\mathfrak{P}$, 
such that at time $\tau$, there is a maintained data structure $\mathfrak{P}$ corresponding to $(G_{\parametertimes, \tau}, \mathfrak{D}_{\parametertimes, \tau})$. 

For every $0 \leq i\leq \parametertimes$, we  maintain 
 $2^{i+1}$ instances of  data structure $\mathfrak{P}$ 
 such that all the instances are  indexed by an $(i+1)$-bit binary number. 
Let 
 $b_{0, j} = j \mod 2$ for all the integer $j$, and \[b_{i, j} = b_{i-1, \lceil j/s\rceil  - 2} \circ (j \mod 2)\] 
for all $0 < i \leq \parametertimes$ and integer $j$, where $\circ$ is the concatenation of two binary strings. 
 
 For $0 \leq i \leq \parametertimes - 1$ and integer $j$, 
we update $\mathfrak{P}_{b_{i, j}}$ to be $(G_{i, j}, \mathfrak{D}_{i, j})$ within the time period $[t_{i, j} + (d_i / 2) + 1, t_{i, j + 1}]$, 
and keep 
$\mathfrak{P}_{b_{i, j}} = (G_{i, j}, \mathfrak{D}_{i, j})$
 within the time period 
$[t_{i, j + 1} + 1, t_{i, j + 2} + (d_i / 2)]$.

Now we are ready to give our preprocessing  algorithm and update algorithm.
In the preprocessing algorithm, we initialize  $\mathfrak{P}_{\beta}$ to be $(G, \mathfrak{D}_{0, 0})$ for all the binary number $\beta$ of length at most $\parametertimes + 1$. 
In the update algorithm, we amortize the updates for each $\mathfrak{P}_{b_{i, j}}$ within the time interval as described above.
\\

 \def\fdpreprocessing{\textsc{Fully-Dynamic-Preprocessing}}
 \def\fdupdate{\textsc{Fully-Dynamic-Update}}
 
\begin{algorithm}[H]
\caption{$\fdpreprocessing(G, \parametertimes)$}
\Input{ access to adjacency list of graph $G$, parameter $\xi$}
\Output{ $\mathfrak{P}_\beta$ for all the binary number $\beta$ of length at most $\parametertimes + 1$}
\For {every $0 \leq i \leq \parametertimes$ and every $(i+1)$-bit binary number $\beta$}{
$G_\beta\gets G$\;
$\mathfrak D_\beta\gets \dsinitialize{(G)}$\;
$\mathfrak{P}_\beta \gets (G_\beta, \mathfrak{D}_\beta)$
 }
\end{algorithm}


\begin{algorithm}[ht]
\caption{ $\fdupdate(\{\mathfrak{P}_\beta\}, \updateseq)$}

\smallskip

\Input{
access to $\mathfrak{P}_\beta = (G_\beta, \mathfrak{D}_\beta)$ for every binary number $\beta$ of at most $\parametertimes + 1$ bits;
access to the first $k$ multigraph updates of the graph update sequence $\updateseq$ at time $k$}

\Output{
access to $\mathfrak{D}_\beta$ for some $\parametertimes+1$ bit binary string $\beta$ such that $\mathfrak{D}_\beta= 
\mathfrak{D}_{\parametertimes, k}$ at time $k$ }

\SetKwFor{Distribute}{distribute}{}{}

\vspace{.2cm}
\Distribute {\textnormal{\textbf{work of the following process evenly  within time interval}} $[t_{0, j}+ d_0 / 2 + 1, t_{0, j+ 1}]$ \textnormal{\textbf{for  every}} $j > 0$:}{

\For{every pair $(G_\beta, \mathfrak{D}_\beta)$ with first bit of $\beta$ equal to $(j \bmod 2)$}{
apply $\updateseq_{[ t_{0, j - 2} + 1, t_{0, j}]}$ to  $G_\beta$\;
 $\mathfrak{D}_\beta \gets \dsinitialize(G_\beta)$\;}

}
\Distribute {\textnormal{\textbf{work of the following process evenly 
 within  time interval}} $[t_{i, j} + d_i / 2 + 1, t_{i, j + 1}]$
\textnormal{\textbf{for every}} $1 \leq i \leq \parametertimes - 1, j > 0$}
{

\For{every binary string $\beta$ of length at least $i + 1$
such that the first $i + 1$ bits of $\beta$ equals to $b_{i, j}$}{ $\beta'\gets$ the first $i$ bits binary number of $b_{i, j}$\;
 reverse $(G_\beta, \mathfrak{D}_\beta)$ back to status that is the same as $(G_{i- 1, \lceil j / s \rceil - 2}, \mathfrak{D}_{i - 1, \lceil j / s \rceil - 2})$\;
run $\dsupdate(G_\beta, \mathfrak{D}_\beta, \updateseq_{[t_{i-1, \lceil j / s \rceil - 2} + 1, t_{i, j}]})$ to update $ \mathfrak{D}_\beta$ to be the same as $\mathfrak{D}_{i, j}$\;
apply $\updateseq_{[t_{i-1, \lceil j / s \rceil - 2} + 1, t_{i, j}]}$ to $G_\beta$\; 
}
}

\Distribute{\textnormal{\textbf{work of the following process at time $j$}}}
{
$\beta'\gets $ the first $\parametertimes$ bits of $b_{\parametertimes, j}$\; 
reverse $(G_{b_{\parametertimes, j}}, \mathfrak{D}_{b_{\parametertimes, j}})$ back to the status that is the same as $(G_{\beta'}, \mathfrak{D}_{\beta'})$\;

run $\dsupdate(G_{b_{\parametertimes, j}}, \mathfrak{D}_{b_{\parametertimes, j}}, \updateseq_{[ t_{i-1, \lceil j / s \rceil - 2} + 1, j]})$
to update $\mathfrak{D}_{b_{\parametertimes, j}}$ to be the same as $\mathfrak{D}_{\parametertimes, j}$\;
apply $\updateseq_{[t_{i-1, \lceil j / s \rceil - 2} + 1, j]}$ to $G_{b_{\parametertimes, j}}$\;
output the access of $ \mathfrak{D}_{b_{\parametertimes, j}}$\;

}
\end{algorithm}

\vspace{.3cm} We use  $\fdpreprocessing$ and $\fdupdate$ to prove Lemma~\ref{lem:fully_framework}. 


\begin{proof}[Proof of Lemma~\ref{lem:fully_framework}]
%
By Claim~\ref{lem:short_distance}  and induction, 
 for $0 \leq i \leq \parametertimes - 1$, 
the pair $(G_{b_{i, j}}, \mathfrak{D}_{b_{i, j}})$ at level $i$ is the same as 
 $(G_{i, j}, \mathfrak{D}_{i, j})$ within the time period 
$[t_{i, j + 1} + 1, t_{i, j + 2} + (d_i / 2)]$.
Hence, the algorithm is correct.

We prove the running time of the preprocessing algorithm and the update algorithm.
Since the running time of $\dsinitialize$ is $\rti$,  the running time of $\fdpreprocessing$ is $O(2^\parametertimes \cdot \rti)$, 
and 
within the time period $[t_{0, j}+ d_0 / 2 + 1, t_{0, j+ 1}]$ for  every $j > 0$,
the total work of line 1-4 of  $\fdupdate$ for every $j> 0$ is 
$O(2^\parametertimes \cdot \rti )$. 
Since \[d_0 = s^{\parametertimes} = \lfloor (w / 2)^{1/\parametertimes}\rfloor^\parametertimes = \Omega(w / 2^{\parametertimes}),\] the amortized  running time for line 1-4 of  $\fdupdate$ is 
\[O\left(\frac{2^\parametertimes \rti }{d_0}\right) = O\left(4^\parametertimes \rti w^{- 1}\right).\]


For any $1 \leq i \leq \parametertimes - 1$, 
the total work of line 5-10 of  $\fdupdate$ for level $i$
within time period $[t_{i, j} + d_i / 2 + 1, t_{i, j + 1}]$ for any 
$j> 0$ is 
$O(2^{\parametertimes }\rtu d_{i - 1})$.
Since the total work is evenly distributed to a time period of length $\Omega(d_{i})$, 
the amortized running time of line 5-10 of  $\fdupdate$ for any $1 \leq i \leq \parametertimes - 1$,  $j> 0$ is 
\[O\left(\frac{2^{\parametertimes }\rtu d_{i - 1} }{ d_i}\right) = O(2^{\parametertimes }\rtu s) = O(2^{\parametertimes }\rtu w^{1 / \parametertimes}).\]
The overall amortized running time of line 5-10 of  $\fdupdate$ is $O( 2^{\parametertimes } \parametertimes\rtu  w^{1/\parametertimes})$. 

The running time of line 11-16 of $\fdupdate$ for every update  is \[O(\rtu d_{\parametertimes - 1}) = O(\rtu s) = O(\rtu w^{1/\parametertimes}).\]

Hence, the overall amortized running time of  $\fdupdate$ per update is \[O(4^\parametertimes \rti w^{- 1} + 2^{\parametertimes } \parametertimes \rtu w^{1/\parametertimes})
= O(4^\parametertimes ( \rti w^{- 1} +   \rtu w^{1/\parametertimes})).\]
%
\end{proof}


\end{document}

%% file: simple.tikzstyles
\tikzstyle{dot}=[fill=black, draw=black, shape=circle, radius=0.1]
\tikzstyle{reddot}=[fill=red, draw=red, shape=circle, radius=0.1]

\tikzstyle{new edge style 0}=[-, draw=red, line width=0.25mm]
\tikzstyle{new edge style 1}=[-, draw=blue, line width=0.25mm]
\tikzstyle{new edge style 2}=[-, draw=black, line width=0.25mm]

%% file: fig1.tikz
\begin{tikzpicture}
	\begin{pgfonlayer}{nodelayer}
		\node [style=none] (0) at (-9, 7) {};
		\node [style=none] (1) at (-10, 6) {};
		\node [style=none] (2) at (-9, 5) {};
		\node [style=none] (3) at (-8, 6) {};
		\node [style=none] (4) at (-6, 7) {};
		\node [style=none] (5) at (-7, 6) {};
		\node [style=none] (6) at (-6, 5) {};
		\node [style=none] (7) at (-5, 6) {};
		\node [style=none] (8) at (-3, 7) {};
		\node [style=none] (9) at (-4, 6) {};
		\node [style=none] (10) at (-3, 5) {};
		\node [style=none] (11) at (-2, 6) {};
		\node [style=none] (12) at (0, 7) {};
		\node [style=none] (13) at (-1, 6) {};
		\node [style=none] (14) at (0, 5) {};
		\node [style=none] (15) at (1, 6) {};
		\node [style=none] (16) at (3, 7) {};
		\node [style=none] (17) at (2, 6) {};
		\node [style=none] (18) at (3, 5) {};
		\node [style=none] (19) at (4, 6) {};
		\node [style=none] (20) at (6, 7) {};
		\node [style=none] (21) at (5, 6) {};
		\node [style=none] (22) at (6, 5) {};
		\node [style=none] (23) at (7, 6) {};
		\node [style=none] (24) at (9, 7) {};
		\node [style=none] (25) at (8, 6) {};
		\node [style=none] (26) at (9, 5) {};
		\node [style=none] (27) at (10, 6) {};
		\node [style=none] (28) at (0, 3.75) {};
		\node [style=none] (29) at (-1, 2.75) {};
		\node [style=none] (30) at (0, 1.75) {};
		\node [style=none] (31) at (1, 2.75) {};
		\filldraw[color=black, fill=black] (-9, 6) circle(0.1);
		\filldraw[color=black, fill=black] (-6, 6) circle(0.1);
		\filldraw[color=black, fill=black] (-3, 6) circle(0.1);
		\filldraw[color=black, fill=black] (-0, 6) circle(0.1);
		\filldraw[color=black, fill=black] (3, 6) circle(0.1);
		\filldraw[color=black, fill=black] (6, 6) circle(0.1);
		\filldraw[color=black, fill=black] (9, 6) circle(0.1);		
		\filldraw[color=black, fill=black] (0, 2.75) circle(0.1);
		\node [style=none] (40) at (-7.5, 6) {};
		\node [style=none] (41) at (-8, 6.25) {};
		\node [style=none] (42) at (-7, 6.25) {};
		\node [style=none] (43) at (-8, 5.75) {};
		\node [style=none] (44) at (-7, 5.75) {};
		\node [style=none] (45) at (-5, 6.25) {};
		\node [style=none] (46) at (-5, 5.75) {};
		\node [style=none] (47) at (-4, 6.25) {};
		\node [style=none] (48) at (-4, 5.75) {};
		\node [style=none] (49) at (-2, 6.25) {};
		\node [style=none] (50) at (-2, 5.75) {};
		\node [style=none] (51) at (-1, 6.25) {};
		\node [style=none] (52) at (-1, 5.75) {};
		\node [style=none] (53) at (7, 6.25) {};
		\node [style=none] (54) at (7, 5.75) {};
		\node [style=none] (55) at (8, 6.25) {};
		\node [style=none] (56) at (8, 5.75) {};
		\node [style=none] (57) at (4, 6.25) {};
		\node [style=none] (58) at (4, 5.75) {};
		\node [style=none] (59) at (5, 6.25) {};
		\node [style=none] (60) at (5, 5.75) {};
		\node [style=none] (61) at (1, 6.25) {};
		\node [style=none] (62) at (1, 5.75) {};
		\node [style=none] (63) at (2, 6.25) {};
		\node [style=none] (64) at (2, 5.75) {};
		\node [style=none] (65) at (-8, 6.25) {};
		\node [style=none] (66) at (-7, 6.25) {};
		\filldraw[color=red, fill=red] (-8.75, 5) circle(0.1);
		\filldraw[color=red, fill=red] (8.75, 5) circle(0.1);
		\filldraw[color=red, fill=red] (-8.5, 5.125) circle(0.1);
		\filldraw[color=red, fill=red] (8.5, 5.125) circle(0.1);
		\filldraw[color=red, fill=red] (-8.25, 5.275) circle(0.1);
		\filldraw[color=red, fill=red] (8.25, 5.275) circle(0.1);
		\node [style=none] (67) at (-1, 3) {};
		\node [style=none] (68) at (1, 3) {};
		\node [style=none] (69) at (-0.88, 3.2) {};
		\node [style=none] (70) at (0.88, 3.2) {};
		\node [style=none] (71) at (-8.75, 5) {};
		\node [style=none] (72) at (-8.5, 5.125) {};
		\node [style=none] (73) at (8.75, 5) {};
		\node [style=none] (74) at (8.5, 5.125) {};
		\node [style=none] (75) at (-8.25, 5.275) {};
		\node [style=none] (76) at (8.25, 5.275) {};
		\node [style=none] (77) at (-0.75, 3.425) {};
		\node [style=none] (78) at (0.75, 3.425) {};
		\node [style=none] (79) at (-11, 8.5) {};
		\node [style=none] (80) at (-11, 4) {};
		\node [style=none] (81) at (11, 4) {};
		\node [style=none] (82) at (11, 8.5) {};
		\node [style=none] (83) at (-5.5, 5.25) {};
		\node [style=none] (84) at (-4.25, 3.5) {$c$};
		\node [style=none] (85) at (4.25, 3.5) {$c$};
		\node [style=none] (86) at (-12, 9) {};
		\node [style=none] (87) at (-12, 0.75) {};
		\node [style=none] (88) at (12, 0.75) {};
		\node [style=none] (89) at (12, 9) {};
		\node [style=none] (90) at (-4.5, 6.75) {$c$};
		\node [style=none] (91) at (-7.5, 6.75) {$c$};
		\node [style=none] (93) at (1.5, 6.75) {$c$};
		\node [style=none] (95) at (4.75, 6.25) {};
		\node [style=none] (96) at (7.5, 6.75) {$c$};
		\node [style=none] (97) at (12, 8.5) {};
		\node [style=none] (98) at (12, 1.25) {};
		\node [style=none] (99) at (11.5, 0.75) {};
		\node [style=none] (100) at (11.5, 9) {};
		\node [style=none] (101) at (-11.5, 9) {};
		\node [style=none] (102) at (-12, 8.5) {};
		\node [style=none] (103) at (-12, 1.25) {};
		\node [style=none] (104) at (-11.5, 0.75) {};
		\node [style=none] (105) at (10.5, 8.5) {};
		\node [style=none] (106) at (10.5, 4) {};
		\node [style=none] (107) at (11, 4.5) {};
		\node [style=none] (108) at (11, 8) {};
		\node [style=none] (109) at (-10.5, 8) {};
		\node [style=none] (110) at (-11, 8) {};
		\node [style=none] (111) at (-10.5, 8.5) {};
		\node [style=none] (112) at (-11, 4.5) {};
		\node [style=none] (113) at (-10.5, 4) {};
		\node [style=none] (115) at (10.5, 4.5) {{\color{blue}$G$}};
		\node [style=none] (116) at (11.25, 1.25) {{$G_0$}};
		\node [style=none] (117) at (-9, 5.5) {$u_0$};
		\node [style=none] (118) at (-6, 5.5) {$u_1$};
		\node [style=none] (119) at (-3, 5.5) {$u_2$};
		\node [style=none] (122) at (6, 5.5) {{$u_{\alpha-3}$}};
		\node [style=none] (123) at (9, 5.5) {{$u_{\alpha-2}$}};
		\node [style=none] (124) at (0, 2.25) {$u_{\alpha-1}$};
		\node [style=none] (126) at (-9, 7.5) {$U_0$};
		\node [style=none] (127) at (-6, 7.5) {$U_1$};
		\node [style=none] (128) at (-3, 7.5) {$U_2$};
		\node [style=none] (129) at (0, 1.25) {$U_{\alpha-1}$};
		\node [style=none] (130) at (6, 7.5) {$U_{\alpha-3}$};
		\node [style=none] (131) at (9, 7.5) {$U_{\alpha -2}$};
	\end{pgfonlayer}
	\begin{pgfonlayer}{edgelayer}
		\draw [bend right=45] (0.center) to (1.center);
		\draw [bend right=45] (1.center) to (2.center);
		\draw [bend right=45] (2.center) to (3.center);
		\draw [bend right=45] (3.center) to (0.center);
		\draw [bend right=45] (4.center) to (5.center);
		\draw [bend right=45] (5.center) to (6.center);
		\draw [bend right=45] (6.center) to (7.center);
		\draw [bend right=45] (7.center) to (4.center);
		\draw [bend right=45] (8.center) to (9.center);
		\draw [bend right=45] (9.center) to (10.center);
		\draw [bend right=45] (10.center) to (11.center);
		\draw [bend right=45] (11.center) to (8.center);
		\draw [bend right=45] (12.center) to (13.center);
		\draw [bend right=45] (13.center) to (14.center);
		\draw [bend right=45] (14.center) to (15.center);
		\draw [bend right=45] (15.center) to (12.center);
		\draw [bend right=45] (16.center) to (17.center);
		\draw [bend right=45] (17.center) to (18.center);
		\draw [bend right=45] (18.center) to (19.center);
		\draw [bend right=45] (19.center) to (16.center);
		\draw [bend right=45] (20.center) to (21.center);
		\draw [bend right=45] (21.center) to (22.center);
		\draw [bend right=45] (22.center) to (23.center);
		\draw [bend right=45] (23.center) to (20.center);
		\draw [bend right=45] (24.center) to (25.center);
		\draw [bend right=45] (25.center) to (26.center);
		\draw [bend right=45] (26.center) to (27.center);
		\draw [bend right=45] (27.center) to (24.center);
		\draw [bend right=45] (28.center) to (29.center);
		\draw [bend right=45] (29.center) to (30.center);
		\draw [bend right=45] (30.center) to (31.center);
		\draw [bend right=45] (31.center) to (28.center);
		\draw [style=new edge style 0] (3.center) to (5.center);
		\draw [style=new edge style 0] (7.center) to (9.center);
		\draw [style=new edge style 0, dotted] (11.center) to (13.center);
		\draw [style=new edge style 0] (15.center) to (17.center);
		\draw [style=new edge style 0, dotted] (19.center) to (21.center);
		\draw [style=new edge style 0] (23.center) to (25.center);
		\draw [style=new edge style 0] (65.center) to (66.center);
		\draw [style=new edge style 0] (43.center) to (44.center);
		\draw [style=new edge style 0] (45.center) to (47.center);
		\draw [style=new edge style 0] (46.center) to (48.center);
		\draw [style=new edge style 0] (61.center) to (63.center);
		\draw [style=new edge style 0] (62.center) to (64.center);
		\draw [style=new edge style 0] (53.center) to (55.center);
		\draw [style=new edge style 0] (54.center) to (56.center);
		\draw [style=new edge style 0] (75) to (77);
		\draw [style=new edge style 0] (72) to (69);
		\draw [style=new edge style 0] (71) to (67);
		\draw [style=new edge style 0] (78) to (76);
		\draw [style=new edge style 0] (70) to (74);
		\draw [style=new edge style 0] (68) to (73);
		\draw [style=new edge style 2] (104.center) to (99.center);
		\draw [style=new edge style 2] (98.center) to (97.center);
		\draw [style=new edge style 2] (102.center) to (103.center);
		\draw [style=new edge style 2] (101.center) to (100.center);
		\draw [style=new edge style 2, bend left=45, looseness=1.25] (102.center) to (101.center);
		\draw [style=new edge style 2, bend left=45, looseness=1.25] (100.center) to (97.center);
		\draw [style=new edge style 2, bend left=45, looseness=1.25] (98.center) to (99.center);
		\draw [style=new edge style 2, bend right=45, looseness=1.25] (103.center) to (104.center);
		\draw [style=new edge style 1] (111.center) to (105.center);
		\draw [style=new edge style 1] (113.center) to (106.center);
		\draw [style=new edge style 1] (108.center) to (107.center);
		\draw [style=new edge style 1] (110.center) to (112.center);
		\draw [style=new edge style 1, bend left=45, looseness=1.25] (110.center) to (111.center);
		\draw [style=new edge style 1, bend right=45, looseness=1.25] (112.center) to (113.center);
		\draw [style=new edge style 1, bend right=45, looseness=1.25] (106.center) to (107.center);
		\draw [style=new edge style 1, bend left=45, looseness=1.25] (105.center) to (108.center);
	\end{pgfonlayer}
\end{tikzpicture}